\documentclass[11pt]{article}
\usepackage{amsmath}
\usepackage{hyperref}
\usepackage{algorithm}
\usepackage{algorithmic}
\usepackage{framed,algorithm,algorithmic}
\usepackage{amsfonts,amsmath,bm,color}

\long\def\symbolfootnote[#1]#2{\begingroup%
\def\thefootnote{\fnsymbol{footnote}}\footnote[#1]{#2}\endgroup}

\newcommand{\R}{\ensuremath\mathbb{R}}
\newcommand{\mb}{\mathbf}

\newcommand{\Expect}[1]{\mbox{}{\mathbb{E}}\left[#1\right]}

\newcommand{\Trace }[1]{\mbox{}{\mathrm{Tr}}\left(#1\right)}

\newcommand{\FNorm }[1]{\mbox{}\|#1\|_\mathrm{F}  }
\newcommand{\FNormS}[1]{\mbox{}\|#1\|_\mathrm{F}^2}

\newcommand{\TNorm }[1]{\mbox{}\|#1\|_2  }
\newcommand{\TNormS}[1]{\mbox{}\|#1\|_2^2}

\newcommand{\pinv}[1]{ {#1}^\dagger}

\newtheorem{theorem}{Theorem}
\newtheorem{lemma}[theorem]{Lemma}
\newtheorem{definition}[theorem]{Definition}

\newtheorem{corollary}[theorem]{Corollary}
\newtheorem{remark}[theorem]{Remark}

\newenvironment{proof}{\begin{trivlist} \item {\bf Proof:~~}}
  {\qed\end{trivlist}}

\newcommand{\transp}{^{\textsc{T}}}
\newcommand{\trace}{\text{\rm Tr}}
\newcommand{\mat}[1]{{\ensuremath{\bm{\mathrm{#1}}}}}

\renewcommand{\vec}[1]{\ensuremath{\bm{#1}}}

\newenvironment{proofof}[1]{\begin{trivlist} \item {\bf Proof
#1:~~}}
  {\qed\end{trivlist}}

\def\rank{\hbox{\rm rank}}

\def\b{{\mathbf b}}
\def\e{{\mathbf e}}

\def\q{{\mathbf q}}
\def\p{{\mathbf p}}

\def\s{{\mathbf s}}
\def\u{{\mathbf u}}
\def\ve{{\mathbf v}}
\def\d{{\mathbf \delta}}

\def\matA{\mat{A}}
\def\matB{\mat{B}}
\def\matC{\mat{C}}
\def\matD{\mat{D}}
\def\matE{\mat{E}}
\def\matF{\mat{F}}
\def\matG{\mat{G}}
\def\matH{\mat{H}}
\def\matI{\mat{I}}
\def\matJ{\mat{J}}
\def\matK{\mat{K}}
\def\matL{\mat{L}}

\def\matM{\mat{M}}
\def\matP{\mat{P}}
\def\matQ{\mat{Q}}
\def\matR{\mat{R}}
\def\matS{\mat{S}}
\def\matT{\mat{T}}
\def\matU{\mat{U}}
\def\matV{\mat{V}}
\def\matW{\mat{W}}
\def\matX{\mat{X}}
\def\matY{\mat{Y}}
\def\matZ{\mat{Z}}
\def\matOmega{\mat{\Omega}}
\def\matSig{\mat{\Sigma}}

\def\matOmega{\mat{\Omega}}
\def\matSig{\mat{\Sigma}}

\def\w{{\mathbf{w}}}

\newcommand\remove[1]{}

\def\nnz{{ \rm nnz }}

\def\math#1{$#1$}

\def\frac#1#2{{#1\over #2}}

\def\eqan#1{\begin{eqnarray*}
#1
\end{eqnarray*}}

\def\qed{\hfill\rule{2mm}{2mm}}

\def\cl#1{{\cal #1}}

\def\argmin{\mathop{\hbox{argmin}}\limits}

\def\g{{\mathbf g}}
\def\h{{\mathbf h}}

\def\x{{\mathbf x}}
\def\y{{\mathbf y}}
\def\z{{\mathbf z}}

\def\a{{\mathbf a}}
\def\b{{\mathbf b}}

\def\norm#1{{\left\|#1\right\|}}

\newcommand{\normF}[1]{{\| #1 \|}_F}
\newcommand{\eps}{\varepsilon}
\newcommand{\poly}{{\mathrm{poly}}}
\newcommand{\polylog}{{\mathrm{polylog}}}
\newcommand{\Exp}{\mathop{\mathrm E}\displaylimits}

\newtheorem{property}{Property}
\newtheorem{fact}{Fact}

\DeclareMathOperator{\tr}{\mathtt{tr}} 

\usepackage{amsmath, amssymb, algorithmic, framed, algorithm}

\title{Sketching as a Tool for Numerical Linear Algebra
\footnote{Version appearing as a monograph in NOW Publishers
``Foundations and Trends in Theoretical Computer Science'' series,
Vol 10, Issue 1--2, 2014, pp 1--157}}

\author{David P. Woodruff\\
IBM Research Almaden\\
dpwoodru@us.ibm.com
}

\begin{document}




\maketitle


\begin{abstract}
This survey highlights the recent advances in algorithms for
numerical linear algebra that have come from the technique of
linear sketching, whereby given a matrix, one first compresses it to
a much smaller matrix by multiplying it by a (usually) random 
matrix with certain
properties. Much of the expensive computation can then be performed on
the smaller matrix, thereby accelerating the solution for the original
problem. 
In this survey we consider least squares as well as robust regression problems,
low rank approximation, and graph sparsification. We also discuss a number
of variants of these problems. Finally, we discuss the limitations of 
sketching methods.
\end{abstract}
%
%
\newpage
\tableofcontents
\newpage
\section{Introduction}
To give the reader a flavor of results in this survey, let us
first consider the classical linear regression problem.
In a special case of this problem one attempts to ``fit'' a line through a set
of given points as best as possible. 

For example, the familiar
Ohm's law states that the voltage $V$ is equal to the resistance
$R$ times the electrical current $I$, or $V = R \cdot I$. Suppose one is
given a set of $n$ example volate-current pairs $(v_j, i_j)$
but does not know the underlying resistance. In this case one is
attempting to find the unknown slope of a line through the origin
which best fits these examples, where best fits can take on a
variety of different meanings. 

More formally, in the standard setting there is one {\it measured variable}
$b$, in the above example this would be the voltage, and a set of
$d$ {\it predictor variables} $a_1, \ldots, a_d$. In the above example $d = 1$
and the single predictor variable is the electrical current. 
Further, it is assumed
that the variables are linearly related up to a noise variable, that is
$b = x_0 + a_1 x_1 + \cdots + a_d x_d + \gamma$, where $x_0, x_1, \ldots, x_d$
are the coefficients of a hyperplane we are trying to learn 
(which does not go through the origin
if $x_0 \neq 0$), and $\gamma$ is a random variable which may be adversarially
chosen, or may come from a distribution which we may have limited or no 
information about. The $x_i$ are also known as the {\it model parameters}. By
introducing an additional predictor variable $a_0$ which is fixed to $1$, 
we can in fact assume that the unknown hyperplane goes through the origin,
that is, it is an unknown subspace of codimension $1$. We will thus assume
that $b = a_1 x_1 + \cdots + a_d x_d + \gamma$ and ignore the affine component
throughout.

In an experiment one is often given $n$ observations, or $n$
$(d+1)$-tuples $(a_{i,1}, \ldots, a_{i,d}, b_i)$, for $i = 1, 2, \ldots, n$. It
is more convenient now to think of the problem in matrix form, where
one is given an $n \times d$ matrix $\matA$ whose rows are the values of the
predictor variables in the $d$ examples, together with an $n \times 1$
column vector $\b$ whose entries are the corresponding observations, and
the goal is to output the coefficient vector $\x$ so that $\matA \x$ and $\b$
are close in whatever the desired sense of closeness may mean. Notice
that as one ranges over all $\x \in \mathbb{R}^d$, $\matA \x$ ranges over
all linear combinations of the $d$ columns of $\matA$, and therefore defines
a $d$-dimensional subspace of $\mathbb{R}^n$, which we refer to as the column
space of $\matA$. Therefore the regression
problem is equivalent to finding the vector $\x$ for which $\matA \x$ is the closest
point in the column space of $\matA$ to the observation vector $\b$. 

Much of the focus of this survey will be on the over-constrained case, in which
the number $n$ of examples is much larger than the number $d$ of predictor
variables. Note that in this case there are more constraints than unknowns,
and there need not exist a solution $\x$ to the equation $\matA \x = \b$. 

Regarding the measure of fit, or closeness of $\matA \x$ to $\b$, one of the most
common is the least squares method, which seeks to find the closest point
in Euclidean distance, i.e., 
$$\textrm{argmin}_{\x} \|\matA \x-\b\|_2 = \sum_{i=1}^n (b_i - \langle \matA_{i,*}, \x \rangle )^2,$$
 where $\matA_{i,*}$ denotes the $i$-th row of $\matA$, and $b_i$ the $i$-th
entry of the vector $\b$. This error measure
has a clear geometric interpretation, 
as the optimal $\x$ satisfies that $\matA \x$ is the standard Euclidean
projection of $\b$ onto the column space of $\matA$. Because of this, it is possible
to write the solution for this problem in a closed form. 
That is,
necessarily one has $\matA^T \matA \x^* = \matA^T \b$ for the optimal solution $\x^*$ by 
considering the gradient at a point $\x$, and observing that in order for it
to be $0$, that is for $\x$ to be a minimum, the above equation has to hold. The
equation $\matA^T \matA \x^* = \matA^T \b$ is known as the {\it normal equation}, which
captures that the line connecting $\matA \x^*$ to $\b$ should be perpendicular to the
columns spanned by $\matA$. If
the columns of $\matA$ are linearly independent, $\matA^T \matA$ is a full rank $d \times d$
matrix and the solution is therefore given by $\x^* = (\matA^T \matA)^{-1} \matA^T \b$. Otherwise,
there are multiple solutions and a solution $\x^*$ of minimum Euclidean norm
is given by $\x^* = \matA^{\dagger} \b$, where $\matA^{\dagger}$ is the Moore-Penrose pseudoinverse
of $\matA$. Recall that if $\matA = \matU \mat\Sigma \matV^T$ is the singular value decomposition (SVD)
of $\matA$, where $\matU$ is $n \times d$ with orthonormal columns, $\mat\Sigma$ is a diagonal
$d \times d$ matrix with non-negative non-increasing diagonal entries, and $\matV^T$
is a $d \times d$ matrix with orthonormal rows, then the Moore-Penrose pseudoinverse of $\matA$
is the $d \times n$ matrix $\matV \mat\Sigma^{\dagger} \matU^T$, where $\mat\Sigma^{\dagger}$ is a $d \times
d$ diagonal matrix with $\mat\Sigma^{\dagger}_{i,i} = 1/\mat\Sigma_{i,i}$ if $\mat\Sigma_{i,i} > 0$,
and is $0$ otherwise. 

The least squares measure of closeness, although popular, is somewhat arbitrary and 
there may be better choices depending on the application at hand. Another popular
choice is the method of least absolute deviation, or $\ell_1$-regression. Here
the goal is to instead find $\x^*$ so as to minimize 
$$\|\matA \x- \b\|_1 = \sum_{i=1}^n |\b_i - \langle \matA_{i,*}, \x \rangle |.$$ 
This measure is known 
to be less sensitive to outliers
than the least squares measure. The reason for this is that one squares
the value $\b_i - \langle \matA_{i,*}, \x \rangle$ in the least squares cost function, while one
only takes its absolute value in the least absolute deviation cost function. Thus,
if $\b_i$ is significantly larger (or smaller) than $\langle \matA_{i,*}, \x \rangle$ for the 
$i$-th observation, due, e.g., to large measurement noise on that observation, this
requires the sought hyperplane $\x$ to be closer to the $i$-th observation 
when using the least squares cost function than when using the least absolute deviation
cost function. While there is no closed-form solution for least absolute deviation
regression, one can solve the problem up to machine precision in polynomial time
by casting it as a linear programming problem and using a generic linear programming
algorithm.

The problem with the above solutions is that on massive data sets, they are
often too slow to be of practical value. Using n\"aive matrix multiplication,
solving the normal equations for least squares would take at least $n \cdot d^2$
time. For least absolute deviation regression, when casting the problem as a linear
program one needs to introduce $O(n)$ variables (these are needed to enforce the absolute
value constraints) and $O(n)$ constraints, and generic solvers would take 
$\poly(n)$ time for an polynomial in $n$ which is at least cubic. 
While these solutions are polynomial time, they are prohibitive for 
large values of $n$. 

The starting point of this survey is a beautiful work by Tam\'{a}s Sarl\'os \cite{S06} which
observed that one could use {\it sketching techniques} to improve upon the
above time complexities, if one is willing to settle for a randomized approximation
algorithm. Here, one relaxes the problem to finding a vector $\x$ so that
$\|\matA \x- \b\|_p \leq (1+\eps)\|\matA \x^*- \b\|_p$, where $\x^*$ is the optimal hyperplane, 
with respect to the $p$-norm, for $p$ either $1$ or $2$ as in the discussion above.
Moreover, one allows the algorithm to fail with some small probability $\delta$, 
which can be amplified by independent repetition and taking the best hyperplane found.

While sketching techniques will be described in great detail in the following sections,
we give a glimpse of what is to come below. Let $r \ll n$, 
and suppose one chooses a
$r \times n$ random matrix $\matS$ from a certain distribution on matrices
to be specified. 
Consider the following algorithm for least squares regression: 
\begin{enumerate}
\item Sample a random matrix $\matS$.
\item Compute $\matS \cdot \matA$ and $\matS \cdot \b$.
\item Output the exact solution $x$ to the regression problem $\min_{\x} \|(\matS \matA)\x-(\mat S\b)\|_2$. 
\end{enumerate}
Let us highlight some key features of this algorithm. First, notice that it is a 
{\it black box} reduction, in the sense that after computing $\matS \cdot \matA$
and $\matS \cdot \b$, we then solve a smaller instance of least squares regression,
replacing the original number $n$ of observations with the smaller value of $r$. For
$r$ sufficiently small, we can then afford to carry out step 3, e.g., by computing
and solving the normal equations as described above. 

The most glaring omission from the above algorithm is which random familes of matrices
$\matS$ will make this procedure work, and for what values of $r$. Perhaps one of the
simplest arguments is the following. Suppose $r = \Theta(d/\eps^2)$ and $\matS$
is a $r \times n$ matrix of i.i.d. normal random variables with mean 
zero and variance $1/r$, denoted $N(0,1/r)$. Let $\matU$ be an $n \times (d+1)$
matrix with orthonormal columns for which the column space of $\matU$ is equal to
the column space of $[\matA, \b]$, that is, the space spanned by the columns of $\matA$
together with the vector $\b$. 

Consider the product $\matS \cdot \matU$. By $2$-stability of the normal
distribution, i.e., if $\matA \sim N(0, \sigma_1^2)$ and $\matB \sim N(0, \sigma_2^2)$,
then $\matA+\matB \sim N(0, \sigma_1^2 + \sigma_2^2)$, each of the entries of $\matS \cdot \matU$
is distributed as $N(0, 1/r)$ (recall that the column norms of $\matU$ are equal to $1$). 
The entries in different rows of $\matS \cdot \matU$ are also independent
since the rows of $\matS$ are independent. The entries in a row
are also independent by rotational invarance of the normal distribution, that is,
if $\g \sim N(0, \matI_n/r)$ is an $n$-dimensional vector of normal random variables
and $\matU_{*,1}, \ldots, \matU_{*,d}$ are orthogonal vectors, then $\langle \g, \matU_{*,1} \rangle,
\langle \g, \matU_{*,2} \rangle, \ldots, \langle \g, \matU_{*,d+1} \rangle$ are independent. 
Here $\matI_n$ is the $n \times n$ identity matrix
(to see this, by rotational invariance, these $d+1$ random variables are equal in
distribution to $\langle \g, \e_1 \rangle, \langle \g, \e_2 \rangle, \ldots,
\langle \g, \e_{d+1} \rangle$, where $\e_1, \ldots, \e_{d+1}$ are the standard unit vectors,
from which independence follows since the coordinates of $g$ are independent). 

It follows that $\matS \cdot \matU$ is an $r \times (d+1)$ matrix of i.i.d. $N(0, 1/r)$
random variables. For $r = \Theta(d/\eps^2)$, it is well-known that with
probability $1-\exp(-d)$, all the singular
values of $\matS \cdot \matU$ lie in the interval 
$[1-\eps, 1+\eps]$. This can be shown by arguing that for any fixed vector $\x$, 
$\|\matS \cdot \matU \x\|_2^2 = (1 \pm \eps)\|\x\|_2^2$ with probability $1-\exp(-d)$, 
since, by rotational invariance of the normal distribution, 
$\matS \cdot \matU \x$ is a vector of $r$
i.i.d. $N(0, \|x\|_2^2)$ random variables, and so one can apply a tail bound for
$\|\matS \cdot \matU \x\|_2^2$, which itself is a $\chi^2$-random variable with $r$ degrees
of freedom. The fact that all singular values of $\matS \cdot \matU$ lie in 
$[1-\eps, 1+\eps]$ then follows by placing a sufficiently fine net on the unit sphere
and applying a union bound to all net points; 
see, e.g., Theorem 2.1 of \cite{RV10} for further details. 

Hence, for all vectors
$\y$, $\|\matS \matU \y\|_2 = (1 \pm \eps)\|\matU \y\|_2$. But now consider the regression
problem $\min_{\x} \|(\matS \matA)\x-(\matS \b)\|_2 = \min_{\x} \|\matS(\matA \x- \b)\|_2$. For each vector $x$,
$\matA \x- \b$ is in the column space of $\matU$, and therefore by the previous paragraph,
$\|\matS(\matA \x-\b)\|_2 = (1 \pm \eps)\|\matA \x-\b\|_2$. It follows that by solving the regression
problem $\min_x \|(\matS \matA)\x-(\matS \b)\|_2$, we obtain a $(1+\eps)$-approximation to the original
regression problem with probability $1-\exp(-d)$. 

The above technique of replacing $\matA$ by $\matS \cdot \matA$ is known as a sketching technique
and $\matS \cdot \matA$ is referred to as a (linear) sketch of $\matA$. While the above
is perhaps the simplest instantiation of sketching, notice that it does not in
fact give us a faster solution to the least squares regression problem. This is because,
while solving the regression problem $\min_{\x} \|(\matS \matA)\x-(\matS \b)\|_2$ can now be done n\"aively
in only $O(rd^2)$ time, which no longer depends on the large dimension $n$, the
problem is that $\matS$ is a dense matrix and computing $\matS \cdot \matA$ may now be too slow,
taking $\Theta(nrd)$ time. 

Thus, the bottleneck in the above algorithm is the time for matrix-matrix multiplication. 
Tam\'{a}s Sarl\'{o}s observed \cite{S06} that one can in fact
choose $\matS$ to come from a much more structured random family of matrices, called
fast Johnson-Lindenstrauss transforms \cite{AC06}. These led to roughly $O(nd \log d) + \poly(d/\eps)$
time algorithms for the least squares regression problem. Recently, Clarkson
and Woodruff \cite{CW13} improved upon the time complexity of this algorithm to obtain 
{\it optimal} algorithms for approximate least squares
regression, obtaining $O(\nnz(\matA)) + \poly(d/\eps)$ time, where $\nnz(\matA)$ denotes
the number of non-zero entries of the matrix $\matA$. We call such algorithms input-sparsity
algorithms, as they exploit the number of non-zero entries of $\matA$. 
The $\poly(d/\eps)$ factors
were subsequently optimized in a number of papers \cite{MM13,JH13,bn13}, leading to optimal algorithms
even when $\nnz(\matA)$ is not too much larger than $d$. 

In parallel, work was done on reducing the dependence on $\eps$ in these algorithms from polynomial to polylogarithmic. This started with work of Rokhlin and Tygert \cite{RT08} (see also the Blendenpik algorithm \cite{amt10}), and combined with the recent input sparsity algorithms give a running time of $O(\nnz(\matA)\log(1/\eps)) + \poly(d)$ for least squares regression \cite{CW13}. This is significant for high precision applications of least squares regression, for example, for solving an equation of the form $\matA^T \matA \x = \matA^T \b$. Such equations frequently arise in interior point methods for linear programming, as well as iteratively reweighted least squares regression, which is a subroutine for many important problems, such as logistic regression; see
\cite{M03Compare} for a survey of such techniques for logistic regression. In these examples $\matA$ is often formed from the Hessian of a Newton step in an iteration. It is clear that such an equation is just a regression problem in disguise (in the form of the normal equations), and the (exact) solution of argmin$_{\x} \|\matA \x - \b\|_2$ provides such a solution. By using high precision approximate regression one can speed up the iterations in such algorithms. 

Besides least squares regression, related sketching techniques have also been 
instrumental in providing better robust $\ell_1$-regression, low rank approximation, 
and graph sparsifiers, as well as a number of variants of these problems. 
We will cover these applications each in more detail. 
\\\\
{\bf Roadmap:} In the next section we will
discuss least squares regression in full detail, which includes applications to
constrained and structured regression. In Section 3, we will then 
discuss $\ell_p$-regression, including least absolute deviation regression. In Section 4
we will dicuss low rank approximation, while in Section 5, we will discuss graph sparsification. 
In Section 6, we will discuss the limitations of sketching techniques. In Section 7, 
we will conclude and briefly discuss a number of other directions in this area.

\section{Subspace Embeddings and Least Squares Regression}
We start with the classical least squares regression problem, which is the following. We are
given an $n \times d$ matrix $\matA$, which is typically overconstrained, that is, $n \gg d$, together
with an $n \times 1$ vector $\b$, and we would like to find an $\x \in \mathbb{R}^d$ 
which minimizes $\|\matA \x- \b\|_2$. Since the
problem is overconstrained, there need not exist a vector $\x$ for which $\matA \x= \b$. We
relax the problem and instead allow for outputting a vector $\x'$ for which with probability $.99$, 
$$\|\matA\x'-\b\|_2 \leq (1+\eps) \|\matA \x-\b\|_2.$$
We are interested in fast solutions to this problem, which we present in \S\ref{sec:regression}. 

{\bf Section Overview:} In \S\ref{sec:se} we introduce the notion of an $\ell_2$-subspace embedding, which is crucial to many of the applications in this book. In this section we will focus on its application to least squares regression. We show several different randomized constructions which vary in the time it takes to construct and apply them, as well as the dimension which they embed into. These constructions turn out to be oblivious to the data set. In \S\ref{sec:mm} we introduce a primitive called matrix product, which is a primitive for performing approximate matrix multiplication. Using this primitive we will show how to construct an $\ell_2$-subspace embedding. The primitive will also play a role in \S\ref{sec:regression} in solving regression with a linear in $1/\eps$ dependence on the accuracy parameter $\eps$, as well as in \S\ref{sec:frobenius} on low rank matrix approximation. In \S\ref{sec:whp} we present a trick which takes any constant probability of success subspace embedding and shows how to obtain a high probability success subspace embedding. Thus, to some extent our earlier treatment of constant subspace embeddings is justified. In \S\ref{sec:leverage} we present a completely different way of achieving a subspace embedding, which is non-oblivious and is obtained by sampling rows of a matrix proportional to their so-called leverage scores. In \S\ref{sec:regression} we present a black box application of subspace embeddings to the least squares regression problem.  In \S\ref{sec:machinePrecision} we show how to use subspace embeddings in a different way to solve least squares regression, leading to an algorithm with only a logarithmic dependence on the error parameter $1/\eps$. This method, while it has a much better dependence on $1/\eps$, does require multiple passes over the data unlike the method of \S\ref{sec:regression}. Finally, in \S\ref{sec:polynomialFit} we show that for regression instances which possess additional structure, such as those that arise in polynomial fitting problems, one can apply subspace embeddings even faster than via generic ways presented before. 

\subsection{Subspace embeddings}\label{sec:se}
We start with the basic notion of an $\ell_2$-subspace embedding for the column space of an $n \times d$ matrix $\matA$. 
As we will see, this will be a powerful hammer for solving least squares regression. 
Throughout, for non-negative real numbers $a$ and $b$, we use the notation $a = (1 \pm \eps)b$ if
$a \in [(1-\eps)b, (1+\eps)b]$. 
\begin{definition}\label{def:subspace}
A $(1 \pm \eps)$ $\ell_2$-subspace embedding for the column space of an 
$n \times d$ matrix $\matA$ 
is a matrix $\matS$ for which for all $\x \in \mathbb{R}^d$
$$\|\matS \matA \x\|_2^2 = (1 \pm \eps)\|\matA \x\|_2^2.$$
\end{definition}
We will often abuse notation and say that $\matS$ is an $\ell_2$-subspace embedding for $\matA$ itself, even though it should be understood from the definition that this property does not depend on a particular basis for the representation of the column space of $\matA$. 

Notice that if $\matS$ is a $(1 \pm \eps)$ $\ell_2$-subspace embedding for $\matA$, then it is also a
$(1 \pm \eps)$ $\ell_2$-subspace embedding for $\matU$, where $\matU$ is an orthonormal basis for the column
space of $\matA$. This is because the sets 
$\{\matA \x \mid \x \in \mathbb{R}^d\}$ and $\{\matU \y \mid \y \in \mathbb{R}^t\}$
are equal, where $t$ is the rank of $\matA$. 
Hence, we could without loss of generality assume that $\matA$ has orthonormal columns. With 
this interpretation, the requirement of Definition \ref{def:subspace} becomes
$$\|\matS \matU \y\|_2^2 = (1 \pm \eps) \|\matU \y\|_2^2 = (1 \pm \eps)\|\y\|_2^2,$$
where the final equality holds since $\matU$ has orthonormal columns. If this requirement is satisfied
for unit vectors $\y$, then it is satisfied for all vectors $\y$ by scaling (since $\matS$ is a linear
map), so the requirement
of Definition \ref{def:subspace} can be further simplified to
\begin{eqnarray}\label{eqn:operatorDiff}
\|\matI_{d} - \matU^T \matS^T \matS \matU\|_2 \leq \eps,
\end{eqnarray}
that is, the operator norm $\sup_{y \textrm{ such that }\|y\|_2 = 1} \|\matI_d - \matU^T\matS^T\matS \matU\|_2$, 
should be at most $\eps$. Here, $\matI_{d}$ is the $d \times d$ identity matrix. 

There are various goals of subspace embeddings. Two of the main goals are finding a matrix $\matS$ with a small
number of rows. Another goal is to be able to compute $\matS \cdot \matA$ quickly, as this is often a bottleneck
in applications. 

There are a number of ways of constructing $\ell_2$-subspace embeddings which achieve various tradeoffs. 
One particularly useful form of an $\ell_2$-subspace embedding is an {\it oblivious} $\ell_2$-subspace 
embedding.

\begin{definition}\label{def:oblSubspace}
Suppose $\Pi$ is a distribution on $r \times n$ matrices $\matS$, where $r$
is a function of $n, d, \eps$, and $\delta$. 
Suppose that with probability at least $1-\delta$, for any fixed $n \times d$ matrix $\matA$,
a matrix $\matS$ drawn from 
distribution $\Pi$ has the property
that $\matS$ is a 
$(1\pm \eps)$ $\ell_2$-subspace embedding for $\matA$. Then
we call $\mat\Pi$ an $(\eps, \delta)$ oblivious $\ell_2$-subspace embedding. 
\end{definition}

Definition \ref{def:oblSubspace} 
will be used in applications throughout this book, and sometimes for convenience
we will drop the word oblivious. 

We do want to note that there are other ways of constructing subspace embeddings
though, such as through sampling the rows of $\matA$ via a certain distribution and reweighting them. 
This is called Leverage Score Sampling \cite{DMM06a,DMM061,DMM062,DMMS07},
which will be discussed later in the section. This also turns out to have a number of applications, for example
to CUR decompositions of a matrix discussed in \S\ref{sec:CUR}. Note that this way of constructing subspace
embeddings is desirable in that it gives an actual ``representative'' subset of rows of $\matA$ which form
a subspace embedding - this is often called a {\it coreset}. Such representations can sometimes lead to better
data interpretability, as well as preserving sparsity. While we do discuss this kind of sampling to some
extent, our main focus will be on sketching. The reader is encouraged to look at the survey by Mahoney
for more details on sampling-based approaches \cite{m11}. See also \cite{lmp13} and \cite{clmmps14} for 
state of the art subspace embeddings based on this approach. 

Returning to Definition \ref{def:oblSubspace}, the first usage of this in the numerical linear 
algebra community,
to the best of our knowledge, was done by S\'arlos, who proposed using Fast Johnson Lindenstrauss transforms
to provide subspace embeddings. We follow the exposition in Sarl\'{o}s for this \cite{S06}.

\begin{definition}
A random matrix $\matS \in \mathbb{R}^{k \times n}$ forms a Johnson-Lindenstrauss transform with parameters 
$\eps, \delta, f$, or JLT$(\eps, \delta, f)$ for short, if with probability at least $1-\delta$,
for any $f$-element subset $V \subset \mathbb{R}^n$, for all $\ve, \ve' \in V$ it holds that
$|\langle \matS \ve, \matS \ve' \rangle - \langle \ve, \ve' \rangle| \leq \eps \|\ve\|_2 \|\ve'\|_2$. 

Note when $\ve = \ve'$ 
we obtain the usual statement that $\|\matS \ve\|_2^2 = (1 \pm \eps)\|\ve\|_2^2$. It turns out that if we scale all $\ve, \ve' \in V$ so that they are unit vectors, we could alternatively require 
$\|\matS \ve\|_2^2 = (1 \pm \eps)\|\ve\|_2^2$ and $\|\matS (\ve + \ve')\|_2^2 = (1 \pm \eps) \|\ve + \ve'\|_2^2$ for all $\ve, \ve' \in V$. That is, the requirement of the definition could be based on norms rather than inner products. To see that this
impies the statement above, we have 
\begin{eqnarray*}
\langle \matS \ve, \matS \ve' \rangle
& = & (\|\matS(\ve + \ve')\|_2^2 - \|\matS \ve\|_2^2 - \|\matS \ve'\|_2^2)/2\\
& = & ((1\pm \eps)\|\ve + \ve'\|_2^2 - (1 \pm \eps)\|\ve\|_2^2 - (1 \pm \eps)\|\ve'\|_2^2)\\
& = & \langle \ve, \ve' \rangle \pm O(\eps), 
\end{eqnarray*}
which implies all inner products are preserved up to $\eps$ by rescaling $\eps$ by a constant. 
\end{definition}

There are many constructions of Johnson-Lindenstrauss transforms, possibly the simplest is given by
the following theorem. 

\begin{theorem}\label{thm:normals} (see e.g., \cite{IM98})
Let $0 < \eps, \delta < 1$ and 
$\matS = \frac{1}{\sqrt{k}} \matR \in \mathbb{R}^{k \times n}$ where the entries
$\matR_{i,j}$ of $\matR$ are independent standard normal random variables. 
Then if $k = \Omega(\eps^{-2} \log (f/\delta))$,
then $\matS$ is a JLT($\eps, \delta, f)$. 
\end{theorem}
We will see a proof of Theorem \ref{thm:normals} in Lemma \ref{lem:jl}. 

We show how Theorem \ref{thm:normals} can be used to provide an $\ell_2$-subspace embedding. To do so, we need
the concept of an $\eps$-net. Let 
$\mathcal{S} = \{\y \in \mathbb{R}^n \mid \y = \matA \x \textrm{ for some } \x \in \mathbb{R}^d \textrm{ and } \|\y\|_2 = 1\}$.
We seek a finite subset of
$\mathcal{S}$, denoted $\mathcal{N}$ so that if 
\begin{eqnarray}\label{eqn:netProperty}
\langle \matS \w, \matS \w' \rangle = \langle \w, \w' \rangle \pm \eps \textrm{ for all } \w, \w' \in \mathcal{N},
\end{eqnarray}
then $\|\matS \y\|_2 = (1 \pm \eps)\|\y\|_2$ for all $\y \in \mathcal{S}$. 

By an argument of \cite{AHK06,fo05,m07}, it suffices to choose $\mathcal{N}$ 
so that for all $\y \in \mathcal{S}$, 
there exists a vector $\w \in \mathcal{N}$ for which $\|\y-\w\|_2 \leq 1/2$. We will refer to $\mathcal{N}$
as a $(1/2)$-net for $\mathcal{S}$. 

To see that $\mathcal{N}$ suffices, if $\y$ is a unit
vector, then we can write 
\begin{eqnarray}\label{eqn:sequence}
\y = \y^0 + \y^1 + \y^2 + \cdots,
\end{eqnarray}
where $\|\y^i\| \leq \frac{1}{2^i}$ and $\y^i$ is a scalar multiple of a vector in $\mathcal{N}$. This is
because we can write $\y = \y^0 + (\y-\y^0)$ where $\y^0 \in \mathcal{N}$ and 
$\|\y-\y^0\| \leq 1/2$ by the definition of $\mathcal{N}$. 
Then, $\y-\y^0 = \y^1 + ((\y-\y^0)-\y^1)$ where $\y^1 \in \mathcal{N}$ and 
$$\|\y-\y^0 - \y^1\|_2 \leq \frac{\|\y-\y^0\|_2}{2} \leq \frac{1}{4}.$$
The expansion in (\ref{eqn:sequence}) then follows by induction. But then, 
\begin{eqnarray*}
\|\matS \y\|_2^2 & = & \|\matS(\y^0 + \y^1 + \y^2 + \cdots)\|_2^2\\
& = & \sum_{0 \leq i < j < \infty} \|\matS \y^i\|_2^2 + 2 \langle \matS \y^i, \matS \y^j \rangle\\
& = & \left (\sum_{0 \leq i < j < \infty} \|\y^i\|_2^2 + 2 \langle \y^i, \y^j \rangle \right ) 
\pm 2\eps \left (\sum_{0 \leq i \leq j < \infty} \|\y^i\|_2 \|\y^j\|_2 \right )\\
& = & 1 \pm O(\eps),
\end{eqnarray*}
where the first equality follows by (\ref{eqn:sequence}), the second equality follows by expanding the
square, the third equality follows from (\ref{eqn:netProperty}), and the fourth equality is what we want
(after rescaling $\eps$ by a constant factor). 

We show the existence of a small $(1/2)$-net $\mathcal{N}$ via a standard argument.  
\begin{lemma}\label{lem:epsL2net}
For any $0 < \gamma < 1$, 
there exists a $\gamma$-net $\mathcal{N}$ of $\mathcal{S}$ for which $|\mathcal{N}| \leq (1 + 4/\gamma)^d$. 
\end{lemma}
\begin{proof}
For $t = \textrm{rank}(\matA) \leq d$, we can equivalently express $\mathcal{S}$ as 
$$\mathcal{S} = \{\y \in \mathbb{R}^n \mid \y = \matU \x \textrm{ for some } \x \in \mathbb{R}^t \textrm{ and } \|\y\|_2 = 1\},$$
where $\matU$ has orthonormal columns and the same column space as $\matA$. 

We choose a $\gamma/2$-net $\mathcal{N}'$ of the unit sphere $\mathcal{S}^{t-1}$, where the $\gamma/2$ net has
size $(1+4/\gamma)^t$. The intuition for this choice is that $\matU$ provides an isometry when operating on $\mathcal{S}^{t-1}$,
and so a net for $\mathcal{S}^{t-1}$ will give us a net for the image of $\mathcal{S}^{t-1}$ under $\matU$. 

This can
be done by choosing a maximal set $\mathcal{N}'$ of points on $\mathcal{S}^{t-1}$ so that no two points are within
distance $\gamma/2$ from each other. It follows that the balls of radius $\gamma/4$ centered at these points are
disjoint, but on the other hand they are all contained in the ball of radius $1+\gamma/4$ centered at the origin. 
The volume of the latter ball is a factor $(1+\gamma/4)^t/(\gamma/4)^t$ larger than the smaller balls, which
implies $|\mathcal{N}'| \leq (1+4/\gamma)^t$. See, e.g., \cite{m02} for more details.

%
%
Define 
$\mathcal{N} = \{\y \in \mathbb{R}^n \mid \y = \matU \x \textrm{ for some } \x \in \mathcal{N}'\}.$
Since the columns of $\matU$ are orthonormal, if there were a point $\matU \x \in \mathcal{S}$ for which there
were no point $\y \in \mathcal{N}$ with $\|\y-\matU \x\|_2 \leq \gamma$, 
then $\x$ would be a point in $\mathcal{S}^{k-1}$
for which there is no point $\z \in \mathcal{N}'$ with $\|\x-\z\|_2 \leq \gamma$, a contradiction. 
\end{proof}

It follows by setting $V = \mathcal{N}$ and $f = 9^d$ in Theorem \ref{thm:normals}, 
we can then apply Lemma \ref{lem:epsL2net} and (\ref{eqn:netProperty}) to obtain the following theorem. 
Note that the net size
does not depend on $\eps$, since we just need a $1/2$-net for the argument, even though the theorem holds for
general $\eps$. 
\begin{theorem}\label{thm:gaussianSE}
Let $0 < \eps, \delta < 1$ and $\matS = \frac{1}{\sqrt{k}} \matR \in \mathbb{R}^{k \times n}$ where the entries
$\matR_{i,j}$ of $\matR$ are independent standard normal random variables. Then if $k = \Theta((d+ \log(1/\delta)) \eps^{-2})$,
then for any fixed $n \times d$ matrix $\matA$, with probability $1-\delta$, $\matS$ is a $(1 \pm \eps)$
$\ell_2$-subspace embedding for $\matA$, that is, simultaneously for all $\x \in \mathbb{R}^d$,
$\|\matS \matA \x\|_2 = (1 \pm \eps)\|\matA \x\|_2$. Here $C > 0$ is an absolute constant. 
\end{theorem}
It turns out, as we will see in Section \ref{chap:lb}, that Theorem \ref{thm:gaussianSE} provides the
optimal number of rows of $\matS$ up to a constant factor, namely $\Theta(k \eps^{-2})$. This is true of
any oblivious $(1 \pm \eps)$ $\ell_2$-subspace embedding, even those achieving only a constant probability
of providing an $\ell_2$-subspace embedding of $\matA$ with constant probability. 

After Theorem \ref{thm:normals} was discovered, there were a number of followups. For instance, it was
shown by Achlioptas that one can replace $\matR$ in Theorem \ref{thm:normals} with a matrix of i.i.d. sign random variables \cite{Ach03}, 
that is, each entry is independently set to $1$ or $-1$ with probability $1/2$. Further, Achlioptas 
showed that one can change the distribution so that for the same value of $k$, one can set each entry
in $\matR$ independently to be $1$ with probability $1/6$, $-1$ with probability $1/6$, and $0$ 
with probability $2/3$.
The latter is important since it results in a sparse matrix $\matS$, for which one can then compute $\matS \cdot \x$ for
a vector $\x \in \mathbb{R}^n$ more quickly. A breakthrough was made by Dasgupta, Kumar, and S\'arlos \cite{dks10}
who showed that it suffices for each column of $\matS$ to have only $\eps^{-1} \poly(\log (f/\delta))$ non-zero entries
per column. Note that if the $\poly(\log f/\delta)$ term is much smaller than $\eps^{-1}$, this is a significant
improvement over the $\Omega(\eps^{-2} \log(f/\delta))$ number of non-zero entries per column achieved by previous
schemes. The $\eps^{-1} \poly(\log(f/\delta))$ sparsity was later optimized by Kane and Nelson \cite{kn14}, who got
$O(\eps^{-1} \log(f/\delta))$ non-zero entries per column. The latter was shown to be almost tight by
Nelson and Nguy$\tilde{\hat{\mbox{e}}}$n \cite{nn13lb}, who showed that $\Omega(\eps^{-1} \log(f/\delta)/\log(1/\varepsilon))$ column sparsity is required. 

In short, the above line of work shows that it is possible to apply a JLT$(\eps, \delta, f)$  matrix $\matS$
to a vector $\x$ in $O(\nnz(\x) \cdot \eps^{-1} \log(f/\delta))$ time, where $\nnz(\x)$ denotes the number
of non-zero entries of the vector $\x$. This results in a significant speedup over Theorem \ref{thm:normals}
when $\eps$ is small. It also leads to improvements in Theorem \ref{thm:gaussianSE}, though regarding
$\ell_2$-subspace embeddings, one can do better as discussed below.

A somewhat different line of work also came about in trying to speed up the basic construction
in Theorem \ref{thm:normals}, and this is due to Ailon and Chazelle \cite{AC06}. Instead of trying to achieve
a sparse matrix $\matS$, they tried to achieve an $\matS$ which could be quickly applied to a vector $\x$. 
The underlying
intuition here is that for a vector $\x \in \mathbb{R}^n$ 
whose $\ell_2$ mass is spread roughly uniformly across its $n$ coordinates, sampling a small number of
its coordinates uniformly at random and rescaling results in a good estimate of the $\ell_2$-norm of $\x$. 
However, if $\x$ does not have this property, e.g., it is sparse, then sampling is a very poor way to estimate
the $\ell_2$-norm of $\x$, as typically most samples will be $0$. By the uncertainty principle, though, 
if $\x$ is sparse, then $\matF \x$ cannot be too sparse, 
where $\mat F$ is the Fourier transform. This is also true for the
Hadamard transform $\matH \x$, and for any bounded orthonormal system (i.e., an orthonormal matrix
whose entry of maximum magnitude is bounded by $O(1/\sqrt{n})$). Indeed, from results in signal processing due
to Donoho and Stark \cite{DS89}, 
if $\matA = [\matI_n \matB]^T$ is a $2n \times n$ matrix such that $\matB$ has orthonormal rows and columns,
and for any distinct rows $\matB_{i*}, \matB_{j*}$ we have $|\matB_{i*}, \matB_{j*}| \leq M$, then for any
$\x \in \mathbb{R}^n$, it holds that $\|\x\|_0 + \|\matB \x\|_0 \geq 1/M$. See, e.g., \cite{i07}, 
for algorithmic applications of this uncertainty principle. 

Unfortunately $\matH \x$ can still be sparse enough that a small number of samples will not work, 
so the intuition
is to re-randomize $\matH \x$ by applying a cheap rotation - namely, computing $\matH \matD \x$ 
for a diagonal matrix $\matD$
with i.i.d. entries $\matD_{i,i}$ in which $\matD_{i,i} = 1$ with probability $1/2$, 
and $\matD_{i,i} = -1$ with probability $1/2$.
If $\matP$ is an $k \times n$ matrix which implements coordinate sampling, 
then $\matP \cdot \matH \cdot \matD \x$ now provides the desired
Johnson-Lindenstrauss transform. Since $\matD$ is a diagonal matrix, 
$\matD \x$ can be computed in $O(n)$ time. The Hadamard
matrix $\matH$ can be applied to an $n$-dimensional vector in $O(n \log n)$ time. Finally, $\matP$ can be applied to
an $n$-dimensional vector in $O(k)$ time. Hence, $\matP \cdot \matH \cdot \matD$ 
can be applied to a vector in $O(n \log n)$
time and to an $n \times d$ matrix in $O(n d \log n)$ time. 
We call this the {\em Fast Johnson Lindenstrauss Transform}. We note that this is not
quite the same as the construction given by Ailon and Chazelle in \cite{AC06}, who form $\matP$
slightly differently to obtain a better dependence on $1/\eps$ in the final dimension. 

The Fast Johnson Lindenstrauss Transform 
is significantly faster than the above $O(\nnz(\x) \cdot \eps^{-1} \log(f/\delta))$ time for many
reasonable settings of the parameters, e.g., 
in a number of numerical linear algebra applications in which $1/\delta$ can be exponentially
large in $d$. Indeed, the Fast Johnson Lindenstrauss Transform was first used by S\'arlos to obtain the first speedups
for regression and low rank matrix approximation with relative error. S\'arlos used a version of the Fast Johnson
Lindenstrauss Transform due to \cite{AC06}. 
We will use a slightly different version called the {\em Subsampled
Randomized Hadamard Transform}, or SRHT for short. Later we will see a significantly faster transform for sparse
matrices. 

\begin{theorem}\label{thm:srht}(Subsampled Randomized Hadamard Transform 
\cite{AC06,S06,DMM06a,DMMS07,Tro11,DMMW12,ldfu13})
Let $\matS = \frac{1}{\sqrt{kn}} \matP \matH_n \matD$, where $\matD$ 
is an $n \times n$ diagonal matrix with i.i.d. diagonal
entries $\matD_{i,i}$ in which $\matD_{i,i} = 1$ with probability $1/2$, and $\matD_{i,i} = -1$ 
with probability $1/2$. 
$\matH_n$ refers to the Hadamard matrix of size $n$, which we assume is a power of $2$. Here, the $(i,j)$-th
entry of $\matH_n$ is given by $(-1)^{\langle i,j \rangle}/\sqrt{n}$, 
where $\langle i,j \rangle = \sum_{z=1}^{\log n} i_z \cdot j_z$,
and where $(i_{\log n}, \ldots, i_1)$ and $(j_{\log n}, \ldots, j_1)$ are the binary representations of $i$ and
$j$ respectively. The $r \times n$ matrix $\matP$ samples $r$ coordinates of an $n$-dimensional vector uniformly at random, 
where 
$$r = \Omega (\eps^{-2} (\log d)(\sqrt{d}+\sqrt{\log n})^2) .$$ 
Then with probability at least $.99$, for any fixed $n \times d$ matrix $\matU$ with orthonormal columns,
$$\|\matI_{d} - \matU^T\mat\Pi^T \mat\Pi \matU\|_2 \leq \eps.$$ 
Further, for any vector $\x \in \mathbb{R}^n$, $\matS \x$ can be computed in $O(n \log r)$ time. 
\end{theorem}
We will not present the proof of Theorem \ref{thm:srht}, instead relying upon the above intuition. The 
proof of Theorem \ref{thm:srht} can be found in the references listed above. 

Using Theorem \ref{thm:srht}, it is possible to compute an oblivious $\ell_2$-subspace embedding of a matrix $\matA$
in $O(n d \log(d(\log n)/\eps))$ time (see Definition 2.2 and Theorem 2.1 of \cite{AL08} for details on obtaining
this time complexity, which is a slight improvement to the $O(nd \log n)$ time mentioned above), which up to the logarithmic factor, is optimal in the matrix dimensions
of $\matA \in \mathbb{R}^{n \times d}$. One could therefore ask if this is the end of the road for subspace embeddings. 
Note that applying Theorem \ref{thm:gaussianSE} to create an oblivious $\ell_2$-subspace embedding $\matS$, 
or also using its optimizations discussed in the paragraphs following Theorem \ref{thm:gaussianSE} due to
Kane and Nelson \cite{kn14}, would
require time at least $O(\nnz(\matA) d \eps^{-1})$, since the number of non-zero entries per column of $\matS$
would be $\Theta(\eps^{-1} \log(f)) = \Theta(\eps^{-1} d)$, 
since the $f$ of Theorem \ref{thm:normals} would need to be set
to equal $\exp(d)$ to apply a net argument. 

It turns out that many matrices $\matA \in \mathbb{R}^{n \times d}$ are sparse, 
that is, the number of non-zero entries, $\nnz(\matA)$, may be much smaller than $n \cdot d$. One could therefore
hope to obtain an oblivious $\ell_2$-subspace embedding $\matS$ in which $\matS \cdot \matA$ can be computed 
in $O(\nnz(\matA))$
time and which the number of rows of $\matS$ is small. 

At first glance this may seem unlikely, since as mentioned
above, it is known that any Johnson Lindenstrauss Transform requires 
$\Omega(\frac{\eps^{-1} \log(f/\delta)}{\log(1/\varepsilon)})$  non-zero entries per column. Moreover, the size of any $C$-net
for constant $C$ is at least $2^{\Omega(d)}$, and therefore applying the arguments above we see that the ``$f$'' in 
the lower bound needs to be $\Omega(d)$. Alternatively, we could try to use an SRHT-based approach, but 
it is unknown how to adapt such approaches to exploit the sparsity of the matrix $\matA$. 

Nevertheless, in work with Clarkson \cite{CW13} we show that it is indeed possible to achieve $O(\nnz(\matA))$ time to
compute $\matS \cdot \matA$ for an oblivious $(1 \pm \eps)$ $\ell_2$ 
subspace embedding $\matS$ with only an $r = \poly(d/\eps)$
number of rows. The key to bypassing the lower bound mentioned above is that $\matS$ will {\it not} be a Johnson
Lindenstrauss Transform; instead it will only work for a set of $f = 2^{\Omega(d)}$ specially chosen points rather
than an arbitrary set of $f$ points. It turns out if we choose $2^{\Omega(d)}$ points from a $d$-dimensional subspace,
then the above lower bound of $\Omega(\eps^{-1} \log(f/\delta)/\log(1/\varepsilon))$ non-zero entries per column
does not apply; that is, this set of $f$ points is far from realizing the worst-case for the lower bound. 

In fact $\matS$ is nothing other than the {\sf CountSkech} matrix from the data stream literature 
\cite{ccf04,tz12}.
Namely, $\matS$ is constructed via the following procedure: for each of the $n$ columns $\matS_{*i}$, 
we first independently choose a 
uniformly random row
$h(i) \in \{1, 2, \ldots, r\}$. Then, we choose a uniformly random element of $\{-1, 1\}$, denoted $\sigma(i)$. We set $\matS_{h(i), i} = \sigma(i)$ and set $\matS_{j,i} = 0$ for all $j \neq i$. 
Thus, $\matS$ has only a single non-zero entry per column. For example, suppose $\matS \in \mathbb{R}^{4 \times 5}$. Then an instance of $\matS$ could
be:
\[ \left( \begin{array}{ccccc}
0 & 0 & -1 & 1 & 0\\
1 & 0 & 0 & 0 & 0\\
0 & 0 & 0 & 0 & 1\\
0 & -1 & 0 & 0 & 0\end{array} \right)\]
We refer to such an $\matS$ as a {\em sparse embedding matrix}. Note that since $\matS$ has only a single non-zero
entry per column, one can compute $\matS \cdot \matA$ for a matrix $A$ in $O(\nnz(\matA))$ time. 
\begin{theorem}\label{thm:cw}(\cite{CW13})
For $\matS$ a sparse embedding matrix with a total of $r = O(d^2/\eps^2 \poly(\log(d/\eps)))$ rows, for any fixed $n \times d$
matrix $\matA$, with probability $.99$, $\matS$ is a $(1 \pm \eps)$ $\ell_2$-subspace embedding for $\matA$. Further,
$\matS \cdot \matA$ can be computed in $O(\nnz(\matA))$ time. 
\end{theorem}
Although the number of rows of $\matS$ is larger than the $d/\eps^2$ using Theorem \ref{thm:gaussianSE}, 
typically $n \gg d$, e.g., in overconstrained regression problems, and so one can reduce $\matS \cdot \matA$
to a matrix containing $O(d/\eps^2)$ rows by composing it with a matrix $\matS'$ sampled using Theorem \ref{thm:normals}
or Theorem \ref{thm:srht},
computing $\matS' \matS \matA$ in time $O(\nnz(A)) + \poly(d/\eps)$, and so provided $\poly(d/\eps) < \nnz(\matA)$, 
this gives an overall $O(\nnz(\matA))$ time algorithm for obtaining an oblivious $(1 \pm \eps)$ $\ell_2$-subspace embedding
with the optimal $O(d/\eps^2)$ number of rows. Note here we can assume that 
$\nnz(A) \geq n$, as otherwise we can delete the rows of all zeros in $\matA$. 

The key intuition behind Theorem \ref{thm:cw}, given in \cite{CW13} 
why a sparse embedding matrix provides a subspace embedding, is 
that $\matS$ need not preserve the norms
of an arbitrary subset of $2^{O(d)}$ vectors in $\mathbb{R}^n$, but rather it need only preserve those norms
of a subset of $2^{O(d)}$ vectors in $\mathbb{R}^n$ which {\it all sit in a $d$-dimensional subspace} of 
$\mathbb{R}^n$. Such a subset of $2^{O(d)}$ vectors is significantly different from an arbitrary such set; 
indeed, the property used in \cite{CW13} which invented this was the following. 
If $\matU \in \mathbb{R}^{n \times d}$ is a matrix
with orthonormal columns with the same column space as $\matA$, then as one ranges over all unit $\x \in \mathbb{R}^d$, 
$\matU \x$ ranges over all unit vectors in the column space of $\matA$. Note though that for any coordinate $i$, by the
Cauchy-Schwarz inequality, 
\begin{eqnarray}\label{eqn:csIneq}
(\matU \x)_i^2 \leq \|\matU_{i*}\|_2^2.
\end{eqnarray}
As $\sum_{i = 1}^n \|\matU_{i*}\|_2^2 = d$, since $\matU$ has orthonormal columns, there is a subset 
$T$ of $[n]$ of size at most $d^2$ for which if $(\matU\x)_i^2 \geq 1/d$, then $i \in T$. Notice that $T$ does not
depend on $\x$, but rather is just equal to those rows $\matU_{i*}$ for which $\|\matU_{i*}\|_2^2 \geq 1/d$. Hence,  
(\ref{eqn:csIneq}) implies that as one ranges over all unit
vectors $\matU \x$, the coordinates of $\matU\x$ that are larger than $1/d$, if any, must lie in this relatively
small set $T$. This
is in sharp contrast to an arbitrary set of $2^{O(d)}$ unit vectors, for which every coordinate could be larger than 
$1/d$ for at least one vector in the collection. It turns out that if $\matU \x$ has no heavy coordinates, then a
sparse subspace embedding does have the Johnson-Lindenstrauss property, 
as shown by Dasgupta, Kumar, and S\'arlos \cite{dks10}.
Hence, provided the set of coordinates of $T$ is perfectly hashed by $\matS$, one can handle the remaining coordinates
by the analysis of \cite{dks10}. 

While the above proof technique has proven useful in generating $\ell_p$ subspace embeddings for other $\ell_p$-norms
(as we will see in Section \ref{chap:robust} for the $\ell_1$-norm), 
and also applies more generally to sets of $2^{O(d)}$ vectors with a fixed small number of heavy coordinates, 
it turns out for $\ell_2$ one can simplify and sharpen the argument by using 
more direct linear-algebraic methods.
In particular, via a simpler second moment calculation, Theorem \ref{thm:cw} 
was improved in \cite{MM13,JH13} to the following.
\begin{theorem}\label{thm:mmnn}\cite{MM13,JH13}
For any $0 < \delta < 1$, and for $\matS$ a sparse embedding matrix with $r = O(d^2/(\delta \eps^2))$ rows, 
then with probability $1-\delta$, for any fixed $n \times d$
matrix $\matA$, $\matS$ is a $(1 \pm \eps)$ $\ell_2$-subspace embedding for $\matA$. 
The matrix product 
$\matS \cdot \matA$ can be computed in $O(\nnz(\matA))$ time. Further, all of this holds if the hash function
$h$ defining $\matS$ is only pairwise independent, and the sign function $\sigma$ defining $\matS$ is only
$4$-wise independent. 
\end{theorem}
The proofs of Theorem \ref{thm:mmnn} given in \cite{MM13,JH13} work by bounding, for even integers $\ell \geq 2$,
\begin{eqnarray*}
\Pr[\|\matI_d-\matU^T\matS^T\matS\matU\|_2 \geq \eps]
& = & \Pr[\|\matI_d -\matU^T\matS^T\matS\matU\|_2^{\ell} \geq \eps^{\ell}]\\
& \leq & \eps^{-{\ell}}{\bf E}[\|\matI_d-\matU^T\matS^T\matS\matU\|_2^{\ell}]\\
& \leq & \eps^{-{\ell}} {\bf E}[\tr((\matI_d-\matU^T\matS^T\matS\matU)^{\ell})],
\end{eqnarray*}
which is a standard way of bounding operator norms of random matrices, see, e.g., \cite{by93}. 
In the bound above, Markov's inequality is used in the first inequality, while the second inequality
uses that the eigenvalues of $(\matI_d - \matU^T\matS^T\matS\matU)^{\ell}$ are non-negative for even integers
$\ell$, one of those eigenvalues is $\|\matI_d-\matU^T\matS^T\matS\matU\|_2^{\ell}$, and the trace is the sum
of the eigenvalues. 
This is also the technique
used in the proof of Theorem \ref{thm:nn} below (we do not present the proof of this), 
though there a larger value of $\ell$ is used while
for Theorem \ref{thm:mmnn} we will see that it suffices to consider $\ell = 2$. 

Rather than proving Theorem \ref{thm:mmnn} directly, we will give a 
alternative proof of it observed by Nguy$\tilde{\hat{\mbox{e}}}$n \cite{n13}
in the next section, showing how it is a consequence of a primitive called approximate matrix multiplication that
had been previously studied, and for which is useful for other applications we consider. 
Before doing so, though, we mention that it is possible to achieve fewer than $O(d^2/\eps^2)$ rows for constant
probability subspace embeddings 
if one is willing to increase the running time of applying the subspace embedding from 
$O(\nnz(\matA))$ to $O(\nnz(\matA)/\eps)$. 
This was shown by Nelson
and Nguy$\tilde{\hat{\mbox{e}}}$n \cite{JH13}. 
They show that for any $\gamma > 0$, one can achieve $d^{1+\gamma}/\eps^2 \poly(1/\gamma)$ dimensions
by increasing the number of non-zero entries in $\matS$ to $\poly(1/\gamma)/\eps$. They also show that 
by increasing the number of non-zero entries in $\matS$ to $\polylog(d)/\eps$, 
one can achieve $d/\eps^2 \polylog(d)$ dimensions. These results also generalize to failure probability $\delta$,
and are summarized by the following theorem. 

\begin{theorem}\label{thm:nn}\cite{JH13}
There are distributions on matrices $S$ with the following properties:

(1) For any fixed $\gamma > 0$ and any fixed $n \times d$ matrix $\matA$, 
$\matS$ is a $(1 \pm \eps)$ oblivious $\ell_2$-subspace embedding for $\matA$
with $r = d^{1+\gamma}/\eps^2$ rows and error probability $1/\poly(d)$. Further, $\matS \cdot \matA$ can be computed
in $O(\nnz(\matA) \poly(1/\gamma)/\eps )$ time. 

(2) There is a $(1 \pm \eps)$ oblivious $\ell_2$-subspace embedding for $\matA$ 
with $r = d \cdot \polylog(d/(\eps \delta))/\eps^2$ rows and error probability $\delta$. 
Further, $\matS \cdot \matA$ can be computed
in $O(\nnz(\matA) \polylog(d/(\eps \delta)))/\eps)$ time. 
\end{theorem}
We note that for certain applications, such as least squares regression, one can
still achieve a $(1+\eps)$-approximation in $O(\nnz(\matA))$ time by applying Theorem \ref{thm:nn}
with the value of $\eps$ in Theorem \ref{thm:nn} set to a fixed constant since the application only
requires a $(1 \pm O(1))$-subspace embedding in order to achieve a $(1+\eps)$-approximation; see 
Theorem \ref{thm:1eps} for further details on this. It is also conjectured in \cite{JH13} that
$r$ can be as small as $O((d+\log(1/\delta))/\eps^2)$ with a time for computing $\matS \cdot \matA$
of $O(\nnz(\matA) \log(d/\delta)/\eps)$, though at the time of this writing the polylogarithmic
factors in Theorem \ref{thm:nn} are somewhat far from achieving this conjecture. 

There has been further work on this by Bourgain and Nelson \cite{bn13}, 
who showed among other things that 
if the columns of $\matU$ form an orthonormal basis for the column space of $\matA$, and
if the coherence $\max_{i \in [n]}\|\matU_{i*}\|_2^2 \leq 1/\polylog(d)$, then a sparse embedding
matrix provides a $(1 \pm \eps)$ $\ell_2$-subspace embedding for $\matA$. Here the column sparsity
remains $1$ given the incoherence assumption, just as in Theorem \ref{thm:mmnn}. The authors also
provide results for unions of subspaces. 

We note that one can also achieve $1-\delta$ success
probability 
bounds in which the sparsity and dimension depend on $O(\log 1/\delta)$ using 
these constructions \cite{CW13,MM13,JH13}. For our
applications it will usually not be necessary, as one can often instead repeat the entire procedure $O(\log 1/\delta)$
times and take the best solution found, such as in regression or low rank matrix approximation. 
We also state a different
way of finding an $\ell_2$-subspace embedding with high success probability in \S\ref{sec:whp}.

\subsection{Matrix multiplication}\label{sec:mm}
In this section we study the approximate matrix product problem. 

\begin{definition}\label{def:matrixProduct}
Let $0 < \eps < 1$ be a given approximation parameter. 
In the {\em Matrix Product} Problem matrices $\matA$ and $\matB$ are given, 
where $\matA$ and $\matB$ each have $n$ rows and a 
total of $c$ columns.  
The goal is to output a 
matrix $\matC$ so that 
\[
\normF{\matA^T \matB - \matC} \leq \varepsilon \normF{\matA} \normF{\matB}.
\]
\end{definition}
There are other versions of approximate matrix product, such as those that replace the Frobenius norms
above with operator norms \cite{Zou10,malik11,CEMMP14,CNW14}. 
Some of these works look at
bounds in terms of the so-called stable rank of $\matA$ and $\matB$, which
provides a continuous relaxation of the rank. 
For our application we will focus on the version of the
problem given in Definition \ref{def:matrixProduct}. 

The idea for solving this problem is to compute $\matA^T \matS^T$ and $\matS \matB$ 
for a sketching matrix $\matS$. We will choose $\matS$ so that
$${\bf E}[\matA^T \matS^T \matS \matB] = \matA^T \matB,$$
and we could hope that the variance of this estimator is small, namely,
we could hope that the standard deviation of the estimator is $O(\varepsilon \normF{\matA} \normF{\matB})$. 
To figure out which matrices $\matS$ are appropriate for this,
we use the following theorem of Kane and Nelson \cite{kn14}. This is a more general result of the 
analogous result for sign matrices of Clarkson and the author \cite{CW09}, and a slight strengthening 
of a result of Sarl\'{o}s \cite{S06}. 

Before giving the theorem, we need a definition.

\begin{definition}\cite{kn14}\label{def:moment}
A distribution $\mathcal{D}$ on matrices $\matS \in \mathbb{R}^{k \times d}$ has the 
$(\eps, \delta, \ell)$-JL moment property if for all $\x \in \mathbb{R}^d$ with $\|\x\|_2 = 1$,
$${\bf E}_{\matS \sim \mathcal{D}} |\|\matS \x\|_2^2-1|^{\ell} \leq \eps^{\ell} \cdot \delta.$$
\end{definition}

We prove the following theorem for a general value of $\ell$, since as mentioned it is used in some
subspace embedding proofs including the ones of Theorem \ref{thm:nn}. However, in this section we will only
need the case in which $\ell = 2$. 

\begin{theorem}\label{thm:jlamp}\cite{kn14}
For $\eps, \delta \in (0,1/2)$, let $\mathcal{D}$ be a distribution over matrices with $d$ columns
that satisfies the $(\eps, \delta, \ell)$-JL moment property for some $\ell \geq 2$. Then for $\matA, \matB$
matrices each with $d$ rows,
$$\Pr_{\matS \sim \mathcal{D}} \left [\FNorm{\matA^T \matS^T \matS \matB - \matA^T \matB} 
> 3 \eps \FNorm{\matA} \FNorm{\matB} \right ] \leq \delta.$$
\end{theorem}
\begin{proof}
We proceed as in the proof of \cite{kn14}. 
For $\x, \y \in \mathbb{R}^d$, we have
$$\frac{\langle \matS\x, \matS\y \rangle}{\|\x\|_2 \|\y\|_2} = \frac{\|\matS\x\|_2^2 + \|\matS\y\|_2^2 - \|\matS(\x-\y)\|_2^2}{2}.$$
For a random scalar $X$, let $\|X\|_p = ({\bf E}|X|^p)^{1/p}$. 
We will sometimes consider $X = \|\matT\|_F$ for a random matrix $\matT$, 
in which case $X$ is a random scalar
and the somewhat cumbersome notation $\|\|\matT\|_F\|_p$ indicates $({\bf E}[\|\matT\|_F^p])^{1/p}$. 

Minkowski's inequality asserts that the triangle
inequality holds for this definition, namely, that $\|\matX+\matY\|_p \leq \|\matX\|_p + \|\matY\|_p$, and as the
other properties of a norm are easy to verify, it follows that $\|.\|_p$ is a norm. Using that it is a norm, we have
for unit vectors $\x$, $\y$, that $\|\langle \matS \x, \matS \y \rangle - \langle \x, \y \rangle\|_{\ell}$ is equal to
\begin{eqnarray*}
& = & \frac{1}{2} \cdot \|(\|\matS \x\|_2^2 - 1 ) + (\|\matS \y\|_2^2-1)
 - (\|\matS(\x-\y)\|_2^2 - \|\x-\y\|_2^2)\|_{\ell}\\
& \leq & \frac{1}{2} \cdot \left (\|\|\matS \x\|_2^2-1\|_{\ell} + \|\|\matS\y\|_2^2-1\|_{\ell}
+ \|\|\matS(\x-\y)\|_2^2 - \|\x-\y\|_2^2\|_{\ell} \right )\\
& \leq & \frac{1}{2} \cdot \left (\eps \cdot \delta^{1/\ell} + \eps \cdot \delta^{1/\ell} + \|\x-\y\|_2^2 \cdot \eps \cdot \delta^{1/\ell} \right )\\
& \leq & 3\eps \cdot \delta^{1/\ell}.
\end{eqnarray*}
By linearity, this implies for arbitrary vectors $\x$ and $\y$ that
$\frac{\|\langle \matS \x, \matS \y \rangle - \langle \x, \y \rangle \|_{\ell}}{\|\x|_2 \|\y\|_2} \leq 3 \eps \cdot \delta^{1/\ell}$. 

Suppose $\matA$ has $n$ columns and $\matB$ has $m$ columns. 
Let the columns of $\matA$ be $\matA_1, \ldots, \matA_n$ and the columns of $\matB$ be $\matB_1, \ldots, \matB_n$. 
Define the random variable 
$$X_{i,j} = \frac{1}{\|A_i\|_2 \|B_j\|_2} \cdot \left (\langle \matS \matA_i, \matS \matB_j \rangle
- \langle \matA_i, \matB_j \rangle \right ).$$
Then, $\FNormS{\matA^T \matS^T \matS \matB - \matA^T \matB} = \sum_{i=1}^n \sum_{j=1}^m \|\matA_i\|_2^2 \cdot \|\matB_j\|_2^2 \cdot X_{i,j}^2$. Again using Minkowski's inequality and that $\ell/2 \geq 1$,
\begin{eqnarray*}
\|\FNormS{\matA^T\matS^T \matS \matB - \matA^T\matB}\|_{\ell/2} & = & 
\|\sum_{i=1}^n \sum_{j=1}^m \|\matA_i\|_2^2 \cdot \|\matB_j\|_2^2 \cdot X_{i,j}^2\|_{\ell/2}\\
& \leq & \sum_{i=1}^n \sum_{j=1}^m \|\matA_i\|_2^2 \cdot \|\matB_j\|_2^2 \cdot \|X_{i,j}^2\|_{\ell/2}\\
& = & \sum_{i=1}^n \sum_{j=1}^m \|\matA_i\|_2^2 \cdot \|\matB_j\|_2^2 \cdot \|X_{i,j}\|_{\ell}^2\\
& \leq & (3\eps \delta^{1/\ell})^2 \cdot \left (\sum_{i=1}^n \sum_{j=1}^m \|\matA_i\|_2^2 \cdot \|\matB_j\|_2^2 \right )\\
& = & (3\eps \delta^{1/\ell})^2 \cdot \FNormS{\matA} \FNormS{\matB}.
\end{eqnarray*}
Using that ${\bf E} \FNorm{\matA^T \matS^T \matS \matB - \matA^T \matB}^{\ell} = \|\FNormS{\matA^T\matS^T \matS \matB - \matA^T \matB\|} \|_{\ell/2}^{\ell/2}$, together with Markov's inequality, we have
\begin{eqnarray*}
\Pr \left [\FNorm{\matA^T\matS^T\matS\matB - \matA^T \matB} > 3\eps \FNorm{\matA} \FNorm{\matB} \right ]
& \leq & \left (\frac{1}{3\eps \FNorm{\matA} \FNorm{\matB}} \right )^{\ell}\\
& \cdot & \ {\bf E} \FNorm{\matA^T\matS^T\matS\matB - \matA^T \matB}^{\ell}\\
& \leq & \delta.
\end{eqnarray*}
\end{proof}
We now show that sparse embeddings matrices satisfy the $(\eps, \delta, 2)$-JL-moment property. This was originally shown
by Thorup and Zhang \cite{tz12}.
\begin{theorem}\label{thm:tz}
Let $\matS$ be a sparse embedding matrix, as defined in \S\ref{sec:se}, with at least $2/(\eps^2 \delta)$ rows. Then 
$\matS$ satisfies the $(\eps, \delta, 2)$-JL moment property. Further, this holds if the hash function $h$ defining the 
sparse embedding matrix is only $2$-wise independent and the sign function $\sigma$ is $4$-wise independent. 
\end{theorem}
\begin{proof}
As per Definition \ref{def:moment}, we need to show for any unit vector $\x \in \mathbb{R}^d$,
\begin{eqnarray}\label{eqn:toshow}
{\bf E}_{\matS} [(\|\matS\x\|_2^2-1)^2] = {\bf E}_{\matS}[\|\matS\x\|_2^4] - 2{\bf E}_{\matS}[\|\matS\x\|_2^2] + 1 \leq \eps^2 \delta.
\end{eqnarray}
For a sparse embedding matrix $\matS$, we let $h:[d] \rightarrow [r]$ be a random $2$-wise independent 
hash function indicating for each column
$j \in [d]$, which row in $\matS$ contains the non-zero entry. Further, we let $\sigma:[d] \rightarrow \{-1,1\}$ be a $4$-wise
independent function, independent of $h$, indicating whether
the non-zero entry in the $j$-th column is $1$ or $-1$. For an event $\mathcal{E}$, let $\delta(\mathcal{E})$ be an indicator variable
which is $1$ if $\mathcal{E}$ occurs, and is $0$ otherwise. Then,
\begin{eqnarray*}
{\bf E}[\|\matS\x\|_2^2] & = & \sum_{i \in [r]} {\bf E} \left [\left (\sum_{j \in [d]} \delta(h(j) = i) \x_j \sigma(j) \right )^2 \right ]\\
& = & \sum_{i \in [r]} \sum_{j, j' \in [d]} \x_j \x_{j'} {\bf E} \left [\delta(h(j) = i)\delta(h(j') = i) \right ] {\bf E} 
\left [\sigma(j) \sigma(j') \right ]\\
& = & \sum_{i \in [r]} \sum_{j \in [d]} \frac{\x_j^2}{r}\\
& = & \|\x\|_2^2\\
& = & 1,
\end{eqnarray*}
where the second equality uses that $h$ and $\sigma$ are independent, while the third equality uses that ${\bf E} [\sigma(j) \sigma(j')] = 1$ if $j = j'$, and otherwise
is equal to $0$. 

We also have,
\begin{eqnarray*}
{\bf E}[\|\matS\x\|_2^4] & = & {\bf E} \left [ \left (\sum_{i \in [r]} \left (\sum_{j \in [d]} \delta(h(j) = i) \x_j \sigma(j) \right )^2 \right )^2 \right ]\\
& = & \sum_{i,i' \in [r]} \sum_{j_1,j_2, j'_1, j'_2  \in [d]} \x_{j_1}\x_{j_2}\x_{j_1'}\x_{j_2'}\\
&& \cdot \ {\bf E} \left [\delta(h(j_1) = i)\delta(h(j_2)=i)\delta(h(j'_1) = i')\delta(h(j'_2) = i') \right ]\\
&& \cdot \ {\bf E} \left [\sigma(j_1) \sigma(j_2) \sigma(j_1')\sigma(j_2') \right ]\\
& = & \sum_{i \in [r]} \sum_{j \in [d]} \frac{\x_j^4}{r} + \sum_{i, i' \in [r]} \sum_{j_1 \neq j_1' \in [d]} \frac{\x_{j_1}^2 \x_{j_1'}^2}{r^2}\\
&& + 2\sum_{i \in [r]} \sum_{j_1 \neq j_2 \in [d]} \frac{\x_{j_1}^2 \x_{j'_1}^2}{r^2}\\
& = & \sum_{j, j' \in [d]} \x_j^2 \x_{j'}^2 + \frac{2}{r}\sum_{j_1 \neq j_2 \in [d]} \x_{j_1}^2 \x_{j'_1}^2\\
& \leq & \|\x\|_2^4 + \frac{2}{r} \|\x\|_2^4\\
& \leq & 1 + \frac{2}{r},
\end{eqnarray*}
where the second equality uses the independence of $h$ and $\sigma$, and the third equality uses that since $\sigma$ is $4$-wise independent, in order for 
${\bf E} \left [\sigma(j_1) \sigma(j_2) \sigma(j_1')\sigma(j_2') \right ]$ not to vanish, it must be that either 
\begin{enumerate}
\item $j_1 = j_2 = j_1' = j_2'$ or
\item $j_1 = j_2$ and $j_1' = j_2'$ but $j_1 \neq j_1'$ or
\item $j_1 = j_1'$ and $j_2 = j_2'$ but $j_1 \neq j_2$ or
\item $j_1 = j_2'$ and $j_1' = j_2$ but $j_1 \neq j_2$.
\end{enumerate}
Note that in the last two cases, for ${\bf E} \left [\delta(h(j_1) = i)\delta(h(j_2)=i)\delta(h(j'_1) = i')\delta(h(j'_2) = i') \right ]$ not to vanish,
we must have $i = i'$. The fourth equality and first inequality are based on regrouping the summations, and the sixth inequality uses that $\|\x\|_2 = 1$. 

Plugging our bounds on $\|\matS \x\|_2^4$ and $\|\matS \x\|_2^2$ into (\ref{eqn:toshow}), the theorem follows. 
\end{proof}
We now present a proof that sparse embedding matrices provide subspace embeddings, as mentioned in \S\ref{sec:se},
as given by Nguy$\tilde{\hat{\mbox{e}}}$n \cite{n13}.  
\begin{proofof}{of Theorem \ref{thm:mmnn}}
By Theorem \ref{thm:tz}, we have that $\matS$ satisfies the $(\eps, \delta, 2)$-JL moment property. We can thus apply Theorem \ref{thm:jlamp}. 

To prove Theorem \ref{thm:mmnn}, 
recall that 
if $\matU$ is an orthonormal basis for the column space of $\matA$ and $\|\matS\y\|_2 = (1 \pm \eps)\|\y\|_2$ for all $\y$ in
the column space of $\matU$, then $\|\matS\y\|_2 = (1 \pm \eps)\|\y\|_2$ for all $\y$ in the column space of $\matA$, since the column
spaces of $\matA$ and $\matU$ are the same. 

We apply Theorem \ref{thm:jlamp} to $\matS$ with the $\matA$ and $\matB$ of that theorem equal to $\matU$, 
and the $\eps$ of that theorem equal to $\eps/d$. Since $\matU^T \matU = \matI_d$ and $\FNormS{\matU} = d$, we have,
$$\Pr_{\matS \sim \mathcal{D}} \left [\FNorm{\matU^T \matS^T \matS \matU - \matI_d} > 3 \eps \right ] \leq \delta,$$
which implies that  
\begin{eqnarray*}
\Pr_{\matS \sim \mathcal{D}} \left [\|\matU^T \matS^T \matS\matU -\matI_d\|_2 > 3 \eps \right]
\leq \Pr_{\matS \sim \mathcal{D}} \left [\FNorm{\matU^T \matS^T \matS\matU -\matI_d} > 3 \eps \right]
\leq \delta.
\end{eqnarray*}
Recall that 
the statement that $\x^T(\matU^T \matS^T \matS\matU -\matI_d)\x \leq 3 \eps$ for all unit $\x \in \mathbb{R}^d$ is
equivalent to the statement that $\|\matS\matU\x\|_2^2 = 1 \pm 3\eps$ for all unit $\x \in \mathbb{R}^d$, that is,
$\matS$ is a $(1 \pm 3\eps)$ $\ell_2$-subspace embedding. The proof follows by rescaling $\eps$ by $3$.  
\end{proofof}

\subsection{High probability}\label{sec:whp}
The dependence of Theorem \ref{thm:mmnn} on the error probability $\delta$ is linear, which is not completely
desirable. One can use Theorem \ref{thm:nn} to achieve a logarithmic dependence, but then the running time
would be at least $\nnz(\matA)\polylog(d/(\eps \delta))/\eps$ and the number of non-zeros per column of $\matS$ would
be at least $\polylog(d/(\eps\delta))/\eps$. Here we describe an alternative way based on \cite{BKLW14} 
which takes $O(\nnz(\matA)\log(1/\delta))$
time, and preserves the number of non-zero entries per column of $\matS$ to be $1$. It is, however, a
non-oblivious embedding. 

In \cite{BKLW14}, an approach (Algorithm~\ref{alg:success} below) to boost the success probability 
by computing $t = O(\log (1/\delta))$
independent sparse oblivious subspace embeddings $\matS_j\matA$ is proposed, $j = 1, 2, \ldots, t$, 
each with only constant success probability,
and then running a cross validation procedure to find one which succeeds with probability $1-1/\delta$.
More precisely, we compute the SVD of all embedded matrices $\matS_j\matA = \matU_j \matD_j \matV_j^\top$,
and find a $j \in [t]$ such that for at least half of the indices $j' \neq j$,
all singular values of $\matD_j \matV_j^\top \matV_{j'} \matD_{j'}^\top$ are in $[1 \pm O(\varepsilon)]$.

The reason why such an embedding $\matS_j \matA$ succeeds with high probability is as follows.
Any two successful embeddings $\matS_j \matA$ and $\matS_{j'} \matA$, by definition, satisfy that
$\| \matS_j \matA \x \|^2_2  = (1 \pm O(\varepsilon)) \| \matS_{j'} \matA \x \|_2^2$ for all $x$,
which we show is equivalent to passing the test on the singular values.
Since with probability at least $1-\delta$, a $9/10$ fraction of the embeddings are successful,
it follows that the one we choose is successful with probability $1-\delta$. One can thus show
the following theorem.

\begin{theorem}\label{thm:success}(\cite{BKLW14})
Algorithm~\ref{alg:success} outputs a subspace embedding with probability at least $1-\delta$.
In expectation step 3 is only run a constant number of times.
\end{theorem}

\begin{algorithm}[t]
\caption{Boosting success probability of embedding}
\label{alg:success}
Input: $\matA \in \R^{n\times d}$, parameters $\varepsilon, \delta$
\begin{enumerate}
\item Construct $t=O(\log \frac{1}{\delta})$ independent constant success probability sparse subspace embeddings $\matS_j \matA$
with accuracy $\varepsilon/6$.
\item Compute SVD $\matS_j\matA = \matU_j \matD_j \matV_j^\top$ for $j\in [t]$.
\item For $j \in [t]$
\begin{enumerate}
\item Check if for at least half $j' \neq j$, $$\sigma_i(\matD_j \matV_j^\top \matV_{j'} \matD_{j'}^\top) \in [1\pm \varepsilon/2],\forall i.$$
\item If so, output $\matS_j \matA$.
\end{enumerate}
\end{enumerate}
\end{algorithm}

\subsection{Leverage scores}\label{sec:leverage}
We now introduce the concept of leverage scores, which provide alternative subspace embeddings based on sampling
a small number of rows of $\matA$. We will see that they play a crucial role in various applications in this book,
e.g., CUR matrix decompositions and spectral sparsification. Here we use the parameter $k$ instead of $d$ for the dimension
of the subspace, as this will match our use in applications. For an excellent survey on leverage scores, we refer
the reader to \cite{m11}. 

\begin{definition}(Leverage Score Sampling)\label{def:lss}
Let $\matZ \in \mathbb{R}^{n \times k}$ have orthonormal columns, and let $p_i = \ell_i^2/k$, 
where $\ell_i^2 = \|\e_i^T \matZ\|_2^2$
is the $i$-th leverage score of $\matZ$. Note that $(p_1, \ldots, p_n)$ is a distribution. 
Let $\beta > 0$ be a parameter, and suppose we have any distribution $q = (q_1, \ldots, q_n)$ for which for
all $i \in [n]$, $q_i \geq \beta p_i$.

Let $s$ be a parameter. Construct an $n \times s$ 
sampling matrix $\mat\Omega$ and an $s \times s$ rescaling matrix $\matD$ as follows. 
Initially, $\mat\Omega = \mat0^{n \times s}$
and $\matD = \mat0^{s \times s}$. For each column $j$ of $\mat\Omega, \matD$, 
independently, and with replacement, pick a row index
$i \in [n]$ with probability $q_i$, and set $\mat\Omega_{i,j} = 1$ and 
$\matD_{jj} = 1/\sqrt{q_i s}$. We denote this procedure
{\textsc RandSampling}$(\matZ, s, q)$. 
\end{definition}
Note that the matrices $\mat\Omega$ and $\matD$ in the {\textsc RandSampling}$(\matZ, s, q)$ procedure 
can be computed in $O(nk + n + s \log s)$ time. 

Definition \ref{def:lss} introduces the concept of the {\it leverage scores} $\ell_i^2 = \|\e_i^T \matZ\|_2^2$ of a matrix 
$\matZ$ with orthonormal columns. For an $n \times k$ matrix $\matA$ whose columns need not be orthonormal, we can still
define its leverage scores $\ell_i^2$ as $\|e_i^T \matZ\|_2^2$, where $\matZ$ is an $n \times r$ matrix with orthonormal
columns having the same column space of $\matA$, where $r$ is the rank of $\matA$. Although there are many choices
$\matZ$ of orthonormal bases for the column space of $\matA$, it turns out that they all give rise to the same values
$\ell_i^2$. Indeed, if $\matZ'$ were another $n \times r$ matrix with orthonormal colums having the same column space
of $\matA$, then $\matZ' = \matZ \matR$ for an $r \times r$ invertible matrix $\matR$. But since $\matZ'$ and $\matZ$
have orthonormal columns, $\matR$ must be orthonormal. Indeed, for every vector $\x$ we have
$$\|\x\|_2 = \|\matZ'\x\|_2 = \|\matZ \matR \x\|_2 = \|\matR \x\|_2.$$
Hence $$\|\e_i^T\matZ'\|_2^2 = \|\e_i^T \matZ \matR\|_2^2 = \|\e_i^t \matZ\|_2^2,$$
so the definition of the leverage scores does not depend on a particular choice of orthonormal basis for the column space
of $\matA$.

Another useful property, though we shall not need it, is that the leverage scores $\ell_i^2$ are at most $1$. This follows
from the fact that any row $\ve$ of $\matZ$ must have squared norm at most $1$, as otherwise 
$$\|\matZ \cdot \frac{\ve}{\|\ve\|_2}\|_2^2 = \frac{1}{\|\ve \|_2^2} \cdot \|\matZ \ve \|_2^2 > 1,$$
contradicting that $\|\matZ \x\|_2 = \|\x\|_2$ for all $\x$ since $\matZ$ has orthonormal columns. 

The following shows that $\matD^T \mat\Omega^T \matZ$ 
is a subspace embedding of the column space of $\matZ$, for $s$ large enough.
To the best of our knowledge, theorems of this form first appeared in \cite{DMM06c,DMM06d}. Here we give a simple
proof along the lines in \cite{Mag10}.   
\begin{theorem}\label{thm:lssPerf}(see, e.g., similar proofs in \cite{Mag10})
Suppose $\matZ \in \mathbb{R}^{n \times k}$ has orthonormal columns. Suppose $s > 144k \ln(2k/\delta)/(\beta \eps^2)$ 
and $\mat\Omega$ and $\matD$
are constructed from the {\textsc RandSampling}$(\matZ,s,q)$ procedure. Then with probability
at least $1-\delta$, simultaneously for all $i$, 
$$1 - \eps \leq \sigma_i^2 (\matD^T \mat\Omega^T \matZ) \leq 1 + \eps,$$
or equivalently,
$$1-\eps \leq \sigma_i^2 (\matZ^T \mat\Omega \matD) \leq 1+\eps.$$
\end{theorem}

\begin{proof}
We will use the following matrix Chernoff bound for a sum of random matrices, 
which is a non-commutative Bernstein bound.

\begin{fact}(Matrix Chernoff)\label{fact:chernoff}
Let $\matX_1, \ldots, \matX_s$ be independent copies of a symmetric random matrix $\matX \in \mathbb{R}^{k \times k}$ 
with ${\bf E}[\matX] = 0$, $\|\matX\|_2 \leq \gamma$, and $\|{\bf E} \matX^T \matX\|_2 \leq s^2$. 
Let $\matW = \frac{1}{s} \sum_{i=1}^s \matX_i$. Then for any $\eps > 0$,
$$\Pr[\|\matW\|_2 > \eps] \leq 2k \exp(-s\eps^2/(2s^2 + 2\gamma \eps/3)).$$ 
\end{fact}
Let $\matU_i \in \mathbb{R}^{1 \times k}$ be the $i$-th sampled row of $\matZ$ by the 
{\textsc RandSampling}$(\matZ, s)$ procedure.
Let
$\z_j$ denote the $j$-th row of $\matZ$. 
Let $\matX_i = \matI_{k} - \matU_i^T \matU_i / q_i$. Then the $\matX_i$ 
are independent copies of a matrix random variable, and  
$${\bf E}[\matX_i] = \matI_{k} - \sum_{j=1}^n q_j \z_j^T \z_j /q_j 
= \matI_{k} - \matZ^T \matZ = \mat0_{k \times k}.$$
For any $j$, $\z_j^T \z_j/q_j$ is a rank-$1$ matrix with 
operator norm bounded by $\|\z_j\|_2^2/q_j \leq k/\beta$. Hence,
\begin{eqnarray}\label{eqn:operator}
\|\matX_i\|_2 \leq \|\matI_{k}\|_2 + \|\matU_i^T \matU_i / q_i\|_2 \leq 1 + \frac{k}{\beta}.
\end{eqnarray}
We also have
\begin{eqnarray}\label{eqn:varBern}
{\bf E}[\matX^T \matX] & = & \matI_{k} - 2 {\bf E}[\matU_i^T \matU_i/q_i] 
+ {\bf E}[\matU_i^T \matU_i \matU_i^T \matU_i/q_i^2] \notag \\
& = & \sum_{j=1}^n \z_j^T \z_j \z_j^T \z_j/q_i - \matI_{k}\\
& \leq & (k/\beta) \sum_{j=1}^n \z_j^T \z_j - \matI_{k}\\
& = & (k/\beta - 1) \matI_{k}.
\end{eqnarray}
It follows that $\|{\bf E}\matX^T \matX\|_2 \leq (k/\beta - 1)$. 
Note that $\matW = \frac{1}{k} \sum_{i=1}^s \matX_i = \matI_{k} - \matZ^T \mat\Omega \matD \matD^T \mat\Omega^T \matZ$. 
Applying Fact \ref{fact:chernoff},
$$\Pr[\|\matI_{k} - \matZ^T \mat\Omega \matD \matD^T \mat\Omega^T \matZ \|_2 > \eps] \leq 2k \exp(-s \eps^2/(2k/\beta + 2k \eps /(3\beta))),$$
and setting $s = \Theta(k \log (k/\delta) /(\beta \eps^2))$ 
implies that with all but $\delta$ probability,
$\|\matI_{k} - \matZ^T \mat\Omega \matD \matD^T \mat\Omega^T \matZ \|_2 \leq \eps$, 
that is, all of the singular values of $\matD^T \mat\Omega^T \matZ$ are within $1 \pm \eps$, as desired. 
\end{proof}
To apply Theorem \ref{thm:lssPerf} for computing subspace embeddings of an $n \times k$ matrix $\matA$, one writes
$\matA = \matZ \mat\Sigma \matV^T$ in its SVD. Then, Theorem \ref{thm:lssPerf} guarantees that for all $\x \in \mathbb{R}^k$,
$$\|\matD^T \mat\Omega^T \matA \x\|_2 = (1 \pm \eps)\|\mat\Sigma\matV^T\|x\|_2 = (1 \pm \eps)\|\matA\x\|_2,$$
where the first equality uses the definition of $\matA$ and the fact that all singular values of $\matD^T\mat\Omega^T\matZ$
are $1 \pm \eps$. The second equality uses that $\matZ$ has orthonormal columns, so $\|\matZ\y\|_2 = \|\y\|_2$ for all
vectors $\y$. 

One drawback of {\textsc RandSampling}$(\matZ, s, q)$ is it requires as 
input a distribution $q$ which well-approximates the leverage
score distribution $p$ of $\matZ$. While one could obtain $p$ exactly by computing the SVD of $\matA$, this would na\"ively take
$O(nk^2)$ time (assuming $k < n$). It turns out, as shown in \cite{DMMW12}, one can compute a distribution $q$ with the approximation
parameter $\beta = 1/2$ in time $O(nk\log n + k^3)$ time. This was further improved in \cite{CW13} to $O(\nnz(A)\log n + k^3)$ time. 

We need a version of the Johnson-Lindenstrauss lemma, as follows. We give a simple proof for completeness.
\begin{lemma}(Johnson-Lindenstrauss)\label{lem:jl}
Given $n$ points $q_1, \ldots, q_n \in \mathbb{R}^d$, if $\matG$ is a $t \times d$ matrix of i.i.d. $N(0,1/t)$ random 
variables, then for $t = O(\log n / \eps^2)$ simultaneously for all $i \in [n]$,
$$\Pr[\forall i, \ \|\matG \q_i\|_2 \in (1 \pm \eps)\|\q_i\|_2] \geq 1-\frac{1}{n}.$$
\end{lemma}
\begin{proof}
For a fixed $i \in [n]$, $\matG \q_i$ is a $t$-tuple of i.i.d. $N(0, \|\q_i\|_2^2/t)$ random variables. Here we use
the fact that for independent standard normal random variables $g$ and $h$ and scalars $a$ and $b$, 
the random variable $a \cdot g + b \cdot h$ has the same distribution as that of the random variable 
$\sqrt{a^2+b^2} z$, where $z \sim N(0,1)$. 

It follows that $\|\matG \q_i\|_2^2$  
is equal, in distribution, to $(\|\q_i\|_2^2/t) \cdot \sum_{i=1}^t g_i^2$, where $g_1, \ldots, g_t$ are independent $N(0,1)$
random variables. 

The random variable $\sum_{i=1}^t g_i^2$ is $\chi^2$ with $t$ degree of freedom. The following tail bounds are known.
\begin{fact}(Lemma 1 of \cite{lm00})\label{fact:c2}
Let $g_1, \ldots, g_t$ be i.i.d. $N(0,1)$ random variables. Then for any $x \geq 0$,
$$\Pr[\sum_{i=1}^t g_i^2 \geq t + 2\sqrt{tx} + 2x] \leq \exp(-x),$$
and
$$\Pr[\sum_{i=1}^t g_i^2 \leq t - 2\sqrt{tx}] \leq \exp(-x).$$
\end{fact}
Setting $x = \eps^2 t/16$, we have that
$$\Pr[|\sum_{i=1}^t g_i^2 - t| \leq \eps t] \leq 2\exp(-\eps^2 t/16).$$
For $t = O(\log n / \eps^2)$, the lemma follows by a union bound over $i \in [n]$.
\end{proof}

\begin{theorem}\label{thm:fastLS}(\cite{CW13})
Fix any constant $\beta \in (0,1)$. 
If $p$ is the leverage score distribution of an $n \times k$ matrix $\matZ$ with orthonormal columns, it is possible to compute
a distribution $q$ on the $n$ rows for which with probability $9/10$, simultaneously 
for all $i \in [n]$, $q_i \geq \beta p_i$. The time complexity is $O(\nnz(A)\log n) + \poly(k)$. 
\end{theorem}
\begin{proof}
Let $\matS$ be a sparse embedding matrix with $r = O(k^2/\gamma^2)$ rows for a constant $\gamma \in (0,1)$ to be specified.  
We can compute $\matS \cdot \matA$ in $O(\nnz(A))$ time. We then compute a QR-factorization of $\matS \matA = \matQ \cdot \matR$, where
$\matQ$ has orthonormal columns. This takes $O(rk^2) = \poly(k/\gamma)$ time. Note that $\matR^{-1}$ is $k \times k$, and can be computed
from $\matR$ in $O(k^3)$ time (or faster using fast matrix multiplication). 

For $t = O(\log n / \gamma^2)$, 
let $G$ by a $k \times t$ matrix of i.i.d. $N(0,1/t)$ random variables. Set $q_i = \|\e_i^T\matA \matR^{-1}\matG\|_2^2$ for all $i \in [n]$.
While we cannot compute $\matA \cdot \matR^{-1}$ very efficiently, we can first compute $\matR^{-1} \matG$ in $O(k^2 \log n / \gamma^2)$
by standard matrix multiplication, and then compute $\matA \cdot (\matR^{-1} \matG)$ in $O(\nnz(A) \log n / \gamma^2)$ time since
$(\matR^{-1} \mat G)$ has a small number of columns. 
Since we will set $\gamma$ to be a constant, the overall time complexity of the theorem follows. 

For correctness, by Lemma \ref{lem:jl}, 
with probability $1-1/n$, simultaneously for all $i \in [n]$, $q_i \geq (1-\gamma) \|\e_i^T \matA \matR^{-1}\|_2^2$, which we condition on.
We now show that $\|\e_i^T \matA \matR^{-1}\|_2^2$ is approximately $p_i$. To do so, first consider $\matA\matR^{-1}$. The claim is that
all of the singular values of $\matA\matR^{-1}$ are in the range $[1-\gamma, 1+\gamma]$. To see this, note that for any $\x \in \mathbb{R}^k$,
\begin{eqnarray*}
\|\matA\matR^{-1}x\|_2^2 & = & (1 \pm \gamma) \|\matS \matA \matR^{-1}\x\|_2^2\\
& = & (1\pm \gamma)\|\matQ \x\|_2^2\\
& = & (1\pm \gamma)\|\x\|_2^2,
\end{eqnarray*}
where the first equality follows since with probability $99/100$, $\matS$ is a $(1 \pm \gamma)$ $\ell_2$-subspace embedding for $\matA$,
while the second equality uses the definition of $\matR$, and the third equality uses that $\matQ$ has orthonormal columns. 

Next, if $\matU$ is an orthonormal basis for the column space of $\matA$, since $\matA\matR^{-1}$ and $\matU$ have the same column
space, $\matU = \matA\matR^{-1}\matT$ for a $k \times k$ change of basis matrix $\matT$. The claim is that the minimum singular
value of $\matT$ is at least $1-2\gamma$. Indeed, since all of the singular values of $\matA\matR^{-1}$ are in 
the range $[1-\gamma, 1+\gamma]$, if there were a singular value of $\matT$ smaller than $1-2\gamma$ with corresponding right singular
vector $\ve$, then $\|\matA \matR^{-1} \matT \ve\|_2^2 \leq (1-2\gamma)(1+\gamma) < 1$, but $\|\matA \matR^{-1} \matT \ve\|_2^2 = \|\matU \ve\|_2^2
= 1$, a contradiction. 

Finally, it follows that for all $i \in [n]$, 
$$\|\e_i^T\matA\matR^{-1}\|_2^2 = \|e_i^T\matU\matT^{-1}\|_2^2 \geq (1-2\gamma)\|e_i^T\matU\|_2^2 = (1-2\gamma)p_i.$$
Hence, $q_i \geq (1-\gamma)(1-2\gamma)p_i$, which for an appropriate choice of constant $\gamma \in (0,1)$, achieves
$q_i \geq \beta p_i$, as desired. 
\end{proof}

\subsection{Regression}\label{sec:regression}
We formally define the regression problem as follows. 
\begin{definition}\label{prob:regression}
In the {\it $\ell_2$-Regression Problem}, an $n \times d$ matrix $\matA$ and an 
$n \times 1$ column vector $\b$ are given, together with an approximation parameter $\eps \in [0,1)$.  
The goal is to output a vector $\x$ so that
$$\norm{\matA\x-\b} \leq (1+\eps)\min_{\x' \in \mathbb{R}^d}\norm{\matA\x'-\b}.$$
\end{definition}

The following theorem is an immediate application of $\ell_2$-subspace embeddings. The proof actually shows
that there is a direct relationship between the time complexity of computing an $\ell_2$-subspace embedding 
and the time complexity of approximately solving $\ell_2$-regression. We give one instantiation of this relationship 
in the following theorem statement.
\begin{theorem}\label{thm:regression}
The $\ell_2$-Regression Problem can be solved with probability $.99$ in $O(\nnz(A)) + \poly(d/\eps)$ time.
\end{theorem}
\begin{proof}
Consider the at most $(d+1)$-dimensional subspace $L$ of $\mathbb{R}^n$ spanned by the columns of $\matA$ 
together with the vector $\b$. Suppose we choose $\matS$ to be a sparse embedding matrix with $r = d^2/\eps^2 \poly(\log(d/\eps))$ rows.
By Theorem \ref{thm:cw}, we have that with probability $.99$, 
\begin{eqnarray}\label{eqn:seBound}
\forall \y \in L, \ \|\matS\y\|_2^2 = (1 \pm \eps)\|\y\|_2^2.
\end{eqnarray}
It follows that we can compute $\matS \cdot \matA$ followed by $\matS \cdot \b$, and then let
$$\x = \argmin_{\x' \in \mathbb{R}^d} \|\matS \matA\x' - \matS \b\|_2.$$
By (\ref{eqn:seBound}), it follows that $\x$ solves the $\ell_2$-Regression Problem. The number of rows
of $\matS$ can be improved to $r = O(d/\eps^2)$ by applying Theorem \ref{thm:mmnn}. 
\end{proof}

We note that Theorem \ref{thm:regression} can immediately be generalized to other versions of regression,
such as {\it constrained regression}. In this problem there is a constraint subset $\mathcal{C} \subseteq \mathbb{R}^d$
and the goal is, given an $n \times d$ matrix $\matA$ and an $n \times 1$ column vector $\b$, to output a vector $\x$
for which
$$\norm{\matA\x-\b} \leq (1+\eps)\min_{\x' \in \mathcal{C}}\norm{\matA\x'-\b}.$$
Inspecting the simple proof of Theorem \ref{thm:regression} we see that (\ref{eqn:seBound}) in particular implies
\begin{eqnarray}\label{eqn:seConstraint}
\forall \x \in \mathcal{C}, \ \|\matS(\matA\x - \b)\|_2 = (1 \pm \eps)\|\matA \x -\b\|_2,
\end{eqnarray}
from which we have the following corollary. This corollary follows by replacing $\matA$ and $\b$
with $\matS\matA$ and $\matS\b$, where $\matS$ has $O(d^2/\eps^2)$ rows using Theorem \ref{thm:mmnn}. 
\begin{corollary}\label{cor:regression}
The constrained $\ell_2$ Regression Problem with constraint set $\mathcal{C}$
can be solved with probability $.99$ in $O(\nnz(A)) + T(d,\eps)$ time, 
where $T(d,\eps)$ is the time to solve constrained $\ell_2$ regression with constraint set $\mathcal{C}$ when $\matA$
has $O(d^2/\eps^2)$ rows and $d$ columns. 
\end{corollary}

It is also possible to obtain a better dependence on $\eps$ than given by Theorem \ref{thm:regression} 
and Corollary \ref{cor:regression} in both the time and space, due to the fact that it is 
possible to choose the sparse subspace embedding
$\matS$ to have only $O(d^2/\eps)$ rows. We present this as its own separate theorem. We only state the time bound
for unconstrained regression. 

The proof is due to Sarl\'{o}s \cite{S06}. The key concept in the proof is that of the {\it normal equations}, which state
for the optimal solution $\x$, $\matA^T\matA\x = \matA^T\b$, or equivalently, $\matA^T(\matA \x - \b) = \mat0$, that is,
$\matA \x -\b$ is orthogonal to the column space of $\matA$. This is easy to see from the fact that the optimal
solution $\x$ is such that $\matA \x$ is the projection of $\b$ onto the column space of $\matA$, which is the closest
point of $\b$ in the column space of $\matA$ in Euclidean distance. 

\begin{theorem}\label{thm:1eps}
If $\matS$ is a sparse subspace embedding with $O(d^2/\eps)$ rows, then with probability $.99$, the solution
$\min_{\x' \in \mathbb{R}^d}\|\matS\matA\x' -\matS \b\|_2
= (1 \pm \eps) \min_{\x \in \mathbb{R}^d}\|\matA\x-\b\|_2.$
\end{theorem}
\begin{proof}
Let $\x'$ be $\textrm{argmin}_{\x' \in \mathbb{R}^d} \|\matS(\matA \x' -\b)\|_2$, and 
let $\x$ be $\textrm{argmin}_{\x \in \mathbb{R}^d} \|\matA \x-\b\|_2$. It will be useful to reparameterize the problem
in terms of an orthonormal basis $\matU$ for the column space of $\matA$. Let $\matU \y' = \matA \x'$ and
$\matU \y = \matA \x$. 

Because of the normal equations, we may apply the Pythagorean theorem,
$$\|\matA\x'-\b\|_2^2 = \|\matA\x-\b\|_2^2 + \|\matA\x' - \matA\x \|_2^2,$$
which in our new parameterization is,
$$\|\matU\y'-\b\|_2^2 = \|\matU \y - \b\|_2^2 + \|\matU(\y'-\y)\|_2^2.$$
It suffices to show
$\|\matU(\y' - \y)\|_2^2 = O(\eps)\|\matU\y-\b\|_2^2$, as then the theorem will follow by rescaling $\eps$ by a constant 
factor. Since $\matU$ has orthonormal columns, it suffices to show $\|\y'-\y\|_2^2 = O(\eps) \|\matU\y-\b\|_2^2.$

Conditioned on $\matS$ being a $(1 \pm 1/2)$ $\ell_2$-subspace embedding,
which by Theorem \ref{thm:mmnn} occurs with probability $.999$ for an $\matS$ with an appropriate $O(d^2/\eps)$ number
of rows, we have 
\begin{eqnarray}\label{eqn:seReg}
\|\matU^T \matS^T \matS \matU - \matI_d\|_2 \leq \frac{1}{2}.
\end{eqnarray}

Hence,
\begin{eqnarray*}
\|\y'-\y\|_2 & \leq & \|\matU^T \matS^T \matS \matU(\y'-\y)\|_2 + \|\matU^T \matS^T \matS \matU (\y'-\y) - \y'-\y\|_2\\
& \leq & \|\matU^T \matS^T \matS \matU(\y'-\y)\|_2 + \|\matU^T \matS^T \matS \matU - \matI_d\|_2 \cdot \|(\y'-\y)\|_2\\
& \leq & \|\matU^T \matS^T \matS \matU(\y'-\y)\|_2 + \frac{1}{2} \cdot \|\y'-\y\|_2,
\end{eqnarray*}
where the first inequality is the triangle inequality, the second inequality uses the sub-multiplicativity of the spectral norm, 
and the third inequality uses (\ref{eqn:seReg}). Rearranging, we have
\begin{eqnarray}\label{eqn:apmp}
\|\y'-\y\|_2 \leq 2 \|\matU^T \matS^T \matS \matU(\y'-\y)\|_2.
\end{eqnarray}
By the normal equations in the sketch space, 
$$\matU^T \matS^T \matS \matU \y' = \matU^T \matS^T \matS \b,$$
and so plugging into (\ref{eqn:apmp}),
\begin{eqnarray}\label{eqn:apmp2}
\|\y'-\y\|_2 \leq 2 \|\matU^T \matS^T \matS (\matU\y - \b)\|_2.
\end{eqnarray}
By the normal equations in the original space, $\matU^T (\matU\y-\b) = \mat0_{\textrm{rank}(\matA) \times 1}$.
By Theorem \ref{thm:tz}, $\matS$ has the $(\eps, \delta, 2)$-JL moment property, and so by Theorem \ref{thm:jlamp},
$$\Pr_{\matS} \left [\FNorm{\matU^T \matS^T \matS (\matU\y-\b)} 
> 3 \frac{\sqrt{\eps}}{d} \FNorm{\matU} \FNorm{\matU\y-\b} \right ] \leq \frac{1}{1000}.$$
Since $\FNorm{\matU} \leq \sqrt{d}$, it follows that with probability $.999$, 
$\FNorm{\matU^T \matS^T \matS (\matU\y-\b)} \leq 3 \sqrt{\eps} \|\matU\y-\b\|_2$, and plugging into (\ref{eqn:apmp2}),
together with a union bound over the two probability $.999$ events, completes the proof. 
\end{proof}

\subsection{Machine precision regression}\label{sec:machinePrecision}
Here we show how to reduce the dependence on $\eps$ to logarithmic in the regression application, 
following the approaches in \cite{RT08,amt10,CW13}. 

A classical approach to finding $\min_{\x} \norm{\matA \x- \b}$
is to solve the normal equations 
$\matA^\top \matA \x = \matA^\top \b$ via Gaussian elimination;
for $\matA\in\R^{n\times r}$ and $\b\in\R^{n\times 1}$,
this requires $O(\nnz(\matA))$ time to
form $\matA^\top \b$, $O(r\nnz(\matA))$ time to form $\matA^\top \matA$,
and $O(r^3)$ time to solve the resulting linear systems.
(Another method is to factor $\matA= \matQ \matW$,
where $\matQ$ has orthonormal columns and $\matW$ is
upper triangular; this typically trades a slowdown for a higher-quality solution.)

Another approach to regression is to apply an iterative method
from the general class of Krylov or conjugate-gradient type algorithms to a 
pre-conditioned version of the problem. In such methods,
an estimate $\x^{(m)}$ of a solution is maintained,
for iterations $m=0,1\ldots$,
using data obtained from previous iterations.
The convergence of these methods depends
on the \emph{condition number}
$\kappa(\matA^\top \matA)= \frac{\sup_{\x,\norm{\x}=1} \norm{\matA \x}^2}{\inf_{\x,\norm{\x}=1} \norm{\matA \x}^2}$
from the input matrix.
A classical result (\cite{Luen} via \cite{MSM} or Theorem 10.2.6,\cite{GvL}),
is that
\begin{equation}\label{eq:CG accuracy}
\frac{\norm{\matA( \x^{(m)} - \x^*)}^2}{\norm{\matA(\x^{(0)} - \x^*)}^2}
	\le 2\left(\frac{\sqrt{\kappa(\matA^\top \matA)} - 1}{\sqrt{\kappa(\matA^\top \matA)} + 1}\right)^m.
\end{equation}
Thus the running time of CG-like methods, such as {\tt CGNR} \cite{GvL},
depends on the (unknown)
condition number. The running time per iteration is the time needed
to compute matrix vector products $\ve = \matA \x$ and $\matA^T \ve$,
plus $O(n+d)$ for vector arithmetic, or $O(\nnz(\matA))$.

Pre-conditioning reduces the number of iterations needed for a given accuracy:
suppose
for a non-singular matrix $\matR$, the condition number $\kappa(\matR^\top \matA^\top \matA \matR)$
is small. Then a conjugate gradient method applied to $\matA \matR$ would converge quickly,
and moreover for iterate $\y^{(m)}$ that has error $\alpha^{(m)} \equiv \norm{\matA \matR \y^{(m)} - \b}$
small, the corresponding $\x\gets \matR \y^{(m)}$ would have $\norm{\matA \x- \b} = \alpha^{(m)}$.
The running time per iteration would have an
additional $O(d^2)$ for computing products involving $\matR$.

Suppose we apply a sparse subspace embedding matrix $\matS$
to $\matA$, and $\matR$ is computed so that $\matS \matA \matR$  has
orthonormal columns, e.g., via a QR-decomposition of $\matS \matA$. 
If $\matS$ is an $\ell_2$-subspace
embedding matrix to constant accuracy $\eps_0$,
for all unit $\x\in\R^d$,
$\norm{\matA \matR \x}^2=(1\pm\eps_0)\norm{\matS \matA \matR \x}^2= (1\pm \eps_0)^2$.
It follows that the condition number
\[
\kappa(\matR^\top \matA^\top \matA \matR)
	\le \frac{(1+\eps_0)^2}{(1-\eps_0)^2}.
\]
That is, $\matA \matR$ is well-conditioned. Plugging this
bound into \eqref{eq:CG accuracy}, after $m$ iterations
$\norm{\matA \matR(\x^{(m)} - \x^*)}^2$ is at most $2\eps_0^m$
times its starting value.

Thus starting with a solution $\x^{(0)}$ with 
relative error at most 1, and applying $1+\log(1/\eps)$ iterations
of a CG-like method with $\eps_0 = 1/e$, the relative error is reduced to $\eps$
and the work is $O((\nnz(\matA)+ d^2)\log(1/\eps))$, 
plus the work to find $\matR$. We have

\begin{theorem}\label{thm:it reg}
The $\ell_2$-regression problem can be solved up to a $(1+\eps)$-factor with probability at least 
$99/100$ in
\[
O(\nnz(\matA)\log (n/\eps) + d^3 \log^2 d + d^2\log(1/\eps))
\]
time.
\end{theorem}
%
%

The matrix $\matA \matR$ is so well-conditioned that a simple iterative improvement scheme
has the same running time up to a constant factor. Again start with a solution $\x^{(0)}$ with 
relative error at most 1, and for $m\ge 0$,
let $\x^{(m+1)} \gets \x^{(m)} + \matR^\top \matA^\top (\b - \matA \matR \x^{(m)})$.
Then using the normal equations,
\begin{align*}
\matA \matR(\x^{(m+1)} - \x^*)
	  & = \matA \matR(\x^{(m)} + \matR^\top \matA^\top (\b - \matA \matR \x^{(m)}) - \x^*)
	\\ & = (\matA\matR - \matA \matR \matR^\top \matA^\top \matA \matR ) (\x^{(m)} - \x^*)
	\\ & = \matU(\mat\Sigma - \mat\Sigma^3)\matV^\top (\x^{(m)} - \x^*),
\end{align*}
where $\matA \matR= \matU \mat\Sigma \matV^\top$ is the SVD of $\matA \matR$.

For all unit $\x\in\R^d$,
$\norm{\matA \matR \x}^2 = (1\pm \eps_0)^2$, and so
we have that all singular values $\sigma_i$ of $\matA \matR$ are $1\pm\eps_0$,
and the diagonal entries of $\mat\Sigma - \mat\Sigma^3$
are all at most $\sigma_i(1- (1 - \varepsilon_0)^2) \le \sigma_i 3\varepsilon_0$ for 
$\varepsilon_0\le 1$. Hence 
\begin{align*}
\norm{\matA \matR(\x^{(m+1)} - \x^*)}
	 & \le 3\eps_0 \norm{\matA \matR(\x^{(m)} - \x^*)},
\end{align*}
and by choosing $\varepsilon_0= 1/2$, say, $O(\log(1/\eps))$ iterations suffice for this
scheme also to attain $\eps$ relative error.

\subsection{Polynomial fitting}\label{sec:polynomialFit}
A natural question is if {\it additional structure in $\matA$} can be non-trivially exploited
to further accelerate the running time of $\ell_2$-regression. Given that $\matA$ is structured,
perhaps we can run in time even faster than $O(\nnz(\matA))$. This was studied in \cite{asw13,anw14},
and we shall present the result in \cite{asw13}. 

Perhaps one of the oldest
regression problems is polynomial fitting. In this case, given a set of samples 
$(z_i, b_i) \in \mathbb{R} \times \mathbb{R}$, for $i = 1, 2, \ldots, n$, we would like to
choose coefficients $\beta_0, \beta_1, \ldots, \beta_q$ of a degree-$q$ univariate polynomial
$b = \sum_{i=1}^q \beta_i z^i$ which best fits our samples. Setting this up as a regression
problem, the corresponding matrix $\matA$ is $n \times (q+1)$ and is a {\it Vandermonde matrix}.
Despite the fact that $\matA$ may be dense, we could hope to solve regression in time faster
than $O(\nnz(\matA)) = O(nq)$ using its Vandermonde structure. 

We now describe the problem more precisely, starting with a definition. 

\begin{definition}(Vandermonde Matrix) Let $x_0, x_1, \ldots, x_{n-1}$ be real numbers. The Vandermonde matrix,
denoted $\matV_{n,q}(x_0, x_1, \ldots, x_{n-1})$, has the form:
\[ \matV_{n,q}(x_1, x_1, \ldots, x_{n-1}) = \left ( \begin{array}{cccc}
1 & x_0 & \ldots & x_0^{q-1}\\
1 & x_1 & \ldots & x_1^{q-1}\\
\ldots & \ldots & \ldots & \ldots\\
1 & x_{n-1} & \ldots & x_{n-1}^{q-1} \end{array} \right ) \]
\end{definition}

Vandermonde matrices of dimension $n \times q$ require only $O(n)$
implicit storage and admit $O(n~\log^2 q)$ matrix-vector multiplication time (see, e.g., Theorem 2.11 of \cite{Tang2004}). 
It is also possible to consider block-Vandermonde matrices as in \cite{asw13}; for simplicity
we will only focus on the simplest polynomial fitting problem here, in which Vandermonde
matrices suffice for the discussion. 

We consider regression problems of the form $\min_{\x \in \mathbb{R}^q} \|\matV_{n,q}\x-\b\|_2$, or the approximate version,
where we would like to output an $\x' \in \mathbb{R}^q$ for which 
$$\|\matV_{n,q}\x'-\b\|_2 \leq (1+\eps) \min_{\x \in \mathbb{R}^d} \|\matV_{n,q}\x-\b\|_2.$$
We call this the {\it $\ell_2$-Polynomial Fitting Problem}.

\begin{theorem}($\ell_2$-Polynomial Fitting)\cite{asw13}\label{thm:FastVand}
There is an algorithm that solves the $\ell_2$-Polynomial Fitting Problem
in time $O(n \log^2~q ) + \poly(q\varepsilon^{-1})$. By
combining sketching methods with preconditioned iterative solvers, we can
also obtain logarithmic dependence on $\varepsilon$.  
\end{theorem}
Note that since $\nnz(\matV_{n,q}) = nq$ and the running time of Theorem \ref{thm:FastVand}
is $O(n\log^2 q)$, this provides a sketching approach that operates faster than ``input-sparsity'' time. It
is also possible to extend Theorem \ref{thm:FastVand} to $\ell_1$-regression, see \cite{asw13} for details. 

The basic intuition behind Theorem \ref{thm:FastVand} is to try to compute $\matS \cdot \matV_{n,q}$ for
a sparse embedding matrix $\matS$. Naively, this would take $O(nq)$ time. However, since $\matS$ contains
a single non-zero entry per column, we can actually think of the product $\mat S \cdot \matV_{n,q}$ as
$r$ {\it vector-matrix products} $x^1 \cdot \matV^1_{n,q}, \ldots, x^r \cdot \matV^r_{n,q}$, where $x^i$ is the vector
with coordinates $j \in [n]$ for which $h(j) = i$, and $\matV^i_{n,q}$ is the row-submatrix of $\matV_{n,q}$
consisting only of those rows $j \in [n]$ for which $h(j) = i$.  To compute each of these vector-matrix products,
we can now appeal to the fast matrix-vector multiplication algorithm associated with Vandermonde matrices, which
is similar to the Fast Fourier Transform. Thus, we can compute each $x^i \cdot \matV^i_{n,q}$ in time proportional
to the number of rows of $\matV^i_{n,q}$, times a factor of $\log^2 q$. In total we can compute all matrix-vector
products in $O(n \log^2 q)$ time, thereby computing $\matS \matV_{n,q}$, which we know is an $\ell_2$-subspace
embedding. We can also compute $\matS \b$ in $O(n)$ time, and now can solve the sketched problem
$\min_x \|\matS\matV_{n,q}x-\matS\b\|_2$ in $\poly(q/\eps)$ time. 

\section{Least Absolute Deviation Regression}\label{chap:robust}
While least squares regression is arguably the most used form of regression in practice, it
has certain non-robustness properties that make it unsuitable for some applications. For
example, oftentimes the noise in a regression problem is drawn from a normal distribution, in 
which case least squares regression would work quite well, but if there is noise due to 
measurement error or a different underlying noise distribution, the least squares 
regression solution may overfit this noise since the cost function squares each of its 
summands. 

A more robust alternative is least absolute deviation regression, or $\ell_1$-regression, 
$\min_{\x} \|\matA \x- \b\|_1 = \sum_{i=1}^n |\b_i - \langle \matA_{i,*}, \x \rangle |$. The $\ell_1$-norm
is much less well-behaved than the $\ell_2$-norm, e.g., it is not invariant under rotation,
not everywhere differentiable, etc. There is also no closed-form solution for an
$\ell_1$-regression problem in general, as a special case of it is the geometric median
or Fermat-Weber problem, for which there is no closed form solution. 

Nevertheless, $\ell_1$-regression is much less sensitive to outliers. It is also the maximum
likelihood estimator (MLE) when the noise in the regression problem is i.i.d. Laplacian of
zero median. In this section we will focus on recent advances in solving $\ell_1$-regression
using sketching. To do so, we first describe a sampling-based solution. We note that many of the
results in this section generalize to $\ell_p$-regression for $p > 1$. See 
\cite{CW13,MM13,wz13} for works on this. This general line of work was introduced by
Clarkson \cite{Cla05}, though our exposition will mostly follow that of \cite{DDHKM09} and the
sketching speedups built on top of it \cite{sw11,CDMMMW13,MM13,wz13}. 

{\bf Section Overview:} In \S\ref{sec:sampling-Based} we show how one can adapt the idea of leverage score sampling in \S\ref{sec:leverage} for $\ell_2$ to provide an initial sampling-based algorithm for $\ell_1$-regression. In \S\ref{sec:l1Subspace} we introduce the notion of a subspace embedding for the $\ell_1$-norm and show how if we had such an object, it could be used in the context of $\ell_1$-regression. We postpone one technical detail in this application to \S\ref{sec:gaussian}, which shows how to combine $\ell_1$-subspace embeddings together with Gaussian sketching to make the technique of using $\ell_1$-subspace embeddings in \S\ref{sec:l1Subspace} efficient. In \S\ref{sec:cauchy} we turn to the task of constructing $\ell_1$-subspace embeddings. We do this using Cauchy random variables. This leads to an $\ell_1$-regression algorithm running in $O(nd^2\log d) + \poly(d/\eps)$. In \S\ref{sec:exponential} we then speed this up even further by replacing the dense matrix of Cauchy random variables in the previous section with a product of a sparse $\ell_2$-subspace embedding and a diagonal matrix of exponential random variables. This leads to an overall time of $O(\nnz(A) \log ) + \poly(d/\eps)$. Finally, in \S\ref{sec:hyperplane} we discuss one application of $\ell_1$-regression to $\ell_1$-Hyperplane Approximation. 

\subsection{Sampling-Based solution}\label{sec:sampling-Based}
One of the most natural ways of solving a regression problem is by sampling. Let us augment
the $n \times d$ design matrix $\matA$ in the regression problem to an $n \times (d+1)$ matrix
by including the $\b$ vector as the $(d+1)$-st column. 

Let $p \in [0,1]^n$. 
Suppose we form a submatrix of $\matA$ by including each row of $\matA$ in the submatrix independently
with probability $p_i$. Let us write this as $\matS \cdot \matA$, where $\matS$ is a diagonal $n \times n$ matrix
with $\matS_{i,i} = 1/p_i$ if row $i$ was included in the sample, and $\matS_{i,i} = 0$ otherwise. Then
${\bf E}[\matS \cdot \matA] = \matA$, and so for any fixed $\x$, ${\bf E}[\matS \cdot \matA \x] = \matA \x$. 

What we would like is that for all 
\begin{eqnarray}\label{eqn:l1ose}
\forall \x \in \mathbb{R}^{d+1}, \ \|\matS \cdot \matA \x\|_1 = (1 \pm \eps) \|\matA \x\|_1,
\end{eqnarray}
that is, 
$\matS$ is an oblivious subspace embedding for $\matA$. Note that although $\matS \cdot \matA$ is an $n \times d$
matrix, in expectation it has only $r = \sum_{i=1}^n p_i$ non-zero rows, and so we can throw away
all of the zero rows. It follows that if $r$ is small,
one could then afford to directly solve the {\it constrained} regression problem:
$$\min_{\x \in \mathbb{R}^{d+1}, \x_{d+1} = -1} \|\matS \matA \x\|_1,$$
using linear programming. This would now take time $\poly(r,d)$, which is a significant savings over
solving the problem directly, e.g., if $r$ is much smaller or independent of $n$. Note that the constraint
$\x_{d+1} = -1$ can be directly incorporated into a linear program for solving this problem, and only
slightly increases its complexity. 

We are left with the task of showing (\ref{eqn:l1ose}). To do so, fix a particular vector
$\x \in \mathbb{R}^d$. Define the random variable $Z_i = |\matA_{i,*} \x|/p_i$, so that for $Z = \sum_{i=1}^n Z_i$,
we have ${\bf E}[Z] = \|\matA \x\|_1$. We would like to understand how large each $Z_i$ can be, and
what the variance of $Z$ is. We would like these quantities to be small, which at first glance
seems hard since $p$ cannot depend on $\x$. 

One way of bounding $Z_i$ is to write $\matA = \matU \cdot \mat\tau$ for an $n \times d$ matrix $\matU$ and
a $d \times d$ change of basis matrix $\mat\tau$. Since $\matU$ does not depend on any particular
vector $\x$, one could hope to define $p$ in terms of $\matU$ for a particularly good choice of basis
$\matU$ for the column space of $\matA$. Note that one has
\begin{eqnarray}\label{eqn:bound}
|\matA_{i,*},\x|/p_i = |\matU_{i,*} \mat\tau \x|/p_i \leq \|\matU_{i,*}\|_1 \cdot \|\mat\tau \x\|_{\infty} / p_i,
\end{eqnarray}
where the inequality follows by H\"older's inequality. 

A natural choice at this point to bound the RHS of (\ref{eqn:bound}) is to define
$p_i = \min(1, r \cdot \frac{\|\matU_{i,*}\|}{\sum_{j=1}^n \|\matU_{j,*}\|_1})$, where recall $r$ is about the expected
number of rows we wish to sample (the expected number of rows sampled may be less than $r$ since $p_i$ is a 
probability and so is upper-bounded by $1$). For later purposes, it will be helpful to instead
allow 
$$p_i \geq \min(1, \zeta \cdot r \cdot \frac{\|\matU_{i,*}\|}{\sum_{j=1}^n \|\matU_{j,*}\|_1}),$$ 
where $\zeta \in (0,1]$
can be thought of as a relaxation parameter which will allow for more
efficient algorithms. 

Note that for those $i$ for which 
$\zeta \cdot r \cdot \frac{\|\matU_{i,*}\|}{\sum_{j=1}^n \|\matU_{j,*}\|_1} \geq 1$, 
the $i$-th row $\matA_{i,*}$ will always be included
in the sample, and therefore will not affect the variance of the sampling process. 

Let us now consider those $i$ for which 
$ \zeta \cdot r \cdot \frac{\|\matU_{i,*}\|}{\sum_{j=1}^n \|\matU_{j,*}\|_1} < 1$. For such $i$ 
one has
\begin{eqnarray}\label{eqn:wcb}
Z_i = |\matA_{i,*}, \x|/p_i \leq (\sum_{j=1}^n \|\matU_{j,*}\|_1) \cdot \|\mat\tau \x\|_{\infty}/(r \zeta)
= \alpha \cdot \beta \|\matA \x\|_1 / (r \zeta),
\end{eqnarray}
where $\alpha = \sum_{j=1}^n \|\matU_{j,*}\|_1$ and $\beta = \sup_{\x} \frac{\|\mat\tau \x\|_{\infty}}{\|\matA \x\|_1}.$

In order for $Z_i$ to never be too large, 
we would like to choose a $\matU$ so that $\alpha$ and $\beta$ are as small as possible. This
motivates the following definition of a well-conditioned basis for the $1$-norm. For ease of notation,
let $\|\matU\|_1 = \sum_{j=1}^n \|\matU_{j,*}\|_1$. 

\begin{definition}(Well-conditioned basis for the $1$-norm)\label{def:wcb}(see \cite{DDHKM09})
Let $\matA$ be an $n \times d$ matrix. An $n \times d$ matrix $\matU$ is an $(\alpha, \beta,1)$-well conditioned
basis for the column space of $\matA$ if (1) $\|\matU\|_1 \leq \alpha$, and (2) for all $\x \in \mathbb{R}^d$,
$\|\x\|_{\infty} \leq \beta \|\matU \x\|_1$. 
\end{definition}
Note that our definition of $\alpha$ and $\beta$ above coincide with that in Definition \ref{def:wcb},
in particular the definition of $\beta$, 
since $\|\matU(\mat\tau \x)\|_1 = \|\matA \x\|_1$ by definition of $\matU$ and $\mat\tau$. 

Fortunately, well-conditioned bases with $\alpha, \beta \leq \poly(d)$ exist and can be efficiently computed. We
will sometimes simply refer to $\matU$ as a well-conditioned basis if $\alpha$ and $\beta$ are both bounded by $\poly(d)$.
That such bases exist is due to a theorem of Auerbach \cite{a30,ak08}, 
which shows that $\alpha = d$ and $\beta = 1$ suffice. 
However,
we are not aware of an efficient algorithm which achieves these values. The
first efficient algorithm for finding a well-conditioned basis is due to Clarkson \cite{Cla05}, 
who achieved a running time of $O(nd^5 \log n) + \poly(d)$. 
The same running time was achieved by 
Dasgupta et al. \cite{DDHKM09}, who improved the concrete values of $\alpha$ and $\beta$. 
We will see that one can in fact compute such bases much faster using sketching techniques below, but let us first
see how these results already suffice to solve $\ell_1$-regression in $O(nd^5 \log n) + \poly(d/\eps)$ time. 

Returning to (\ref{eqn:wcb}), we have the bound
$$Z_i = |\matA_{i,*}, \x|/p_i \leq \poly(d) \|\matA \x\|_1/(r \zeta).$$
Using this bound together with independence of the sampled rows, 
\begin{eqnarray*}
{\bf Var}[Z] & = & \sum_{i=1}^n {\bf Var}[Z_i] = \sum_{i \mid p_i < 1} {\bf Var}[Z_i] \leq \sum_{i \mid p_i < 1} {\bf E}[Z_i^2]
= \sum_{i \mid p_i < 1} \frac{|\matA_{i,*}, \x|^2}{p_i}\\
& \leq & \max_{i \mid p_i < 1} \frac{|\matA_{i,*}, \x|}{p_i} \sum_{i \mid p_i < 1} |\matA_{i,*}, \x|\\
& \leq & \frac{\poly(d) \|\matA \x\|_1^2}{r \zeta}.
\end{eqnarray*}
We have computed ${\bf E}[Z]$ and bounded ${\bf Var}[Z]$ as well as $\max_{i \mid p_i < 1} Z_i$, 
and can now use strong tail bounds to bound the deviation of
$Z$ from its expectation. We use the following tail inequalities.

\begin{theorem}[Bernstein inequality \cite{m03}]
\label{thm:bernstein}
Let $Z_i\ge 0$ be independent random variables with 
$\sum_{i}\Exp[Z_i^2]<\infty$, and define
$Z=\sum_{i}Z_i$. 
Then, for any $t>0$,
$$\Pr[Z\le \Exp[Z]-t]\le\exp\left(\frac{-t^2}{2\sum_i\Exp[Z_i^2]}\right).$$
Moreover, if $Z_i - {\bf E}[Z_i] \leq \Delta$ for all $i$, we have
$$\Pr[Z \geq {\bf E}[Z] + \gamma] \leq 
\exp \left (\frac{-\gamma^2}{2{\bf Var}[Z] + 2\gamma \Delta/3} \right ).$$
\end{theorem}
%
%
Plugging our bounds into \ref{thm:bernstein}, we have
\begin{eqnarray*}
\Pr[Z \leq  \|\matA \x\|_1 - \eps \|\matA \x\|_1] & \le & 
\exp \left (\frac{-\eps^2 \|\matA \x\|_1^2r \zeta}{2 \poly(d) \|\matA \x\|_1^2} \right )\\
& \le & \exp \left (\frac{-\eps^2 r \zeta}{2 \poly(d)} \right ),
\end{eqnarray*}
and also
\begin{eqnarray*}
\Pr[Z \geq \|\matA \x\|_1 + \eps \|\matA \x\|_1] & \le & 
\exp \left (\frac{-\eps^2 \|\matA \x\|_1^2}{2\frac{\poly(d) \|\matA \x\|_1^2}{r \zeta} 
+ 2\eps \frac{\|\matA \x\|_1^2 \poly(d)}{3r \zeta}} \right ) \\
& \le & \exp \left (\frac{-\eps^2 r \zeta}{2 \poly(d) + 2 \eps \poly(d)/3} \right ).
\end{eqnarray*}
Setting $r = \eps^{-2} \poly(d)/\zeta$ 
for a large enough polynomial in $d$ allows us to conclude that for any fixed $\x \in \mathbb{R}^d$,
\begin{eqnarray}\label{eqn:tail}
\Pr[Z \in (1 \pm \eps)\|\matA \x\|_1] \geq 1- (\eps/4)^d.
\end{eqnarray}
While this implies that $Z$ is a $(1+\eps)$-approximation with high probability for a fixed $\x$, 
we now need an argument for all $\x \in \mathbb{R}^d$.
To prove that $\|\matS \matA \x\|_1 = (1\pm \eps)\|\matA \x\|_1$ 
for all $\x$, it suffices to prove the statement for all $\y \in \mathbb{R}^n$
for which $\y = \matA \x$ for some $\x \in \mathbb{R}^d$ and $\|\y\|_1 = 1$. 
Indeed, since $\matS \matA$ is a linear map, it will follow that $\|\matS \matA \x\|_1 = (1\pm \eps)\|\matA \x\|_1$
for all $\x$ by linearity. 

Let $\mathcal{B} = \{\y \in \mathbb{R}^n \mid \y 
= \matA \x \textrm{ for some } \x \in \mathbb{R}^d \textrm{ and } \|\y\|_1 = 1\}$. We seek a finite subset of
$\mathcal{B}$, denoted $\mathcal{N}$, which is an $\eps$-net, 
so that if $\|\matS \w\|_1 = (1 \pm \eps)\|\w\|_1$ for all $\w \in \mathcal{N}$, then
it implies that $\|\matS \y\|_1 = (1 \pm \eps)\|\y\|_1$ for all $\y \in \mathcal{B}$. The argument will be similar
to that in \S\ref{sec:se} for the $\ell_2$ norm, though the details are different. 

It suffices to choose $\mathcal{N}$ so that for all $\y \in \mathcal{B}$, 
there exists a vector $\w \in \mathcal{N}$ for which $\|\y-\w\|_1 \leq \eps$.
Indeed, in this case note that 
$$\|\matS \y\|_1 = \|\matS \w + \matS(\y-\w)\|_1 \leq \|\matS \w\|_1 + \|\matS(\y-\w)\|_1 \leq 1+\eps + \|\matS(\y-\w)\|_1.$$
If $\y-\w = 0$, we are done. Otherwise, suppose $\alpha$ is such that $\alpha \|\y-\w\|_1 = 1$. 
Observe that $\alpha \geq 1/\eps$, since, $\|\y-\w\|_1 \leq \eps$ yet $\alpha \|\y-\w\|_1 = 1$. 

Then $\alpha(\y-\w) \in \mathcal{B}$, and we can
choose a vector $\w^2 \in \mathcal{N}$ for which $\|\alpha(\y-\w) - \w^2\|_1 \leq \eps$, or equivalently,
$\|\y-\w - \w^2/\alpha\|_1 \leq \eps/\alpha \leq \eps^2$. Hence,
\begin{eqnarray*}
\|\matS(\y-\w)\|_1 &= &\|\matS \w^2/\alpha + \matS(\y-\w-\w^2/\alpha)\|_1\\
&\leq & (1+\eps)/\alpha + \|\matS(\y-\w-\w^2/\alpha)\|_1.
\end{eqnarray*}
Repeating this argument, we inductively have that 
$$\|\matS \y\|_1 \leq \sum_{i \geq 0} (1+\eps) \eps^i \leq (1+\eps)/(1-\eps) \leq 1+O(\eps).$$
By a similar argument, we also have that
$$\|\matS \y\|_1 \geq 1-O(\eps).$$
Thus, by rescaling $\eps$ by a constant factor, we have that $\|\matS \y\|_1 = 1 \pm \eps$ 
for all vectors $\y \in \mathcal{B}$. 

\begin{lemma}\label{lem:eps-net}
There exists an $\eps$-net $\mathcal{N}$ for which $|\mathcal{N}| \leq (2/\eps)^d$. 
\end{lemma}
\begin{proof}
For a parameter $\gamma$ and point $\p \in \mathbb{R}^n$, 
define $$B(\p, \gamma, \matA) = \{\q = \matA \x \textrm{ for some } \x \textrm{ and } \|\p-\q\|_1 \leq \gamma\}.$$
Then $B(\eps, 0)$ is a $d$-dimensional polytope with a ($d$-dimensional) volume denoted $|B(\eps, 0)|$. 
Moreover, 
$B(1,0)$ and $B(\eps/2, 0)$ are similar polytopes, namely, $B(1,0) = (2/\eps)B(\eps/2,0)$. 
As such, $|B(1,0)| = (2/\eps)^d |B(\eps/2,0)|$. 

Let $\mathcal{N}$ be a maximal subset of $\y \in \mathbb{R}^n$ in the column space of $\matA$ for which $\|\y\|_1 = 1$ and 
for all $\y \neq \y' \in \mathcal{N}$, $\|\y-\y'\|_1 > \eps$. Since $\mathcal{N}$ is maximal, it follows that for all 
$\y \in \mathcal{B}$,
there exists a vector $\w \in \mathcal{N}$ for which $\|\y-\w\|_1 \leq \eps$. 
Moreover, for all $\y \neq \y' \in \mathcal{N}$, 
$B(\y, \eps/2, \matA)$ and $B(\y',\eps/2,\matA)$ are disjoint, 
as otherwise by the triangle inequality, $\|\y-\y'\|_1 \leq \eps$, a contradicition. 
It follows by the previous paragraph that $\mathcal{N}$ can contain at most $(2/\eps)^d$ points. 
\end{proof}

By applying (\ref{eqn:tail}) and a union bound over the points in $\mathcal{N}$, and rescaling $\eps$ by a constant factor, we have thus
shown the following theorem.
\begin{theorem}\label{thm:slow}
The above sampling algorithm is such that with probability at least $1-2^{-d}$, simultaneously for all $\x \in \mathbb{R}^d$,
$\|\matS \matA \x\|_1 = (1\pm \eps)\|\matA \x\|_1$. The expected number of non-zero rows of $\matS \matA$ is at most $r = \eps^{-2} \poly(d) / \zeta$. The 
overall time complexity is $T_{wcb} + \poly(d/\eps)$, where $T_{wcb}$ is the time to compute a well-conditioned basis. Setting
$T_{wcb} = O(nd^5 \log n)$ suffices. 
\end{theorem}

\subsection{The Role of subspace embeddings for L1-Regression}\label{sec:l1Subspace}
The time complexity of the sampling-based algorithm for $\ell_1$-Regression in the previous section is dominated by the computation of 
a well-conditioned basis. In this section we will design subspace embeddings with respect to the $\ell_1$-norm and show how they 
can be used to speed up this computation. Unlike for $\ell_2$, the distortion of our vectors in our subspace will not be $1+\eps$, but
rather a larger factor that depends on $d$. Still, the distortion does not depend on $n$, and this will be sufficient for our applications. This
will be because, with this weaker distortion, we will still be able to form a well-conditioned basis, and then we can apply 
Theorem \ref{thm:slow} to obtain a $(1+\eps)$-approximation to $\ell_1$-regression. 

\begin{definition}(Subspace Embedding for the $\ell_1$-Norm)
We will say a matrix $\matS$ is an $\ell_1$-subspace embedding for an $n \times d$ matrix 
$\matA$ if there are constants $c_1, c_2 > 0$ so that for all $\x \in \mathbb{R}^d$,
$$\|\matA \x\|_1 \leq \|\matS \matA \x\|_1 \leq d^{c_1} \|\matA \x\|_1,$$
and $\matS$ has at most $d^{c_2}$ rows. 
\end{definition}
Before discussing the existence of such embeddings, let us see how they can be used to speed up the computation of a well-conditioned basis.
\begin{lemma}\label{lem:wcbse}(\cite{sw11})
Suppose $\matS$ is an $\ell_1$-subspace embedding for an $n \times d$ matrix $\matA$. Let $\matQ \cdot \matR$ 
be a QR-decomposition of $\matS \matA$, i.e., $\matQ$ has orthonormal
columns (in the standard $\ell_2$ sense) and $\matQ \cdot \matR = \matS \matA$. Then $\matA \matR^{-1}$ 
is a $(d^{c_1+c_2/2+1}, 1, 1)$-well-conditioned basis.
\end{lemma}
\begin{proof}
We have 
\begin{eqnarray*}
\alpha & = & \|\matA \matR^{-1}\|_1 = \sum_{i=1}^d \|\matA \matR^{-1}e_i\|_1 \leq d^{c_1} \sum_{i=1}^d \|\matS \matA \matR^{-1} \e_i\|_1\\
& \leq & d^{c_1 + c_2/2} \sum_{i=1}^d \|\matS \matA \matR^{-1} \e_i\|_2\\
& \leq & d^{c_1+c_2/2} \sum_{i=1}^d \|\matQ \e_i\|_2\\
& = & d^{c_1+c_2/2+1}.
\end{eqnarray*}
Recall we must bound $\beta$, where $\beta$ is minimal for which for all $\x \in \mathbb{R}^d$,
$\|\x\|_{\infty} \leq \beta \|\matA \matR^{-1}\x\|_1$. We have
\begin{eqnarray*}
\|\matA \matR^{-1} \x\|_1 & \geq & \|\matS \matA \matR^{-1} \x\|_1\\
& \geq & \|\matS \matA \matR^{-1} \x\|_2\\
& = & \|\matQ \x\|_2\\
& = & \|\x\|_2\\
& \geq & \|\x\|_{\infty},
\end{eqnarray*}
and so $\beta = 1$. 
\end{proof}
Note that $\matS \matA$ is a $d^{c_2} \times d$ matrix, and therefore its QR decomposition can be computed in $O(d^{c_2+2})$ time. One can also
compute $\matA \cdot \matR^{-1}$ in $O(nd^2)$ time, which could be sped up with fast matrix multiplication, though we will see a better way
of speeding this up below. By Lemma \ref{lem:wcbse}, provided $\matS$ is a subspace embedding for $\matA$ with constants $c_1, c_2 > 0$, $\matA \matR^{-1}$
is a $(d^{c_1+c_2/2+1}, 1, 1)$-well-conditioned basis, and so we can improve
the time complexity of Theorem \ref{thm:slow} to $T_{mm} + O(nd^2) + \poly(d/\eps)$, 
where $T_{mm}$ is the time to compute the matrix-matrix
product $\matS \cdot \matA$. 

We are thus left with the task of producing an $\ell_1$-subspace embedding for $\matA$. 
There are many ways to do this non-obliviously \cite{l78,s87,blm89,t90,js03}, 
but they are slower than the time bounds we can achieve using
sketching. 

We show in \S\ref{sec:cauchy} that by using sketching we can achieve 
$T_{mm} = O(nd^2 \log d)$, which illustrates several main ideas and improves upon Theorem \ref{thm:slow}. 
We will then show how to improve
this to $T_{mm} = O(\nnz(\matA))$ in \S\ref{sec:exponential}. 
Before doing so, let us first see how, given $\matR$ for which $\matA \matR^{-1}$
is well-conditioned, we can 
improve the $O(nd^2)$ time for computing a representation of 
$\matA \matR^{-1}$ which is sufficient to perform the sampling in Theorem \ref{thm:slow}. 

\subsection{Gaussian sketching to speed up sampling}\label{sec:gaussian}
Lemma \ref{lem:wcbse} shows that if $\matS$ is an $\ell_1$-subspace embedding for an $n \times d$ matrix $\matA$, 
and $\matQ \cdot \matR$ is a QR-decomposition of $\matS \matA$, then $\matA \matR^{-1}$ is a well-conditioned basis.

Computing $\matA \matR^{-1}$, on the other hand, naively takes $O(\nnz(\matA) d)$ time. 
However, observe that to do the sampling
in Theorem \ref{thm:slow}, we just need to be able to compute the probabilities $p_i$, for $i \in [n]$, where recall
\begin{eqnarray}\label{sampleBound}
p_i \geq \min(1, \zeta \cdot r \cdot \frac{\|\matU_{i,*}\|_1}{\sum_{j=1}^n \|\matU_{j,*}\|_1}),
\end{eqnarray}
where $\zeta \in (0,1]$, and $\matU = \matA \matR^{-1}$ 
is the well-conditioned basis. This is where $\zeta$ comes in to the picture.

%
%
Instead of computing the matrix product $\matA \cdot \matR^{-1}$ directly, one can choose a $d \times t$ matrix
$\matG$ of i.i.d. $N(0,1/t)$ random variables, for $t = O(\log n)$ 
and first compute $\matR^{-1} \cdot \matG$. This matrix can be 
computed in $O(t d^2)$ time and only has $t$ columns, and so now 
computing $\matA \matR^{-1} \matG = \matA \cdot (\matR^{-1} \cdot \matG)$ 
can be computed in $O(\nnz(\matA) t) = O(\nnz(\matA) \log n)$ time. By choosing
the parameter $\eps = 1/2$ of Lemma \ref{lem:jl} we have
for all $i \in [n]$, that $\frac{1}{2} \|(\matA \matR^{-1})_i\|_2 \leq \|(\matA \matR^{-1} \matG)_i\|_2 
\leq 2 \|(\matA \matR^{-1})_i\|_2.$ Therefore,
$$\sum_{j=1}^n \|(\matA \matR^{-1} \matG)_j\|_1 \leq \sqrt{d} \sum_{j=1}^n \|(\matA \matR^{-1} \matG)_j\|_2 
\leq 2\sqrt{d} \sum_{j=1}^n \|(\matA \matR^{-1})_j\|_1,$$ and
also $$\|(\matA \matR^{-1} \matG)_j\|_1 \geq \|(\matA \matR^{-1} \matG)_j\|_2 
\geq \frac{1}{2} \|(\matA \matR^{-1})_j\|_2 \geq \frac{1}{2\sqrt{d}} \|(\matA \matR^{-1})_j\|_1.$$
It follows that for 
$$p_i = \min(1, r \cdot \frac{\|(\matA \matR^{-1} \matG)_i\|_1}{\sum_{j=1}^n \|(\matA \matR^{-1} \matG)_j\|_1} ),$$
we have that (\ref{sampleBound}) holds with $\zeta = 1/(4d)$. 

We note that a tighter anaylsis is possible, in which $\matG$ need only have $O(\log(d \eps^{-1} \log n))$ columns, as shown 
in \cite{CDMMMW13}. 

\subsection{Subspace embeddings using cauchy random variables}\label{sec:cauchy}
The Cauchy distribution, having density function 
$p(x)=\frac{1}{\pi} \cdot \frac{1}{1+x^2}$, 
is the unique $1$-stable distribution.
That is to say, if $C_1,\ldots,C_M$ are independent Cauchys, then 
$\sum_{i\in[M]}\gamma_iC_i$ is distributed as a Cauchy scaled by
$\gamma=\sum_{i\in[M]}|\gamma_i|$.

The absolute value of a Cauchy distribution has density function $f(x) = 2p(x) = \frac{2}{\pi} \frac{1}{1+x^2}$.
The cumulative distribution function $F(z)$ of it is
$$F(z) = \int_0^z f(z) dz = \frac{2}{\pi} \arctan(z).$$
Note also that since $\tan(\pi/4) = 1$, we have $F(1) = 1/2$, so that $1$ is the median of this distribution. 

Although Cauchy random variables do not have an expectation, and have infinite variance, some control over them can be
obtained by clipping them. The first use of such a truncation 
technique in algorithmic applications that we are aware of
is due to Indyk \cite{i06}. 

\begin{lemma}\label{lem:truncated}
Consider the event $\mathcal{E}$ that a Cauchy random variable $X$ satisfies
$|X| \leq M$, for some parameter $M \geq 2$. Then there is a constant $c> 0$ for which 
$\Pr[\mathcal{E}] \geq 1- \frac{2}{\pi M}$ and ${\bf E}[|X| \mid \mathcal{E}] \leq c \log M,$
where $c > 0$ is an absolute constant.  
\end{lemma}
\begin{proof}
\begin{eqnarray}\label{eqn:condition}
\Pr[\mathcal{E}] = F(M) = \frac{2}{\pi} \tan^{-1}(M) = 1-\frac{2}{\pi} \tan^{-1} \left (\frac{1}{M} \right ) \geq 1- \frac{2}{\pi M}.
\end{eqnarray}
Hence, for $M \geq 2$, 
\begin{eqnarray*}\label{eqn:expectation}
{\bf E}[|X| \mid \mathcal{E}] = \frac{1}{\Pr[\mathcal{E}]} \int_{0}^M \frac{2}{\pi} \frac{x}{1+x^2}
= \frac{1}{\Pr[\mathcal{E}]} \frac{1}{\pi} \log(1+M^2)
\leq C \log M,
\end{eqnarray*}
where the final bound uses (\ref{eqn:condition}). 
\end{proof}
We will show in Theorem \ref{thm:l1embed} below 
that a matrix of i.i.d. Cauchy random variables is an $\ell_1$-subspace embedding. Interestingly, we will use the existence of a
well-conditioned basis in the proof, though we will not need an algorithm for constructing it. This lets us use well-conditioned bases with
slightly better parameters. In particular, we will use the following Auerbach basis.

\begin{definition}(see ``Connection to Auerbach bases'' in Section 3.1 of \cite{DDHKM09})\label{thm:auerbach}
There exists a $(d, 1, 1)$-well-conditioned basis. 
\end{definition}

For readability, it is useful to separate out the following key lemma that is used in Theorem \ref{thm:l1embed} below.
This analysis largely follows that in \cite{sw11}. 
\begin{lemma}\label{lem:key}(Fixed Sum of Dilations)
Let $\matS$ be an $r \times n$ matrix of i.i.d. Cauchy random variables, and let $\y_1, \ldots, \y_d$ be $d$ arbitrary vectors
in $\mathbb{R}^n$. Then 
$$\Pr[\sum_{i=1}^d \|\matS \y_i\|_1 \geq C r \log(rd) \sum_{i=1}^d \|\y_i\|_1] \leq \frac{1}{100},$$
where $C > 0$ is an absolute constant. 
\end{lemma}
\begin{proof}
Let the rows of $\matS$ be denoted $\matS_{1*}, \ldots, \matS_{r*}$. 
For $i = 1, \ldots, r$, let $\mathcal{F}_i$ be the event that 
$$\forall j \in [d], \ |\langle \matS_{i*}, \y_j \rangle | \leq C'rd \|\y_j\|_1,$$
where $C'$ is a sufficiently large positive constant. Note that by the $1$-stability of the Cauchy distribution,
$\langle \matS_{i*}, \y_j \rangle$ is distributed as $\|\y_j\|_1$ times a Cauchy random variable. 
By Lemma \ref{lem:truncated} applied to $M = C'rd$, together with a union bound, we
have $$\Pr[\mathcal{F}_i] \geq 1 - d \cdot \frac{2}{\pi C'rd} \geq 1 - \frac{2}{\pi C' r}.$$
Letting $\mathcal{F} = \wedge_{i=1}^r \mathcal{F}_i$, we have by another union bound that
$$\Pr[\mathcal{F}] \geq 1 - \frac{2r}{\pi C' r} = 1 - \frac{2}{\pi C'}.$$

Given $\mathcal{F}$, we would then like to appeal
to Lemma \ref{lem:truncated} to bound the expectations ${\bf E}[|\langle \matS_{i*}, \y_j \rangle| \mid \mathcal{F}]$. 
The issue
is that the expectation bound in Lemma \ref{lem:truncated} cannot be applied, since the condition
$\mathcal{F}$ additionally conditions $\matS_{i*}$ through the remaining columns $\matA_{*j'}$ 
for $j' \neq j$. A first observation
is that by independence, we have
$${\bf E}[|\langle \matS_{i*}, \y_j \rangle| \mid \mathcal{F}] 
= {\bf E}[|\langle \matS_{i*}, \y_j \rangle| \mid \mathcal{F}_i].$$
We also know from Lemma \ref{lem:truncated} that if $\mathcal{F}_{i,j}$ is the event that
$|\langle \matS_{i*} \y_j \rangle | \leq C'rd \|\y_j\|_1$, then 
${\bf E}[|\langle \matS_{i*}, \y_j \rangle| \mid \mathcal{F}_{i,j}] \leq C \log(C'rd) \|\y_j\|_1$, where $C$ is the constant
of that lemma.

We can perform the following manipulation (for an event $\mathcal{A}$, we use the notation
$\neg \mathcal{A}$ to denote the occurrence of the complement of $\mathcal{A}$):
\begin{eqnarray*}
C \log(C'rd) \|\y_j\|_1 & = & {\bf E}[|\langle \matS_{i*}, \y_j \rangle | \mid \mathcal{F}_{i,j}]\\
& = & {\bf E}[|\langle \matS_{i*}, \y_j \rangle | \mid \mathcal{F}_i] \cdot \Pr[\mathcal{F}_i \mid \mathcal{F}_{i,j}]\\
& + & {\bf E}[\langle \matS_{i*}, \y_j \rangle | \mid \mathcal{F}_{i,j} \wedge \neg \mathcal{F}_i] \cdot \Pr[\neg \mathcal{F}_i \mid \mathcal{F}_{i,j}]\\
& \geq & {\bf E}[|\langle \matS_{i*}, \y_j \rangle | \mid \mathcal{F}_i] \cdot \Pr[\mathcal{F}_i \mid \mathcal{F}_{i,j}].
\end{eqnarray*}
We also have 
$$\Pr[\mathcal{F}_i \mid \mathcal{F}_{i,j}] \cdot \Pr[\mathcal{F}_{i,j}]
= \Pr[\mathcal{F}_i] \geq 1-O(1/(C'r)),$$
and $\Pr[\mathcal{F}_{i,j}] \geq 1-O(1/(C'rd))$. Combining these two, we have
\begin{eqnarray}\label{eqn:depend}
\Pr[\mathcal{F}_i \mid \mathcal{F}_{i,j}] \geq \frac{1}{2},
\end{eqnarray}
for $C' > 0$ a sufficiently large constant. Plugging (\ref{eqn:depend}) into the above,
$$C \log(C'rd) \|\y_j\|_1 
\geq  {\bf E}[|\langle \matS_{i*}, \y_j \rangle | \mid \mathcal{F}_i] \cdot \Pr[\mathcal{F}_i \mid \mathcal{F}_j] \cdot \frac{1}{2},$$
or equivalently, 
\begin{eqnarray}\label{eqn:below}
{\bf E}[|\langle \matS_{i*}, \y_j \rangle| \mid \mathcal{F}] 
= {\bf E}[|\langle \matS_{i*}, \y_j \rangle| \mid \mathcal{F}_i] 
\leq C \log(C'rd) \|\y_j\|_1,
\end{eqnarray}
as desired. 

We thus have, combining (\ref{eqn:below}) with Markov's inequality,
\begin{eqnarray*}
&& \Pr[\sum_{j=1}^d \|\matS \y_j\|_1 \geq r C' \log(C'rd) \sum_{ji=1}^d \|\y_j\|_1]\\
& \leq & \Pr[\neg \mathcal{F}] + 
\Pr[\sum_{j=1}^d \|\matS \y_j\|_1 \geq r C' \log(C'rd) \sum_{j=1}^d \|\y_j\|_1 \mid \mathcal{F}]\\
& \leq & \frac{2}{\pi C'} + 
\frac{{\bf E}[\sum_{j=1}^d \|\matS \y_j\|_1 \mid \mathcal{F}] }{r C' \log(C'rd) \sum_{j=1}^d \|\y_j\|_1}\\
& = & \frac{2}{\pi C'} + 
\frac{\sum_{i=1}^r \sum_{j=1}^d {\bf E}[|\langle \matS_{i*} \y_j \rangle | \mid \mathcal{F}]}
{r C' \log(C'rd) \sum_{j=1}^d \|\y_j\|_1}\\
& \leq & \frac{2}{\pi C'} + \frac{r C \log(C'rd)}{r C' \log(C'rd)}\\
& \leq & \frac{2}{\pi C'} + \frac{C}{C'}.
\end{eqnarray*}
As $C'$ can be chosen sufficiently large, while $C$ is the fixed constant of Lemma \ref{lem:truncated},
we have that
$$\Pr[\sum_{j=1}^d \|\matS \y_j\|_1 \geq r C' \log(C'rd) \sum_{ji=1}^d \|\y_j\|_1] \leq \frac{1}{100}.$$
The lemma now follows by appropriately setting the constant $C$ in the lemma statement. 
\end{proof}

\begin{theorem}\label{thm:l1embed}
A matrix $\matS$ 
of i.i.d. Cauchy random variables with 
$r = O(d \log d)$ rows is an $\ell_1$-subspace embedding with constant probability, that is,
with probability at least $9/10$ simultaneously 
for all $x$, $$\|\matA \x\|_1 \leq 4\|\matS \matA \x\|_1/r = O(d \log d) \|\matA \x\|_1.$$
\end{theorem}
\begin{proof}
Since we will show that with probability $9/10$, for all $\x$ we have 
$\|\matA \x\|_1 \leq 3\|\matS \matA \x\|_1/r \leq C d \log d \|\matA \x\|_1,$ we are free to
choose whichever basis of the column space of $\matA$ that we like. 
In particular, we can assume the $d$ columns $\matA_{*1}, \ldots, \matA_{*d}$ 
of $\matA$ form an Auerbach basis. We will first bound the dilation, and then bound the contraction.
\\\\
{\bf Dilation:}
We apply Lemma \ref{lem:key} with $\y^i = \matA_{*i}$ for $i = 1, 2, \ldots, d$. 
We have with probability at least $99/100$,
\begin{eqnarray}\label{eqn:apply}
\sum_{j=1}^d \|\matS \y_j\|_1 \leq r C \log(rd) \sum_{j=1}^d \|\y_j\|_1 = r C d \log(rd),
\end{eqnarray}
where the last equality used that $\y^1, \ldots, \y^d$ is an Auerbach basis. 

Now let $\y = \matA \x$ be an arbitrary vector in the column space of $\matA$. Then,
\begin{eqnarray*}
\|\matS \y\|_1 & = & \sum_{j=1}^d \|\matS \matA_{*j} \cdot \x_j\|_1\\
& \leq & \sum_{j=1}^d \|\matS \matA_{*j}\|_1 \cdot |\x_j|\\
& \leq & \|\x\|_{\infty} \sum_{j=1}^d \|\matS \matA_{*j}\|_1\\
& \leq & \|\x\|_{\infty} r C d \log(rd)\\ 
& \leq & \|\matA \x\|_1 r C d \log(rd),
\end{eqnarray*}
where the third inequality uses (\ref{eqn:apply}) and the fourth inequality uses a property of 
$\matA$ being a $(d,1,1)$-well-conditioned 
basis. It follows that $4\|\matS \matA \x\|_1/r \leq 4C d \log(rd) \|\matA \x\|_1$, as needed in the statement of the theorem.
\\\\
{\bf Contraction:} 
We now argue that no vector's norm shrinks by more than a constant factor. 
Let $\y = \matA \x$ be an arbitrary vector in the column space of $\matA$. By the $1$-stability of the Cauchy distribution, 
each entry of $\matS \y$ is distributed as a Cauchy scaled by $\|\y\|_1$. 

Since the median of the distribution of the absolute value of a Cauchy random variable is $1$, we have that with 
probability at least $1/2$, $|\langle \matS_i \y \rangle | \geq \|\y\|_1$. Since the entries of $\matS \y$ are independent, 
it follows by a Chernoff bound that the probability that fewer than a $1/3$ fraction
of the entries are smaller than $\|\y\|_1$ is at most $\exp(-r)$. Hence, with probability $1 - \exp(-r)$,
$\|\matS \y\|_1$ is at least $r \|\y\|_1/3$, or equivalently, $4\|\matS \matA \x\|_1/r \geq (4/3) \|\matA \x\|_1$. 

We now use a net argument as in \cite{sw11}. By Lemma \ref{lem:eps-net},
there exists an $\eps$-net $\mathcal{N} \subset \{\matA \x \mid \|\matA \x\|_1 = 1\}$ 
for which $|\mathcal{N}| \leq (24C d \log(rd))^d$ and for any $\y=\matA \x$ with $\|\y\|_1 = 1$, 
there exists a $\w \in \mathcal{N}$ with
$\|\y-\w\|_1 \leq \frac{1}{12C d \log(rd)}$. Observe that for a sufficiently
large $r = O(d \log d)$ number of rows of $\matS$, we have by a union bound, that with probability
$1-\exp(-r) |\mathcal{N}| \geq 99/100$, simultaneously for all $\z \in \mathcal{N}$, 
$4\|\matS \w\|_1/r \geq (4/3) \|\w\|_1$. 

For an arbitrary $\y = \matA \x$ with $\|\y\|_1 = 1$, we can write $\y = \w + (\y-\w)$ for a $\w \in \mathcal{N}$ and 
$\|\y-\w\|_1 \leq \frac{1}{12C d \log(rd)}$. 
By the triangle inequality,  
\begin{eqnarray*}
\frac{4\|\matS \y\|_1}{r} & \geq & \frac{4\|\matS \w\|_1}{r} - \frac{4\|\matS(\y-\w)\|_1}{r}\\
& \geq & \frac{4}{3} \|\w\|_1 - \frac{4\|\matS(\y-\w)\|_1}{r}\\
& = & \frac{4}{3} - \frac{4\|\matS(\y-\w)\|_1}{r}.
\end{eqnarray*}
Since we have already shown that $4\|\matS \matA \x\|_1/r \leq 4C d \log(rd) \|\matA \x\|_1$ for all $\x$, it follows
that 
$$\frac{4\|\matS(\y-\w)\|_1}{r} \leq 4C d \log(rd) \|\y-\w\|_1 \leq \frac{4Cd\log(rd)}{12Cd\log(rd)} \leq \frac{1}{3}.$$
It follows now that $4\|\matS \y\|_1/r \geq 1 = \|\y\|_1$ for all vectors $\y = \matA \x$ with $\|\y\|_1 = 1$. 

Hence, the statement of the theorem holds with probability at least $9/10$, by a union bound over the events in the 
dilation and contraction arguments. This concludes the proof. 
\end{proof}

\begin{corollary}\label{cor:first}
There is an $O(nd^2 + nd \log(d \eps^{-1}\log n)) + \poly(d/\eps)$ time algorithm for solving the $\ell_1$-regression problem
up to a factor of $(1+\eps)$ and with error probability $1/10$. 
\end{corollary}
\begin{proof}
The corollary follows by combining Theorem \ref{thm:slow}, Lemma \ref{lem:wcbse} and its optimization in
\S\ref{sec:gaussian}, and 
Theorem \ref{thm:l1embed}. Indeed, we can compute $\matS \cdot \matA$ in $O(nd^2)$ time, then a QR-factorization as
well as $\matR^{-1}$ in $\poly(d)$ time. Then we can compute $\matA \matR^{-1} \matG$ as well as perform
the sampling in Theorem \ref{thm:slow} in $nd \log(d \eps^{-1} \log n)$ time. Finally, we can solve the $\ell_1$-regression
problem on the samples in $\poly(d/\eps)$ time. 
\end{proof}
While the time complexity of Corollary \ref{cor:first} can be improved to roughly
$O(nd^{1.376}) + \poly(d/\eps)$ using algorithms for fast matrix multiplication, there are better
ways of speeding this up, as we shall see in the next section. 

\subsection{Subspace embeddings using exponential random variables}\label{sec:exponential}
We now describe a speedup over the previous section using exponential random variables, as in \cite{wz13}. Other
speedups are possible, using \cite{CDMMMW13,CW13,MM13}, though the results in \cite{wz13} additionally also 
slightly improve the sampling complexity. The use of exponential random variables in \cite{wz13} is inspired
by an elegant work of Andoni, Onak, and Krauthgamer on frequency moments \cite{AKO11,Andoni12}. 

An exponential distribution has support $x \in [0, \infty)$, probability density function 
$f(x) = e^{-x}$ and cumulative distribution function $F(x) = 1 - e^{-x}$. 
We say a random variable $X$ is exponential if $X$ is chosen from the exponential distribution. 
The exponential distribution has the following max-stability property.
\begin{property}
\label{prop:exp}
If $U_1, \ldots, U_n$ are exponentially distributed, and $\alpha_i > 0\ (i = 1, \ldots, n)$ are real numbers, then 
$\textstyle \max\{\alpha_1 / U_1, \ldots, \alpha_n / U_n\} \simeq \left(\sum_{i \in [n]} \alpha_i \right) \left/ U \right.$,
where $U$ is exponential.
%
\end{property}
%
The following lemma shows a relationship between the Cauchy distribution and the exponential distribution. 
%
%
\begin{lemma}
\label{lem:tail-squared}
Let $y_1, \ldots, y_d \geq 0$ be scalars. 
Let $U_1, \ldots, U_d$ be $d$ independendent exponential random variables, 
and let $X = (\sum_{i \in [d]} y_i^2/U_i^2)^{1/2}$. Let $C_1, \ldots, C_d$ be $d$ independent Cauchy
random variables, and let $Y = (\sum_{i \in [d]} y_i^2 C_i^2)^{1/2}$. 
There is a constant $\gamma > 0$ for which for any $t > 0$. 
$$\Pr[X \geq t] \leq \Pr[Y \geq \gamma t].$$
\end{lemma}
\begin{proof}
We would like the density function $h$ of $y_i^2 C_i^2$. Letting $t = y_i^2 C_i^2$, the inverse function is $C_i = t^{1/2}/y_i$. Taking the derivative,
we have $\frac{dC_i}{dt} = \frac{1}{2y_i} t^{-1/2}$. Letting $f(t) = \frac{2}{\pi} \frac{1}{1+t^2}$ be the density function of the absolute value of a Cauchy random variable,
we have by the change of variable technique,
\begin{eqnarray*}
h(t) & = & \frac{2}{\pi} \frac{1}{1+t/y_i^2} \cdot \frac{1}{2y_i} t^{-1/2} = \frac{1}{\pi} \frac{1}{y_i t^{1/2}+t^{3/2}/y_i}. 
\end{eqnarray*}
We would also like the density function $k$ of $y_i^2 E_i^2$, where $E_i \sim 1/U_i$. Letting $t = y_i^2 E_i^2$, the inverse function is $E_i = t^{1/2}/y_i$. Taking
the derivative, $\frac{dE_i}{dt} = \frac{1}{2y_i} t^{-1/2}$. Letting $g(t) = t^{-2} e^{-1/t}$ be the density function of the reciprocal of an exponential random variable,
we have by the change of variable technique,
\begin{eqnarray*}
k(t) & = & \frac{y_i^2}{t} e^{-y_i/t^{1/2}} \cdot \frac{1}{2y_i} t^{-1/2} = \frac{y_i}{2t^{3/2}} e^{-y_i/t^{1/2}}.
\end{eqnarray*}
We claim that $k(t) \leq h(\gamma t)/\gamma$ for a sufficiently small constant $\gamma > 0$. This is equivalent to showing that
$$\frac{\pi}{2} \frac{y_i}{t^{3/2}} e^{-y_i/t^{1/2}} \gamma \leq \frac{1}{\gamma^{1/2} y_i t^{1/2} + \gamma^{3/2} t^{3/2} / y_i},$$
which for $\gamma < 1$, is implied by showing that
$$\frac{\pi}{2} \frac{y_i}{t^{3/2}} e^{-y_i/t^{1/2}} \gamma \leq \frac{1}{\gamma^{1/2} y_i t^{1/2} + \gamma^{1/2} t^{3/2} / y_i}.$$
We distinguish two cases: first suppose $t \geq y_i^2$. In this case, $e^{-y_i/t^{1/2}} \leq 1$. Note also that 
$y_i t^{1/2} \leq t^{3/2}/y_i$ in this case. Hence,
$\gamma^{1/2} y_i t^{1/2} \leq \gamma^{1/2} t^{3/2}/y_i$. Therefore, the above is implied by showing
$$\frac{\pi}{2} \frac{y_i}{t^{3/2}} \gamma \leq \frac{y_i}{2\gamma^{1/2} t^{3/2}},$$
or 
$$\gamma^{3/2} \leq \frac{1}{\pi},$$
which holds for a sufficiently small constant $\gamma \in (0,1)$. 

Next suppose $t < y_i^2$. In this case $y_i t^{1/2} > t^{3/2}/y_i$, and it suffices to show
$$\frac{\pi}{2} \frac{y_i}{t^{3/2}} e^{-y_i/t^{1/2}} \gamma \leq \frac{1}{2\gamma^{1/2} y_i t^{1/2}},$$
or equivalently,
$$\pi y_i^2 \gamma^{3/2} \leq t e^{y_i/t^{1/2}}.$$
Using that $e^x \geq x^2/2$ for $x \geq 0$, it suffices to show
$$\pi y_i^2 \gamma^{3/2} \leq y_i^2/2,$$
which holds for a small enough $\gamma \in (0,1)$. 

We thus have,
\begin{eqnarray*}
\Pr[X \geq t] & = & \Pr[X^2 \geq t^2]\\
& = & \Pr[\sum_{i=1}^d y_i^2/U_i^2 \geq t^2]\\
& = & \int_{\sum_{i=1}^d t_i \geq t^2} k(t_1) \cdots k(t_d) dt_1 \cdots dt_d\\
& \leq & \int_{\sum_{i=1}^d t_i \geq t^2} \kappa^{-d} h(\kappa t_1) \cdots h(\kappa t_d) dt_1 \cdots dt_d\\
& \leq & \int_{\sum_{i=1}^d s_i \geq \kappa t^2} f(s_1) \cdots f(s_d) ds_1 \cdots ds_d\\
& = & \Pr[Y^2 \geq \kappa t^2]\\
& = & \Pr[Y \geq \kappa^{1/2} t],
\end{eqnarray*}
where we made the change of variables $s_i = \kappa t_i$. Setting $\gamma = \kappa^{1/2}$ completes the proof. 
\end{proof}
We need a bound on $\Pr[Y \geq t]$, where $Y = (\sum_{i \in [d]} y_i^2 C_i^2)^{1/2}$ is 
as in Lemma \ref{lem:tail-squared}.

\begin{lemma}\label{lem:cauchy-l2}
There is a constant $c > 0$ so that for any $r > 0$,
$$\Pr[Y \geq r \|y\|_1] \leq \frac{c}{r}.$$
\end{lemma}
\begin{proof}
For $i \in [d]$, let $\sigma_i \in \{-1,+1\}$ be i.i.d. random variables with $\Pr[\sigma_i = -1] = \Pr[\sigma_i = 1] = 1/2$. 
Let $Z = \sum_{i \in [d]} \sigma_i y_i C_i$. We will obtain tail bounds for $Z$ in two different ways, and use this to establish
the lemma.

On the one hand, by the $1$-stability of the Cauchy distribution, we have that $Z \sim \|y\|_1 C$, where $C$ is a standard
Cauchy random variable. Note that this holds for any fixing of the $\sigma_i$. The cumulative distribution function
of the Cauchy random variable is $F(z) = \frac{2}{\pi} \arctan(z).$ Hence for any $r > 0$,
\begin{eqnarray*}
\Pr[Z \geq r \|y\|_1] & = & \Pr[C \geq r] = 1 - \frac{2}{\pi} \arctan(r).
\end{eqnarray*}
We can use the identity 
$$\arctan(r) + \arctan \left (\frac{1}{r} \right) = \frac{\pi}{2},$$
and therefore using the Taylor series for $\arctan$ for $r > 1$,
$$\arctan(r) \geq \frac{\pi}{2} - \frac{1}{r}.$$
Hence,
\begin{eqnarray}\label{eqn:upper}
\Pr[Z \geq r \|y\|_1] \leq \frac{2}{\pi r}.
\end{eqnarray}
On the other hand, for any fixing of $C_1, \ldots, C_d$, we have
$${\bf E}[Z^2] = \sum_{i \in [d]} y_i^2 C_i^2,$$
and also 
$${\bf E}[Z^4] = 3 \sum_{i \neq j \in [d]} y_i^2 y_j^2 C_i^2 C_j^2 + \sum_{i \in [d]} y_i^4 C_i^4.$$
We recall the Paley-Zygmund inequality.
\begin{fact}
If $R \geq 0$ is a random variable with finite variance, and $0 < \theta < 1$, then
$$\Pr[R \geq \theta {\bf E}[R]] \geq (1-\theta)^2 \cdot \frac{{\bf E}[R]^2}{{\bf E}[R^2]}.$$
\end{fact}
Applying this inequality with $R = Z^2$ and $\theta = 1/2$, we have
\begin{eqnarray*}
\Pr[Z^2 \geq \frac{1}{2} \cdot \sum_{i \in [d]} y_i^2 C_i^2] & \geq & 
\frac{1}{4} \cdot \frac{\left (\sum_{i \in [d]} y_i^2 C_i^2 \right )^2}{3 \sum_{i \neq j \in [d]} y_i^2 y_j^2 C_i^2 C_j^2 + \sum_{i \in [d]} y_i^4 C_i^4}\\
& \geq & \frac{1}{12},
\end{eqnarray*}
or equivalently 
\begin{eqnarray}\label{eqn:applyPZ}
\Pr[Z \geq \frac{1}{\sqrt{2}} (\sum_{i \in [d]} y_i^2 C_i^2)^{1/2}] \geq \frac{1}{12}.
\end{eqnarray}
Suppose, towards a contradiction, that $\Pr[Y \geq r \|y\|_1] \geq c/r$ for a sufficiently large constant $c > 0$. By independence of the $\sigma_i$ and the $C_i$, by
(\ref{eqn:applyPZ}) this implies
$$\Pr[Z \geq \frac{r \|y\|_1}{\sqrt{2}}] \geq \frac{c}{12r}.$$
By (\ref{eqn:upper}), this is a contradiction for $c > \frac{24}{\pi}$. It follows that $\Pr[Y \geq r \|y\|_1] < c/r$, as desired. 
\end{proof} 

\begin{corollary}\label{cor:upper}
Let $y_1, \ldots, y_d \geq 0$ be scalars. 
Let $U_1, \ldots, U_d$ be $d$ independendent exponential random variables, 
and let $X = (\sum_{i \in [d]} y_i^2/U_i^2)^{1/2}$. There is a constant $c > 0$ for which for any $r > 0$,
$$\Pr[X > r \|y\|_1] \leq c/r.$$
\end{corollary}
\begin{proof}
The corollary follows by combining Lemma \ref{lem:tail-squared} with Lemma \ref{lem:cauchy-l2}, and rescaling the constant $c$ from Lemma \ref{lem:cauchy-l2}
by $1/\gamma$, where $\gamma$ is the constant of Lemma \ref{lem:tail-squared}. 
\end{proof}

\begin{theorem}\label{thm:expWork}(\cite{wz13})
Let $\matS$ be an $r \times n$ CountSketch matrix with $r = d \cdot \poly\log d$, and $\matD$ an $n \times n$ diagonal matrix
with i.i.d. entries $\matD_{i,i}$ distributed as a reciprocal of a standard exponential random variable. 
Then, with probability at least $9/10$ simultaneously for all $x$,
$$\Omega\left (\frac{1}{d \log^{3/2} d} \right )\|\matA \x\|_1 \leq \|\matS \matD \matA \x\|_1 \leq O(d \log d) \|\matA \x\|_1.$$
\end{theorem}
\begin{proof}
By Theorem \ref{thm:nn}, with probability at least $99/100$ over the choice of $\matS$, 
$\matS$ is an $\ell_2$-subspace embedding for the matrix $\matD \cdot \matA$, 
that is, simultaneously for all $x \in \mathbb{R}^d$,
$\|\matS \matD \matA \x\|_2 = (1 \pm 1/2) \|\matD \matA \x\|_2$. We condition $\matS$ on this event.

For the dilation, we need Khintchine's inequality.
\begin{fact}\label{fact:khintchine}(\cite{h81}). 
Let $Z = \sum_{i=1}^r \sigma_i z_i$ for i.i.d. random variables $\sigma_i$
uniform in $\{-1,+1\}$, and $z_1, \ldots, z_r$ be scalars. 
There exists a constant $c > 0$ for which for all $t > 0$
$$\Pr[|Z| > t \|\y\|_2] \leq \exp(-ct^2).$$
\end{fact}
Let $\y^1, \ldots, \y^d$ be $d$ vectors in an Auerbach basis for the column
space of $\matA$. 
Applying Fact \ref{fact:khintchine} to a fixed entry $j$ of $\matS \matD \y^i$ for a
fixed $i$, and letting $\z^{i,j}$ denote the vector whose $k$-th coordinate
is $\y^i_k$ if $\matS_{j,k} \neq 0$, and otherwise $\z^{i,j}_k = 0$, 
we have for a constant $c' > 0$, 
$$\Pr[|(\matS \matD \y^i)_j| > c' \sqrt{\log d} \|\matD \z^{i,j}\|_2] \leq \frac{1}{d^3}.$$
By a union bound, with probability 
$$1- \frac{r d}{d^3} = 1-\frac{d^2 \poly \log d}{d^3}
= 1- \frac{\poly \log d}{d},$$ for all $i$ and $j$,
$$|(\matS \matD \y^i)_j| \leq c' \sqrt{\log d} \|\matD \z^{i,j}\|_2,$$
which we denote by event $\mathcal{E}$ and condition on. 
Notice that the probability is taken only over
the choice of the $\sigma_i$, and therefore conditions only the $\sigma_i$ random variables. 

In the following, $i \in [d]$ and $j \in [r]$. 
Let $\mathcal{F}_{i,j}$ be the event that 
$$\|\matD \z^{i,j}\|_2 \leq 100 d r \|\z^{i,j}\|_1,$$
We also
define $$\mathcal{F}_j = \wedge_{i} \mathcal{F}_{i,j}, \ \ 
\mathcal{F} = \wedge_{j} \mathcal{F}_j = \wedge_{i, j} \mathcal{F}_{i,j}.$$
By Corollary \ref{cor:upper} and union bounds,
$$\Pr[\mathcal{F}_j] \geq 1 - \frac{1}{100r},$$
and union-bounding over $j \in [r]$, 
$$\Pr[\mathcal{F}] \geq \frac{99}{100}.$$
We now bound ${\bf E}[\|\matD \z^{i,j}\|_2 \mid \mathcal{E}, \mathcal{F}].$
By independence, 
$${\bf E}[\|\matD \z^{i,j}\|_2 \mid \mathcal{E}, \mathcal{F}]
= {\bf E}[\|\matD \z^{i,j}\|_2 \mid \mathcal{E}, \mathcal{F}_j].$$
Letting $p = \Pr[\mathcal{E} \wedge \mathcal{F}_{i,j}] \geq 99/100$, 
we have by Corollary \ref{cor:upper},
\begin{eqnarray*}
{\bf E}[\|\matD \z^{i,j}\|_2 \mid \mathcal{E}, \mathcal{F}_{i,j}]
& = & \int_{u = 0}^{100dr} 
\Pr[\|\matD \z^{i,j}\|_2 \geq u \|\z^{i,j}\|_1 
\mid \mathcal{E}, \mathcal{F}_{i,j}] \cdot du \\
& \leq & \frac{1}{p} (1 + \int_{u = 1}^{100dr} \frac{c}{u}) \cdot du\\
& \leq & \frac{c}{p} (1 + \ln (100dr)).
\end{eqnarray*}
We can perform the following manipulation:
\begin{eqnarray*}
\frac{c}{p}(1 + \ln(100dr)) & \geq & {\bf E}[\|\matD \z^{i,j}\|_2 \mid \mathcal{E}, \mathcal{F}_{i,j}]\\
& \geq & {\bf E}[\|\matD \z^{i,j}\|_2 \mid \mathcal{E}, \mathcal{F}_j] \cdot \Pr[\mathcal{F}_j \mid \mathcal{F}_{i,j}]\\
& = & {\bf E}[\|\matD \z^{i,j}\|_2 \mid \mathcal{E}, \mathcal{F}_j] \cdot \Pr[\mathcal{F}_j]/ \Pr[\mathcal{F}_{i,j}]\\
& \geq & \frac{1}{2} {\bf E}[\|\matD \z^{i,j}\|_2 \mid \mathcal{E}, \mathcal{F}_j] \cdot \Pr[\mathcal{F}_j]\\
& = & \frac{1}{2} {\bf E}[\|\matD \z^{i,j}\|_2 \mid \mathcal{E}, \mathcal{F}_j] \cdot \Pr[\mathcal{F}].
\end{eqnarray*}
It follows by linearity of expectation that,
\begin{eqnarray*}
{\bf E}[\sum_{i \in [d], j \in [d \poly \log d]} \|\matD \z^{i,j}\|_2 \mid \mathcal{E}, \mathcal{F}]
& \leq & \frac{c}{p} (1 + \ln(100dr)) \sum_{i=1}^d \|\y^i\|_1.
\end{eqnarray*}
Consequently, by a Markov bound, and using that $p \geq 1/2$, 
conditioned on $\mathcal{E} \wedge \mathcal{F}$, with probability at least $9/10$, we have the
occurrence of the event $\mathcal{G}$
that 
\begin{eqnarray}\label{eqn:dilate}
\sum_{i=1}^d \|\matD \y^i\|_2 \leq 10 \frac{c}{p}(1 +  \ln(100dr)) \sum_{i=1}^d \|\y^i\|_1.
\leq 40 c d \ln(100dr)
\end{eqnarray}
To bound the dilation, consider a fixed vector $\x \in \mathbb{R}^d$.
Then conditioned on $\mathcal{E} \wedge \mathcal{F} \wedge \mathcal{G}$,
and for $\matA = [\y^1, \ldots, \y^d]$ 
an Auerbach basis (without loss of generality),
\begin{eqnarray*}
\|\matS \matD \matA \x\|_1 & \leq & \|\x\|_{\infty} \sum_{i = 1}^d \|\matS \matD \y^i\|_1\\
& \leq & \|\matA \x\|_1 \sum_{i = 1}^d \|\matS \matD \y^i\|_1\\
& \leq & \|\matA \x\|_1 c' \sqrt{\log d} 40 cd \ln(100dr)\\
& \leq & c'' d (\log^{3/2} d) \|\matA \x\|_1,
\end{eqnarray*}
where the first inequality follows from the triangle inequality,
the second inequality uses that $\|\x\|_{\infty} \leq \|\matA \x\|_1$ for a
well-conditioned basis $\matA$, the third inequality uses 
(\ref{eqn:dilate}), and in the fourth inequality $c'' > 0$ is a
sufficiently large constant. Thus for all $\x \in \mathbb{R}^d$,
\begin{eqnarray}\label{eqn:dilateFinal}
\|\matS \matD \matA \x\|_1 \leq c'' d (\log^{3/2} d) \|\matA \x\|_1.
\end{eqnarray}

For the contraction, we have
\begin{eqnarray*}
\|\matS \matD \matA \x\|_1 & \geq & \|\matS \matD \matA \x\|_2\\
& \geq & \frac{1}{2} \|\matD \matA \x\|_2 \\
& \geq & \frac{1}{2} \|\matD \matA \x\|_{\infty}\\
& = & \frac{1}{2} \frac{\|\matA\x\|_1}{U},
\end{eqnarray*}
where $U$ is a standard exponential random variables, and 
where the first inequality uses our conditioning on $\matS$, the second inequality uses a standard norm inequality, and
the third inequality uses the max-stability of the exponential distribution. Thus, since the cumulative distribution
of an exponential random variable $F(x) = 1-e^{-x}$, we have that for any fixed $x$,
\begin{eqnarray}\label{eqn:tailExp}
\Pr \left [\|\matS \matD \matA \x\|_1 \geq \frac{1}{4} \frac{\|\matA \x\|_1}{d \log (2d^3)} \right ] \geq 1- (2d^2)^{2d}. 
\end{eqnarray} 
By Lemma \ref{lem:eps-net}, there exists a $\frac{1}{d^3}$-net $\mathcal{N}$ for which $|\mathcal{N}| \leq (2d^3)^d$,
where $\mathcal{N}$ is a subset of 
$\{\y \in \mathbb{R}^n \mid \y = \matA \x \textrm{ for some } \x \in \mathbb{R}^d \textrm{ and }
\|\y\|_1 = 1\}$. Combining this with (\ref{eqn:tailExp}), by a union bound we have the event $\mathcal{E}$ that 
simultaneously for all 
$\w \in \mathcal{N}$, 
$$\|\matS \matD \w\|_1 \geq \frac{1}{4} \frac{\|\w\|_1}{d \log (2d^3)}.$$
Now consider an arbitrary vector $\y$ of the form $\matS \matD \matA \x$ 
with $\|\y\|_1 = 1$. By definition of $\mathcal{N}$, one can
write $\y = \w + (\y-\w)$, where $\w \in \mathcal{N}$ and $\|\y-\w\|_1 \leq \frac{1}{d^3}$. We have,
\begin{eqnarray*}
\|\matS \matD \y\|_1 & \geq & \|\matS \matD \w\|_1 - \|\matS \matD(\y-\w)\|_1\\
& \geq & \frac{1}{4 d \log(2d^3)} - \|\matS\matD(\y-\w)\|_1\\
&  \geq & \frac{1}{4 d \log(2d^3)} - \frac{O(d \log^{3/2} d)}{d^3}\\
& \geq & \frac{1}{8 d \log(2d^3)},
\end{eqnarray*}
where the first inequality uses the triangle inequality, the second the occurrence of $\mathcal{E}$, and the third
(\ref{eqn:dilateFinal}). This completes the proof. 
\end{proof}
\begin{corollary}\label{cor:second}
There is an $O(\nnz(\matA)\log n) + \poly(d/\eps)$ time algorithm for computing $\mat\Pi \matA$, where 
$\mat\Pi$ is a $\poly(d/\eps)$
by $n$ matrix satisfying, with probability at least $9/10$, 
$\|\mat\Pi \matA \x\|_1 = (1 \pm \eps) \|\matA \x\|_1$ for all $x$. Therefore, there is also
an $O(\nnz(\matA) \log n) + \poly(d/\eps)$ time algorithm for solving the $\ell_1$-regression problem
up to a factor of $(1+\eps)$ with error probability $1/10$. 
\end{corollary}
\begin{proof}
The corollary follows by combining Theorem \ref{thm:slow}, Lemma \ref{lem:wcbse} and its optimization 
in \S\ref{sec:gaussian}, and 
Theorem \ref{thm:expWork}. 
\end{proof}

\subsection{Application to hyperplane fitting}\label{sec:hyperplane}
One application of $\ell_1$-regression is to finding the best hyperplane to find a 
set of $n$ points in $\mathbb{R}^d$, presented as an $n \times d$ matrix $\matA$ 
\cite{bd09L1,bdb10,kk03,kk05,sw11,CDMMMW13}. 
One seeks
to find a hyperplane $H$ so that the sum of $\ell_1$-distances of the rows $\matA_{i*}$ to
$H$ is as small as possible. 

While in general, the points on $H$ are those $\x \in \mathbb{R}^d$ for which $\langle \x, \w \rangle = \gamma$,
where $\w$ is the normal vector of $H$ and $\gamma \in \mathbb{R}$, we can in fact assume that 
$\gamma = 0$. Indeed, this follows by increasing the dimension $d$ by one, placing the value $1$ on
all input points in the new coordinate, and placing the value $\gamma$ on the new coordinate in $\w$. As
this will negligibly affect our overall time complexity, we can therefore assume $\gamma = 0$ in what
follows, that is, $H$ contains the origin. 

A nice feature of the $\ell_1$-norm is that if one grows an $\ell_1$-ball around a point $\x \in \mathbb{R}^d$,
it first touches a hyperplane $H$ at a vertex of the $\ell_1$-ball. Hence, there is a coordinate
direction $i \in [d]$ for which the point of closest $\ell_1$-distance to $\x$ on $H$ is obtained by
replacing the $i$-th coordinate of $\x$ by the unique real number $v$ so that 
$(x_1, \ldots, x_{i-1}, v, x_{i+1}, \ldots, x_d)$ is on $H$. 

An interesting observation is that this coordinate direction $i$ only depends on $H$, that is,
it is independent of $\x$, as shown in Corollary 2.3 of \cite{m97}. Let $\matA^{-j}$ denote the matrix
$\matA$ with its $j$-th column removed. Consider a hyperplane $H$ with normal vector $\w$. Let $\w^{-j}$
denote the vector obtained by removing its $j$-th coordinate. Then the
sum of $\ell_1$-distances of the rows $\matA_{i*}$ of $\matA$ to $H$ is given by
$$\min_j \|-\matA^{-j} \w^{-j} - \matA_{*j}\|_1,$$
since $\matA^{-j} \w^{-j}$ is the negation of the vector of $j$-th coordinates of the points projected
(in the $\ell_1$ sense) onto $H$, using that $\langle \w, \x \rangle = 0$ 
for $\x$ on the hyperplane. It follows that an optimal
hyperplane $H$ can be obtained by solving
$$\min_j \min_{\w} \|-\matA^{-j} \w^{-j} -\matA_{*j}\|_1,$$
which characterizes the normal vector $\w$ of $H$. Hence, by solving $d$ $\ell_1$-regression problems,
each up to a $(1+\eps)$-approximation factor and each on an $n \times (d-1)$ matrix, 
one can find a hyperplane whose cost is at most $(1+\eps)$ times the cost of the optimal hyperplane.

One could solve each of the $d$ $\ell_1$-regression problems independently up to $(1+\eps)$-approximation with error
probability $1/d$, each taking $O(\nnz(\matA)\log n) + \poly(d/\eps)$ time. This would lead to
an overall time of $O(\nnz(\matA) d \log n) + \poly(d/\eps)$, but we can do better by reusing computation. 

That is, it suffices to compute a subspace embedding $\mat\Pi \matA$ once, using Corollary \ref{cor:second}, which
takes only $O(\nnz(\matA) \log n) + \poly(\d/\eps)$ time. 

%
For the subspace approximation problem, we can write
\begin{eqnarray*}
\min_j \min_{\w} \|-\matA^{-j} \w^{-j} -\matA_{*j}\|_1
& = & \min_j \min_{\w \mid \w_j = 0} \|-\matA \w - \matA_{*j}\|_1\\
& = & (1 \pm \eps) \min_j \min_{\w \mid \w_j = 0} \|-\mat\Pi \matA \w- \mat\Pi \matA_{*j}\|_1.
\end{eqnarray*}
Thus, having computed $\mat\Pi \matA$ once, one can solve the subspace approximation problem with an
additional $\poly(d/\eps)$ amount of time. We summarize our findings in the following theorem. 

\begin{theorem}($\ell_1$-Hyperplane Approximation)
There is an $O(\nnz(\matA)\log n) + \poly(d/\eps)$ time algorithm for solving the $\ell_1$-Hyperplane
approximation problem with constant probability. 
\end{theorem}

\section{Low Rank Approximation}\label{chap:lowRank}
In this section we study the low rank approximation problem. We are given
an $n \times d$ matrix $\matA$, and would like to find a matrix $\tilde{\matA}_k$
for which 
$$\|\matA-\tilde{\matA}_k\| \leq (1+\eps) \|\matA-\matA_k\|,$$
where $\matA_k$ is the best rank-$k$ approximation to $\matA$ with respect to some
matrix norm, and $\tilde{\matA}_k$ has rank $k$. 

Low rank approximation can be used for a variety of problems, such as
Non-Negative Matrix Factorization (NNMF)~\cite{seung2001algorithms}, 
Latent Dirichlet Allocation (LDA)~\cite{blei2003latent}, and face recognition. 
It has also recently been used for $\ell_2$-error shape-fitting
problems~\cite{dan2013tiny}, such as $k$-means and projective clustering. 

Here we demonstrate an application to latent semantic analysis (LSA). We define a {\it term-document}
matrix $\matA$ in which the rows correpond to terms (e.g., words) and columns correspond to documents. The entry
$\matA_{i,j}$ equals the number of occurrences of term $i$ in document $j$. Two terms $i$ and $j$ can be regarded 
as correlated
if the inner product $\langle \matA_{i,*},\mat A_{j,*} \rangle$ of their corresponding rows of $\matA$ is large. The matrix
$\matA\matA^T$ contains all such inner products. Similarly, one can look at document correlation by looking at $\matA^T\matA$.
By writing $\matA = \matU \mat\Sigma \matV^T$ in its SVD, we have $\matA\matA^T = \matU\mat\Sigma^2 \matU^T$. 

By taking the SVD of low rank approximation $\tilde{\matA}_k$ to a matrix $\matA$, 
one obtains $\tilde{\matA}_k = \matL \matU \matR^T$, where
$\matL$ and $\matR$ have orthonormal columns, 
and $\matU$ is a rank-$k$ matrix. One can view the columns of $\matL$ and $\matR$ as approximations
to the top $k$ left and right singular vectors of $\matA$.  Note that, as we will see below, the algorithm for generating
$\tilde{\matA}_k$ usually generates its factorization into the product of $\matL$, $\matU$, and $\matR^T$ 
so one does not need to perform
an SVD on $\tilde{\matA}_k$ (to achieve $O(\nnz(\matA)) + (n+d)\poly(k/\eps)$ time algorithms for low rank approximation, one
cannot actually afford to write down $\tilde{\matA}_k$ other than in factored form, since $\tilde{\matA}_k$ may be dense). 

There are two well-studied norms in this context, the Frobenius
and the spectral (operator) norm, both of which have the same minimizer 
$\matA_k$ given by the singular value decomposition of $\matA$. That is, if one
writes $\matA = \matU \mat\Sigma \matV^T$ in its SVD, where $\matU$ and $\matV$ are orthonormal
and $\mat\Sigma$ is a non-negative diagonal matrix with 
$\mat\Sigma_{1,1} \geq \mat\Sigma_{2,2} \geq \cdots \mat\Sigma_{n,n} \geq 0$, 
then $\matA_k = \matU \mat\Sigma_k \matV^T$, where $\mat\Sigma_k$ agrees with $\mat\Sigma$ on its
top $k$ diagonal entries, but is $0$ otherwise. Clearly this is a rank-$k$
matrix, and the Eckart-Young Theorem guarantees that it is the minimizer
for any rotationally-invariant norm, which includes the Frobenius and
spectral norms. The top $k$ rows of $\matV^T$ are known as the top $k$ principal 
components of $\matA$. 

We will show how
to use sketching to speed up algorithms for both problems, and further variants.
Our exposition is based on combinations of several works in this area by S\'arlos,
Clarkson, and the author \cite{S06,CW09,CW13}. 

{\bf Section Overview:} In \S\ref{sec:frobenius} we give an algorithm for computing a low rank approximation achieving error proportional to the Frobenius norm. In \S\ref{sec:CUR} we give a different kind of low rank approximation, called a CUR decomposition, which computes a low rank approximation also achieving Frobenius norm error but in which the column space equals the span of a small subset of columns of the input matrix, while the row space equals the span of a small subset of rows of the input matrix. A priori, it is not even clear why such a low rank approximation should exist, but we show that it not only exists, but can be computed in nearly input sparsity time. We also show that it can be computed deterministically in polynomial time. This algorithm requires several detours into a particular kind of spectral sparsification given in \S\ref{sec:bss}, as well as an adaptive sampling technique given in \S\ref{sec:adaptiveSampling}. Finally in \S\ref{sec:CURWrapup} we show how to put the pieces together to obtain the overall algorithm for CUR factorization. One tool we need is a way to compute the best rank-$k$ approximation of the column space of a matrix when it is restricted to lie within a prescribed subspace; we defer the details of this to \S\ref{sec:dislra}, where the tool is developed in the context of an application called Distributed Low Rank Approximation. In \S\ref{sec:spectral} we show how to perform low rank approximation with a stronger guarantee, namely, an error with respect to the spectral norm. While the solution quality is much better than in the case of the Frobenius norm, it is unknown how to compute this as quickly, though one can still compute it much more quickly than the SVD. In \S\ref{sec:dislra} we present the details of the Distributed Low Rank Approximation algorithm.

\subsection{Frobenius norm error}\label{sec:frobenius}
We will say a $k$-dimensional subspace of $\mathbb{R}^d$ spans a 
$(1+\eps)$ rank-$k$ approximation to $\matA$ if 
$$\FNorm{\matA-\matA\matL\matL^T} \leq (1+\eps) \FNorm{\matA-\matA_k},$$
where $\matL\matL^T$ is the projection operator onto that subspace. We will
sometimes abuse notation and refer to $\matL$ as the subspace as well,
meaning the $k$-dimensional subspace of $\mathbb{R}^d$ spanned
by the rows of $\matL^T$. 

One way of interpreting the Frobenius low rank problem is to treat each
of the $n$ rows of $\matA$ as a point in $\mathbb{R}^d$. A 
particularly nice property about the Frobenius norm is that if one is given
a subspace $L$ of $\mathbb{R}^d$ which is guaranteed to contain a
rank-$k$ subspace $L' \subseteq L$ spanning a $(1+\eps)$ rank-$k$
approximation to $\matA$, then it can be found by projecting each of the
rows of $\matA$ onto $\matV$, and then finding the best rank-$k$ approximation
to the projected points inside of $\matV$. This is a simple, but very useful
corollary of the Pythagorean theorem.

\begin{lemma}\label{lem:pythagorean}
The best rank-$k$ approximation to $\matA$ in Frobenius norm in the row space
of a matrix $\matU^T$ with orthonormal rows is given by $[\matA\matU]_k \matU^T$, where 
$[\matA\matU]_k$ denotes the best rank-$k$ approximation to $\matA\matU$.
\end{lemma}
\begin{proof}
Let $\matZ$ be an arbitrary matrix of rank $k$ of the same dimensions as $\matA\matU$. Then,
\begin{eqnarray*}
\FNormS{\matA\matU\matU^T - [\matA\matU]_k \matU^T} & = & \FNormS{\matA\matU - [\matA\matU]_k}\\
& \leq & \FNormS{\matA\matU - \matZ}\\
& = & \FNormS{\matA\matU\matU^T - \matZ\matU^T},
\end{eqnarray*}
where the equalities use that the rows of $\matU^T$ are orthonormal, while the 
inequality uses that $[\matA\matU]_k$ is the best rank-$k$ approximation to $\matA\matU$. 

Hence, 
\begin{eqnarray*}
\FNormS{\matA - [\matA\matU]_k \matU^T} & = & \FNormS{\matA-\matA\matU\matU^T} + \FNormS{\matA\matU\matU^T - [\matA\matU]_k \matU^T}\\
& \leq & \FNormS{\matA-\matA\matU\matU^T} + \FNormS{\matA\matU\matU^T - \matZ\matU^T}\\
& = & \FNormS{\matA-\matZ\matU^T},
\end{eqnarray*}
where the equalities use the Pythagorean theorem and the inequality uses
the bound above. It follows that the best rank-$k$ approximation to $\matA$ in the rowspace of $\matU^T$
is $[\matA\matU]_k \matU^T$. 
\end{proof}

The following lemma shows how to use sketching to find a good space $L$. For a matrix
$\matA$, if its SVD is $\matU \mat\Sigma \matV^T$, then the Moore-Penrose pseudoinverse $\matA^{\dagger}$ of $\matA$
is equal to $\matV \mat\Sigma^{\dagger} \matU^T$, where $\mat\Sigma^{\dagger}$ for a diagonal matrix $\mat\Sigma$
satisfies $\mat\Sigma^{\dagger}_{i,i} = 1/\mat\Sigma_{i,i}$ if $\mat\Sigma_{i,i} > 0$, and is $0$ otherwise. 
\begin{lemma}\label{lem:sketching}
Let $\matS$ be an $\ell_2$-subspace embedding for any fixed $k$-dimensional subspace $M$ 
with probability at least $9/10$, so that
$\|\matS\y\|_2 = (1 \pm 1/3) \|\y\|_2$ for all $\y \in M$. Further, suppose $\matS$ satisfies
the $(\sqrt{\eps/k}, 9/10, \ell)$-JL moment property for some $\ell \geq 2$ of Definition \ref{def:moment}, so
that the conclusion of Theorem \ref{thm:jlamp} holds, namely, that for any fixed matrices $\matA$ and $\matB$ each with 
$k$ rows, 
$$\Pr_{\matS}[\FNorm{\matA^T\matS^T\matS\matB -\matA^T\matB} > 3\sqrt{\eps/k} \FNorm{\matA} \FNorm{\matB}] \leq \frac{1}{10}.$$
Then the rowspace of $\matS\matA$ contains a $(1+\eps)$ rank-$k$ approximation to $\matA$. 
\end{lemma}
\begin{remark}
An example of a matrix $\matS$ having both properties as required by the lemma 
is a sparse embedding matrix with
$O(k^2 + k/\eps)$ rows, as follows by Theorem \ref{thm:mmnn} and Theorem \ref{thm:tz}. One
can also use a matrix $\matS$ of i.i.d. normal random variables with $O(k/\eps)$ rows, 
which follows by Theorem
\ref{thm:gaussianSE} and Theorem \ref{thm:jlamp}; for the latter one needs to show that Definition
\ref{def:moment} is satisfied. Also, the product of subspace embeddings is a 
a subspace embedding, and one can show the product of a sparse embedding matrix with
a matrix of i.i.d. normal random variables satisfies Theorem \ref{thm:jlamp}. One advantage of
using the product is that one obtains fast time complexity as well as $O(k/\eps)$ overall rows.
See, e.g., \cite{CW13}, where the product of a sparse embedding matrix with the Subsampled
Randomized Hadamard Transform was used. 
\end{remark}
\begin{proof}
Let $\matU_k$ denote the $n \times k$ matrix of top $k$ left singular vectors of $\matA$. 
Consider the quantity
\begin{eqnarray}\label{eqn:central}
\FNormS{\matU_k (\matS \matU_k)^{\dagger} \matS \matA - \matA}.
\end{eqnarray}
The goal is to show (\ref{eqn:central}) is at most $(1+\eps)\FNormS{\matA-\matA_k}$. Note that this
implies the lemma, since $\matU_k (\matS \matU_k)^{\dagger} \matS \matA$ is a rank-$k$ matrix inside of the
rowspace of $\matS \matA$. 

Since the columns of $\matA-\matA_k$ are orthogonal to the columns of $\matU_k$, by the matrix
Pythagorean theorem (applied to columns),
\begin{eqnarray*}
&& \FNormS{\matU_k (\matS \matU_k)^{\dagger} \matS \matA - \matA}\\
 & = & \FNormS{\matU_k (\matS\matU_k)^{\dagger} \matS \matA - \matA_k} + \FNormS{\matA-\matA_k}\\
& = & \FNormS{(\matS\matU_k)^{\dagger} \matS \mat A - \mat\Sigma_k \matV_k^T} + \FNormS{\matA-\matA_k},
\end{eqnarray*}
where the second equality uses that the columns of $\matU_k$ are orthonormal, and that $\matA_k = \matU_k\mat\Sigma_k\matV_k^T$. 

It suffices to show 
$\FNormS{(\matS\matU_k)^{\dagger} \matS\matA - \mat\Sigma_k\matV_k^T} = O(\eps) \FNormS{\matA-\matA_k}$. We use 
that $\matA = \matU\mat\Sigma\matV^T = \matU_k\mat\Sigma_k\matV_k^T + \matU_{n-k}\mat\Sigma_{r-k}\matV_{d-k}^T$, where 
$r$ is the rank of $A$, the columns of $\matU_{n-k}$ correspond to the bottom $n-k$ left singular vectors, while the rows of $\matV_{d-k}^T$
correspond to the bottom $d-k$ right singular vectors. Hence, it suffices to show
$\FNormS{(\matS\matU_k)^{\dagger} \matS\matU_k\mat\Sigma_k\matV_k^T + (\matS\matU_k)^{\dagger} \matS\matU_{n-k}\mat\Sigma_{r-k}\matV_{d-k}^T 
- \mat\Sigma_k\matV_k^T} = O(\eps) \FNormS{\matA-\matA_k}$. Now, $(\matS\matU_k)^{\dagger} \matS\matU_k = \matI_k$, and so it suffices to show
$\FNormS{(\matS\matU_k)^{\dagger} \matS\matU_{n-k}\mat\Sigma_{r-k}\matV_{d-k}^T} = O(\eps) \FNormS{\matA-\matA_k}$.

Note that $(\matS \matU_k)^{\dagger}$ and $(\matS \matU_k)^T$ have the same row space, namely $\matU_k^T\matS^T$, 
and so we can
write $(\matS \matU_k)^T = \matB (\matS \matU_k)^{\dagger}$ where $\matB$ is a $k \times k$ change of basis matrix (of full rank). Hence, it 
is equivalent to show 
\begin{eqnarray}\label{eqn:fixBlah}
\FNormS{\matB^{-1}(\matS\matU_k)^T \matS\matU_{n-k}\mat\Sigma_{r-k}\matV_{d-k}^T} = O(\eps) \FNormS{\matA-\matA_k}.
\end{eqnarray}

We seek an upper bound on $\|\matB^{-1}\|_2$, or equivalently a lower bound on $\sigma_r(\matB)$. Since 
$(\matS \matU_k)^{\dagger}$ has full row rank, we can find an $\x$ for which $(\matS \matU_k)^{\dagger}\x = \ve$, where
$\ve$ is the right singular vector of $\matB$ of minimum singular value. 

With probability at least $9/10$, $\matS$ is an $\ell_2$-subspace embedding for the column space
of $\matU_k$, that is $\|\matS\matU_k\z\|_2 = (1\pm 1/3) \|\matU_k\z\|_2$ for all $\z$. 
Since $\matU_k$ has orthonormal columns,
this implies that all of the singular values of $\matS \matU_k$ are in the range $[2/3, 4/3]$.
Hence, the singular values of
$(\matS \matU_k)^{\dagger}$ are in $[4/3, 3/2]$, so we can choose the $\x$ above so that $\|\x\|_2 \in [2/3, 3/4]$. It follows that
$\|\matB (\matS \matU_k)^{\dagger}\x\|_2 = \sigma_r(\matB)$. But $\matB (\matS \matU_k)^{\dagger} = (\matS \matU_k)^T$,
and so $\sigma_r(\matB) \geq \sigma_r( (\matS \matU_k)^T) 2/3$. 
The minimum singular value of $(\matS \matU_k)^T$ is at least $2/3$, and so $\sigma_r(\matB) \geq \frac{4}{9}$, or
equivalently $\|\matB^{-1}\|_2 \leq \frac{9}{4}$. Returning to (\ref{eqn:fixBlah}), it suffices to show
$\FNormS{(\matS\matU_k)^T \matS\matU_{n-k}\mat\Sigma_{r-k}\matV_{d-k}^T} = O(\eps) \FNormS{\matA-\matA_k}$.

Since $\matS$ satisfies the conclusion of Theorem \ref{thm:jlamp}, with probability at least $9/10$,
$$\FNormS{\matU_k^T \matS^T \matS \matU_{n-k}\mat\Sigma_{r-k}\matV_{d-k}^T} \leq 9 \cdot \frac{\eps}{k} \FNormS{\matU_k} \FNormS{\matA-\matA_k} 
\leq 9\eps \FNormS{\matA-\matA_k}.$$
Rescaling $\eps$ by a constant factor completes
the proof. 
\end{proof}
Lemma \ref{lem:pythagorean} and Lemma \ref{lem:sketching} give a natural way of using sketching to 
speed up low rank approximation. 
Namely, given $\matA$, first compute $\matS \cdot \matA$, which is a small number of random linear combinations
of the rows of $\matA$. Using efficient $\ell_2$-subspace embeddings, this can be done in $O(\nnz(\matA))$ time,
and $\matS$ need only have $\tilde{O}(k/\eps)$ rows. 
Next, compute an orthogonal basis $\matU^T$ for the rowspace of $\matS \cdot \matA$, 
which can be done in $\tilde{O}((k/\eps)^2 d)$
time. 

Next, compute $\matA \matU$ in $\tilde{O}(\nnz(\matA)k/\eps)$ time. 
By invoking Lemma \ref{lem:pythagorean}, we can now compute $[\matA \matU]_k$, and our
overall low rank approximation will be $[\matA \matU]_k \matU^T$, which is a $(1+\eps)$-approximation. Note that
we can compute the SVD of $\matA \matU$ in $\tilde{O}((k/\eps)^2 n)$ time, thereby giving us $[\matA \matU]_k$. This allows us
to obtain the SVD $\tilde{\matU} \tilde{\mat\Sigma} \tilde{\matV}^T$ of $[\matA \matU]_k$ in this amount of time as well. 
We don't require
explicitly outputting the product of $[\matA \matU]_k$ and $\matU^T$, since this may be a dense matrix and so would
require at least $nd$ time to write down. In applications, it is usually better to have a factored form.
Notice that $\tilde{\matV}^T \matU^T$ has orthonormal rows, since $\tilde{\matV}^T \matU^T \matU \tilde{\matV}$ 
is the identity.
Therefore, we can output $\tilde{\matU}, \tilde{\mat\Sigma}, \tilde{\matV}^T \matU^T$, which is the SVD of a rank-$k$
matrix providing a $(1+\eps)$-approximation. 

The overall time of the above algorithm is $\tilde{O}(\nnz(\matA))k/\eps + (n+d)(k/\eps)^2)$. While this is a significant
improvement over computing the SVD of $\matA$, which would take $\min(nd^2, n^2d)$ time, we could still
hope to achieve a leading order running time of $O(\nnz(\matA))$ as opposed to $\tilde{O}(\nnz(\matA) k/\eps)$. The dominant
cost is actually in computing $\matA\matU$, 
the coordinate representation of the rows of $\matA$ in the rowspace of $\matU^T$. That
is, it is inefficient to directly project the rows of $\matA$ onto $\matU$. 

Fortunately, we can cast this projection problem as a regression problem, and solve it approximately. 
\begin{theorem}\label{thm:lowrank}
Now let $\matR$ be a $(1+O(\eps))$-approximate $\ell_2$-subspace embedding for the row space of $\matS\matA$, where $\matS$
is as in Lemma \ref{lem:sketching}. Then $$\FNormS{\matA\matR(\matS\matA\matR)^{\dagger}\matS\matA-\matA} \leq (1+\eps) \FNormS{\matA-\matA_k}.$$ Furthermore, 
$$\FNormS{[\matA\matR\matU]_k \matU^T (\matS \matA \matR)^{\dagger}\matS \matA -\matA} \leq (1+\eps)\FNormS{\matA-\matA_k},$$ 
where $\matU \matU^T = (\matS \matA \matR)^{\dagger} \matS \matA \matR$ is the projection
matrix onto the rowspace of $\matS \matA \matR$. 

The time to compute the factorizations 
$(\matA \matR), (\matS\matA\matR)^{\dagger}, (\matS \matA)$ or 
$[\matA \matR \matU]_k$, $\matU^T$, $(\matS \matA \matR)^{\dagger}$, 
$(\matS \matA)$,
is $O(\nnz(\matA)) + (n+d)\poly(k/\eps)$. Note that in general one can only hope to output
the factorizations in this amount of time, as performing the actual multiplications of the 
factors of the low rank approximation may result in a dense
$n \times d$ matrix. 
\end{theorem}
\begin{proof}
Lemma \ref{lem:sketching} implies that 
$$\min_{\matY} \FNorm{\matY \matS \matA-\matA} \leq (1+\eps)\FNorm{\matA-\mat A_k},$$
The minimizer
of the regression problem $\min_{\mat Y} \FNorm{\matY \matS \matA \matR - \matA \matR}$ 
is equal to $\matY = \matA\matR(\matS \matA \matR)^{\dagger}$, and since $\matR$ is a subspace
embedding we have 
$$\FNorm{\matA \matR(\matS \matA \matR)^{\dagger}\matS \matA - \matA} \leq (1+\eps) \min_{\matY} \FNorm{\matY \matS \matA-\matA} 
\leq (1+\eps)^2 \FNorm{\matA-\matA_k},$$
implying the first part of the theorem after rescaling $\eps$ by a constant factor. 

For the second part of the theorem, note that Lemma \ref{lem:sketching} gives the stronger guarantee that
$$\min_{\textrm{rank }k \matY} \FNorm{\matY \matS \matA- \matA} \leq (1+\eps)\FNorm{\matA-\matA_k}.$$
By the properties of an $\ell_2$-subspace embedding, we thus have if $\matZ$ is the solution to the regression
problem
$$\min_{\textrm{rank k} \matZ} \FNorm{\matZ \matS \matA \matR- \matA \matR},$$
then 
$$\FNorm{\matZ \matS \matA- \matA} \leq (1+\eps) \min_{\textrm{rank }k \matY}\FNorm{\matY \matS \matA- \matA} 
\leq (1+\eps)^2 \FNorm{\matA- \matA_k}.$$
Therefore, it suffices to find $\matZ$. Note that $\matZ \matS \matA \matR$ is the best rank-$k$ approximation to $\matA \matR$ in the rowspace
of $\matS \matA \matR$. Therefore, by Lemma \ref{lem:pythagorean}, $\matZ \matS \matA \matR = [(\matA \matR)\matU]_k \matU^T$, where $\matU \matU^T = (\matS \matA \matR)^{\dagger} \matS \matA \matR$ is the projector onto
the row space of $\matS \matA \matR$. Note that $\matS \matA \matR (\matS \matA \matR)^{\dagger} = \matI$ since $\matS$ has fewer rows than columns, and therefore
$$\matZ \matS \matA = [(\matA \matR) \matU]_k \matU^T (\matS \matA \matR)^{\dagger} \matS \matA.$$ 
For the time complexity, the dominant cost is in computing $\matA \matR$ and $\matS \matA$, 
both of which can be done in $O(\nnz(\matA))$ time. 
The remaining operations are on matrices for which at least one dimension is $\poly(k/\eps)$, and therefore can be computed
in $(n+d)\poly(k/\eps)$ time. 
\end{proof}
While the representation $\matA \matR(\matS \matA \matR)^{\dagger}\matS \matA$ in Theorem \ref{thm:lowrank} 
might be useful in its own right as a low rank approximation to $\matA$, given
it is technically a bicriteria solution since its rank may be $O(k/\eps + k \poly \log k)$,
whereas the original problem formulation wants our representation to have rank at most $k$. 
The second part of Theorem \ref{thm:lowrank}
gives a rank-$k$ approximation. 

\subsection{CUR decomposition}\label{sec:CUR}
We now give an alternative to low rank approximation which involves finding a decomposition of an $n \times n$ matrix $\matA$
into $\matC \cdot \matU \cdot \matR$, where $\matC$ is a subset of columns of $\matA$, $\matR$ is a subset of rows of 
$\matA$, and $\matU$ is a low rank matrix. 
Ideally, we would like the following properties:
\begin{enumerate}
\item $\FNorm{\matC \matU \matR - \matA} \leq (1+\eps) \FNorm{\matA-\matA_k}.$
\item $\matC$ is $n \times c$ for a small value of $c$. Similarly, $\matR$ is $r \times n$ for a small value of $r$.
\item $\matU$ has rank $k$.
\item The matrices $\matC$, $\matU$, and $\matR$ can be found quickly, ideally in $\nnz(\matA) + \poly(k/\eps)n$ time. 
\end{enumerate}
A CUR decomposition of a matrix is thus a rank-$k$ approximation whose column space and row space are spanned by a small subset of 
actual rows and columns
of $\matA$. This often makes it more interpretable than a generic low rank approximation, or even the SVD, whose column and row spaces are
spanned by arbitrary linear combinations of all of the columns and rows of $\matA$, respectively, see, e.g., \cite{MD09} for a discussion
of this. 

Before discussing the details of some of the available CUR algorithms in~\cite{DK03,DKM06c,DMM08,DM09,GM13,WZ13CUR,BW14}, we briefly mention a 
similar problem
which constructs factorizations of the form $\matA = \matC \matX + \matE$, where $\matC$ contains columns of $\matA$ and $\matX$ has rank at
most $k$. There are also optimal algorithms for this problem~\cite{BDM11a,GS12}, in both the spectral and the Frobenius norm. Indeed, to
obtain a relative-error optimal CUR, one uses a sampling method from~\cite{BDM11a}, which allows to select $O(k)$ columns and rows.
For a more detailed discussion of this CX problem, which is also known as CSSP (Column Subset Selection Problem) see~\cite{BMD09a,BDM11a,GS12}.

Drineas and Kannan brought CUR factorizations to the theoretical computer science community in~\cite{DK03}; we
refer the reader to the jounral version of their work together with Mahoney \cite{DKM06c}. 
Their main algorithm (see Theorem 5 in \cite{DKM06c}) 
is randomized and samples columns and rows from $\matA$ with probabilities proportional
to their Euclidean length. The running time of this algorithm is linear in $m$ and $n$ and proportional to a
small-degree polynomial in $k$ and $1/\varepsilon,$ for some $\varepsilon >0$,
but the approximation bound is additive rather than relative 
(see Theorem 3.1 in~\cite{DK03}): with $c = O(k/\varepsilon^4)$ columns and $r=O(k/\varepsilon^2)$ rows
the bound is 
$
\FNormS{\matA - \matC \matU \matR} \le \FNormS{\matA - \matA_k} + \varepsilon \FNormS{\matA}.
$
%

The first relative-error CUR algorithm appeared in~\cite{DMM08} (see Theorem 2 of \cite{DMM08}). The
algorithm of~\cite{DMM08} is based on subspace sampling and requires 
$c = O(k \log(k / \varepsilon^2)  \log \delta^{-1})$ columns and
$r=O(c \log (c / \varepsilon^2) \log \delta^{-1})$ rows to construct a relative-error CUR with failure probability $\delta$.
The running time of the method in~\cite{DMM08} is $O( m n \min\{m,n\} )$, since subspace sampling is based on sampling with probabilities
proportional to the so-called leverage scores, i.e., the row norms of the matrix $\matV_k$ from the SVD of $\matA$.

Mahoney and Drineas~\cite{DM09}, using again subspace sampling, improved slightly upon the number of columns and rows,
compared to~\cite{DMM08}, but achieved only a constant factor error (see Eqn.(5) in~\cite{DM09}).
Gittens and Mahoney~\cite{GM13} discuss CUR decompositions on symmetric positive semidefinite (SPSD) 
matrices and present approximation bounds for Frobenius,
trace, and spectral norms (see Lemma 2 in~\cite{GM13}).
Using the near-optimal column subset selection methods in~\cite{BDM11a} along with a novel adaptive sampling technique, Wang and Zhang~\cite{WZ13CUR} present a CUR
algorithm selecting
$c = (2k/\varepsilon)(1+o(1))$
columns and
$r = (2k/\varepsilon^2)(1+\varepsilon)(1+o(1))$
rows from $\matA$
(see Theorem 8 in~\cite{WZ13CUR}).
The running time of this algorithm is
$$
O( mnk\varepsilon^{-1} + mk^3\varepsilon^{-\frac{2}{3}} + nk^3\varepsilon^{-\frac{2}{3}} + mk^2\varepsilon^{-2} + nk^2\varepsilon^{-4}).
$$
Boutsidis and the author \cite{BW14} improve this to achieve a simlutaneously optimal $c = r = O(k/\eps)$, and rank$(\matU) = k$. 
This in fact is optimal up to constant factors, as shown by \cite{BW14} by presenting a matching lower bound. 
Boutsidis and the author also show how to do this in $O(\nnz(\matA) \log n) + n \cdot \poly(k/\eps)$ time. 
There is also some desire to make the CUR decomposition deterministic. 
We will see that this is possible as well, as shown in \cite{BW14}. 

%
Finally, there are several interesting results on CUR developed within 
the numerical linear algebra community~\cite{Tyr96, Tyr00, GTZ97a, GTZ97b,HP97, Pan03, MG03, GM04,BPSS05,Ste99}.
For example,~\cite{Tyr96, Tyr00, GTZ97a, GTZ97b} discuss the so-called skeleton approximation, which focuses on the spectral norm
version of the CUR problem via selecting exactly $k$ columns and $k$ rows. The algorithms there are deterministic, run
in time proportional to the time to compute the rank $k$ SVD of $\matA$, and achieve bounds of the order,
$$\TNorm{\matA-\matC\matU\matR} \le O( \sqrt{k(n-k)} + \sqrt{k(m-k)} )\TNorm{\matA-\matA_k}.$$

We now outline the approach of Boutsidis and the author \cite{BW14}. A key lemma we need is the following, which
is due to Boutsidis, Drineas, and Magdon-Ismail \cite{BDM11a}. 
\begin{lemma}
\label{lem:structural}
Let $\matA = \matA \matZ \matZ^T + \matE \in \mathbb{R}^{m \times n}$ 
be a low-rank matrix factorization of $\matA$, with $\matZ \in \mathbb{R}^{n \times k},$
and $\matZ^T \matZ= \matI_{k}$.
Let $\matS\in\mathbb{R}^{n\times c}$ ($c \ge k$) be any matrix such that $rank(\matZ^T \matS) =
rank(\matZ)=k.$
Let $\matC = \matA \matS \in \mathbb{R}^{m \times c}$. Then,
$$
\FNormS{\matA - \matC \matC^{\dagger} \matA} \le
  \FNormS{\matA - \Pi_{C,k}(\matA)} \le
\FNormS{\matE} +\FNormS{\matE \matS (\matZ^T\matS)^{\dagger}}.
$$
Here, $\Pi_{C,k}(\matA) = \matC \matX_{opt} \in \mathbb{R}^{m \times n},$ 
where $\matX_{opt}\in \mathbb{R}^{c \times n}$ has rank at most $k$, 
$\matC \matX_{opt}$ is the
best rank $k$ approximation to $\matA$ in the column space of $\matC$, and $(\matZ^T \matS)^{\dagger}$ denotes the Moore-Penrose 
pseudoinverse of $\matZ^T \matS$.
\end{lemma}
\begin{proof}
First note that 
$\FNormS{\matA - \matC \matC^{\dagger} \matA} \le
  \FNormS{\matA - \Pi_{C,k}(\matA)}$ since $\matC \matC^{\dagger} \matA$ is the projection of the columns of $\matA$ onto the column
space of $\matC$, whereas $\Pi_{C, k}(\matA)$ is the best rank-$k$ approximation to $\matA$ in the column space of $\matC$. 

For the second inequality in the lemma statement, 
the main idea in the proof is to consider the matrix $\matX = \matC(\matZ^T \matS)^{\dagger} \matZ^T$. 
Since this matrix is in the column space
of $\matC$, we have 
$$\FNormS{\matA - \Pi_{C,k}(\matA)} \leq \FNormS{\matA-\matX},$$
since $\Pi_{C,k}(\matA)$ is the best rank-$k$ approximation to $\matA$ inside the column space of $\matC$. 

Manipulating $\matA-\matX$, we have that $\FNormS{\matA-\matC(\matZ^T\matS)^{\dagger} \matZ^T}$ is equal to 
\begin{eqnarray*}
& = & \FNormS{\matA\matZ\matZ^T + \matE - (\matA\matZ\matZ^T + \matE)\matS (\matZ^T \matS)^{\dagger} \matZ^T}\\
& = & \FNormS{\matA\matZ\matZ^T - \matA\matZ\matZ^T \matS(\matZ^T \matS)^{\dagger} \matZ^T + \matE - \matE \matS (\matZ^T \matS)^{\dagger} \matZ^T}\\
& = & \FNormS{\matE - \matE\matS(\matZ^T \matS)^{\dagger} \matZ^T}\\
& = & \FNormS{\matE} + \FNormS{\matE\matS(\matZ^T \matS)^{\dagger} \matZ^T},
\end{eqnarray*}
where the first equality uses that $\matA = \matA\matZ\matZ^T + \matE$ and that $\matC = \matA\matS$, 
the second equality is a rearrangement of terms,
the third equality uses that rank$(\matZ^T \matS) = k$ and so $(\matZ^T \matS)(\matZ^T \matS)^{\dagger} = \matI_{k}$, 
and the last equality follows from
the Pythagorean theorem since $\matE = \matA(\matI_{k}-\matZ\matZ^T)$ 
has rows orthogonal to the row space of $\matZ^T$, while 
$\matE\matS(\matZ^T \matS)^{\dagger} \matZ^T$ has rows in the row space of $\matZ^T$. Finally, noting that 
$$\FNormS{\matE\matS(\matZ^T \matS)^{\dagger} \matZ^T} \leq \FNormS{\matE\matS(\matZ^T \matS)^{\dagger}} \FNormS{\matZ^T},$$
by submultiplicativity, and that $\|\matZ^T\|_2 = 1$, completes the proof. 
\end{proof}
We will apply Lemma \ref{lem:structural} twice, and adaptively. First, we compute an $n \times k$ matrix $\matZ$ with
orthonormal columns for which $\FNormS{\matA-\matA\matZ\matZ^T} \leq (1+\frac{1}{9})\FNormS{\matA-\matA_k}$. 
This can be done in $O(\nnz(\matA)) + n \cdot \poly(k)$
time as shown in the second part of Theorem \ref{thm:lowrank} of the previous section. Specifically, from the statement of that theorem,
for an $n \times d$ matrix $\matA$, the column space of $[\matA\matR\matU]_k$ spans a $(1+\eps)$ rank-$k$ approximation to $\matA$, where
$\matU$ satisfies $\matU\matU^T = (\matS\matA\matR)^{\dagger} \matS\matA\matR$. We can apply that theorem to $\matA^T$ to obtain a $k \times d$
matrix $\matZ^T$ which spans a $(1+\eps)$ rank-$k$ approximation to $\matA$. 

Given $\matZ$, we will sample $O(k \log k)$ columns of $\matZ$ proportional to the squared row norms, or leverage scores of $\matZ$.
Let $\ell_i^2 = \|e_i^T\matZ\|_2^2$ be the $i$-th leverage score. 
Since $\sum_{i=1}^n \ell_i^2 = k$, the $p_i = \ell_i^2/k$ values
define a probability distribution. 

We now invoke the {\textsc RandSampling}$(\matZ, s, p)$ algorithm of
Definition \ref{def:lss} with $s = \Theta(k (\log k))$. By the guarantee of Theorem \ref{thm:lssPerf},
we obtain matrices $\mat\Omega$ and $\matD$ for which with probability $1-1/\poly(k)$, for all $i \in [k]$,
\begin{eqnarray}\label{eqn:rank1}
\frac{1}{2} \leq \sigma_i^2 (\matZ^T \mat\Omega \matD) \leq \frac{3}{2}.
\end{eqnarray}
Here $\mat\Omega \matD$ implements sampling $s$ columns of $\matZ^T$ and re-scaling them by the coefficients
in the diagonal matrix $\matD$. We also record the following simple fact about the {\textsc RandSampling}
algorithm.
\begin{lemma}\label{lem:FrobLSS}
With probability at least $.9$ over the randomness in the algorithm {\textsc RandSampling}$(\matZ, s, p)$,
$$\FNormS{\matZ^T \mat\Omega \matD} \leq 10\FNormS{\matZ^T}.$$
\end{lemma}
\begin{proof}
We show ${\bf E}[\FNormS{\matZ^T \mat\Omega \matD}] = \FNormS{\matZ^T}$. 
By linearity of expectation, it suffices to show for a fixed column $j \in [s]$,
${\bf E}[\FNormS{(\matZ^T \mat\Omega \matD)_{*j}}] = \FNormS{\matZ^T}/s.$ We have, 
\begin{eqnarray*}
{\bf E}[\FNormS{(\matZ^T \mat\Omega \matD)_{*j}}] & = & \sum_{i=1}^n p_i \cdot \frac{1}{p_i s} \|\matZ^T_{*i}\|_2^2 = \frac{1}{s} \FNormS{\matZ^T},
\end{eqnarray*}
as needed. The lemma now follows by Markov's bound. 
\end{proof}

Our algorithm thus far, is given $\matA$, to compute $\matZ$, and then to compute $\mat\Omega$ and $\matD$ via
{\textsc RandSampling}$(\matZ, s, p)$, 
where $s = O(k \log k)$. At this point, we could look at $\matA \mat\Omega \matD$, which
samples $s$ columns of $\matA$. While this set of columns can be shown to have good properties, namely, 
its column space contains a $k$-dimensional subspace spanning an $O(1)$ rank-$k$ approximation to $A$, which can then
be used as a means for obtaining a $(1+\eps)$-approximation by adaptive sampling, as will be seen in 
\S\ref{sec:adaptiveSampling}. However, the number of columns is $O(k \log k)$, which would result in an overall
CUR decomposition with at least $O(k \log k / \eps)$ columns and rows using the techniques below, 
which is larger by a $\log k$ factor than what we would like (namely, $O(k/\eps)$ columns and rows). 

We therefore
wish to first downsample the $s$ columns of $\matA$ we have now to $O(k)$ columns by right-multiplying by a matrix 
$\matS \in \mathbb{R}^{s \times k}$, so that $\matZ^T \mat\Omega \matD \matS$ 
has rank $k$ and has a reasonably large $k$-th singular
value. 

To proceed, we need an algorithm in the next subsection, 
which uses a method of Batson, Spielman, and Srivastava \cite{BSS09}
refined for this application by Boutsidis, Drineas, and Magdon-Ismail \cite{BDM11a}.

\subsubsection{Batson-Spielman-Srivastava sparsification}\label{sec:bss}
Define
the parameters
$$\delta_{LOW} = 1, \ \ \delta_{UP} = \frac{\FNormS{\matA}}{1-\sqrt{\frac{k}{r}}}.$$
Define the function
$$\phi(L, \matM) = \Trace{(\matM-L\matI)^{-1}} = \sum_{i=1}^k \frac{1}{\lambda_i(\matM)-L}.$$
Note that $\phi(L, \matM)$ measures how far the eigenvalues of $\matM$ are from $L$, since the closer they
are to $L$, the more $\phi$ ``blows up''. 

Also, define the functions
$$UP(a, \delta_{UP}) = \delta_{UP}^{-1} a^Ta,$$
and
\begin{eqnarray*}
LOW(\ve_j, \delta_{LOW}, \matM, L) 
& = & \frac{\ve^T (\matM-(L+\delta_{LOW})\matI_{k})^{-2} \ve}{\phi(L+\delta_{LOW}, \matM) - \phi(L, \matM)}\\
&& - \ve^T(\matM-(L+\delta_{LOW})\matI_{k})^{-1}\ve.
\end{eqnarray*}

These two functions will be used in each iteration of Algorithm \ref{alg1} below to make progress in each iteration. 
What we will be able to show is that in each iteration of the algorithm, the current value of our potential function $UP$ will be less 
than the current value of our potential function $LOW$, and this will enable us to choose a new vector $\ve_i$ and add $\ve_i \ve_i^T$ 
in our decomposition of the identity. This corresponds to a rank-one update of our current decomposition $\matM$, 
and we use the Sherman-Morrison-Woodbury formula to analyze how the eigenvalues of $\matM$ change, as well as how
the values of $UP$ and $LOW$ change, given this rank-one update. 

\begin{algorithm}[t]
\caption{Deterministic Dual Set Spectral Sparsification}
\label{alg1}
{\bf Input:} 
\begin{itemize}
\item $\mathcal{V} = \{\ve_1, \ldots, \ve_n\}$ with $\sum_{i=1}^n \ve_i \ve_i^T = \matI_{k}$
\item $\mathcal{A} = \{\a_1, \ldots, \a_n\}$. 
\end{itemize}
{\bf Output:}
A set of $n$ non-negative weights $s_i$, at most $r$ of which are non-zero. 
\begin{enumerate}
\item Initialize $s_0 = \mat0_{n \times 1}$, $\matM_0 = \mat0_{k \times k}$
\item For $\tau = 0, \ldots, r-1$
\begin{itemize}
\item Set $L_{\tau} = \tau - \sqrt{rk}$ and $U_{\tau} = \tau \delta_{UP}$. 
\item Find an index $j \in \{1, 2, \ldots, n\}$ such that
$UP(\a_j, \delta_{UP}) \leq LOW(\ve_j, \delta_{LOW}, \matM_{\tau}, L_{\tau})$. 
\item Let $t^{-1} = \frac{1}{2} \left (UP(\a_j, \delta_{UP}) + LOW(\ve_j, \delta_{LOW}, \matM_{\tau}, L_{\tau}) \right )$. 
\item Update the $j$-th component of $\s$ and $\matM_{\tau}$:
$$s_{\tau+1}[j] = s_{\tau}[j]+t, \ \matM_{\tau+1} = \matM_{\tau} + t\ve_j \ve_j^T.$$
\end{itemize}
\item Return $\s = r^{-1}(1-\sqrt{k/r}) \cdot \s_r$.
\end{enumerate}
\end{algorithm}
The following theorem shows correctness of 
the Deterministic Dual Set Spectral Sparsification algorithm
described in Algorithm \ref{alg1}. 

\begin{theorem}(Dual Set Spectral-Frobenius Sparsification)
\label{thm:dualset}
Let $\ve_i \in \mathbb{R}^{k}$ for $i = 1, \ldots, n$ with $k < n$, and
$\sum_{i=1}^n \ve_i \ve_i^T = \matI_{k}$. Let $A=\{\a_1,\ldots,\a_n\}$ be an arbitrary set
of vectors, where $\a_i\in \mathbb{R}^{\ell}$ for all $i$. Then,
given an integer $r$ such that $k < r \le n$, 
there exists a set of weights $s_i\ge 0$ ($i=1, \ldots, n$), 
at most $r$ of which are non-zero, such that
\eqan{
\lambda_{k}\left(\sum_{i=1}^n s_i \ve_i \ve_i^T\right)
&\ge&
\left(1 - \sqrt{\frac{k}{r}}\right)^2,
\qquad \\
\trace\left(\sum_{i=1}^n s_i \a_i\a_i^T\right)
&\le&
\trace\left(\sum_{i=1}^n \a_i \a_i^T\right)
=
\sum_{i=1}^n \| \a_i\|_2^2.
}
Equivalently,
if $\matV \in \mathbb{R}^{n \times k}$ is a matrix whose rows are the vectors $\ve_i^T$,
$\matA \in \mathbb{R}^{n \times \ell}$ is a matrix whose rows are the vectors $\a_i^T$, and
$\matS \in \mathbb{R}^{n \times r}$ is the sampling matrix containing the weights $s_i > 0, $ then:
\eqan{
\sigma_{k}\left(\matV^T \matS\right)
\ge
(1 - \sqrt{{k}/{r}})^2,
\qquad
\FNormS{\matA^T \matS}
\le \FNormS{\matA^T}.
}
The weights $s_i$ can be computed in $O\left(rnk^2+n\ell\right)$ time. We denote this procedure as
$$\matS = BssSampling(\matV, \matA, r).$$
\end{theorem} 
\begin{proof}
First, regarding the time complexity, 
in each of $r$ iterations we need to compute $L$ on each of the $n$ vectors $\ve_j$. The costly matrix
inversion in the definition of $L$ can be performed once in $O(k^3)$ time, which also upper bounds
the time to compute $\phi(L+\delta_{LOW}, \matM)$ and $\phi(L,\matM)$. Given these quantities, 
computing $L$ for a single vector $\ve_j$ takes $O(k^2)$ time and so for all $n$ vectors $\ve_j$
$O(nk^2)$ time, and across all $r$ iterations, $O(rnk^2)$ time. Computing $UP(\a_j, \delta_{UP})$
just corresponds to computing the Euclidean norm of $\a_j$, and these can be computed once at the
beginning of the algorithm in $O(n \ell)$ time. This implies the overall time complexity of the lemma.

We now turn to correctness. The crux of the analysis turns out to be 
to show there always exists an index $j$
in each iteration for which 
$$UP(\a_j, \delta_{UP}) \leq LOW(\ve_j, \delta_{LOW}, \matM_{\tau}, L_{\tau}).$$ 

For a real symmetric matrix $\matM$ we let $\lambda_i(\matM)$ denote 
its $i$-th largest eigenvalue of matrix $\matM$. It will be 
useful to observe that for $L_{\tau}$ and $U_{\tau}$ as defined by the algorithm, we have chosen the
definitions so that $L_{\tau} + \delta_{LOW} = L_{\tau + 1}$ and $U_{\tau} + \delta_{UP} = U_{\tau + 1}$. 

We start with a lemma which uses the 
Sherman-Morrison-Woodbury identity to analyze a rank-$1$ perturbation. 
\begin{lemma}\label{lem:bss}
Fix $\delta_{LOW} > 0$, $\matM \in \mathbb{R}^{k \times k}$, $\ve \in \mathbb{R}^k$, and $L < \lambda_k(\matM)$. If
$t > 0$ satisfies $$t^{-1} \leq LOW(\ve, \delta_{LOW}, \matM, L),$$ 
then 
\begin{enumerate}
\item $\lambda_k(\matM + t\ve\ve^T) \geq L+\delta_{LOW},$ and
\item $\phi(L+\delta_{LOW}, \matM + t\ve\ve^T) \leq \phi(L, \matM)$. 
\end{enumerate} 
\end{lemma}
\begin{proof}
Note that by definition of $\phi(L, \matM)$, given that $\lambda_k(\matM) > L$, 
and $\phi(L, \matM) \leq \frac{1}{\delta_{LOW}}$, this implies that $\lambda_k(\matM) > L + \delta_{LOW}$, and so
for any $t > 0$, $\lambda_k(\matM + t \ve\ve^T) > L + \delta_{LOW}$. This proves the first part of the lemma.

For the second part, we use the following well-known formula.
\begin{fact}(Sherman-Morrison-Woodbury Formula)
If $\matM$ is an invertible $n \times n$ matrix and $\ve$ is an $n$-dimensional vector, then
$$(\matM+\ve\ve^T)^{-1} = \matM^{-1} - \frac{\matM^{-1}\ve\ve^T \matM^{-1}}{1+\ve^T\matM^{-1}\ve}.$$
\end{fact}

Letting $L' = L + \delta_{LOW}$, we have
\begin{eqnarray*}
\phi(L+\delta_{LOW}, \matM + t\ve\ve^T) & = & \Trace{(\matM + t\ve\ve^T - L' \matI)^{-1}}\\
 & = & \Trace{(\matM-L'\matI)^{-1}}\\
&& - \Trace{\frac{t(\matM-L'\matI)^{-1}\ve\ve^T (\matM-L'\matI)^{-1}}{1 + t\ve^T(\matM-L'\matI)^{-1}\ve}}\\
& = & \Trace{(\matM-L'\matI)^{-1}}\\
&& - \frac{t \Trace{\ve^T(\matM-L'\matI)^{-1}(\matM-L'\matI)^{-1}\ve}}{1+t\ve^T(\matM-L'\matI)^{-1}\ve}\\
& = & \phi(L',\matM) - \frac{t\ve^T(\matM-L'\matI)^{-2}\ve}{1+t\ve^T(\matM-L'\matI)^{-1}\ve}\\
& = & \phi(L, \matM) + (\phi(L', \matM) - \phi(L, \matM))\\
&& - \frac{\ve^T(\matM-L'\matI)^{-2}\ve}{1/t + \ve^T(\matM-L'\matI)^{-1}\ve}\\
& \leq & \phi(L,\matM),
\end{eqnarray*}
where the first equality uses the definition of $\phi$, the second equality
uses the Sherman-Morrison Formula, the third equality uses that the trace is a linear operator
and satisfies $\Trace{\matX\matY} = \Trace{\matY\matX}$, 
the fourth equality
uses the definition of $\phi$ and that the trace of a number is the number itself, the fifth
equality follows by rearranging terms, and the final inequality follows by assumption
that $t^{-1} \leq LOW(\ve, \delta_{LOW}, \matM, L)$. 
\end{proof}
%
We also need the following lemma concerning properties of the $UP$ function.  
\begin{lemma}\label{lem:U}
Let $\matW \in \mathbb{R}^{\ell \times \ell}$ be a symmetric positive semi-definite matrix, let 
$\a \in \mathbb{R}^{\ell}$ be a vector, and let $U \in \mathbb{R}$ satisfy $U > Tr(\matW)$. If $t > 0$ 
satisfies $$UP(\a, \delta_{UP}) \leq t^{-1},$$ 
then
$$Tr(\matW + t\ve\ve^T) \leq U + \delta_{UP}.$$
\end{lemma}
\begin{proof}
By the assumption of the lemma, 
$$UP(\a, \delta_{UP}) = \delta_{UP}^{-1} \a^T \a \leq t^{-1},$$
or equivalently, $t \a^T \a \leq \delta_{UP}$. Hence, 
\begin{eqnarray*}
\Trace{\matW + t\a\a^T} - U - \delta_{UP} & = & \Trace{\matW} - U + (t\a^T\a - \delta_{UP})\\
& \leq & \Trace{\matW} - U < 0.
\end{eqnarray*}
\end{proof}
Equipped with Lemma \ref{lem:bss} and Lemma \ref{lem:U}, we now prove the main lemma we need.

\begin{lemma}\label{lem:sandwich}
At every iteration $\tau = 0, \ldots, r-1$, there exists an index $j \in \{1, 2, \ldots, n\}$
for which 
$$UP(\a_j, \delta_{UP}) \leq t^{-1} \leq LOW(\ve_j, \delta_{LOW}, \matM_{\tau}, L_{\tau}).$$
\end{lemma}
\begin{proof}
It suffices to show that
\begin{eqnarray}\label{eqn:average}
\sum_{i=1}^n UP(\a_j, \delta_{UP}) = 1 - \sqrt{\frac{k}{r}} \leq \sum_{i=1}^n LOW(\ve_i, \delta_{LOW}, \matM_{\tau}, L_{\tau}).
\end{eqnarray}
Indeed, if we show (\ref{eqn:average}), then by averaging there must exist an index $j$ for which
$$UP(\a_j, \delta_{UP}) \leq t^{-1} \leq LOW(\ve_j, \delta_{LOW}, \matM_{\tau}, L_{\tau}).$$
We first prove the equality in (\ref{eqn:average}) using the definition of $\delta_{UP}$. Observe that
it holds that 
\begin{eqnarray*}
\sum_{i=1}^n UP(\a_i, \delta_{UP}) & = & \delta_{UP}^{-1} \sum_{i=1}^n \a_i^T \a_i\\
& = & \delta_{UP}^{-1} \sum_{i=1}^n \|\a_i\|_2^2\\
& = & 1- \sqrt{\frac{k}{r}}.
\end{eqnarray*}
We now prove the inequality in (\ref{eqn:average}). Let $\lambda_i$ denote the $i$-th
largest eigenvalue of $\matM_{\tau}$. 
Using that $\Trace{\ve^T \matY \ve} = \Trace{\matY\ve\ve^T}$ and $\sum_i \ve_i \ve_i^T = \matI_{k}$,
we have
\begin{eqnarray*}
\sum_{i=1}^n LOW(\ve_i, \delta_{LOW}, \matM_{\tau}, L_{\tau}) & = &
\frac{\Trace{(\matM_{\tau} - L_{\tau + 1} \matI_{k})^{-2}}}{\phi(L_{\tau+1}, \matM_{\tau}) - \phi(L_{\tau}, \matM_{\tau})}\\
&& - \phi(L_{\tau + 1}, \matM_{\tau})\\
& = & \frac{\sum_{i=1}^k \frac{1}{(\lambda_i - L_{\tau+1})^2}}{\delta_{LOW} \sum_{i=1}^k \frac{1}{(\lambda_i - L_{\tau + 1})(\lambda_i - L_{\tau})}}\\
&& - \sum_{i=1}^k \frac{1}{(\lambda_i - L_{\tau + 1})}\\
& = & \frac{1}{\delta_{LOW}} - \phi(L_{\tau}, M_{\tau}) + \mathcal{E},
\end{eqnarray*}
where
\begin{eqnarray*}
\mathcal{E} & = & \frac{1}{\delta_{LOW}} \left (\frac{\sum_{i=1}^k \frac{1}{(\lambda_i - L_{\tau +1})^2}}{\sum_{i=1}^k \frac{1}{(\lambda_i - L_{\tau} + 1)(\lambda_i - L_{\tau})}} - 1 \right )\\
&& - \delta_{LOW} \sum_{i=1}^k \frac{1}{(\lambda_i -L_{\tau})(\lambda_i - L_{\tau +1})}.
\end{eqnarray*}
We will show $\mathcal{E} \geq 0$ below. Given this, we have
\begin{eqnarray*}
\phi(L_{\tau}, \matM_{\tau}) & \leq & \phi(L_0, \matM_0)
 = \phi(- \sqrt{rk}, \mat0_{k \times k})
 = \frac{-k}{-\sqrt{rk}}
 = \sqrt{\frac{k}{r}},
\end{eqnarray*}
where the inequality uses Lemma \ref{lem:bss}. Since $\delta_{LOW} = 1$, we 
have 
$$\sum_{i=1}^n LOW(\ve_i, \delta_{LOW}, \matM_{\tau}, L_{\tau}) \geq 1-\sqrt{\frac{k}{r}},$$
which will complete the proof. 

We now turn to the task of showing $\mathcal{E} \geq 0$. The Cauchy-Schwarz inequality implies
that for $a_i, b_i \geq 0$, one has $(\sum_i a_i b_i)^2 \leq (\sum_i a_i^2 b_i)(\sum_i b_i)$, 
and therefore
\begin{eqnarray}\label{eqn:E}
\mathcal{E} \sum_{i=1}^k \frac{1}{(\lambda_i - L_{\tau + 1})(\lambda_i - L_{\tau})}
& = & \frac{1}{\delta_{LOW}} \sum_{i=1}^k \frac{1}{(\lambda_i - L_{\tau + 1})^2 (\lambda_i - L_{\tau})} \notag \\
&& - \delta_{LOW} \left (\sum_{i=1}^k \frac{1}{(\lambda_i - L_{\tau})(\lambda_i - L_{\tau + 1})} \right )^2 \notag \\
& \geq & \frac{1}{\delta_{LOW}} \sum_{i=1}^k \frac{1}{(\lambda_i - L_{\tau + 1})^2 (\lambda_i - L_{\tau})} \notag \\
&& - \sum_{i=1}^k \frac{\delta_{LOW}}{(\lambda_i - L_{\tau + 1})^2(\lambda_i - L_{\tau})} \sum_{i=1}^k \frac{1}{\lambda_i - L_{\tau}} \notag \\
& = & 
\sum_{i=1}^k \frac{\left (\frac{1}{\delta_{LOW}} - \delta_{LOW} \cdot \phi(L_{\tau}, \matM_{\tau}) \right )}{(\lambda_i - L_{\tau+1})^2(\lambda_i - L_{\tau})}.
\end{eqnarray}
Since $\delta_{LOW} = 1$ and we have computed above that $\phi(L, \matM) \leq \sqrt{k/r}$, 
we have $\frac{1}{\delta_{LOW}} - \delta_{LOW} \cdot \phi(L_{\tau}, \matM_{\tau}) \geq 1-\sqrt{k/r} > 0$
since $r > k$. 

Also, 
\begin{eqnarray*}
\lambda_i & \geq & \lambda_k(\matM_{\tau})\\
& \geq & L_{\tau} + \frac{1}{\phi(L_{\tau}, \matM_{\tau})}\\
& \geq & L_{\tau} + \frac{1}{\phi(L_0, \matM_0)}\\
& \geq & L_{\tau} + \sqrt{\frac{r}{k}}\\
& > & L_{\tau} + 1\\
& = & L_{\tau+1}.
\end{eqnarray*}
Plugging into (\ref{eqn:E}), we conclude that $\mathcal{E} \geq 0$, as desired. 
\end{proof}

By Lemma \ref{lem:sandwich}, the algorithm is well-defined, finding a $t \geq 0$ at each iteration
(note that $t \geq 0$ since $t^{-1} \geq UP(a_j, \delta_{UP}) \geq 0$). 

It follows by Lemma \ref{lem:bss} and induction that for every $\tau$, we have $\lambda_k(\matM_{\tau}) \geq L_{\tau}$. 
Similarly, by
Lemma \ref{lem:U} and induction, for every $\tau$ it holds that $\Trace{\matW_{\tau}} \leq U_{\tau}$. 

In particular, for $\tau = r$ we have
\begin{eqnarray}\label{eqn:Lfinal}
\lambda_k(\matM_r) \geq L_r = r(1-\sqrt{k/r}), 
\end{eqnarray}
and
\begin{eqnarray}\label{eqn:Ufinal}
\Trace{\matW_r} \leq U_r = r ( 1- \sqrt{k/r})^{-1} \FNormS{\matA}.
\end{eqnarray}
Rescaling by $r^{-1}(1-\sqrt{k/r})$ in Step 3 of the algorithm therefore results in the guarantees on $\lambda_k(\matM_r)$
and $\Trace{\matW_r}$ claimed in the theorem statement. 

Finally, note that Algorithm \ref{alg1} runs in $r$ steps. 
The vector $s$ of weights is initialized to the all-zero 
vector, and one of its entries is updated in each iteration. Thus, $s$ will contain at most $r$
non-zero weights upon termination. As shown above, the value $t$ chosen in each iteration 
is non-negative, so the weights in $s$ are non-negative.

This completes the proof. 
\end{proof}

We will also need the following corollary, which shows how to perform the dual set sparsification much more
efficiently if we allow it to be randomized. 

\begin{corollary}\label{lem:dualnnz}
Let \math{\cl V=\{\ve_1,\ldots,\ve_n\}} be a decomposition of the identity, where \math{\ve_i\in \mathbb{R}^{k}} ($k < n$) and
$\sum_{i=1}^n\ve_i\ve_i^T=\matI_{k}$; let \math{\cl A=\{\a_1,\ldots,\a_n\}} be an arbitrary set
of vectors, where \math{\a_i\in\mathbb{R}^{\ell}}. Let $\matW \in \mathbb{R}^{\xi \times \ell}$ be a randomly chosen
sparse subspace embedding with $\xi = O( n^2 / \varepsilon^2 ) < \ell$, for some $0 < \varepsilon < 1$. 
Consider a new set of vectors  \math{\cl B = \{\matW\a_1,\ldots, \matW\a_n\}},
with \math{\matW\a_i \in \mathbb{R}^{\xi}}. Run Algorithm \ref{alg1} with \math{\cl V=\{\ve_1,\ldots,\ve_n\}},
\math{\cl B = \{\matW\a_1,\ldots, \matW\a_n\}}, and some integer \math{r} such that \math{k < r \le n}. Let the output
of this be a set of weights \math{s_i\ge 0} ($i=1\ldots n$), at most \math{r} of which are non-zero. 
Then, with probability at least $0.99,$
\eqan{
\lambda_{k}\left(\sum_{i=1}^ns_i \ve_i \ve_i^T\right)
&\ge&
\left(1 - \sqrt{\frac{k}{r}}\right)^2,
\qquad\\
\trace\left(\sum_{i=1}^n s_i \a_i \a_i^T\right)
&\le&
\frac{1 + \varepsilon }{1 - \varepsilon} \cdot
\trace\left(\sum_{i=1}^n \a_i \a_i^T\right) \\
&=&
\frac{1 + \varepsilon }{1 - \varepsilon} \cdot
\sum_{i=1}^n \TNormS{\a_i}.
}
Equivalently,
if $\matV \in \mathbb{R}^{n \times k}$ is a matrix whose rows are the vectors $\ve_i^T$,
$\matA \in \mathbb{R}^{n \times \ell}$ is a matrix whose rows are the vectors $\a_i^T$,
$\matB = \matA \matW^T \in \mathbb{R}^{n \times \xi}$ is a matrix whose rows are the vectors $\a_i^T \matW^T$,
and
$\matS \in \mathbb{R}^{n \times r}$ is the sampling matrix containing 
the weights $s_i > 0, $ then with probability at least $0.99,$
\eqan{
\sigma_{k}\left( \matV^T \matS\right)
\ge
1 - \sqrt{{k}/{r}},
\qquad
\FNormS{\matA^T\matS}
\le
\frac{1 + \varepsilon }{1 - \varepsilon} \cdot
\FNormS{\matA}.
}
The weights $s_i$ can be computed in $O\left( \nnz(\matA) + rnk^2+n \xi \right)$ time. We denote this procedure as
$$\matS = BssSamplingSparse(\matV, \matA, r, \varepsilon).$$
\end{corollary}
\begin{proof}
The algorithm constructs $\matS$ as follows,
$$ \matS = BssSampling(\matV, \matB, r). $$
The lower bound for the smallest singular value of $\matV$ is immediate from Theorem~\ref{thm:dualset}. That theorem also ensures,
$$
\FNormS{ \matB^T \matS } \le \FNormS{\matB^T}, 
$$ 
i.e.,
$$
\FNormS{ \matW \matA^T \matS } \le \FNormS{\matW \matA^T}.
$$
Since $\matW$ is an $\ell_2$-subspace embedding, we have that with probability at least $0.99$ 
and for all vectors $\y \in \mathbb{R}^{n}$ simultaneously,
$$
\left( 1 - \varepsilon \right) \TNormS{\matA^T\y} \le \TNormS{\matW \matA^T \y}.
$$
Apply this $r$ times for $\y \in \mathbb{R}^n$ being columns from $\matS \in \mathbb{R}^{n \times r}$ 
and take a sum on the resulting inequalities,
$$\left( 1 - \varepsilon \right) \FNormS{\matA^T \matS} \le \FNormS{\matW \matA^T \matS}.$$
Now, since $\matW$ is an $\ell_2$-subspace embedding, 
$$ \FNormS{\matW \matA^T} \le  \left( 1 + \varepsilon \right)   \FNormS{\matA^T}, $$
which can be seen by applying $\matW$ to each of the vectors $\matW\matA^T\e_i$. 
Combining all these inequalities together, we conclude that with probability at least $0.99,$
$$ \FNormS{\matA^T \matS}  \le \frac{1 + \varepsilon }{1 - \varepsilon} \cdot \FNormS{\matA^T}. $$
\end{proof}

\paragraph{Implications for CUR.} 
Returning to our CUR decomposition algorithm, letting $\matM = \matZ_1^T \mat\Omega_1\matD_1$ where
$\mat\Omega_1$ and $\matD_1$ are found using {\textsc RandSampling}$(\matZ, s, p)$, we apply 
Corollary \ref{lem:dualnnz} to compute 
$\matS_1 = ${\textsc BssSamplingSparse}$(\matV_{\matM},  (\matA-\matA\matZ_1\matZ_1^T)^T \mat\Omega_1\matD_1, 4k, .5)$,
where $\matV_{\matM}$ is determined by writing $\matM$ in its SVD as $\matM = \matU \mat\Sigma \matV_{\matM}^T$. 

At this point we set $\matC_1 = \matA\matOmega_1\matD_1\matS_1 \in \mathbb{R}^{m \times 4k}$ which contains 
$c_1 = 4k$ rescaled columns of $\matA$. 

\begin{lemma}\label{lem:conk}
With probability at least $.8$, 
$$\FNormS{\matA - \matC_1 \pinv{\matC_1}\matA} \leq 90 \cdot \FNormS{\matA-\matA_k}.$$
\end{lemma}
\begin{proof}
We apply Lemma \ref{lem:structural} with $\matZ = \matZ_1 \in \mathbb{R}^{n \times k}$ and 
$\matS = \mat\Omega_1 \matD_1 \matS_1 \in \mathbb{R}^{n \times c_1}$. First, we show that with probability $.9$, 
the rank assumption of Lemma \ref{lem:structural} is satisfied for our choice of $\matS$, namely, that 
rank$(\matZ^T \matS) = k$. We have 
$$\rank(\matZ^T \matS) =  \rank(\matZ_1^T\mat\Omega_1 \matD_1 \matS_1) = 
\rank(\matM\matS_1) = \rank(\matV_{\matM}^T\matS_1) = k,$$
where the first two equalities follow from the definitions, the third equality follows assuming the $1-\frac{1}{\poly(k)}$
event of (\ref{eqn:rank1}) that $\rank(\matM) = k$, and the last equality follows from
the fact that Corollary \ref{lem:dualnnz} guarantees that with probability at least $.98$,
$\sigma_{k}\left( \matV_M^T \matS\right) \geq \frac{1}{2}$. 

Now applying Lemma \ref{lem:structural} with the $\matC$ there equal
to $\matC_1$ and the $\matE$ there equal to $\matE_1 = \matA-\matA\matZ_1\matZ_1^T$, we have
\begin{eqnarray*}
\FNormS{\matA - \matC_1 \pinv{\matC_1}\matA} & \leq & \FNormS{\matA - \Pi_{\matC_1,k}(\matA)}\\
& \leq & \FNormS{\matA-\matC_1 \pinv{(\matZ_1 \mat\Omega_1 \matD_1 \matS_1)} \matZ_1^T}\\
& \leq & \FNormS{\matE_1} + \FNormS{\matE_1 \mat\Omega_1\matD_1\matS_1 \pinv{(\matZ_1\mat\Omega_1\matD_1\matS_1)}}.
\end{eqnarray*}
We have that $\FNormS{\matE_1 \matOmega_1 \matD_1 \matS_1(\matZ_1\matOmega_1 \matD_1 \matS_1)^{\dagger}}$ is
at most
\eqan{
&\buildrel(a)\over\le& \FNormS{\matE_1 \matOmega_1 \matD_1 \matS_1} \cdot \TNormS{(\matZ_1\matOmega_1 \matD_1 \matS_1)^{\dagger}} \\
&\buildrel(b)\over=  & \FNormS{\matE_1 \matOmega_1 \matD_1 \matS_1} \cdot \TNormS{(\matU_{\matM} \matSig_{\matM} \matV_{\matM}\transp \matS_1)^{\dagger}} \\
&\buildrel(c)\over=  & \FNormS{\matE_1 \matOmega_1 \matD_1 \matS_1} \cdot \TNormS{ \left(\matV_{\matM}\transp \matS_1\right)^{\dagger} \left( \matU_{\matM} \matSig_{\matM} \right)^{\dagger}} \\
&\buildrel(d)\over\le& \FNormS{\matE_1 \matOmega_1 \matD_1 \matS_1} \cdot \TNormS{ \left(\matV_{\matM}\transp \matS_1\right)^{\dagger}} \cdot \TNormS{ \left( \matU_{\matM} \matSig_{\matM} \right)^{\dagger}}  \\
&\buildrel(e)\over=  & \FNormS{\matE_1 \matOmega_1 \matD_1 \matS_1}  \cdot \frac{1}{ \sigma_k^2\left(\matV_{\matM}\transp \matS_1\right)}  \cdot \frac{1}{ \sigma_k^2\left( \matU_{\matM} \matSig_{\matM} \right)} \\
&\buildrel(f)\over\le & \FNormS{\matE_1 \matOmega_1 \matD_1 \matS_1} \cdot 8 \\
&\buildrel(g)\over\le&  \FNormS{\matE_1 \matOmega_1 \matD_1} \cdot 8 \\
&\buildrel(h)\over\le& 80 \FNormS{\matE_1}
}
where
(a) follows by the sub-multiplicativity property of matrix norms,
(b) follows by replacing $\matZ_1\matOmega_1 \matD_1 = \matM = \matU_{\matM} \matSig_{\matM} \matV_{\matM}\transp$,
(c) follows by the fact that $\matU_{\matM} \matSig_{\matM}$ is a full rank $k \times k$ matrix
assuming the $1-\frac{1}{\poly(k)}$ probability event of (\ref{eqn:rank1})
(d) follows by the sub-multiplicativity property of matrix norms,
(e) follows by the 
connection of the spectral norm of the pseudo-inverse with the singular values of the matrix to be pseudo-inverted,
(f) follows if the $1-\frac{1}{\poly(k)}$ event of (\ref{eqn:rank1}) occurs and the probability $.98$ event of 
Corollary \ref{lem:dualnnz} occurs, 
(g) follows by Corollary \ref{lem:dualnnz}, and
(h) follows by Lemma~\ref{lem:FrobLSS} and by adding a $0.1$ to the overall failure probability.
So, overall with probability at least $0.8,$
$$ \FNormS{\matE_1 \matOmega \matD \matS_1(\matZ_1\matOmega \matD \matS_1)^{\dagger}} \le 80 \FNormS{\matE_1},$$
Hence,
with the same probability,
$$\FNormS{ \matA - \matC_1 \pinv{\matC}_1\matA } \le\FNormS{\matE_1} +  80 \FNormS{\matE_1}.$$
By our choice of $\matZ$, $\FNormS{\matE_1} \leq \frac{10}{9} \FNormS{\matA - \matA_k}.$
Hence, with probability at least $0.8$,
$$ \FNormS{ \matA - \matC_1 \pinv{\matC}_1\matA } \le 90 \FNormS{\matA - \matA_k}.$$
\end{proof}

Lemma \ref{lem:conk} gives us a way to find $4k$ columns providing an $O(1)$-approximation. We would like
to refine this approximation to a $(1+\eps)$-approximation using only an additional $O(k/\eps)$ number of columns.
To do so, we perform a type of residual sampling from this $O(1)$-approximation, as described in the next section.

\subsubsection{Adaptive sampling}\label{sec:adaptiveSampling}
Given $O(k)$ columns providing a constant factor approximation, we can sample $O(k/\eps)$ additional 
columns from their ``residual'' to obtain a $(1+\eps)$-approximation. This was shown in the following
lemma of Deshpande, Rademacher, Vempala, and Wang. It is actually more general in the sense that the matrix $\matV$
in the statement of the theorem need not be a subset of columns of $\matA$. 

\begin{theorem}[Theorem 2.1 of~\cite{DRVW06}]
\label{thm:adaptivecolumns}
Given $\matA \in \mathbb{R}^{m \times n}$ and $\matV \in \mathbb{R}^{m \times c_1}$ (with $c_1 \le n, m$),
define the residual
$$\matB = \matA - \matV \matV^{\dagger} \matA \in \mathbb{R}^{m \times n}.$$
For $i=1,\ldots,n$, and some fixed constant  $\alpha >0 ,$ let $p_i$ be a probability distribution such that for each $i:$
$$p_i \ge \alpha {\TNormS{\b_{i}}}/{\FNormS{\matB}},$$
where $\b_i$ is the $i$-th column of the matrix $\matB$. Sample
$c_2$ columns from $\matA$ in \math{c_2} i.i.d. trials, where in each trial the $i$-th column is chosen with probability $p_i$.
Let $\matC_2 \in \mathbb{R}^{m \times c_2}$ contain the $c_2$ sampled columns and let $\matC = [\matV\ \ \matC_2] \in \mathbb{R}^{m \times (c_1+c_2)}$
contain the columns of $\matV$ and $\matC_2$. 
Then, for any integer $k > 0$,
$$\Expect{ \FNormS{ \matA - \Pi_{\matC,k}^{\mathrm{F}}(\matA) } } \le \FNormS{\matA - \matA_k } 
+ \frac{k}{\alpha \cdot c_2} \FNormS{\matA - \matV \matV^{\dagger} \matA}.$$
We denote this procedure as
$$\matC_2 = AdaptiveCols(\matA, \matV, \alpha, c_2).$$
Given $\matA$ and $\matV,$ the above algorithm requires $O( c_1 m n + c_2 \log c_2 )$ arithmetic operations to find $\matC_2$.
\end{theorem}
Rather than prove Theorem \ref{thm:adaptivecolumns} directly, we will prove the following theorem of Wang and Zhang 
which generalizes
the above theorem. One can think
of the following theorem as analyzing the deviations of $\matA$ from an
arbitrary space - the row space of $\matR$, that occur via sampling additional
columns according to the residual from a given set of columns. These columns
may have nothing to do with $\matR$. This therefore generalizes the result
of Deshpande and Vempala which considered $\matR$ to be the top $k$ right singular
vectors of $\matA$ (we will make this generalization precise below). 
\begin{theorem}[Theorem 4 in~\cite{WZ13CUR}]
\label{thm:adaptiverows}
Given $\matA \in \mathbb{R}^{m \times n}$ and a matrix $\matR \in \mathbb{R}^{r \times n}$ such that 
$$\rank(\matR) = \rank(\matA \pinv{\matR} \matR) = \rho,$$
with $\rho \le r \le n,$
we let $\matC_1 \in \mathbb{R}^{m \times c_1}$ consist of $c_1$ columns of $\matA$ 
and define the residual 
$$\matB = \matA - \matC_1 \pinv{\matC_1} \matA \in \mathbb{R}^{m \times n}.$$
For $i=1,\ldots,n$ 
let $p_i$ be a probability distribution such that for each $i:$
$$p_i \geq \alpha {\TNormS{\b_{i}}}/{\FNormS{\matB}},$$
where $\b_i$ is the $i$-th column of $\matB$. Sample
$c_2$ columns from $\matA$ in \math{c_2} i.i.d. trials, 
where in each trial the $i$-th column is chosen with probability $p_i$.
Let $\matC_2 \in \mathbb{R}^{m \times c_2}$ 
contain the $c_2$ sampled columns and let $\matC = [\matC_1 , \matC_2] \in \mathbb{R}^{m \times c_2}$.
Then,
$$\Expect{ \FNormS{ \matA - \matC\pinv{\matC}\matA \pinv{\matR}\matR } }
\le \FNormS{\matA - \matA \pinv{\matR}\matR } + \frac{\rho}{\alpha c_2} \FNormS{\matA - \matC_1 \pinv{\matC_1} \matA}.$$
We denote this procedure as
$$\matC_2 = AdaptiveCols(\matA, \matR, \matC_1, \alpha, c_2).$$
Given $\matA,$ $\matR,$ $\matC_1,$
the above algorithm requires $O( c_1 m n + c_2 \log c_2  )$ arithmetic operations to find $\matC_2$.
\end{theorem}
\begin{proof}
Write $\matA\matR^{\dagger} \matR$ in its SVD as $\matU \mat\Sigma \matV^T$. 
The key to the proof is to define the following matrix
$$\matF = \left (\sum_{q=1}^{\rho} \sigma_q^{-1} \w_q \u_q^T \right )\matA\matR^{\dagger} \matR,$$
where $\sigma_q$ is the $q$-th singular value of $\matA\matR^{\dagger} \matR$ with corresponding
left singular vector $\u_q$. The $\w_q \in \mathbb{R}^m$ are random
column vectors which will depend on the sampling and have certain desirable properties described
below.  

To analyze the expected error of the algorithm with respect to the choices made in the sampling
procedure, we have
\begin{eqnarray}\label{eqn:open}
{\bf E}\FNormS{\matA-\matC\matC^{\dagger}\matA\matR^{\dagger} \matR}
& = & {\bf E} \FNormS{\matA-\matA\matR^{\dagger} \matR + \matA\matR^{\dagger}\matR - \matC\matC^{\dagger}\matA\matR^{\dagger}\matR} \notag \\
& = & {\bf E} \FNormS{\matA-\matA\matR^{\dagger} \matR} \notag \\
&& + {\bf E} \FNormS{\matA\matR^{\dagger}\matR-\matC\matC^{\dagger}\matA\matR^{\dagger}\matR}.
\end{eqnarray}
where the second equality uses the Pythagorean theorem. 

One property of the $\w_q$ we will ensure below is that the $\w_q$ each lie in the span of the columns
of $\matC$. Given this, we have
\begin{eqnarray}\label{eqn:close}
\FNormS{\matA\matR^{\dagger}\matR-\matC\matC^{\dagger}\matA\matR^{\dagger}\matR} 
& \leq & \FNormS{\matA\matR^{\dagger} \matR - \matW\matW^{\dagger}\matA\matR^{\dagger} \matR} \notag \\
& \leq & \FNormS{\matA\matR^{\dagger} \matR - \matF}. 
\end{eqnarray}
Plugging (\ref{eqn:close}) into (\ref{eqn:open}), we have
\begin{eqnarray}\label{eqn:combined}
{\bf E}[\FNormS{\matA-\matC\matC^{\dagger}\matA\matR^{\dagger} \matR}] 
\leq \FNormS{\matA-\matA\matR^{\dagger}\matR} + {\bf E}\left [\FNormS{\matA\matR^{\dagger} \matR - \matF} \right ],
\end{eqnarray}
where note that $\matR$ is deterministic so we can remove the expectation. 

Let $\ve_1, \ldots, \ve_{\rho}$ be the right singular vectors of $\matA\matR^{\dagger}\matR$. As both the rows of
$\matA\matR^{\dagger}\matR$ and of $\matF$ lie in the span of $\ve_1, \ldots, \ve_{\rho}$, we can decompose
(\ref{eqn:combined}) as follows:
\begin{eqnarray}\label{eqn:ultimate}
{\bf E}[\FNormS{\matA-\matC\matC^{\dagger}\matA\matR^{\dagger} \matR}] 
& \leq & \FNormS{\matA-\matA\matR^{\dagger}\matR} \notag \\
& + & {\bf E}\left [\FNormS{\matA\matR^{\dagger} \matR - \matF} \right ] \notag \\
& \leq & \FNormS{\matA-\matA\matR^{\dagger}\matR} \notag \\
&+& \sum_{j=1}^{\rho} {\bf E}\|(\matA\matR^{\dagger}\matR-\matF)\ve_j\|_2^2 \notag \\
& = & \FNormS{\matA-\matA\matR^{\dagger}\matR} \notag \\
&+& \sum_{j=1}^{\rho} {\bf E}\|\matA\matR^{\dagger}\matR\ve_j - \sum_{q=1}^{\rho}\sigma_q^{-1}\w_q \u_q^T \sigma_j \u_j\|_2^2 \notag \\
& = & \FNormS{\matA-\matA\matR^{\dagger}\matR} \notag \\
& + & \sum_{j=1}^{\rho} {\bf E} \|\matA\matR^{\dagger}\matR\ve_j - \w_j\|_2^2 \notag \\
& = & \FNormS{\matA-\matA\matR^{\dagger}\matR} \notag \\
& + & \sum_{j=1}^{\rho} {\bf E} \|\matA\ve_j - \w_j\|_2^2,
\end{eqnarray}
where the final equality follows from the fact that $\ve_j$ is, by definition, in the row space of $\matR$, 
and so $$\matA\ve_j - \matA\matR^{\dagger}\matR\ve_j = \matA(\matI-\matR^{\dagger}\matR)\ve_j = \vec0.$$
Looking at (\ref{eqn:ultimate}), it becomes clear what the properties of the $\w_j$ are that we want. Namely, 
we want them to be in the column space of $\matC$ and to have the property that ${\bf E}\|\matA\ve_j - \w_j\|_2^2$
is as small as possible. 

To define the $\w_j$ vectors, we begin by defining auxiliary random variables $\x_{j, (\ell)} \in \mathbb{R}^m$,
for $j = 1, \ldots, \rho$ and $\ell = 1, \ldots, c_2$:
$$\x_{j, (\ell)} = \frac{\ve_{i,j}}{p_i}\b_i = \frac{\ve_{i,j}}{p_i} \left (\a_i - \matC_1 \matC_1^{\dagger}\a_i \right ),$$
with probability $p_i$, for $i = 1, \ldots, n$. 
We have that $\x_{j, (\ell)}$ is a deterministic linear combination of a random column sampled from the distribution
defined in the theorem statement. Moreover,
\begin{eqnarray*}
{\bf E}[\x_{j, (\ell)}] & = & \sum_{i=1}^n p_i \frac{\ve_{i,j}}{p_i} \b_i = \matB \ve_j,
\end{eqnarray*}
and
\begin{eqnarray*}
{\bf E}\|\x_{j, (\ell)}\|_2^2 & = & \sum_{i=1}^n p_i \frac{\ve_{i,j}^2}{p_i^2} \|\b_i\|_2^2
\leq \sum_{i=1}^n \frac{\ve_{i,j}^2}{\alpha \|\b_i\|_2^2/\FNormS{\matB}} \|\b_i\|_2^2
= \frac{\FNormS{\matB}}{\alpha}.
\end{eqnarray*}
We now define the average vector
$$\x_j = \frac{1}{c_2} \sum_{\ell = 1}^{c_2} \x_{j, (\ell)},$$
and we have
\begin{eqnarray*}
{\bf E}[\x_j] & = & {\bf E}[\x_{j, (\ell)}] = \matB\ve_j,
\end{eqnarray*}
and
\begin{eqnarray}\label{eqn:variance}
{\bf E}\|\x_j-\matB\ve_j\|_2^2 & = & {\bf E}\|\x_j - {\bf E}[\x_j]\|_2^2 \notag \\
& = & \frac{1}{c_2} {\bf E}\|\x_{j, (\ell)}-{\bf E}[\x_{j, (\ell)}]\|_2^2\\
& = & \frac{1}{c_2} {\bf E}\|\x_{j, (\ell)}-\matB\ve_j\|_2^2 \notag,
\end{eqnarray}
where (\ref{eqn:variance}) follows from the fact that the samples are independent. In fact, pairwise independence
suffices for this statement, which we shall use in our later derandomization of this theorem
(note that while for
fixed $j$, the $\x_{j, (\ell)}$ are pairwise independent, for two different $j, j'$ we have that $\x_j$ 
and $\x_{j'}$ are dependent). 

Notice that $\x_j$ is in the span of the columns of $\matC$, for every $j = 1, \ldots, n$.
This follows since $\x_{j, (\ell)}$ is a multiple of a column in $\matC_2$. 

For $j = 1, \ldots, \rho$, we now define
\begin{eqnarray}\label{eqn:w}
\w_j = \matC_1 \matC_1^{\dagger} \matA \ve_j + \x_j,
\end{eqnarray}
and we also have that $\w_1, \ldots, \w_{\rho}$ are in the column space of $\matC$, as required above. It remains
to bound ${\bf E}\|\w_j - \matA\ve_j\|_2^2$ as needed for (\ref{eqn:ultimate}). 

We have
\begin{eqnarray}\label{eqn:relation}
{\bf E}[\w_j] & = & \matC_1\matC_1^{\dagger} \matA\ve_j + {\bf E}[\x_j] \notag \\
& = & \matC_1 \matC_1^{\dagger}\matA\ve_j + \matB\ve_j \notag \\
& = & \matA\ve_j, 
\end{eqnarray}
where (\ref{eqn:relation}) together with (\ref{eqn:w}) imply that
$$\w_j - \matA\ve_j = \x_j - \matB\ve_j.$$
At long last we have
\begin{eqnarray}\label{eqn:varBound}
{\bf E}\|\w_j - \matA\ve_j\|_2^2 & = & {\bf E}\|\x_j - \matB\ve_j\|_2^2 \notag \\
& = & \frac{1}{c_2} {\bf E}\|\x_{j, (\ell)}-\matB\ve_j\|_2^2 \notag \\
& = & \frac{1}{c_2} {\bf E}\|\x_{j, (\ell)}\|_2^2 - \frac{2}{c_2}(\matB\ve_j)^T {\bf E}[\x_{j, (\ell)}] + \frac{1}{c_2} \|\matB\ve_j\|_2^2\notag \\
& = & \frac{1}{c_2} {\bf E}\|\x_{j, (\ell)}\|_2^2 - \frac{1}{c_2} \|\matB\ve_j\|_2^2 \notag \\
& \leq & \frac{1}{\alpha c_2} \FNormS{\matB} - \frac{1}{c_2} \|\matB\ve_j\|_2^2 \notag \\
& \leq & \frac{1}{\alpha c_2} \FNormS{\matB}.
\end{eqnarray}
Plugging (\ref{eqn:varBound}) into (\ref{eqn:ultimate}), we obtain
$${\bf E}[\FNormS{\matA-\matC\matC^{\dagger}\matA\matR^{\dagger} \matR}] \leq 
\FNormS{\matA-\matA\matR^{\dagger}\matR} + \frac{\rho}{\alpha c_2} \FNormS{\matA-\matC_1 \matC_1^{\dagger}\matA},$$
which completes the proof. 
\end{proof}

\subsubsection{CUR wrapup}\label{sec:CURWrapup}
\paragraph{Obtaining a Good Set of Columns.}
We will apply 
Theorem \ref{thm:adaptivecolumns} with the $\matV$ 
of that theorem set to $\matC_1$. For the distribution $p$, we need
to quickly approximate the column norms of 
$\matB = \matA-\matC_1 \pinv{\matC_1} \matA$. To do so, 
by Lemma \ref{lem:jl}
it suffices
to compute $\matG \cdot \matB$, where $\matG$ is an $t \times m$
matrix of i.i.d. $N(0,1/t)$ random variables, for $t = O(\log n)$. By
Lemma \ref{lem:jl}, with probability at least $1-1/n$, simultaneously
for all $i \in [n]$, 
$$\frac{\|\b_i\|_2^2}{2} \leq \|(\matG\matB)_{*i}\|_2^2 \leq \frac{3}{2} \|\b_i\|_2^2,$$
where $\b_i = \matB_{*i}$ is the $i$-th column of $\matB$. It follows that we
can set $\alpha = \frac{1}{3}$ in Theorem \ref{thm:adaptivecolumns} using
the distribution $p$ on $[n]$ given by
$$\forall i \in [n], \ p_i = \frac{\|(\matG\matB)_{*i}\|^2_2}{\FNormS{\matG \matB}}.$$
Hence, for a parameter $c_2 > 0$, if we set
$$\matC_2 = AdaptiveCols(\matA, \matC_1, \frac{1}{3}, c_2),$$
if $\matC = [\matC_1, \matC_2]$, where $\matC_2$ are the columns sampled by 
$AdaptiveCols(\matA, \matC_1, \frac{1}{3}, c_2),$
then by the conclusion of Theorem \ref{thm:adaptivecolumns}, 
$$\Expect{ \FNormS{ \matA - \Pi_{\matC,k}^{\mathrm{F}}(\matA) } } \le \FNormS{\matA - \matA_k } 
+ \frac{3k}{c_2} \FNormS{\matA - \matC_1 \matC_1^{\dagger} \matA}.$$
By Lemma \ref{lem:conk}, with probability at least $.8$, 
$\FNormS{\matA - \matC_1 \matC_1^{\dagger} \matA} \leq 90 \FNormS{\matA - \matA_k }$, 
which we condition on. It follows by setting $c_2 = 270k/\eps$, then 
taking expectations with respect to the randomness
in $AdaptiveCols(\matA, \matC_1, \frac{1}{3}, c_2),$ we have
$$\Expect{ \FNormS{ \matA - \Pi_{\matC,k}^{\mathrm{F}}(\matA) } } \le (1+\eps) \FNormS{\matA - \matA_k }.$$

\paragraph{Running Time.}
A few comments on the running time are in order. We can compute $\matZ$ in
$O(\nnz(A)) + (m+n) \cdot \poly(k)$ time via Theorem \ref{thm:lowrank}. Given $\matZ$, we
can run {\textsc RandSampling}$(\matZ, s, p)$, where $s = O(k \log k)$ and $p$ is the leverage
score distribution defined by $\matZ$. This can be done in $n \cdot \poly(k/\eps)$ time. 

We then run 
{\textsc BssSamplingSparse}$(\matV_{\matM},  (\matA-\matA\matZ_1\matZ_1^T)^T \mat\Omega_1\matD_1, 4k, .5)$.
To do this efficiently, we can't afford to explicitly compute the matrix
$(\matA-\matA\matZ_1\matZ_1^T)^T \mat\Omega_1\matD_1$. 
We only form  $\matA\matOmega \matD$ and $\matZ_1\transp \matOmega \matD$ in $O(\nnz(A)) + n \cdot \poly(k)$ time. 
Then, {\textsc BssSamplingSparse} multiplies 
$\left(\matA - \matA \matZ_1\matZ_1\transp\right)\matOmega \matD$
from the left with a sparse subspace embedding matrix $\matW \in \R^{\xi \times m}$ with 
$\xi = O( k^2 \log^2 k)$. Computing $\matW \matA$ takes $O\left(  \nnz(\matA) \right)$ time.
Then, computing $(\matW\matA) \matZ_1$ and $(\matW\matA\matZ_1)\matZ_1\transp$ takes another 
$O(\xi m k)$ + $O(\xi n k)$ time, respectively. Finally, the sampling algorithm on
$\matW (\matA - \matA \matZ_1\matZ_1\transp)\matOmega \matD$ is $O( k^4 \log k + m k \log k)$ time.

Given $\matA, \matZ_1, \mat\Omega, \matD$ and $\matS_1$ we then know the matrix $\matC_1$ needed to run
$AdaptiveCols(\matA, \matC_1, \frac{1}{3}, c_2).$ The latter algorithm samples columns of $\matA$,
which can be done in $O(\nnz(A)) + n k/\eps$ time given the distribution $p$ to sample from. Here to find
$p$ we need to compute $\matG \cdot \matB$, where $\matB = \matA-\matC_1 \pinv{\matC_1} \matA$ and $\matG$ is 
an $O(\log n) \times m$ matrix. We can compute this matrix product in time $O(\nnz(A) \log n) + (m+n)\poly(k/\eps)$. 

It follows that the entire procedure to find $\matC$ is $O(\nnz(A) \log n) + (m+n) \poly(k/\eps)$ time. 

\paragraph{Simultaneously Obtaining a Good Set of Rows.}
At this point we have a set $\matC$ of $O(k/\eps)$ columns of $\matA$ for which
$$\Expect{ \FNormS{ \matA - \Pi_{\matC,k}^{\mathrm{F}}(\matA) } } \le (1+\eps) \FNormS{\matA - \matA_k }.$$
If we did not care about running time, we could now find the best $k$-dimensional subspace of the columns of
$\matC$ for approximating the column space of $\matA$, that is, if $\matU$ has orthonormal columns with the
same column space as $\matC$, then by Lemma \ref{lem:pythagorean},
$$\Expect{ \FNormS{ \matA - \matU [\matU^T\matA]_k } }\le (1+\eps) \FNormS{\matA - \matA_k },$$
where $[\matU^T\matA]_k$ denotes the best rank-$k$ approximation to $\matU^T\matA$ in Frobenius norm. So if
$\matL$ is an $m \times k$ matrix with orthonormal columns with the same column space as 
$\matU [\matU^T\matA]_k$, we could then attempt to execute the analogous algorithm to the one that 
we just ran. That algorithm was for finding a good
set of columns $\matC$ starting with $\matZ$, and now we would like to find a good set $\matR$ of rows 
starting with $\matL$. This is the proof strategy used by Boutsidis and the author in \cite{BW14}. 

Indeed, the algorithm of \cite{BW14} works by first sampling $O(k \log k)$ rows of $\matA$ according
to the leverage scores of $\matL$. It then downsamples this to $O(k)$ rows using {\textsc BssSamplingSparse}. Now,
instead of using Theorem \ref{thm:adaptivecolumns}, the algorithm invokes Theorem \ref{thm:adaptiverows}, applied to
$\matA^T$, to find $O(k/\eps)$ rows. 

Applying Theorem \ref{thm:adaptiverows} to $\matA^T$, the error has the form:
\begin{eqnarray}\label{eqn:rows}
{\bf E} \FNormS{ \matA - \matV\pinv{\matV}\matA \pinv{\matR}\matR } 
\le \FNormS{\matA - \matV\pinv{\matV}\matA } + \frac{\rho}{r_2} \FNormS{\matA - \matA\pinv{\matR}_1\matR_1}
\end{eqnarray}
where $\rho$ is the rank of $\matV$. Note that had we used $\matU$ here in place of $\matL$, $\rho$ could be $\Theta(k/\eps)$,
and then the number $r_2$ of samples we would need in (\ref{eqn:rows}) would be $\Theta(k/\eps^2)$, which is more than
the $O(k/\eps)$ columns and $O(k/\eps)$ rows we could simultaneously hope for. It turns out that these procedures
can also be implemented in $O(\nnz(\matA))\log n + (m+n)\poly(k/\eps)$ time.

We glossed over the issue of how to find the best $k$-dimensional subspace $\matL$ of the columns of $\matC$ for approximating
the column space of $\matA$, as described above. Na\"ively doing this would involve projecting the columns of $\matA$
onto the column space of $\matC$, which is too costly. Fortunately, by Theorem \ref{thm:kv}
in \S\ref{sec:dislra}, 
in $O(\nnz(\matA)) + (m+n)\poly(k/\eps)$ time it is 
possible to find an $m \times k$ matrix $\matL'$ with orthonormal columns so that
$$\FNorm{\matA - \matL'(\matL')^T\matA} \leq (1+\eps)\FNorm{\matA-\matL\matL^T \matA}.$$
Indeed, Theorem \ref{thm:kv} implies that if $\matW$ is an $\ell_2$-subspace embedding, then we can take
$\matL'$ to be the top $k$ left singular vectors of $\matU\matU^T\matA\matW$, and since $\matU\matU^T$ has rank $O(k/\eps)$,
this matrix product can be computed in $O(\nnz(\matA)) + (m+n) \cdot \poly(k/\eps)$ time using sparse subspace embeddings.  
We can thus use $\matL'$ in place of $\matL$ in the algorithm for selecting a subset of $O(k/\eps)$ rows of $\matA$. 

\paragraph{Finding a $\matU$ With Rank $k$.}
The above outline shows how to simultaneously obtain a matrix $\matC$ and a matrix $\matR$ with $O(k/\eps)$ columns
and rows, respectively. Given such a $\matC$ and a $\matR$, we need to find a rank-$k$ matrix $\matU$ which is the
minimizer to the problem
$$\min_{\textrm{rank}-k \matU}\FNorm{\matA - \matC\matU\matR}.$$
We are guaranteed that there is such a rank-$k$ matrix $\matU$ since crucially, when we apply 
Theorem \ref{thm:adaptiverows}, we apply it with $\matV = \matL$, which has rank $k$. Therefore, the resulting
approximation $\matV\pinv{\matV}\matA\pinv{\matR}\matR$ is a rank-$k$ matrix, and since $\matL$ is in the span
of $\matC$, can be expressed as $\matC\matU\matR$. It turns out one can quickly find $\matU$, as shown in \cite{BW14}. We omit the details. 

\paragraph{Deterministic CUR Decomposition.}
The main idea in \cite{BW14} 
to achieve a CUR Decomposition with the same $O(k/\eps)$ columns and rows and a rank-$k$ matrix $\matU$
deterministically is to derandomize Theorem \ref{thm:adaptivecolumns} and Theorem \ref{thm:adaptiverows}. The 
point is that the proofs involve the second moment method, and therefore by a certain discretization
of the sampling probabilities, one can derandomize the algorithm using pairwise-independent samples (of either columns
or rows, depending on whether one is derandomizing Theorem \ref{thm:adaptivecolumns} or Theorem \ref{thm:adaptiverows}). This
increases the running time when applied to an $n \times n$ matrix $\matA$ to $n^4 \cdot \poly(k/\eps)$, 
versus, say, $n^3$ using
other deterministic algorithms such as the SVD, but gives an actual subset of rows and columns. 

\subsection{Spectral norm error}\label{sec:spectral}
Here we show how to quickly obtain a $(1+\eps)$ rank-k approximation with respect to the spectral norm
$\|\matA\|_2 = \sup_{\x} \frac{\|\matA\x\|_2}{\|\x\|_2}$. That is, given an $m \times n$ matrix $\matA$, compute a rank-$k$
matrix $\tilde{\matA}_k$, where $\|\matA-\tilde{\matA}_k\|_2 \leq (1+\eps) \|\matA-\matA_k\|_2$. 

It is well-known that $\|\matA-\matA_k\|_2 = \sigma_{k+1}(\matA)$, where $\sigma_{k+1}(\matA)$ is the $(k+1)$-st singular
value of $\matA$, 
and that $\matA_k$ is the matrix $\matU \mat\Sigma_k \matV^T$, where $\matU \mat\Sigma \matV^T$ is the SVD of $\matA$ and
$\mat\Sigma_k$ is a diagonal matrix with first $k$ diagonal entries equal to those of $\mat\Sigma$, and $0$ otherwise. 

Below we present an algorithm, proposed by Halko, Martinsson and Tropp \cite{HMT}, that was shown 
by the authors to be a bicriteria
rank-$k$ approximation. That is, they efficiently find an $n \times 2k$ matrix $\matZ$ with orthonormal columns for which
 $\|\matA-\matZ\matZ^T\matA\|_2 \leq (1+\eps)\|\matA-\matA_k\|_2$. By slightly modifying their analysis, 
this matrix $\matZ$ can be shown to have dimensions $n \times (k+4)$ with the same error 
guarantee. 
The analysis of this algorithm was somewhat simplified by
Boutsidis, Drineas, and Magdon-Ismail \cite{BDM11a}, and by slightly modifying their analysis, this results in an 
$n \times (k+2)$ matrix $\matZ$ with orthonormal
columns for which $\|\matA-\matZ\matZ^T\matA\|_2 \leq (1+\eps)\|\matA-\matA_k\|_2$. 
We follow the analysis of \cite{BDM11a}, but 
simplify and improve it slightly in order to output 
a true rank-$k$ approximation, that is, an $n \times k$ matrix $\matZ$ with orthonormal columns for which 
$\|\matA-\matZ\matZ^T\matA\|_2 \leq (1+\eps)\|\matA-\matA_k\|_2$. This gives us a new result
which has not appeared in the literature to the best of our knowledge. 

Before presenting the algorithm, we need the following lemma. 
Suppose we have an $n \times k$ matriz $\matZ$ with orthonormal
columns for which there exists an $\matX$ for which $\|\matA-\matZ\matX\|_2 \leq (1+\eps)\|\matA-\matA_k\|_2$. How do we find
such an $\matX$? It turns out the optimal such $\matX$ is equal to $\matZ^T \matA$.

\begin{lemma}\label{lem:kinside}
If we let $\matX^* = \textrm{argmin}_{\matX} \|\matA-\matZ\matX\|_2$, then $\matX^*$ 
satisfies $\matZ\matX^* = \matZ\matZ^T \matA$. 
\end{lemma} 
\begin{proof}
On the one hand, $\|\matA-\matZ\matX^*\|_2 \leq \|\matA-\matZ\matZ^T \matA\|_2$, since $\matX^*$ is the minimizer. On the other hand,
for any vector $\ve$, by the Pythagorean theorem, 
\begin{eqnarray*}
\|(\matA-\matZ\matX^*)\ve\|_2^2 & = & \|(\matA-\matZ\matZ^T\matA)\ve\|_2^2 + \|(\matZ\matZ^T\matA - \matZ\matX^*)\ve\|_2^2\\
& \geq & \|(\matA-\matZ\matZ^T\matA)\ve\|_2^2,
\end{eqnarray*}
and so $\|\matA-\matZ\matX^*\|_2 \geq \|\matA-\matZ\matZ^T\matA\|_2$. 
\end{proof}
We also collect a few facts about the singular values of a Gaussian matrix. 
\begin{fact}(see, e.g., \cite{RV10})\label{fact:gaussian}
Let $\matG$ be an $r \times s$ matrix of i.i.d. normal random variables with mean $0$ and variance $1$. 
There exist constants $C, C' > 0$ for which

(1) The maximum singular value $\sigma_1(\matG)$ satisfies $\sigma_1(\matG) \leq C\sqrt{\max(r,s)}$ with
probability at least $9/10$. 

(2) If $r = s$, then the minimum singular value $\sigma_r(\matG)$ satisfies 
$\sigma_r(\matG) \geq C'/\sqrt{r}$ with probability at least $9/10$.
\end{fact}

The algorithm, which we call {\sf SubspacePowerMethod} is as follows. The intuition is, like
the standard power method, if we compute $(\matA \matA^T)^q \matA \g$ for a 
random vector $\g$, then for large
enough $q$ this very quickly converges to the top left singular vector of $\matA$. If we instead
compute $(\matA \matA^T)^q \matA \matG$ for a random $n \times k$ matrix $\matG$, for large
enough $q$ this also very quickly converges to an $n \times k$ matrix which is close, in a certain
sense, to the top $k$ left singular vectors of $\matA$. 
\begin{enumerate}
\item Compute $\matB = (\matA\matA^T)^q \matA$ and $\matY = \matB\matG$, where $\matG$ 
is an $n \times k$ matrix of i.i.d. $N(0,1)$ random variables.
\item Let $\matZ$ be an $n \times k$ matrix with orthonormal columns whose column space is equal to that of $\matY$.
\item Output $\matZ\matZ^T \matA$. 
\end{enumerate}

In order to analyze {\sf SubspacePowerMethod}, we need a key lemma shown in \cite{HMT} concerning powering of a matrix.
\begin{lemma}\label{lem:tropp}
Let $\matP$ be a projection matrix, i.e., $\matP = \matZ\matZ^T$ for a matrix $\matZ$ with orthonormal columns. For any matrix $\matX$ of
the appropriate dimensions and integer $q \geq 0$,
$$\|\matP\matX\|_2 \leq (\|\matP(\matX\matX^T)^q\matX\|_2)^{1/(2q+1)}.$$
\end{lemma}
\begin{proof}
Following \cite{HMT}, we first show that if $\matR$ is a projection matrix and $\matD$ a non-negative diagonal matrix,
then $\|\matR\matD\matR\|_2^t \leq \|\matR\matD^t\matR\|_2$. To see this, suppose $\x$ is a unit vector for which
$\x^T\matR\matD\matR\x = \|\matR\matD\matR\|_2$. We can assume that $\|\matR\x\|_2 = 1$, as otherwise since $\matR$ is a projection matrix,
$\|\matR\x\|_2 < 1$, and taking the unit vector $\z = \matR\x/\|\matR\x\|_2$, we have
$$\z^T \matR \matD\matR \z = \frac{\x\matR^2 \matD \matR^2 \x}{\|\matR\x\|_2^2} = \frac{\x^T\matR\matD\matR\x}{\|\matR\x\|_2^2} > \x^T\matR\matD\matR\x,$$
contradicting that $\x^t\matR\matD\matR\x = \|\matR\matD\matR\|_2$. We thus have,
\begin{eqnarray*}
\|\matR\matD\matR\|^t & = & (\x^T \matR\matD\matR \x)^t = (\x^T\matD\x)^t = (\sum_j \matD_{j,j} \x_j^2)^t\\
& \leq & \sum_j \matD_{j,j}^t \x_j^2 = \x^T \matD^t \x = (\matR\x)^T \matD^t \matR\x\\
& \leq & \|\matR\matD^t \matR\|_2,
\end{eqnarray*}
where we have used Jensen's inequality to show that $(\sum_j \matD_{j,j} \x_j^2)^t \leq \sum_j \matD_{j,j}^t \x_j^2$, noting
that $\sum_j \x_j^2 = 1$ and the function $z \rightarrow |z|^t$ is convex.

Given this claim, let $\matX = \matU \mat\Sigma \matV^T$ be a decomposition of $\matX$
in which $\matU$ and $\matV^T$ are square matrices with orthonormal columns and rows, 
and $\Sigma$ has non-negative entries on the diagonal (such a decomposition can be obtained
from the SVD). Then,
\begin{eqnarray*}
\|\matP\matX\|_2^{2(2q+1)} & = & \|\matP \matX\matX^T \matP\|_2^{2q+1}\\
& = & \|(\matU^T \matP \matU) \mat\Sigma^2 (\matU^T \matP \matU)\|_2^{2q+1}\\
& \leq & \|(\matU^T \matP \matU) \mat\Sigma^{2(2q+1)} (\matU^T \matP \matU)\|_2\\
& = & \|\matP (\matX\matX^T)^{(2q+1)}\matP\|_2 \\
& = & \|\matP(\matX\matX^T)^q \matX \matX^T (\matX\matX^T)^q \matP\|_2\\
& = & \|\matP(\matX\matX^T)^q \matX\|_2^2,
\end{eqnarray*}
where the first equality follows since $\|\matP\matX\|_2^2 = \|\matP\matX\matX^T\matP\|_2$,
the second equality uses that $\matX\matX^T = \matU \mat\Sigma^2 \matU^T$ and rotational
invariance given that $\matU$ has orthonormal rows and columns, the first inequality
uses the claim
above with $\matR = \matU^T \matP \matU$, the next equality uses that $\matX\matX^T =
\matU \mat\Sigma^2 \matU^T$, the next equality regroups terms, and the final equality
writes the operator norm as the equivalent squared operator norm. 

If we raise both sides to the $1/(2(2q+1))$-th power, then this completes the proof. 
\end{proof}
We can now prove the main theorem about {\sf SubspacePowerMethod}
\begin{theorem}\label{thm:tropp}
For appropriate $q = O(\log(mn)/\eps)$, 
with probability at least $4/5$, {\sf SubspacePowerMethod} outputs a rank-$k$ matrix $\matZ\matZ^T \matA$ for which
$\|\matA-\matZ\matZ^T\matA\|_2 \leq (1+\eps)\|\matA-\matA_k\|_2$.
Note that {\sf SubspacePowerMethod} can be implemented in $O(\nnz(\matA)k \log(mn)/\eps)$ time. 
\end{theorem}
\begin{proof}
By Lemma \ref{lem:kinside}, $\matZ\matZ^T\matA$ is the best rank-$k$ approximation of $\matA$ in the column space of $\matZ$
with respect to the spectral norm. Hence, 
\begin{eqnarray*}
\|\matA - \matZ\matZ^T\matA\|_2 & \leq & \|\matA-(\matZ\matZ^T\matB)(\matZ\matZ^T\matB)^{\dagger}\matA\|_2\\
& = & \|(\matI-(\matZ\matZ^T\matB)(\matZ\matZ^T\matB)^{\dagger})\matA\|_2,
\end{eqnarray*}
where the inequality follows since $\matZ\matZ^T\matB$ is of rank $k$ and in the column space of $\matZ$. Since
$\matI-(\matZ\matZ^T\matB)(\matZ\matZ^T\matB)^{\dagger}$ is a projection matrix, we can apply Lemma \ref{lem:tropp} to infer that $\|(\matI - (\matZ\matZ^T\matB)(\matZ\matZ^T\matB)^{\dagger})\matA\|_2$ is at most
 $\|(\matI- (\matZ\matZ^T\matB)(\matZ\matZ^T\matB)^{\dagger})(\matA\matA^T)^q\matA\|_2^{1/(2q+1)}$, which is
equal to
\begin{eqnarray*}
& = & \|\matB - (\matZ\matZ^T\matB)(\matZ\matZ^T\matB)^{\dagger}\matB\|_2^{1/(2q+1)}\\
& = & \|\matB-\matZ\matZ^T\matB\|_2^{1/(2q+1)},
\end{eqnarray*} 
where we use that $(\matZ\matZ^T\matB)^{\dagger} = (\matZ^T \matB)^{\dagger} \matZ^T$ since $\matZ$ has orthonormal columns, and thus
$$(\matZ\matZ^T\matB)(\matZ\matZ^T\matB)^{\dagger}\matB = (\matZ\matZ^T \matB) (\matZ^T \matB)^{\dagger} (\matZ^T \matB) = \matZ\matZ^T \matB.$$
Hence, 
\begin{eqnarray}\label{eqn:abound}
\|\matA-\matZ\matZ^T\matA\|_2 \leq \|\matB-\matZ\matZ^T\matB\|_2^{1/(2q+1)}.
\end{eqnarray}

Let $\matU\mat\Sigma \matV^T$ be the SVD of $\matB$. Let $\mat\Omega_U = \matV_k^T\matG$
and $\mat\Omega_L = \matV_{n-k}^T\matG$, where $\matV_k^T$ denotes the top $k$ rows of $\matV^T$, and $\matV_{n-k}^T$ the remaining
$n-k$ rows. Since the rows of $\matV^T$ are orthonormal, by rotational invariance of the Gaussian distribution,
$\mat\Omega_U$ and $\mat\Omega_L$ are independent matrices of i.i.d. $N(0,1)$ entries. 

We now apply Lemma \ref{lem:structural} with the $\matC$ of that lemma equal to $\matZ$ above, the $\matZ$ of that
lemma equal to $\matV_k$, and the $\matA$ of that lemma equal to $\matB$ above. This implies the $\matE$ of that
lemma is equal to $\matB-\matB_k$. Note that to apply the lemma we need $\matV_k^T \matG$ to have full rank, which holds with probability
$1$ since it is a $k \times k$ matrix of i.i.d. $N(0,1)$ random variables. We thus have, 
\begin{eqnarray*}
\|\matB-\matZ\matZ^T\matB\|_2^2 & \leq & \|\matB-\matB_k\|_2^2 + \|(\matB-\matB_k)\matG(\matV_k^T\matG)^{\dagger}\|_2^2\\
& = & \|\matB-\matB_k\|_2^2 + \|\matU_{n-k} \mat\Sigma_{n-k}\matV_{n-k}^T \matG (\matV_k^T\matG)^{\dagger}\|_2^2\\
& = & \|\matB-\matB_k\|_2^2 + \|\mat\Sigma_{n-k} \matV_{n-k}^T \matG (\matV_k^T\matG)^{\dagger}\|_2^2\\
& \leq & \|\matB-\matB_k\|_2^2 \left (1 + \|\mat\Omega_L\|_2^2 \|\mat\Omega_U^{\dagger}\|_2^2 \right ),
\end{eqnarray*}
where $\mat\Sigma_{n-k}$ denotes the $(n-k) \times (n-k)$ diagonal matrix whose entries are the bottom $n-k$ diagonal
entries of $\mat\Sigma$, and $\matU_{n-k}$ denotes the rightmost $n-k$ columns of $\matU$. Here 
in the second equality we use unitary invariance of $\matU_{n-k}$, while in the inequality we use
sub-multiplicativity of the spectral norm. By Fact \ref{fact:gaussian} and independence of $\mat\Omega_L$ and
$\mat\Omega_U$, we have that $\|\mat\Omega_L\|_2^2 \leq C(n-k)$ and $\|\mat\Omega_1^{\dagger}\|_2^2 \leq \frac{k}{(C')^2}$
with probability at least $(9/10)^2 > 4/5$. Consequently for a constant $c > 0$, 
\begin{eqnarray}\label{eqn:bBound}
\|\matB-\matZ\matZ^T\matB\|_2^2 \leq \|\matB-\matB_k\|_2^2 \cdot c(n-k)k.
\end{eqnarray}
Combining (\ref{eqn:bBound}) with (\ref{eqn:abound}), we have
$$\|\matA-\matZ\matZ^T\matA\|_2 \leq \|\matB-\matB_k\|_2^{1/(2q+1)} \cdot \left (c(n-k)k \right )^{1/(4q+2)}.$$
Noting that $\|\matB-\matB_k\|_2 = \|\matA-\matA_k\|_2^{2q+1}$, and setting $q = O((\log n)/\eps)$ so that
$$(c(n-k)k)^{1/(4q+2)} = (1+\eps)^{\log_{1+\eps} c(n-k)k/((4q+2))} \leq 1+\eps,$$ completes the proof
\end{proof}

\subsection{Distributed low rank approximation}\label{sec:dislra}
In this section we study an algorithm for distributed low rank approximation. The model is called the 
{\it arbitrary partition model}. In this model there are $s$ players (also called servers), 
each locally holding an $n \times d$ 
matrix $\matA^t$, and we let $\matA = \sum_{t \in [s]} \matA^t$. We would like for each player to obtain a rank-$k$ projection
matrix $\matW\matW^T \in \mathbb{R}^{d \times d}$, for which 
$$\FNormS{\matA-\matA\matW\matW^T} \leq (1+\eps)\FNormS{\matA-\matA_k}.$$
The motivation is that each player can then locally project his/her matrix $\matA^t$ by computing 
$\matA^t\matW\matW^T$. It is often useful to have such a partition of the original input matrix $\matA$. For instance,
consider the case when a customer coresponds to a row of $\matA$, and a column to his/her purchase of a specific item. 
These purchases could be distributed across servers corresponding to different vendors. The communication is point-to-point,
that is, all pairs of players can talk to each other through a private channel for which the other $s-2$ players do not
have access to. The assumption is that $n \gg d$, though $d$ is still large, so having communication independent of $n$
and as small in $d$ as possible is ideal. In \cite{kvw14} an $\Omega(sdk)$ bit communication lower bound was shown. 
Below we show an algorithm of Kannan, Vempala, and the author \cite{kvw14} using $O(sdk/\eps)$ words of communication, 
assuming a word is $O(\log n)$ bits and the entries of each $\matA^t$ are $O(\log n)$-bit integers. 

We first show the following property about the top $k$ right 
singular vectors of $\matS\matA$ for a subspace
embedding $\matS$, as shown in \cite{kvw14}. The property shows that 
the top $k$ right singular
vectors $\ve_1, \ldots, \ve_k$ of $\matS \matA$ provide a 
$(1+\eps)$-approximation to the best rank-$k$ approximation
to $\matA$. This fact quickly follows from the fact that 
$\|\matS \matA \ve_i\|_2 = (1 \pm \eps)\|\matA \ve_i\|_2$ for the bottom $d-k$ right singular vectors
$\ve_{k+1}, \ldots, \ve_{d}$ of $\matS \matA$. It is crucial that $\matS$ is an $\ell_2$-subspace
embedding for $\matA$, as otherwise there is a dependency issue since 
the vectors $\matA\ve_{k+1}, \ldots, \matA\ve_d$ depend on $\matS$. 

\begin{theorem}\label{thm:kv}
Suppose $\matA$ is an $n \times d$ matrix. Let $\matS$ be an $m \times d$ matrix for which
$(1-\eps)\|\matA\x\|_2 \leq \|\matS\matA\x\|_2 \leq (1+\eps)\|\matA\x\|_2$ for all $\x \in \mathbb{R}^d$, that is,
$\matS$ is a subspace embedding for the column space of $\matA$. Suppose
$\matV\matV^T$ is a $d \times d$ matrix which projects vectors in $\mathbb{R}^d$ onto the space of
the top $k$
singular vectors of $\matS\matA$. Then
$\FNorm{\matA-\matA\matV\matV^T} \leq (1+O(\eps))\cdot \FNorm{\matA-\matA_k}.$
\end{theorem}
\begin{proof}
Form an orthonormal basis of ${\bf R}^d$ using the right singular vectors of $\matS\matA$. Let
$\ve_1,\ve_2,\ldots ,\ve_d$ be the basis.
\begin{eqnarray*}
\FNormS{\matA-\matA\sum_{i=1}^k\ve_i\ve_i^T}&=& \sum_{i=k+1}^d \|\matA\ve_i\|_2^2
\leq (1+\eps)^2\sum_{i=k+1}^d\|\matS\matA\ve_i\|_2^2\\
&= & (1+\eps)^2\FNormS{\matS\matA-[\matS\matA]_k},
\end{eqnarray*}
where the first equality follows since $\ve_1, \ldots, \ve_d$ is an orthonormal basis
of $\mathbb{R}^d$,
the inequality follows using the fact that $(1-\eps)\|\matA \x\|_2 \leq \|\matS \matA \x\|_2$
for all $\x \in \mathbb{R}^d$, and the final equality follows using that the $\ve_1, \ldots, \ve_d$
are the right singular vectors of $\matS \matA$. 

Suppose now $\u_1,\u_2,\ldots ,\u_d$ is an orthonormal basis consisting of the singular vectors
of $\matA$. Then, we have
\begin{eqnarray*}
\FNormS{\matS\matA-[\matS\matA]_k} &\leq & \FNormS{\matS\matA-\matS\matA\sum_{i=1}^k\u_i\u_i^T}\\
& = &\sum_{i=k+1}^d \|\matS\matA\u_i\|_2^2\\
&\leq & (1+\varepsilon )^2\sum_{i=k+1}^d\|\matA\u_i\|_2^2\\
& = &(1+\eps)^2\FNormS{\matA-\matA_k},
\end{eqnarray*}
where the first inequality uses that the rank-$k$ matrix $\sum_i \u_i \u_i^T$ is no better
at approximating $\matS \matA$ than $[\matS\matA]_k$, the first equality uses that
$\u_1, \ldots, \u_d$ is an orthonormal basis of $\mathbb{R}^d$, the second inequality uses
that $\|\matS \matA \x \|_2 \leq (1+\eps)\|\matA \x\|_2$ for all $\x \in \mathbb{R}^d$, 
and the final equality uses that $\u_1, \ldots, \u_d$ are the right singular vectors of $\matA$. 

Thus,
\[
\FNormS{\matA - \matA \sum_{i=1}^k \ve_i \ve_i^T} \le (1+\eps)^4 \FNormS{\matA-\matA_k},
\]
and the theorem follows.
\end{proof}
We will also need a variant of Lemma \ref{lem:sketching} from Section \ref{chap:lowRank}
which intuitively states that for a class of random matrices $\matS$, if we project
the rows of $\matA$ onto the row space of $\matS \matA$, we obtain a good low rank approximation.   
Here we use an $m \times n$ matrix $\matS$ in which each of the
entries is $+1/\sqrt{m}$ or $-1/\sqrt{m}$ with probability $1/2$, 
and the entries of $\matS$ are $O(k)$-wise independent. We
cite a theorem of Clarkson and Woodruff \cite{CW09} which shows what we need.
It can be shown by showing the following properties:
\begin{enumerate}
\item $\matS$ is an $\ell_2$-subspace embedding for any fixed $k$-dimensional subspace with probability at least $9/10$, and
\item $\matS$ has the $(\eps, \delta, \ell)$-JL moment property for some $\ell \geq 2$ (see Definition \ref{def:moment}). 
\end{enumerate}
\begin{theorem}\label{thm:sketch}(combining Theorem 4.2 and the second part of Lemma 4.3 of \cite{CW09})
Let $\matS \in \mathbb{R}^{m \times n}$ be a random sign matrix with $m = O(k \log (1/\delta)/\eps)$ in which
the entries are $O(k + \log(1/\delta))$-wise independent.
Then with probability at least $1-\delta$, if $\matU \matU^T$ is the $d \times d$ projection matrix onto the row space of $\matS \matA$,
then if $[\matA \matU]_k$ is the best rank-$k$ approximation to matrix $\matA \matU$, we have
$$\FNorm{[\matA\matU]_k \matU^T-\matA} \leq (1+O(\eps))\FNorm{\matA-\matA_k}.$$
\end{theorem}
The main algorithm {\sc AdaptiveCompress} of \cite{kvw14} is given in Algorithm 
{\sf AdaptiveCompress} below.  

Here we state the key idea behind Theorem \ref{thm:low-rank-arbitrary} below. The idea is that
if each of the servers projects their matrix $\matA^t$ to $\matP\matA^t$ using an $\ell_2$
subspace embedding $\matP$, then $\matP \matA = \sum_t \matP \matA^t$ and by Theorem
\ref{thm:kv}, if we compute the top $k$ right singular vectors of $\matP\matA$, we can send
these to each server to locally project their data on. Since $\matP\matA^t$ is more efficient
to communicate than $\matA^t$, this provides a considerable improvement in communication. However,
the communication is proportional to $d^2$ and we can make it proportional to only $d$ 
by additionally using Theorem \ref{thm:sketch} to first ``replace'' the $\matA^t$ matrices
with $\matA^t \matU$ matrices, where the columns of $\matU$ are an orthonormal basis
containing a $(1+\eps)$ rank-$k$ approximation. 

\begin{algorithm}[p]
\caption{The {\sf AdaptiveCompress}($k$,$\eps$,$\delta$) protocol}
\begin{enumerate}
\item Server $1$ chooses a random seed for an $m \times n$ sketching matrix $\matS$ 
as in Theorem \ref{thm:sketch},
given parameters $k, \eps,$ and $\delta$, where $\delta$ is a small positive constant.
It communicates the seed to the other servers.
\item Server $t$ uses the random seed to compute $\matS$, and then $\matS \matA^t$, 
and sends it to Server $1$.
\item Server $1$ computes $\sum_{t=1}^s \matS \matA^t = \matS \matA$. It computes an $m \times d$
orthonormal basis $\matU^T$ for the row space of $\matS \matA$, and sends $\matU$ to all the servers.
\item Each server $t$ computes $\matA^t \matU$. 
\item Server 1 chooses another random seed for a $O(k/\varepsilon^3)\times n$ matrix $\matP$ 
which is to be $O(k)$-wise independent and communicates this seed to all servers.
\item The servers then agree on a subspace embedding matrix $\matP$ of Theorem \ref{thm:kv} 
for $\matA \matU$, 
where $\matP$ is an $O(k/\eps^3) \times n$ matrix which can be described with $O(k \log n)$ bits.
\item Server $t$ computes $\matP \matA^t \matU$ and send it to Server $1$.
\item Server $1$ computes $\sum_{t=1}^s \matP \matA^t \matU = \matP \matA \matU$. 
It computes $\matV \matV^T$, which is an $O(k/\eps) \times O(k/\eps)$ projection 
matrix onto the top $k$ singular vectors of $\matP \matA \matU$, and sends $\matV$ to all the servers.
\item Server $t$ outputs $\matC^t = \matA^t \matU \matV \matV^T \matU^T$. 
Let $\matC=\sum_{t=1}^s \matC^t$. $\matC$ is not computed explicitly.
\end{enumerate}
\end{algorithm}

\begin{theorem}\label{thm:low-rank-arbitrary}
Consider the arbitrary partition model where an $n\times d$ matrix $\matA^t$ resides in server $t$ and the data matrix
$\matA=\matA^1+\matA^2+\cdots +\matA^s$.
For any $1 \ge \eps > 0$, there is a protocol {\sc AdaptiveCompress} that, on termination, leaves an
$n\times d$ matrix $\matC^t$ in server $t$ such that the matrix $\matC=\matC^1+\matC^2+\cdots +\matC^s$ 
with arbitrarily large constant probability achieves
$
\FNorm{\matA-\matC} \le (1+\eps) \min_{\matX: \text {rank}(\matX)\leq k} \FNorm{\matA-\matX},$
using linear space, polynomial time and with total communication complexity $O(sdk/\eps + sk^2/\eps^4)$ real numbers.
Moreover, if the entries of each $\matA^t$
are $b$ bits each, then the total communication is $O(sdk/\eps + sk^2/\eps^4)$ words each consisting of $O(b + \log(nd))$ bits.
\end{theorem}
\begin{proof}
By definition of the {\sc AdaptiveCompress} protocol, we have 
$\FNorm{\matA-\matC} = \FNorm{\matA- \matA\matU \matV \matV^T \matU^T}.$

Notice
that $\matU\matU^T$ and $\matI_{d}-\matU\matU^T$ are projections onto orthogonal subspaces. 
It follows by the Pythagorean theorem applied to each row that
\begin{eqnarray}\label{eqn:first}
&& \FNormS{\matA\matU \matV \matV^T \matU^T - \matA} \nonumber \\
& = & \FNormS{(\matA\matU \matV \matV^T \matU^T -\matA)(\matU \matU^T)}\\ 
& +&  \FNormS{(\matA\matU \matV \matV^T \matU^T-\matA)(\matI_d - \matU\matU^T)}\nonumber \\
& = & \FNormS{\matA\matU\matV\matV^T\matU^T - \matA\matU\matU^T\|^2 + \|\matA-\matA\matU\matU^T},
\end{eqnarray}
where the second equality uses that $\matU^T\matU = \matI_c$, where $c$ is the number of columns of $\matU$.

Observe that the row spaces of $\matA\matU\matV\matV^T\matU^T$ and $\matA\matU\matU^T$ are both in the row space of $\matU^T$, and
therefore in the column space of $\matU$. It follows that since $\matU$ has orthonormal columns,
$\FNorm{\matA\matU\matV\matV^T\matU^T-\matA\matU\matU^T} = \FNorm{(\matA\matU\matV\matV^T\matU^T-\matA\matU\matU^T)\matU},$ 
and therefore
\begin{eqnarray}\label{eqn:second}
&& \FNormS{\matA\matU\matV\matV^T\matU^T-\matA\matU\matU^T} + \FNormS{\matA-\matA\matU\matU^T} \nonumber \\ 
& = & \FNormS{(\matA\matU\matV\matV^T\matU^T-\matA\matU\matU^T)\matU} + 
\FNormS{\matA-\matA\matU\matU^T} \nonumber \\
& = & \FNormS{\matA\matU\matV\matV^T - \matA\matU} + \FNormS{\matA-\matA\matU\matU^T},
\end{eqnarray}
where the second equality uses that $\matU^T \matU = \matI_c$. Let $(\matA\matU)_k$ be the best rank-$k$ approximation to
the matrix $\matA\matU$. By Theorem \ref{thm:kv}, with probability $1-o(1)$,
$\FNormS{\matA\matU\matV\matV^T - \matA\matU} \leq (1+O(\eps)) \FNormS{(\matA\matU)_k - \matA\matU},$
and so
\begin{eqnarray}\label{eqn:third}
&& \FNormS{\matA\matU\matV\matV^T-\matA\matU} + \FNormS{\matA-\matA\matU\matU^T}\nonumber \\
& \leq & (1+O(\eps)) \FNormS{(\matA\matU)_k - \matA\matU} +  \FNormS{\matA-\matA\matU\matU^T} \nonumber \\
& \leq & (1+O(\eps)) ( \FNormS{(\matA\matU)_k - \matA\matU} +  \FNormS{\matA-\matA\matU\matU^T}).
\end{eqnarray}
Notice that the row space of $(\matA\matU)_k$ is spanned by the top $k$ right singular
vectors of $\matA\matU$, which are in the row space of $\matU$. Let us write $(\matA\matU)_k = \matB \cdot \matU$, 
where $\matB$ is a rank-$k$ matrix.

For any vector $\ve \in \mathbb{R}^d$
, $\ve\matU\matU^T$ is in the row space of $\matU^T$, and since the columns of $\matU$ are orthonormal, 
$\FNormS{\ve\matU\matU^T} = \FNormS{\ve\matU\matU^T\matU} = \FNormS{\ve\matU}$,
and so 
\begin{eqnarray}\label{eqn:fourth}
&& \FNormS{(\matA\matU)_k - \matA\matU} + \FNormS{\matA-\matA\matU\matU^T} \nonumber \\
& = & \FNormS{(\matB -\matA)\matU} + \FNormS{\matA(\matI-\matU\matU^T)} \nonumber \\
& = & \FNormS{\matB\matU\matU^T - \matA\matU\matU^T} + \ \FNormS{\matA\matU\matU^T - \matA}).
\end{eqnarray}
We apply the Pythagorean theorem to each row in the expression in (\ref{eqn:fourth}), noting that
the vectors $(\matB_i-\matA_i)\matU\matU^T$ and $\matA_i\matU\matU^T-\matA_i$ are orthogonal, where $\matB_i$ and $\matA_i$ are the $i$-th rows of $\matB$ and $\matA$,
respectively. Hence,
\begin{eqnarray}\label{eqn:fifth}
&& \FNormS{\matB\matU\matU^T - \matA\matU\matU^T} + \FNormS{\matA\matU\matU^T - \matA}\\
& = & \FNormS{\matB\matU\matU^T - \matA}\\
& = & \FNormS{(\matA\matU)_k \matU^T - \matA},
\end{eqnarray}
where the first equality uses that
\begin{eqnarray*}
&& \FNormS{\matB\matU\matU^T-\matA} \nonumber \\
& = & \FNormS{(\matB\matU\matU^T-\matA)\matU\matU^T} + \FNormS{(\matB\matU\matU^T-\matA)(\matI-\matU\matU^T)}\\
& = & \FNormS{\matB\matU\matU^T - \matA\matU\matU^T} + \FNormS{\matA\matU\matU^T - \matA},
\end{eqnarray*}
and 
the last equality uses the definition of $\matB$. By Theorem \ref{thm:sketch},
with constant probability arbitrarily close to $1$, we have
\begin{eqnarray}\label{eqn:sixth}
\FNormS{[\matA\matU]_k \matU^T - \matA} & \leq & (1+O(\eps)) \FNormS{\matA_k - \matA}.
\end{eqnarray}
It follows by combining (\ref{eqn:first}), (\ref{eqn:second}), (\ref{eqn:third}), (\ref{eqn:fourth}), (\ref{eqn:fifth}),
(\ref{eqn:sixth}),
that $\FNormS{\matA\matU \matV \matV^T \matU^T - \matA} \leq (1+O(\eps))\FNormS{\matA_k-\matA}$, 
which shows the correctness property of {\sc AdaptiveCompress}.

We now bound the communication. In the first step, by Theorem \ref{thm:sketch}, $m$ can be set to $O(k/\eps)$
and the matrix $\matS$ can be described using a random seed that is $O(k)$-wise independent. The communication of
steps 1-3 is thus $O(sdk / \eps)$ words. By Theorem \ref{thm:kv}, the remaining steps take $O(s(k/\eps)^2/\eps^2) =
O(sk^2/\eps^4)$ words of communication.

To obtain communication with $O(b + \log(nd))$-bit words if the entries of the matrices $\matA^t$ are specified by $b$ bits, Server 1 can
instead send $\matS\matA$ to each of the servers. The $t$-th server then computes $\matP\matA^t (\matS\matA)^T$ and sends this to Server 1.
Let $\matS\matA = \matR\matU^T$, where $\matU^T$ is an orthonormal
basis for the row space of $\matS\matA$, and $\matR$ is an $O(k/\eps) \times O(k/\eps)$ change of basis matrix.
Server 1 computes 
$\sum_t \matP\matA^t (\matS\matA)^T = P\matA(\matS\matA)^T$ and sends this to each of the servers. Then, since each of the servers
knows $\matR$, it can compute $\matP\matA(\matS\matA)^T (\matR^T)^{-1} = \matP\matA\matU$. 
It can then compute the SVD of this matrix, from which it obtains
$\matV\matV^T$, the projection onto its top $k$ right singular vectors. Then, since Server $t$ knows $\matA^t$ and $\matU$, it can compute
$\matA^t \matU(\matV\matV^T)\matU^T$, as desired. Notice that in this variant of the algorithm what is sent is 
$\matS\matA^t$ and $\matP\matA^t(\matS\matA)^T$, which
each can be specified with $O(b + \log(nd))$-bit words if the entries of the $\matA^t$ are specified by $b$ bits.
\end{proof}

\section{Graph Sparsification}
{\bf Section Overview:} This section is devoted to showing how sketching can be used to perform spectral sparsification of graphs. While $\ell_2$-subspace embeddings compress tall and skinny matrices to small matrices, they are not particularly useful at compressing roughly square matrices, as in the case of a graph Laplacian. This section shows how related sketching techniques can still be used to sparsify such square matrices, resulting in a useful compression.

While $\ell_2$-subspace embeddings are a powerful tool, such embeddings compress an $n \times d$ 
matrix to a $\poly(d/\eps) \times d$ matrix. 
This is not particularly useful if $n$ is not too much larger than $d$. For instance, one natural problem is to compress
a graph $G = (V, E)$ on $n$ vertices using linear sketches so as to preserve all spectral information. In this case one is interested in a subspace embedding of the Laplacian of $G$, which is an $n \times n$ matrix, for which an $\ell_2$-subspace embedding does not provide compression.  
In this section we explore
how to use linear sketches for graphs. 

We formally define the problem as follows, following the notation and outlines of \cite{KLMMS14}. 
Consider an ordering on the $n$ vertices, denoted $1, \ldots, n$. We will only
consider undirected graphs, though we will often talk about edges $e = \{u,v\}$ as $e = (u,v)$, where here $u$ is less than
$v$ in the ordering we have placed on the edges. This will be for notational convenience only; the underlying graphs are 
undirected. 

Let $\matB_n \in \mathbb{R}^{\binom{n}{2} \times n}$ be the vertex edge incidence of the
undirected, unweighted complete graph on $n$ vertices, where the $e$-th row $\b_e$ for edge $e = (u,v)$ has a $1$ in column
$u$, a $(-1)$ in column $v$, and zeroes elsewhere. 

One can then write the vertex edge incidence matrix of an arbitrary undirected graph $G$ as 
$\matB = \matT \cdot \matB_n$, where $\matT \in \mathbb{R}^{\binom{n}{2} \times \binom{n}{2}}$ is a diagonal matrix with a $\sqrt{w_e}$ in the $e$-th
diagonal entry if and only if $e$ is an edge of $G$ and its weight is $w_e$. 
The remaining diagonal entries of $\matT$ are equal to $0$. The
Laplacian is $\matK = \matB^T\matB$. 

The spectral sparsification problem can then be defined as follows: find a weighted subgraph $H$ of $G$ so that if 
$\tilde{\matK}$ is the Laplacian of $H$, then 
\begin{eqnarray}\label{eqn:spectral}
\forall \x \in \mathbb{R}^n, \ (1-\eps)\x^T \matK \x \leq \x^T \tilde{\matK} \x \leq (1+\eps) \x^T \matK \x.
\end{eqnarray}
We call $H$ a {\it spectral sparsifier} of $G$ 
The usual notation for (\ref{eqn:spectral}) is
$$(1-\varepsilon) \matK \preceq \tilde{\matK} \preceq (1+\varepsilon) \matK,$$
where $\matC \preceq \matD$ means that $\matD-\matC$ is positive semidefinite. We also sometimes use the notation
$$(1-\varepsilon) \matK \preceq_R \tilde{\matK} \preceq_R (1+\varepsilon) \matK,$$ 
to mean that $(1-\varepsilon)\x^T \matK \x \leq \x^T \tilde{\matK} \x \leq (1+\varepsilon) \x^T\matK\x$ for all vectors $\x$
in the row space of $\matK$, which is a weaker notion since there is no guarantee for vectors $\x$
outside of the row space of $\matK$. 

One way to solve the spectral sparsification problem is via leverage score sampling. Suppose
we write the above matrix $\matB$ in its SVD as $\matU \mat\Sigma \matV^T$. Let us look at the leverage scores of $\matU$,
where recall the $i$-th leverage score $\ell_i = \|\matU_{i*}\|_2^2$. Recall the definition of Leverage
Score Sampling given in Definition \ref{def:lss}. By Theorem \ref{thm:lssPerf}, if we take
$O(n \eps^{-2} \log n)$ samples of rows of $\matU$, constructing the sampling and rescaling matrices
of Definition \ref{def:lss}, then with probability $1-1/n$, simultaneously for all $i \in [n]$,
\begin{eqnarray}\label{eqn:guaranteeLSS}
(1-\varepsilon/3) \leq \sigma_i^2(\matD^T \mat\Omega^T \matU) \leq (1+\varepsilon/3).
\end{eqnarray}
Suppose we set
\begin{eqnarray}\label{eqn:tildeK}
\tilde{\matK} = (\matD^T \mat\Omega^T \matB)^T (\matD^T \mat\Omega^T \matB).
\end{eqnarray}
\begin{theorem}\label{thm:lssSimple}
For $\tilde{\matK}$ defined as in (\ref{eqn:tildeK}), with probability $1-1/n$, 
$$(1-\varepsilon) \matK \preceq \tilde{\matK} \preceq (1+\varepsilon)\matK.$$
\end{theorem}
\begin{proof}
Using that $\matK = \matB^T\matB$ and the definition of $\tilde{\matK}$, it suffices to show for all $\x$, 
$$\|\matB\x\|_2^2 = (1 \pm \varepsilon/3) \|\matD^T \mat\Omega^T \matB \x\|_2^2.$$
By (\ref{eqn:guaranteeLSS}), and using that $\matB = \matU \mat\Sigma \matV^T$, 
$$\|\matD^T \mat\Omega^T \matB \x\|_2^2 = (1 \pm \varepsilon/3) \|\mat\Sigma \matV^T \x\|_2^2,$$
and since $\matU$ has orthonormal columns,
$$\|\mat\Sigma \matV^T \x\|_2^2 = \|\matU \mat\Sigma \matV^T \x\|_2^2 = \|\matB\x\|_2^2,$$
which completes the proof. 
\end{proof}
Notice that Theorem \ref{thm:lssSimple} shows that if one knows the leverage scores, then
by sampling $O(n \eps^{-2} \log n)$ edges of $G$ and reweighting them, one obtains a spectral
sparsifier of $G$. One can use algorithms for approximating the leverage scores of general matrices \cite{DMMW12},
though more efficient algorithms, whose overall running time is near-linear in the number of edges
of $G$, are known \cite{SS08,st11}. 

A beautiful theorem of Kapralov, Lee, Musco, Musco, and Sidford is the following \cite{KLMMS14}. 
\begin{theorem}\label{thm:linear}
There exists a distribution $\Pi$ on $\eps^{-2} \polylog(n) \times \binom{n}{2}$ matrices
$\matS$ for which with probability $1-1/n$, from $\matS \cdot \matB$, it is possible to recover a
weighted subgraph $H$ with $O(\eps^{-2} n \log n)$ edges such that $H$ is a spectral sparsifier of $G$. 
The algorithm runs in $O(\eps^{-2} n^2 \polylog(n))$ time.  
\end{theorem}

We note that Theorem \ref{thm:linear} is not optimal in its time complexity or the number of edges
in $H$. Indeed, Spielman and Srivastava \cite{SS08} show that in $\tilde{O}(m (\log n) \eps^{-2})$ time
it is possible to find an $H$ with the same number $O(\eps^{-2} n \log n)$ of edges as in Theorem \ref{thm:linear},
where $m$ is the number of edges of $H$. For sparse graphs, this results in significantly less time for finding
$H$. Also, Batson, Spielman, and Srivastava \cite{BSS09} show that it is possible to deterministically find
an $H$ with $O(\eps^{-2} n)$ edges, improving the $O(\eps^{-2} n \log n)$ number of edges in Theorem \ref{thm:linear}. This
latter algorithm is a bit slow, requiring $O(n^3 m \eps^{-2})$ time, with some improvements for dense graphs 
given by Zouzias \cite{z12}, though these are much slower than Theorem \ref{thm:linear}. 

Despite these other works, the key feature of Theorem \ref{thm:linear} is that it is a {\it linear sketch}, namely, it is formed by choosing a random oblivious (i.e., independent of $\matB$) linear map $\matS$ and storing $\matS \cdot \matB$. Then, the sparsifier
$H$ can be found using only $\matS \cdot \matB$, i.e., without requiring
access to $\matB$. This gives
it a number of advantages, such that it implies the first algorithm for maintaining a spectral sparsifier in a
data stream in the presence of insertions and deletions to the graph's edges. That is, for the other works, it was
unknown how to rebuild the sparsifier if an edge is deleted; in the case when linear sketches are used to summarize
the graph, it is trivial to update the sketch in the presence of an edge deletion. 

In the remainder of the section, we give an outline of the proof of Theorem \ref{thm:linear}, following the exposition given
in \cite{KLMMS14}. We restrict to unweighted graphs for the sake of presentation; the arguments generalize in a natural
way to weighted graphs. 

The main idea, in the author's opinion, is the 
use of an elegant technique due to Li, Miller and Peng \cite{lmp13} called ``Introduction and Removal of Artificial Bases''. 
We suspect this technique should have a number of other applications; Li, Miller and Peng use it for
obtaining approximation algorithms for $\ell_2$ and $\ell_1$ regression. 
Intuitively, the theorem states that if you take any PSD matrix $\matK$, you can form a sequence
of matrices $\matK(0), \matK(1), \ldots, \matK(d)$, where $\matK(0)$ has a spectrum which is
within a factor of $2$ of the identity, $\matK(d)$ has a spectrum within a factor of 
$2$ of $\matK$, and for
each $\ell$, $\matK(\ell-1)$ has a spectrum within a factor of $2$ of $\matK(\ell)$. Furthermore
if $\matK$ is the Laplacian of an unweighted graph, $d = O(\log n)$. 

The proof of the following
theorem is elementary. We believe the power in the theorem is its novel
use in algorithm design. 

\begin{theorem}(Recursive Sparsification of \cite{lmp13}, as stated in \cite{KLMMS14})\label{thm:mp}
Let $\matK$ be a PSD matrix with maximum eigenvalue bounded above by $\lambda_u$ and minimum eigenvalue bounded from
below by $\lambda_{\ell}$. Let $d = \lceil \log_2 (\lambda_u/\lambda_{\ell}) \rceil$. For $\ell \in \{0, 1, 2, \ldots, d\}$, set
$$\gamma(\ell) = \frac{\lambda_u}{2^{\ell}}.$$
Note that $\gamma(d) \leq \lambda_{\ell}$ and $\gamma(0) = \lambda_u$. Consider the sequence of PSD matrices
$\matK(0), \matK(1), \ldots, \matK(d)$, where
$$\matK(\ell) = \matK + \gamma(\ell) \matI_{n}.$$
Then the following conditions hold.
\begin{enumerate}
\item $\matK \preceq_R \matK(d) \preceq_R 2\matK$.
\item $\matK(\ell) \preceq \matK(\ell-1) \preceq 2\matK(\ell)$ for $\ell = 1, 2, \ldots, d$. 
\item $\matK(0) \preceq 2 \gamma(0)\matI \preceq 2\matK(0)$.
\end{enumerate}
If $\matK$ is the Laplacian of an unweighted graph, then its maximum eigenvalue is at most $2n$ and its minimum
eigenvalue is at least $8/n^2$. We can thus set $d = \lceil \log_2 \lambda_u/\lambda_{\ell} \rceil  = O(\log n)$
in the above. 
\end{theorem}
\begin{proof}
For the first condition, for all $\x$ in the row space of $\matK$, 
$$\x^T \matK \x \leq \x^T \matK \x + \x^T (\gamma(d) \matI )\x \leq \x^T \matK \x + \x^T \lambda_{\ell} \x \leq 2 \x^T \matK \x.$$
For the second condition, for all $\x$,
$$\x^T \matK(\ell) \x = \x^T \matK \x + \x^T \gamma(\ell) \matI \x \leq \x^T \matK \x + \x^T \gamma(\ell-1) \x = \x^T \matK(\ell-1)\x,$$
and
$$\x^T \matK(\ell-1)\x = \x^T \matK \x + \x^T \gamma(\ell-1) \matI \x = \x^T \matK \x + 2 \x^T \gamma(\ell) \matI \x \leq 2 \x^T\matK(\ell)\x.$$
Finally, for the third condition, for all $\x$,
\begin{eqnarray*}
\x^T \matK(0) \x & = & \x^T \matK \x + \x^T \lambda_u \matI \x\\
& \leq & \x^T (2\lambda_u \matI) \x\\
& \leq & 2\x^T \matK \x + 2 \x^T \lambda_u \matI \x\\
& \leq & 2\x^T \matK(0) \x.
\end{eqnarray*} 
The bounds on the eigenvalues of a Laplacian are given in \cite{st04} (the bound on the maximum eigenvalue follows from
the fact that $n$ is the maximum eigenvalue of the Laplacian of the complete graph on $n$ vertices. The bound on the minimum
eigenvalue follows from Lemma 6.1 of \cite{st04}). 
\end{proof}
The main idea of the algorithm is as follows. We say a PSD matrix $\tilde{\matK}$ is a $C$-approximate 
row space sparsifier of a PSD matrix $\matK$ if $\matK \preceq_R \tilde{\matK} \preceq_R C \cdot \matK$. If we also
have the stronger condition that $\matK \preceq \tilde{\matK} \preceq C \cdot \matK$ we say that $\tilde{\matK}$
is a $C$-approximate sparsifier of $\matK$. 

By the first condition of Theorem \ref{thm:mp}, if we had a matrix $\tilde{\matK}$ which is a $C$-approximate
row space sparsifier of $\matK(d)$, then $\tilde{\matK}$ is also a $2C$-approximate row space sparsifier to $\matK$. 

If we were not
concerned with compressing the input graph $G$ with a linear sketch, at this point we could perform
Leverage Score Sampling to obtain a $(1+\eps)$-approximate sparsifier to $\matK(d)$. Indeed, by 
Theorem \ref{thm:lssPerf}, it is enough to construct a distribution $q$
for which $q_i \geq p_i/\beta$ for all $i$, where $\beta > 0$ is a constant. 

To do this, first observe that the leverage score for a potential edge $i = (u,v)$ is given by
\begin{eqnarray}\label{eqn:newLabel}
\|\matU_{i*}\|_2^2 & = & \matU_{i*} \mat\Sigma \matV^T (\matV \mat\Sigma^{-2} \matV^T) \matV \mat\Sigma \matU_{i*}^T\\
& = & \b_i^T \matK^{\dagger} \b_i.
\end{eqnarray}
%
As $\b_i$ is in the row space of $\matB$, it is also in the row space of $\matK = \matB^T \matB$,
since $\matB$ and $\matK$ have the same row space (to see this, write $\matB = \matU \mat\Sigma \matV^T$
in its SVD and then $\matK = \matV \mat\Sigma^2 \matV^T$). Since $\tilde{\matK}$ is a 
$2C$-approximate row space sparsifier of $\matK$, for all $\u$ in the row space of $\matK$, 
$$\u^T \matK \u \leq \u^T \tilde{\matK} \u \leq 2C \u^T \matK \u,$$
which implies since $\matK^+$ has the same row space as $\matK$ (to see this, again look at the SVD), 
$$\frac{1}{2C} \u^T \matK^+ \u \leq \u^T \tilde{K}^+ \u \leq \u^T \matK^+ \u.$$
Since this holds for $\u = \b_i$ for all $i$, it follows that $\b_i^T \tilde{\matK}^+ \b_i$ is within a factor
of $2C$ of $\b_i^T \matK^+ \b_i$ for all $i$. It follows 
by setting $q_i = \b_i^T \tilde{\matK}^{\dagger} \b_i/n$, we have that $q_i \geq p_i/2C$, where $p_i$
are the leverage scores of $\matB$. Hence, by Theorem \ref{thm:lssPerf}, 
it suffices to take $O(n \eps^{-2} \log n)$ samples of the edges of $G$ according to $q_i$, reweight them,
and one obtains a spectral sparsifier to the graph $G$. 

Hence, if we were not concerned with compressing $G$ with a linear sketch, i.e., 
of computing the sparsifier $H$ from $\matS \matB$ for a random oblivious
mapping $\matS$, one approach using
Theorem \ref{thm:mp} would be the following. By the third condition of Theorem \ref{thm:mp}, 
we can start with a sparsifier $\tilde{\matK} = 2 \gamma(0)\matI$ which provides a $2$-approximation to $\matK(0)$, 
in the sense that $\matK(0) \preceq \tilde{\matK} \preceq 2\matK(0)$. Then, we can apply Leverage Score Sampling
and by Theorem \ref{thm:lssPerf}, obtain a sparsifier $\tilde{\matK}$ for which 
$$\matK(0) \preceq \tilde{\matK} \preceq (1+\eps)\matK(0).$$
Then, by the second property of Theorem \ref{thm:mp}, 
\begin{eqnarray*}
\matK(1) \preceq \matK(0) \preceq \tilde{\matK} \preceq 2(1+\eps)\matK(1).
\end{eqnarray*}
Hence, $\tilde{\matK}$ is a $2(1+\eps)$-approximation to $\matK(1)$. We can now apply Leverage Score Sampling
again, and in this way obtain $\tilde{\matK}$ which is a $2(1+\eps)$-approximation to $\matK(2)$, etc. Note that
the errors do not compound, and the number of samples in $\tilde{\matK}$ is always $O(n \eps^{-2} \log n)$. By
the argument above, when $\tilde{\matK}$ becomes a $2(1+\eps)$-approximation to $\matK(d)$, it is a $4(1+\eps)$
approximation to $\matK$, and we obtain a spectral sparsifier of $G$ by sampling $O(n \eps^{-2} \log n)$ edges
according to the leverage scores of $\tilde{\matK}$. 

Thus, the only task left is to implement this hierarchy of leverage score sampling using linear sketches. 

For this, we need the following standard theorem from the sparse recovery literature. 
\begin{theorem}(see, e.g., \cite{ccf04,GLPS})\label{thm:hh}
For any $\eta > 0$, there is an algorithm $D$ and a 
distribution on matrices $\mat\Phi$ in $\mathbb{R}^{O(\eta^{-2} \polylog(n)) \times n}$
such that for any $\x \in \mathbb{R}^n$, with probability $1-n^{-100}$ over the choice of $\mat\Phi$, the output of
$D$ on input $\mat\Phi \x$ is a vector $\w$ with $\eta^{-2} \polylog(n)$ non-zero entries which satisfies the 
guarantee that 
$$\|\x-\w\|_{\infty} \leq \eta \|\x\|_2.$$
\end{theorem}
Several standard consequences of this theorem, as observed in \cite{KLMMS14}, 
can be derived by setting $\eta = \frac{\eps}{C \log n}$ for a constant $C > 0$, which is the setting of $\eta$
we use throughout.  
Of particular interest is that for $0 < \eps < 1/2$, from $w_i$ one can determine
if $\x_i \geq \frac{1}{C \log n} \|\x\|_2$ or $\x_i < \frac{1}{2C \log n} \|\x\|_2$ given that it satisfies one of 
these two conditions. We omit the proof of this fact which can be readily verified from the statement
of Theorem \ref{thm:hh}, as shown in \cite{KLMMS14}. 

The basic idea behind the sketching algorithm is the following intuition. 
Let $\x_e = \matT \matB_n \matK^{\dagger} \b_e$ for an edge $e$ which may or may not occur in $G$, which as we will
see below is a vector with the important property that its $e$-th coordinate is either $\ell_e$ or $0$. Then, 
$$\ell_e = \|\matU_e\|_2^2 = \b_e^T \matK^{\dagger} \matK \matK^{\dagger} \b_e = \|\matB \matK^{\dagger} \b_e\|_2^2 = \|\matT \matB_n \matK^{\dagger} \b_e\|_2^2 = \|\x_e\|_2^2,$$
where the first equality follows by definition of the leverage scores, the second equality 
follows by (\ref{eqn:newLabel}) and using that $\matK^+ = \matK^+ \matK \matK^+$, the third equality
follows by definition of $\matK = \matB^T \matB$, the fourth equality follows from $\matT\matB_n = \matB$, and 
the final equality follows by definition of $\x_e$. 

Moreover, by definition of $\matT$, the $e$-th coordinate of $\x_e$ is $0$ if $e$ does not occur in $G$. 
Otherwise, it is $\b_e \matK^{\dagger} \b_e = \ell_e$. 
We in general could have that $\ell_e \ll \|\x_e\|_2$, that is, there can be many other
non-zero entries among the coordinates of $\x_e$ other than the $e$-th entry.

This is where sub-sampling and Theorem \ref{thm:hh} come to the rescue. 
At a given level in the Leverage Score Sampling hierarchy, that is, when trying to construct $\matK(\ell+1)$ from
$\matK(\ell)$, we have a $C$-approximation
$\tilde{\matK}$ to $\matK(\ell)$ for a given $\ell$, and would like a $(1+\eps)$-approximation to $\matK(\ell)$. Here
$C > 0$ is a fixed constant. To do this,
suppose we sub-sample the edges of $G$ at rates $1, 1/2, 1/4, 1/8, \ldots, 1/n^2$, where sub-sampling
at rate $1/2^i$ means we randomly decide to keep each edge independently with probability $1/2^i$. Given $\tilde{\matK}$,
if $\hat{\ell}_e$ is our $C$-approximation to $\ell_e$, if we sub-sample at rate $1/2^i$ where
$2^i$ is within a factor of $2$ of $\hat{\ell}_e$, then we would expect $\|\x_e\|_2^2$ to drop by a factor
of $\Theta(\ell_e)$ to $\Theta(\ell_e^2)$. Moreover, if edge $e$ is included in the sub-sampling at rate $1/2^i$,
then we will still have $\x_e = \ell_e$. Now we can apply Theorem \ref{thm:hh} on the sub-sampled vector $\x_e$
and we have that $\x_e = \Omega(\|\x_e\|_2)$, which implies that in the discussion after 
Theorem \ref{thm:hh}, we will be able to find edge $e$. 
What's more is that the process of dropping each edge with probability $1/2^i = 1/\Theta(\ell_e)$
can serve as the leverage score sampling step itself. Indeed, this process sets the $e$-th coordinate of $\x_e$
to $0$ with probability $1-\Theta(\ell_e)$, that is, it finds edge $e$ with probability $\Theta(\ell_e)$, which
is exactly the probability that we wanted to sample edge $e$ with in the first place. 

Thus, the algorithm is to sub-sample the edges of $G$ at rates $1, 1/2, \ldots, 1/n^2$, and for each rate
of sub-sampling, maintain the linear sketch given by Theorem \ref{thm:hh}. This involves computing $\mat\Phi^i \matT \matB_n$
where $\mat\Phi^i$ is a linear sketch of the form $\mat\Phi \cdot \matD^i$, where $\mat\Phi$ is as in Theorem \ref{thm:hh},
and $\matD^i$ is a diagonal matrix with each diagonal entry set to $1$ with probability $1/2^i$ and set to $0$ otherwise.
We do this entire process independently
$O(\log n)$ times, as each independent repetition will allow us to build a $\matK(\ell+1)$ from a $\matK(\ell)$ for one
value of $\ell$. 
Then, for each level of the Leverage Score Sampling hierarchy of Theorem \ref{thm:mp}, we
have a $\tilde{\matK}$. For each possible edge $e$, we compute $\hat{\ell}_e$ using $\tilde{\matK}^{\dagger}$ which determines
a sub-sampling rate $1/2^i$. By linearity, we can compute $(\mat\Phi^i \matT \matB_n) \tilde{\matK}^{\dagger} \b_e$, which is the
sub-sampled version of $\x_e$. We sample edge $e$ 
if it is found by the algorithm $\matD$ in the discussion surrounding Theorem \ref{thm:hh}, for that sub-sampling level. We 
can thus use these sketches to walk up the Leverage Score Sampling hierarchy of Theorem \ref{thm:mp} and 
obtain a $(1+\eps)$-approximate spectral sparsifier to $G$. 
Our discussion has omitted a number of details, but hopefully gives a flavor of the result. 
We refer the reader to \cite{KLMMS14} for futher details
on the algorithm. 

\section{Sketching Lower Bounds for Linear Algebra}\label{chap:lb}
While sketching, and in particular subspace embeddings, have been used for a wide variety of applications, there
are certain limitations. In this section we explain some of them.

{\bf Section Overview:} In \S\ref{sec:schatten} we introduce the Schatten norms as a natural family of matrix norms including the Frobenius and operator norms, and show that they can be approximated pretty efficiently given non-oblivious methods and multiple passes over the data. In \S\ref{sec:oded} we ask what we can do with just a single oblivious sketch of the data matrix, and show that unlike the Frobenius norm, where it can be compressed to a vector of a constant number of dimensions, for approximating the operator norm of an $n \times n$ matrix from the sketch one cannot compress to fewer than $\Omega(n^2)$ dimensions. In \S\ref{sec:streaming} we discuss streaming lower bounds for numerical linear algebra problems, such as approximate matrix product, $\ell_2$-regression, and low rank approximation. In \S\ref{sec:seLB} we mention lower bounds on the dimension of $\ell_2$-subspace embeddings themselves. Finally, in \S\ref{sec:adaptiveLower} we show how algorithms which sketch input data, then use the same sketch to adaptively query properties about the input, typically cannot satisfy correctness. That is, we show broad impossibility results for sketching basic properties such as the Euclidean norm of an input vector when faced with adaptively chosen queries. Thus, when using sketches inside of complex algorithms, one should make sure they are not queried adaptively, or if they are, that the algorithm will still succeed. 

\subsection{Schatten norms}\label{sec:schatten}
A basic primitive is to be able to use a sketch to estimate a norm. This is a very well-studied problem with
inspirations from the data stream literature, where sketching $\ell_p$-norms has been extensively studied. 

For problems on matrices one is often interested in error measures that depend on a matrix norm. An 
appealing class of such norms is the Schatten $p$-norms of a matrix $\matA$, which we shall denote $\|\matA\|_p$. 
\begin{definition}\label{def:schatten}
For $p \geq 1$, 
the $p$-th Schatten norm $\|\matA\|_p$ of a rank-$\rho$ matrix $\matA$ is defined to be 
$$\|\matA\|_p = \left (\sum_{i=1}^{\rho} \sigma_i^p \right)^{1/p},$$
where $\sigma_1 \geq \sigma_2 \geq \cdots \geq \sigma_{\rho} > 0$ are the singular values of $\matA$. For
$p = \infty$, $\|\matA\|_{\infty}$ is defined to be $\sigma_1$. 
\end{definition}
Two familiar norms immediately stand out: the Schatten $2$-norm is just the Frobenius norm of $\matA$, while
the Schatten $\infty$-norm is the operator norm of $\matA$. Note that typical convention is to let
$\|\matA\|_2$ denote the operator norm of $\matA$, but in this section we shall use $\|\matA\|_{\infty}$ to denote
the operator norm to distinguish it from the Schatten $2$-norm, which is the Frobenius norm. 

The Schatten norms are particularly useful
in that they are rotationally invariant. That is, if $\matA$ is an $m \times n$ matrix, and if $\matJ$ is
an $m \times m$ orthonormal matrix while $\matK$ is an $n \times n$ orthonormal matrix, then 
$\|\matJ\matA\matK\|_p = \|\matA\|_p$ for any $p \geq 1$. To see this, we may write $\matA = \matU \mat\Sigma \matV^T$ in its
SVD. Then $\matJ\matU$ has orthonormal columns, while $\matV^T\matK$ has orthonormal rows. It follows that
the SVD of the matrix $\matJ\matA\matK$ is given in factored form as $(\matJ\matU) \mat\Sigma (\matV^T\matK)$, and so it has
the same singular values as $\matA$, and therefore the same Schatten $p$-norm for any $p \geq 1$. 

One reason one is interested in estimating a matrix norm is to evaluate the quality of an approximation.
For instance, suppose one finds a matrix $\tilde{\matA}$ which is supposed to approximate $\matA$ in a certain
norm, e.g., one would like $\|\matA-\tilde{\matA}\|_p$ to be small. To evaluate the quality of the approximation
directly one would need to compute $\|\matA-\tilde{\matA}\|_p$. This may be difficult to do if one is interested
in a very fast running time or using a small amount of space and a small number of passes over the data. 
For instance, for $p \notin \{2, \infty\}$
it isn't immediately clear there is an algorithm other than computing the SVD of $\matA-\tilde{\matA}$. 

While our focus in this section is on lower bounds, we mention that for integers $p \geq 1$, there
is the following simple algorithm for estimating Schatten norms which has a good running time but requires multiple passes over
the data. This is given in \cite{lnw14}. 
\begin{theorem}
For any integer $p \geq 1$, given an $n \times d$ matrix $\matA$, 
there is an $O(p \cdot \nnz(\matA)/\eps^{-2})$ time algorithm for obtaining a $(1+\eps)$-approximation to
$\|\matA\|_p^p$ with probability at least $9/10$. Further,
the algorithm makes $\lceil p/2 \rceil$ passes over the data. 
\end{theorem}
\begin{proof}
Let $r = C/\eps^2$ for a positive constant $C > 0$. 
Suppose $\g^1, \ldots, \g^{r}$ are independent $N(0,1)^d$ vectors, that is, they are independent
vectors of i.i.d. normal random variables with mean $0$ and variance $1$. 

We can assume $\matA$ is symmetric by replacing $\matA$ with the matrix
\[ \matB = \left ( \begin{array}{cc}
0 & \matA^T\\
\matA & 0 .\end{array} \right ) \]
A straightforward calculation shows $\|\matB\|_p = 2^{1/p} \|\matA\|_p$ for all Schatten $p$-norms, and
that the rank of $\matB$ is $2\rho$, where $\rho$ is the rank of $\matA$.

In the first pass we compute $\matB \g^1, \ldots, \matB \g^r$. In the second pass we compute $\matB(\matB \g^1), \ldots, \matB(\matB \g^r)$, and
in general in the $i$-th pass we compute $\matB^i\g^1, \ldots, \matB^i\g^r$. 

If $\matB = \matU \mat\Sigma \matU^T$ is the SVD
of the symmetric matrix $\matB$, 
then after $s = \lceil p/2 \rceil$ passes we will have computed $\matU \mat\Sigma^s \matU^T\g^i$ for each $i$,
as well as $\matU \mat\Sigma^t\matU^T \g^i$ for each $i$ where $t = \lfloor p/2 \rfloor$. Using that $\matU^T\matU = \matI$,
we can compute $(\g^i)^T\matU\mat\Sigma^p\matU^T\g^i$ for each $i$. By rotational invariance, these $r$
values are equal $(\h^1)^T \mat\Sigma^p \h^1, \ldots, (\h^r)^T \mat\Sigma^p \h^r$, where $\h^1, \ldots, \h^r$ are 
independent vectors of independent $N(0,1)$ random variables. 

For every $i$, we have
\begin{eqnarray*}
{\bf E}[(\h^i)^T \mat\Sigma^p \h^i] & = & \sum_{j=1}^{2\rho} {\bf E}[(\h^i)_j^2 \sigma_j^p] = \|\matB\|_p^p,
\end{eqnarray*}
where we use that ${\bf E}[(\h^i)_i^2] = 1$ for all $i$. We also have that
\begin{eqnarray*}
{\bf E}[((\h^i)^T \mat\Sigma^p \h^i)^2] & = & \sum_{j,j'} \Sigma^p_{j,j} \Sigma^p_{j',j'} {\bf E}[(\h^i_{j})^2 (\h^i_{j'})^2]\\
& = & 3 \sum_j \Sigma_{j,j}^{2p} + \sum_{j \neq j'}  \Sigma_{j,j}^p \Sigma_{j',j'}^p {\bf E}[(\h^i_{j})^2] {\bf E}[(\h^i_{j'})^2]\\
& = & 3 \sum_j \Sigma_{j,j}^{2p} + \sum_{j \neq j'}  \Sigma_{j,j}^p \Sigma_{j',j'}^p \\
& \leq & 4 \|\matB\|_p^{2p},
\end{eqnarray*}
where the second equality uses independence of the coordinates of $\h^i$ and that the $4$-th moment of an $N(0,1)$ random variable is $3$, while the third equality uses that the variance
of an $N(0,1)$ random variable is $1$. It follows by Chebyshev's inequality that if
$r \geq 40/\eps^2$ and let $Z = \frac{1}{r}\sum_{i \in [r]} ((\h^i)^T \mat\Sigma^p \h^i)^2$, then
$$\Pr[|Z - \|\matB\|_p^p| > \eps \|\matB\|_p^p] \leq \frac{4 \|\matB\|_p^{2p}}{\eps^2 \|\matB\|_p^{2p}} \cdot \frac{\eps^2}{40}
\leq \frac{1}{10}.$$
This shows correctness. The running time follows from our bound on $r$ and the number $s$ of passes. 
\end{proof}

\subsection{Sketching the operator norm}\label{sec:oded}
The algorithm in the previous section 
has the drawback that it is not a linear sketch, and therefore requires multiple passes
over the data. This is prohibitive in certain applications. We now turn our focus to linear sketches. A first
question is what it means to have a linear sketch of a matrix $\matA$. While some applications could require
a sketch of the form $\matS \cdot \matA \cdot \matT$ for random matrices $\matS$ and $\matT$, we will not restrict ourselves
to such sketches and instead consider treating an $n \times d$ matrix $\matA$ as an $nd$-dimensional vector, 
and computing $\matL(\matA)$, where $\matL:\mathbb{R}^{nd} \rightarrow \mathbb{R}^k$ is a random linear operator, i.e.,
a $k \times nd$ matrix which multiplies $\matA$ on the left, where $\matA$ is treated as an $nd$-dimensional vector.
Since we will be proving lower bounds, this generalization is appropriate. 

While the Frobenius norm is easy to approximate up to a $(1+\eps)$-factor 
via a sketch, e.g., by taking $\matL$ to be a random Gaussian matrix with
$k = O(1/\eps^2)$ rows, another natural, very important Schatten norm is the Schatten-$\infty$, or operator norm
of $\matA$. Can the operator norm of $\matA$ be sketched efficiently?

Here we will prove a lower bound of $k = \Omega(\min(n,d)^2)$ for obtaining a fixed constant factor approximation.
Note that this is tight, up to a constant factor, since if $\matS$ is an $\ell_2$-subspace embedding with $O(d)$
rows, then $\matS \matA$ preserves all the singular values of $\matA$ up to a $(1 \pm 1/3)$ factor. We prove this formally 
with the following lemma.

The idea is to use the min-max principle for singular values. 
	\begin{lemma}\label{lem:interlace}
Suppose $\matS$ is a $(1 \pm \varepsilon)$ $\ell_2$-subspace embedding for $\matA$.  
Then, 
it holds that $(1-\varepsilon)\sigma_i(\matS \matA)\leq \sigma_i(\matA)\leq (1+\varepsilon)\sigma_i(\matS \matA)$ for all $1\leq i\leq d$.
	\end{lemma}
\begin{proof}
The min-max principle for singular values says that
\[
\sigma_i(\matA) = \max_{Q_i} \min_{\substack{\x\in Q_i\\ \|\x\|_2=1}} \|\matA\x\|,
\]
where $Q_i$ runs through all $i$-dimensional subspaces. 
Observe that the range of $\matA$ is a subspace of dimension at most $d$. It follows from the definition of a subspace embedding that 
\[
(1-\varepsilon)\|\matA\x\|\leq \|\matS \matA\x\| \leq (1+\varepsilon)\|\matA\x\|,\quad \forall \x\in\mathbb{R}^d.
\]
The lemma now follows from the min-max principle for singular values, since every vector in the range
has its norm preserved up to a factor of $1+\eps$, and so this also holds for any $i$-dimensional subspace
of the range, for any $i$. 
\end{proof}
Similarly, if $\matT$
is an $\ell_2$ susbpace embedding with $O(d)$ columns, then $\matA \matT$ preserves all the singular values of $\matA$
up to a $(1 \pm 1/3)$ factor, so $O(\min(n,d)^2)$ is achievable. 

Hence, we shall, for simplicity focus on the case when $\matA$ is a square $n \times n$ matrix. The following
$\Omega(n^2)$ lower bound on the sketching dimension $t$ was shown by Oded Regev \cite{r14}, improving an 
earlier $\Omega(n^{3/2})$ lower bound of Li, Nguy$\tilde{\hat{\mbox{e}}}$n, and the author \cite{lnw14}. We will need to describe
several tools before giving its proof. 

Define two distributions: 
\begin{itemize}
\item $\mu$ is the distribution on $n \times n$ matrices with i.i.d. $N(0,1)$ entries. 
\item $\nu$ is the distribution on $n \times n$ matrices obtained by 
(1) sampling $\matG$ from $\mu$, (2) sampling $\u, \ve \sim N(0,\matI_n)$ to be independent $n$-dimensional
vectors with i.i.d. $N(0,1)$ entries, and (3) outputting $\matG + \frac{5}{n^{1/2}}\u\ve^T$.
\end{itemize}

We will show that any linear sketch $\matL$ for 
distinguishing a matrix drawn from $\mu$ from a matrix drawn from 
$\nu$ requires $\Omega(n^2)$ rows. For this to imply a lower bound for approximating the operator norm of
a matrix, we first need to show that with good probability, a matrix drawn from $\mu$ has an operator
norm which differs by a constant factor from a matrix drawn from $\nu$. 

\begin{lemma}\label{lem:gap}
Let $\matX$ be a random matrix drawn from distribution $\mu$, while $\matY$ is a random matrix drawn from distribution $\nu$.
With probability $1-o(1)$, $\|\matY\|_{\infty} \geq \frac{4}{3} \|\matX\|_{\infty}$.   
\end{lemma}
\begin{proof}
It suffices to show that for $\matG$ drawn from $\mu$ and $u,\ve \sim N(0,\matI_n)$, that 
$\|\matG\|_{\infty}$ and $\|\matG+\frac{5}{n^{1/2}}\u\ve^T\|_{\infty}$ differ by a constant factor with probability $1-o(1)$. 
We use the following tail bound. 
\begin{fact}[Operator norm of Gaussian Random Matrix  \cite{vershynin2010}]\label{lem:operator norm}
Suppose that $\matG\sim \mu$. Then with probability at least $1 - e^{-t^2/2}$, it holds that $\|\matX\|_{\infty} \leq 2n^{1/2}+t$.
\end{fact}
By Fact \ref{lem:operator norm}, with probability $1-e^{-\Theta(n)}$, $\|\matG\|_{\infty} \leq 2.1n^{1/2}$. 

Let $\matX = \frac{5}{n^{1/2}}\u\ve^T$. 
Since $\matX$ is of rank one, the only non-zero singular value of $\matX$ is equal to $\FNorm{\matX}$. We also have 
$\FNorm{\matX} \geq 4.9\cdot n^{1/2}$ with probability $1-1/n$, 
since $\FNormS{\u\ve^T}\sim (\chi^2(n))^2$, where $\chi^2(n)$ is the $\chi^2$-distribution with $n$ degrees of freedom,
which is tightly concentrated around $n$.

It follows with probability $1-O(1/n)$, by the triangle inequality
$$\|\matG+\frac{5}{n^{1/2}}\u\ve^T\|_{\infty} \geq 4.9n^{1/2} - 2.1n^{1/2} \geq 2.8n^{1/2} \geq \frac{4}{3}\|\matG\|_{\infty}.$$
%
%
\end{proof}

In our proof we need the following tail bound due to Lata{\l}a
%
%
%
Suppose that $g_{i_1},\dots,g_{i_d}$ are i.i.d. $N(0,1)$ random variables. The following result, due to Lata\l{}a \cite{latala}, bounds the tails of Gaussian chaoses $\sum a_{i_1}\cdots a_{i_d}g_{i_1}\cdots g_{i_d}$. 
The proof of Lata\l{}a's tail bound was later simplified by Lehec \cite{Lehec}.

Suppose that $\matA = (a_{\mb{i}})_{1\leq i_1,\dots,i_d\leq n}$ is a finite multi-indexed matrix of order $d$. For $\mb{i}\in [n]^d$ and $I\subseteq [d]$, define $i_I = (i_j)_{j\in I}$. For disjoint nonempty subsets $I_1,\dots,I_k\subseteq [d]$ define $\|\matA\|_{I_1,\dots,I_k}$ to be:
\[\sup \left\{
	\sum_{\mb{i}} a_{\mb{i}} \x_{i_{I_1}}^{(1)}\cdots \x_{i_{I_k}}^{(k)}: \sum_{i_{I_1}} \left(\x_{i_{I_1}}^{(1)}\right)^2 \leq 1,\dots,
	\sum_{i_{I_k}} \left(\x_{i_{I_k}}^{(1)}\right)^2 \leq 1
\right\}.
\]
Also denote by $S(k,d)$ the set of all partitions of $\{1,\dots,d\}$ into $k$ nonempty disjoint sets $I_1,\dots,I_k$. It is not hard to
show that if a partition 
$\{I_1,\dots,I_k\}$ is finer than another partition $\{J_1,\dots,J_\ell\}$, 
then $\|\matA\|_{I_1,\dots,I_k}\leq \|\matA\|_{J_1,\dots,J_\ell}$.

\begin{theorem}\label{thm:latala}
For any $t > 0$ and $d\geq 2$,
\begin{eqnarray*}
\Pr\left [ \left| \sum_{\mb{i}} a_{\mb{i}} \prod_{j=1}^d g_{i_j}^{(j)} \right| \geq t \right ] \leq C_d\\
\cdot \exp\left\{ - c_d \min_{1\leq k\leq d} \min_{(I_1,\dots,I_k)\in S(k,d)} \left( \frac{t}{\|\matA\|_{I_1,\dots,I_k}}\right)^\frac{2}{k} \right\},
\end{eqnarray*}
where $C_d, c_d > 0$ are constants depending only on $d$.
\end{theorem}

To prove the main theorem of this section, we need a few facts about distances between distributions. 

Suppose $\mu$ and $\nu$ are two probability measures on $\R^n$. 
For a convex function $\phi:\R\to \R$ such that $\phi(1)=0$, we define the $\phi$-divergence
\[
D_\phi(\mu || \nu) = \int \phi\left(\frac{d\mu}{d\nu}\right)d\nu.
\]
In general $D_\phi(\mu||\nu)$ is not a distance because it is not symmetric. 

The \textit{total variation distance} between $\mu$ and $\nu$, denoted by $d_{TV}(\mu,\nu)$, 
is defined as $D_\phi(\mu||\nu)$ for $\phi(x) = |x-1|$. It can be verified that this is indeed a distance. It 
is well known that if $d_{TV}(\mu,\nu) \leq c < 1$, then the probability that any, possibly randomized
algorithm, can distinguish the two distributions is at most $(1+c)/2$. 

The \textit{$\chi^2$-divergence} between $\mu$ and $\nu$, denoted by $\chi^2(\mu||\nu)$, 
is defined as $D_\phi(\mu||\nu)$ for $\phi(x) = (x-1)^2$ or $\phi(x) = x^2-1$. 
It can be verified that these two choices of $\phi$ give exactly the same value of $D_\phi(\mu||\nu)$.

We can upper bound total variation distance in terms of the $\chi^2$-divergence using the next proposition.
\begin{fact}[{\cite[p90]{Tsybakov}}] \label{prop:TV_chi^2}
$d_{TV}(\mu,\nu) \leq \sqrt{\chi^2(\mu||\nu)}$.
\end{fact}

The next proposition gives a convenient upper bound on the $\chi^2$-divergence between a Gaussian distribution
and a mixture of Gaussian distributions. 
\begin{fact}[{\cite[p97]{IS}}] \label{prop:chi^2}
$\chi^2(N(0,\matI_n)\ast \mu||N(0,\matI_n)) \leq {\bf E}[e^{\langle \x,\x'\rangle}-1]$, where $\x,\x'\sim \mu$ are independent.
\end{fact}

We can now prove the main impossibility theorem about sketching the operator norm up to a constant factor. 
\begin{theorem}\cite{r14}
Let $\matL \in \mathbb{R}^{k \times n^2}$ be drawn from a distribution on matrices for which for any fixed $n \times n$
matrix $\matA$, with probability at least $9/10$ 
there is an algorithm which given $\matL(\matA)$, can estimate $\|\matA\|_{\infty}$ up to a constant factor $C$, 
with $1 \leq C < \frac{2}{\sqrt{3}}$. Recall here that $\matL(\matA)$ is the image (in $\mathbb{R}^k$) of the linear map $\matL$ which
takes as input $\matA$ represented as a vector in $\mathbb{R}^{n^2}$.
Under these conditions, it holds that $k = \Omega(n^2)$. 
\end{theorem}
\begin{proof}
We can assume the rows of $\matL$ are orthonormal vectors in $\mathbb{R}^{n^2}$. Indeed, this just corresponds to a change
of basis of the row space of $\matL$, which can be performed in post-processing. That is, given $\matL(\matA)$ one
can compute $\matR \cdot \matL(\matA)$ where $\matR$ is a $k \times k$ change of basis matrix.

Let the orthonormal rows of
$\matL$ be denoted $\matL_1, \ldots, \matL_k$. Although these are vectors in $\mathbb{R}^{n^2}$, we will sometimes think
of these as $n \times n$ matrices with the orthonormal property expressed by $\tr(\matL_i^T \matL_j) = \delta_{ij}$. 

Suppose $\matA \sim \mu$. Then, since the rows of $\matL$ are orthonormal, it follows by rotational invariance 
that $\matL(\matA)$ is distributed as a $k$-dimensional Gaussian vector $N(0, \matI_k)$. 
On the other hand, if $\matA \sim \nu$, then 
$\matL(\matA)$ is distributed as a $k$-dimensional Gaussian vector with a {\it random mean}, that is, as
$N(\matX_{\u, \ve}, \matI_k)$ where
\[
\matX_{\u,\ve} =: \frac{5}{n^{1/2}}\begin{pmatrix}
\u^T \matL_1 \ve, & \u^T \matL_2 \ve, & \cdots &, \u^T \matL_k \ve
\end{pmatrix} =: \frac{5}{n^{1/2}} \matY_{\u,\ve}.
\]
We denote the distribution of $N(\matX_{\u,\ve}, \matI_k)$ by $\mathcal{D}_{n,k}$. 
By Lemma \ref{lem:gap} and the definition of total variation distance, 
to prove the theorem it suffices to upper bound $d_{TV}(N(0,\matI_n), \mathcal{D}_{n,k})$ by a constant $C \leq 4/5$. We shall
do so for $C = 1/4$. 

Without loss of generality we may assume that $k\geq 16$. 
Consider the event $\mathcal{E}_{\u,\ve} = \{\|\matY_{\u,\ve}\|_2\leq 4\sqrt{k}\}$. 
Since ${\bf E}\|\matY_{\u,\ve}\|_2^2 = k$, it follows by Markov's inequality that 
$\Pr_{\u,\ve}\{\mathcal{E}_{\u,\ve}\}\geq 15/16$. 
Let $\widetilde{\mathcal{D}}_{n,k}$ be the marginal distribution of $\mathcal{D}_{n,k}$ conditioned on $\mathcal{E}_{\u,\ve}$. Then
\[
d_{TV}(\widetilde{\mathcal{D}}_{n,k},\mathcal{D}_{n,k})\leq \frac{1}{8},
\]
and it suffices to bound $d_{TV}(N(0,\matI_n),\widetilde{\mathcal{D}}_{n,k})$. 
Resorting to $\chi^2$-divergence by invoking Proposition~\ref{prop:TV_chi^2} and Proposition~\ref{prop:chi^2}, we have that
\[
d_{TV}(N(0,\matI_n),\widetilde{\mathcal{D}}_{n,k})\leq \sqrt{{\bf E} e^{\langle \matX_{\u,\ve},\matX_{\u',\ve'}\rangle}-1},
\]
where $\u,\ve,\u',\ve'\sim N(0,\matI_n)$ conditioned on $\mathcal{E}_{\u,\ve}$ and $\mathcal{E}_{\u',\ve'}$. We first see that 
\begin{eqnarray*}
\langle \matX_{u,v}, \matX_{u', v'}\rangle & = & \frac{25}{n}\sum_{a,b,c,d=1}^n\sum_i (\matL^i)_{ab}(\matL^i)_{cd} \u_a \u'_b \ve_c \ve'_d \\
& =: & D\sum_{a,b,c,d} \matA_{a,b,c,d} \u_a \u'_b \ve_c \ve'_d,
\end{eqnarray*}
where $D = \frac{25}{n}$ and $\matA_{a,b,c,d}$ is an array of order $4$ such that
\[
\matA_{a,b,c,d} = \sum_{i=1}^k \matL^i_{ab} \matL^i_{cd}.
\]
We shall compute the partition norms of $\matA_{a,b,c,d}$ as needed in Lata{\l}a's tail bound Theorem \ref{thm:latala}. 

\noindent\textbf{Partition of size 1.}
The only possible partition is $\{1,2,3,4\}$. We have
\begin{align*}
\|\matA\|_{\{1,2,3,4\}} & = & \left(\sum_{a,b,c,d} \left(\sum_{i=1}^k \matL^i_{a,b}\matL^i_{c,d}\right)^2\right)^{1/2}\\
&=&
\left(\sum_{a,b,c,d} \sum_{i,j=1}^k \matL^i_{a,b}\matL^i_{c,d}\matL^j_{a,b}\matL^j_{c,d}\right)^{1/2}\\
&=&\left(\sum_{a,b,c,d} \sum_{i=1}^k (\matL^i_{a,b})^2 (\matL^i_{c,d})^2\right)^{1/2}\\
&=&\sqrt{k}
\end{align*}

\noindent\textbf{Partitions of size 2.} The norms are automatically upper-bounded by $\|\matA\|_{\{1,2,3,4\}} = \sqrt{k}$.

\noindent\textbf{Partitions of size 3.} Up to symmetry, there are only two partitions to consider:
$\{1,2\}, \{3\}, \{4\}$, and $\{1,3\}, \{2\}, \{4\}$. We first consider the partition $\{1, 2\}, \{3\}, \{4\}$. We have
\begin{eqnarray*}
\|\matA\|_{\{1,2\},\{3\},\{4\}} &=& \sup_{\matW \in \mathbb{S}^{n^2-1}, \u^T, \ve \in \mathbb{S}^{n-1}} \sum_{i=1}^k 
\langle L^i W \rangle \cdot \u^T \matL^i \ve\\
&\leq & \left (\sum_{i=1}^k (\langle \matL^i, \matW \rangle \right )^{1/2} \cdot \left (\sum_{i=1}^k (\u^T \matL^i \ve)^2 \right )^{1/2}\\
&\leq & 1,
\end{eqnarray*}
where the first inequality follows from Cauchy-Schwarz. We now consider the partition $\{1, 3\}, \{2\}, \{4\}$. We have
\begin{eqnarray}\label{eqn:oded}
\|\matA\|_{\{1,3\},\{2\},\{4\}} &= \sup_{\matW \in \mathbb{S}^{n^2-1}, \u^T, \ve \in \mathbb{S}^{n-1}} \sum_{i=1}^k \langle W, ((\matL^i \u) \otimes
(\matL^i \ve)) \rangle \notag \\
& = \|\sum_{i=1}^k ((\matL^i \u) \otimes (\matL^i \ve))\|_{F},
\end{eqnarray}
where the second equality follows by Cauchy-Schwarz.

Let $\matT$ be the $n^2 \times k$ matrix whose columns are the $\matL^i$, interpreted as column vectors in $\matR^{n^2}$. 
Let $\matT'$ be the $n \times k$ matrix whose columns are $\matL^i \u$. 
Similarly, $\matT''$ is the $n \times k$ matrix whose columns are $\matL^i \ve$. Then 
$$ \|\sum_{i=1}^k ((\matL^i \u) \otimes (\matL^i \ve))\|_{\infty} = \|\matT' (\matT'')^T\|_{\infty}.$$
Since $\matT'$ and $\matT''$ are obtained from $\matT$ by applying a contraction, we have that
$\|\matT'\|_{\infty} \leq \|\matT\|_{\infty} \leq 1$, and 
$\|\matT''\|_{\infty} \leq \|\matT\|_{\infty} \leq 1$. Therefore, $\|\matT' (\matT'')^T\|_{\infty} \leq 1$. Consequently, 
since $\matT' (\matT'')^T$ is an $n \times n$ matrix, $\FNorm{\matT' (\matT'')^T} \leq \sqrt{n}$. 

\noindent\textbf{Partition of size 4.}
The only partition is $\{1\},\{2\},\{3\},\{4\}$. Using that for integers $a,b$, $a \cdot b \leq (a^2+b^2)/2$, we have
\begin{align*}
\|\matA\|_{\{1\},\{2\},\{3\},\{4\}} &= \sup_{\u,\ve,\u',\ve'\in \mathbb{S}^{n-1}} \sum_{i=1}^k \u^T \matL^i \ve \u'^T \matL^i \ve'\\
&\le \sup_{\u,\ve,\u',\ve'} \frac{1}{2}\left(\sum_{i=1}^k \langle \u\ve^T, \matL^i\rangle^2 + \langle \u'\ve'^T, \matL^i\rangle^2\right) \\
&\le 1
\end{align*}
The last inequality follows the fact that $\u\ve^T$ is a unit vector in $\R^{n^2}$ and $\matL^i$'s are orthonormal vectors in $\R^{n^2}$.

Lata{\l}a's inequality (Theorem~\ref{thm:latala}) states that for $t > 0$, 
\begin{eqnarray*}
\Pr\left [\left|\sum_{a,b,c,d} \matA_{a,b,c,d}\u_a\u_b'\ve_c\ve_d'\right| > t\right ] & \leq & C_1\\
& \cdot & \exp\left(-c\min\left\{\frac{t}{\sqrt k},\frac{t^2}{k},\frac{t^\frac23}{n^\frac13},\sqrt{t}\right\}\right)
\end{eqnarray*}
The above holds with no conditions imposed on $\u,\ve,\u',\ve'$. For convenience, we let 
$$f(t) = \min\left\{\frac{t}{\sqrt k},\frac{t^2}{k},\frac{t^\frac23}{n^\frac13},\sqrt{t}\right\}.$$
It follows that
\begin{eqnarray*}
\Pr\left [|\langle \matY_{\u,\ve},\matY_{\u',\ve'}\rangle| > t \big| \mathcal{E}_{\u,\ve}\mathcal{E}_{\u',\ve'}\right ]
& \leq &  
\frac{\Pr\left [|\langle \matY_{\u,\ve},\matY_{\u',\ve'} \rangle| > t\right ]}{\Pr [\mathcal{E}_{\u',\ve'}\}\Pr\{\mathcal{E}_{\u,\ve} ]}\\
& \leq & C_2 \exp\left(-c\cdot f(t) \right).
\end{eqnarray*}
Note that conditioned on $\mathcal{E}_{\u,\ve}$ and $\mathcal{E}_{\u',\ve'}$,
\[
|\langle Y_{\u,\ve},Y_{\u',\ve'}\rangle| \leq \|Y_{\u,\ve}\|_2\|Y_{\u',\ve'}\|_2\leq 16k.
\]

We now claim that $tD = \frac{25t}{n} \leq c f(t) /2$ for all $\sqrt{k} < t < 16k$, provided $k = o(n^2)$. First, note
that since $t < 16k$, $\frac{t}{\sqrt{k}} = O(\sqrt{t})$. Also, since $t > \sqrt{k}$, $\frac{t}{\sqrt{k}} \leq \frac{t^2}{k}$. 
Hence, $f(t) = \Theta(\min(t/\sqrt{k}, t^{2/3}/n^{1/3}))$. Since $k = o(n^2)$, if $t/\sqrt{k}$ achieves the minimum, then it is larger
than $\frac{25t}{n}$. On the other hand, if $t^{2/3}/n^{1/3}$ achieves the minimum, then $c f(t)/2 \geq \frac{25t}{n}$ whenever
$t = o(n^2)$, which since $t < 16k = o(n^2)$, always holds. 

Integrating the tail bound gives that
\begin{align*}
{\bf E}[e^{\matX_{\u,\ve},\matX_{\u',\ve'}}] &= 1 + D\int_0^{16k}e^{tD}\Pr[|\langle \matY_{\u,\ve},\matY_{\u',\ve'}\rangle | > t]dt\\
&\leq 1 + D\int_{0}^{\sqrt{k}}e^{tD}dt + D\int_{\sqrt{k}}^{16k}e^{tD - cf(t)}dt\\
&\leq 1 + D\sqrt{k}e^{\sqrt{k}D} + D\int_{\sqrt{k}}^{16k} e^{-cf(t)/2}dt\\
& \leq 1 + o(1),
\end{align*}
where the first inequality uses that $\Pr[|\langle \matY_{\u,\ve},\matY_{\u',\ve'}\rangle | > t] \leq 1$, the second
inequality uses the above bound that $tD \leq cf(t)/2$, and the first part of the third inequality uses that $k = o(n^2)$. For
the second part of the third inequality, since $\sqrt{k} < t < 16k$, we have that $f(t) = \Theta(\min(t/\sqrt{k}, t^{2/3}/n^{1/3}))$. Also, 
$f(t) \geq 1$ (assuming $k \geq n$, which we can assume, since if there is a linear sketch with $k < n$ 
there is also one with $k > n$), and so $D \int_{\sqrt{k}}^{16k} e^{-cf(t)/2}dt \leq D \sqrt{k}$ since the integral is dominated by a geometric
series. Also, since $k = o(n^2)$, $D\sqrt{k} = o(1)$. 

It follows that $d_{TV}(N(0,\matI_{n}),\widetilde{\mathcal{D}}_{n,k})\leq 1/10$ and thus 
\[
d_{TV}(N(0,\matI_{n}),\mathcal{D}_{n,k})\leq 1/10+1/8 < 1/4.
\]
\end{proof}

\subsection{Streaming lower bounds}\label{sec:streaming}
In this section, we explain some basic communication complexity,
and how it can be used to prove bit lower bounds for the space
complexity of linear algebra problems in the popular {\it streaming model} of computation.
We refer the reader to Muthukrishnan's survey \cite{m05} for a comprehensive overview on the streaming
model. We state the definition of the model that we need as follows. These results are by Clarkson
and the author \cite{CW09}, and we follow the exposition in that paper.

In the {\it turnstile model} of computation for linear algebra problems, there is an input matrix
$\matA \in \mathbb{R}^{n \times d}$ which is initialized to all zeros. We then see a finite sequence of 
additive updates to the coordinates of $\matA$, where each update has the form $(i,j,\delta)$
for $i \in [n], j \in [d]$, and $\delta \in \mathbb{R}$, with the meaning that 
$\matA_{i,j} \leftarrow \matA_{i,j} + \delta$. We will restrict ourselves to the case when at all
times, all entries of $\matA$ are integers bounded by $M \leq \poly(nd)$, for some fixed polynomial
in $n$ and $d$. We note that the sequence of updates is adversarially chosen, and multiple
updates to each entry $\matA_{i,j}$ may occur and in an arbitrary order (interleaved with other
updates to entries $\matA_{i', j'}$). One of the main goals in this computational model is to
compute or approximate a function of $\matA$ using as little space in bits as possible. 

\subsection{Communication complexity}
For lower bounds in the turnstile model, we use a few definitions and basic results
from two-party communication complexity, described below. We refer the reader to the
book by Kushilevitz and Nisasn for more information \cite{kn97}.
We will call the two parties Alice and Bob.

For a function $f: \mathcal{X} \times 
\mathcal{Y} \rightarrow \{0,1\}$, we use 
$R_{\delta}^{1-way}(f)$ to denote the randomized communication complexity with 
two-sided error at most $\delta$ in which only a single message is sent from 
Alice to Bob. Here, only a single message $M(X)$ is sent from Alice to Bob,
where $M$ is Alice's message function of her input $X$ and her random coins. Bob
computes $f(M(X), Y)$, where $f$ is a possibly randomized function of $M(X)$
and Bob's input $Y$. For every input pair $(X,Y)$, Bob should output a correct
answer with probability at least $1-\delta$, where the probability is taken over the joint
space of Alice and Bob's random coin tosses. If this holds, we say the protocol is correct.
The communication complexity $R_{\delta}^{1-way}(f)$ is then the minimum over correct protocols, 
of the maximum length of Alice's message $M(X)$, over all inputs and all settings to the 
random coins. 

We also use $R_{\mu, \delta}^{1-way}(f)$ to denote the minimum 
communication of a protocol, in which a single message from Alice to Bob is 
sent, for solving $f$ with probability at least $1-\delta$, where the 
probability now is taken over both the coin tosses of the protocol and an input 
distribution $\mu$.

In the augmented indexing problem, which we call 
$AIND$, Alice is given $\x \in \{0,1\}^n$, while Bob is given 
both an $i \in [n]$ together with $\x_{i+1}, \x_{i+2}, \ldots, \x_n$. Bob should 
output $\x_i$.
\begin{theorem}(\cite{MNSW98})\label{thm:mnsw}
$R_{1/3}^{1-way}(AIND) = \Omega(n)$ and also
$R_{\mu, 1/3}^{1-way}(AIND) = \Omega(n)$,
where $\mu$ is uniform on $\{0,1\}^n \times [n]$.
\end{theorem}

\subsubsection{Matrix product}
We start with an example showing how to use Theorem \ref{thm:mnsw} for proving
space lower bounds for the {\sf Matrix Product} problem, which is the same
as given in Definition \ref{def:matrixProduct}. Here we also include in the definition
the notions relevant for the streaming model.

\begin{definition}\label{def:matrixProduct2}
In the {\em Matrix Product} Problem matrices $\matA$ and $\matB$ are presented as an arbitrary stream of
additive updates to their entries, where $\matA$ and $\matB$ each have $n$ rows and a 
total of $c$ columns.  At all times in the stream we assume the entries of $\matA$ and $\matB$
are specified by $O(\log nc)$-bit numbers.
The goal is to output a 
matrix $\matC$ so that 
\[
\normF{\matA^T\matB - \matC} \leq \varepsilon \normF{\matA} \normF{\matB}.
\]
\end{definition}

\begin{theorem}
Suppose $n \geq \beta (\log_{10} cn)/\varepsilon^2$ for an absolute constant $\beta > 0$, and
that the entries of $\matA$ and $\matB$ are represented by $O(\log (nc))$-bit numbers.
Then any randomized $1$-pass algorithm which solves Matrix Product with probability 
at least $4/5$ uses $\Omega(c\varepsilon^{-2}\log (nc))$ bits of space.
\end{theorem}
\begin{proof}
Throughout we shall assume that $1/\eps$ is an integer, and that
$c$ is an even integer. These conditions can be removed with minor modifications.
Let $Alg$ be a $1$-pass algorithm which solves Matrix Product with 
probability at least $4/5$. Let $r = \log_{10} (cn)/(8\varepsilon^2)$. We 
use $Alg$ to solve instances of $AIND$ on strings of size $cr/2$. It will 
follow by Theorem \ref{thm:mnsw} that the space complexity of $Alg$ must be 
$\Omega(cr) = \Omega(c\log (cn))/\eps^2$. 

Suppose Alice has $\x \in \{0,1\}^{cr/2}$. 
She creates a $c/2 \times n$ matrix $\matU$ as follows. We will have that
$\matU = (\matU^0, \matU^1, \ldots, \matU^{\log_{10}(cn) - 1}, \mat0_{c/2 \times (n-r)})$, where for each 
$k \in \{0, 1, \ldots, \log_{10} (cn)-1\}$, $\matU^k$ is a 
$c/2 \times r/(\log_{10} (cn))$ matrix with entries in the set $\{-10^k, 10^k\}$.
Also, $\mat0_{c/2 \times (n-r)}$ is a $c/2 \times (n-r)$ matrix consisting of all zeros.

Each entry of $\x$ is associated with a unique entry in a unique $\matU^k$. If the
entry in $\x$ is $1$, the associated entry in $\matU^k$ is $10^k$, otherwise it is
$-10^k$. Recall that $n \geq \beta (\log_{10}(cn))/\eps^2$, so we can assume 
that $n \geq r$ provided that $\beta > 0$ is a sufficiently large constant.

Bob is given an index in $[cr/2]$, and suppose this index of $\x$ is associated
with the $(i^*, j^*)$-th entry of $\matU^{k^*}$. By the definition of the AIND problem,
we can assume that Bob is given all entries of $\matU^k$ for all $k > k^*$.
Bob creates a $c/2 \times n$ matrix $\matV$ as follows. In $\matV$, all entries in the first
$k^*r/(\log_{10} (cn))$ columns are set to $0$. The entries in the remaining columns
are set to the negation of their corresponding entry in $\matU$. This is possible 
because Bob has $\matU^k$ for all $k > k^*$. 
The remaining $n - r$ columns of 
$\matV$ are set to $0$. We define $\matA^T = \matU + \matV$. Bob also creates the 
$n \times c/2$ matrix $\matB$ which is $0$ in all but the 
$((k^*-1)r/(\log_{10} (cn)) + j^*,1)$-st entry, which is $1$. Then,
$$\normF{\matA}^2 = \normF{\matA^T}^2 = \left (\frac{c}{2} \right ) \left (\frac{r}{\log_{10}(cn)} \right ) \sum_{k=1}^{k^*} 100^k
\leq \left (\frac{c}{16\eps^2} \right ) \frac{100^{k^*+1}}{99}.$$

Using that $\normF{\matB}^2 = 1$, 
$$\eps^2 \normF{\matA}^2 \normF{\matB}^2 
\leq \eps^2 \left (\frac{c}{16\eps^2} \right ) \frac{100^{k^*+1}}{99}
= \frac{c}{2} \cdot 100^{k^*} \cdot \frac{25}{198}.
$$

$\matA^T \matB$ has first column equal to the $j^*$-th column of $\matU^{k^*}$, 
and remaining columns equal to zero. Let $\matC$ be the $c/2 \times c/2$ approximation 
to the matrix $\matA^T\matB$. We say an entry
$\matC_{\ell, 1}$ is {\it bad} if its sign disagrees with the sign of $(\matA^T \matB)_{\ell, 1}$. 
If an entry $\matC_{\ell, 1}$ is bad, then 
$((\matA^T\matB)_{\ell, 1} - \matC_{\ell, 1})^2 \geq 100^{k^*}$. Thus, the fraction of bad
entries is at most $\frac{25}{198}$. Since we may assume that $i^*, j^*$, and $k^*$
are chosen independently of $\x$, with probability at least $173/198$, 
sign$(\matC_{i^*, 1}) = $ sign$(\matU^{k^*}_{i^*, j^*})$. 

Alice runs $Alg$ on $\matU$ in an arbitrary order, transmitting the state to Bob, who 
continues the computation on $\matV$ and then on $\matB$, again feeding the entries into 
$Alg$ in an arbitrary order. Then with probability at least $4/5$, over $Alg$'s 
internal coin tosses, $Alg$ outputs a matrix $\matC$ for which 
$\normF{\matA^T\matB-\matC}^2 \leq \eps^2 \normF{\matA}^2 \normF{\matB}^2$. 

It follows that the parties can solve the AIND problem with probability at least
$4/5 - 25/198 > 2/3$. The theorem now follows by Theorem \ref{thm:mnsw}.
\end{proof}

\subsubsection{Regression and low rank approximation}
One can similarly use communication complexity to obtain lower bounds in the streaming
model for Regression and Low Rank Approximation. The results are again obtained by reduction
from Theorem \ref{thm:mnsw}. They are a bit more involved than those for matrix product, and
so we only state several of the known theorems regarding these lower bounds. We begin with
the formal statement of the problems, which are the same as defined earlier, 
specialized here to the streaming setting. 

\begin{definition}
In the {\it $\ell_2$-Regression Problem}, an $n \times d$ matrix $\matA$ and an 
$n \times 1$ column vector $\b$ are presented as a sequence of additive updates to their
coordinates. We assume that at all points in the stream, the entries of $\matA$ and $\b$ 
are specified by $O(\log nd)$-bit 
numbers. The goal is to output a vector $\x$ so that
$$\norm{\matA\x-\b} \leq (1+\eps)\min_{\x' \in \mathbb{R}^d}\norm{\matA\x'-\b}.$$
\end{definition}

\begin{theorem}(\cite{CW09})
Suppose $n \geq d (\log_{10} (nd))/(36\varepsilon)$ and $d$ is sufficiently large. Then any 
randomized $1$-pass algorithm which solves the $\ell_2$-Regression Problem with 
probability at least $7/9$ needs $\Omega(d^2\varepsilon^{-1}\log (nd))$ bits of space.
\end{theorem}

\begin{definition}\label{prob:lowrank}
In the {\it Rank-$k$ Approximation Problem}, 
we are given an integer $k$, value $\varepsilon>0$, and $n\times d$ matrix $\matA$ which is presented
as a sequence of additive updates to its coordinates. The goal is to 
find a matrix $\tilde \matA_k$ of rank at most $k$ so that
\[
\normF{\matA-\tilde \matA_k} \le (1+\varepsilon)\normF{\matA-\matA_k},
\]
where $\matA_k$ is the best rank-$k$ approximation to $\matA$.
\end{definition}

\begin{theorem}\label{thm:rankLB}(\cite{CW09})
Let $\varepsilon > 0$ and $k \geq 1$ be arbitrary. Then, 

(1) Suppose $d >\beta k/\varepsilon$ for an absolute constant $\beta> 0$. Then any
randomized $1$-pass algorithm which solves the Rank-$k$ Approximation Problem with 
probability at least $5/6$, and which receives
the entries of $\matA$ in row-order, must use $\Omega(nk/\varepsilon)$ bits of space.

(2) Suppose $n > \beta k/\varepsilon$ for an absolute constant $\beta > 0$. Then any
randomized $1$-pass algorithm which solves the Rank-$k$ Approximation Problem with
probability at least $5/6$, and which receives
the entries of $\matA$ in column-order must use $\Omega(dk/\varepsilon)$ bits of space.

\end{theorem}

\subsection{Subspace embeddings}\label{sec:seLB}
We have seen that $\ell_2$-subspace embeddings have a number of important applications, 
especially ones that are oblivious to the matrix $\matA$ they are being applied to. A natural
question is what the minimum dimension of such subspace embeddings needs to be. That is,
we seek to design a distribution $\Pi$ over $r \times n$ matrices $\matS$, with $r$ as small
as possible, so that for any
fixed $n \times d$ matrix $\matA$, we have with constant probability over $\matS$ drawn from $\Pi$,
\begin{eqnarray}\label{eqn:selb}
\forall \x \in \mathbb{R}^d, (1-\eps)\|\matA \x\|_2 \leq \|\matS \matA\x\|_2 \leq (1+\eps)\|\matA \x\|_2.
\end{eqnarray}
We have seen that by choosing $\matS$ to be a matrix of i.i.d. Gaussians, it suffices to set
$r = O(d/\eps^2)$, which also achieves (\ref{eqn:selb}) with probability $1-\exp(-d)$.

A theorem of Nelson and Nguy$\tilde{\hat{\mbox{e}}}$n \cite{nn14} shows that the above setting of $r$ is best possible
for achieving (\ref{eqn:selb}), even if one desires only constant error probability. 

\begin{theorem}\label{thm:lbnn}
For $n \geq Cd/\eps^2$ for a sufficiently large constant $C > 0$, 
any distribution $\Pi$ satisfying (\ref{eqn:selb}) with constant probability over $\matS$
drawn from $\Pi$, satisfies $r = \Omega(d/\eps^2)$. 
\end{theorem}

While Theorem \ref{thm:lbnn} gives a nice dimension lower bound for subspace embeddings,
it turns out that often one can do better than it in specific applications, such as the $\ell_2$ Regression
Problem, where it is possible to achieve a dependence on $1/\eps$ that is linear. This is because in the
analysis, only a subspace embedding with constant $\eps$ is needed, while additional
other properties of the sketching matrix $\matS$ are used that only incur a $1/\eps$ factor in the dimension.  

\subsection{Adaptive algorithms}\label{sec:adaptiveLower}
In this section we would like to point out a word of caution of using a sketch for multiple, adaptively chosen
tasks. 

Suppose, for example, that 
we have a $k \times n$ sketching matrix $\matS$, with $k \ll n$, drawn from some distribution
with the property that there is a, possibly randomized {\it reconstruction function} $f$ 
such that for any fixed vector $\x \in \mathbb{R}^n$, 
\begin{eqnarray}\label{eqn:JLproperty}
(1-\eps)\|\x\|_2^2 \leq \|f(\matS\x)\|_2^2 \leq (1+\eps)\|\x\|_2^2,
\end{eqnarray}
with probability at least $1-\delta$ for some parameter $\delta > 0$. In this section we will focus
on the case in which $\delta< n^{-c}$ for every constant $c > 0$, that is, $\delta$ shrinks faster
than any inverse polynomial in $n$

The property in (\ref{eqn:JLproperty}) is a basic property 
that one could ask for a sketching matrix $\matS$ to satisfy, and we will refer to an $(\matS, f)$ pair satisfying
this property as 
an {\it $\ell_2$-sketch}. It is
clear, for example, that an $\ell_2$-subspace embedding has this property for certain $\eps$ and $\delta$, 
where the function $f(\matS\x) = \|\matS\x\|_2^2$. As
we have seen, such embeddings have a number of applications in linear algebra. 

A natural question is if an $\ell_2$-sketch can be reused, 
in the following sense. Suppose we compute
$$\matS \cdot \x^1, \matS \cdot \x^2, \matS \cdot \x^3, \ldots, \matS \cdot \x^r,$$
where $\x^1, \ldots, \x^r$ is an adaptively chosen sequence of vectors in $\mathbb{R}^n$. We will also assume
$r \leq n^c$ for some fixed constant $c > 0$.  
For each $\matS\cdot \x^i$, we obtain $f(\matS\cdot \x^i)$. 
Note that if the $\x^i$ were fixed before the choice of $\matS$, 
by a union bound over the $r$ vectors $\x^1, \ldots, \x^r$,
we should have that with probability at least $1-\delta n^c$, 
simultaneously for all $i$, 
$$(1-\eps)\|\x^i\|_2^2 \leq \|f(\matS\x^i)\|_2^2 \leq (1+\eps)\|\x^i\|_2^2.$$
A natural question though, is what happens in the adaptive case, where the $\x^i$ 
can depend on $f(\matS\x^1), f(\matS\x^2), \ldots, f(\matS\x^{i-1})$. 
As an illustrative example that this is a nontrivial issue, 
we first consider the standard estimator $f(\matS\x) = \|\matS\x\|_2^2$. An $\ell_2$ sketch with this choice of
estimator is often called a {\it Johnson-Lindenstrauss transform}. 

\begin{theorem}\label{thm:jlBreak}
For any $\eps > 0$, and any Johnson-Lindenstrauss transform $\matS$ with $k$ rows and $n$ columns, $k < n$, 
there is an algorithm which makes
$r = \binom{k+1}{2} + (k+1) + 1$ query vectors $\x^1, \ldots, \x^r$, for which with probability $1$, 
$$f(\matS\x^r) \notin \left [(1-\eps)\|\x^r\|^2_2, \ (1+\eps)\|\x^r\|^2_2 \right ].$$
Further, the algorithm runs in $O(k^3)$ time and the first $r-1$ queries can be chosen 
non-adaptively (so the algorithm makes a single adaptive query, namely, $\x^r$). 
\end{theorem}
\begin{proof}
The algorithm first queries the sketch on the vectors
$$\e_i, \e_{i}+\e_j \ \ \textrm{for all } i,j \in [k+1],$$
where the $\e_i$ are the standard unit vectors in $\mathbb{R}^n$. Since $\matS$ is a Johnson-Lindenstrauss
transform, it learns $\|\matS_i\|_2^2$ and $\|\matS_i + \matS_j\|_2^2$ for all $i,j \in [k+1]$, where $\matS_i$
denotes the $i$-th column of $\matS$. Since
$\|\matS_i +\matS_j\|_2^2 = \|\matS_i\|_2^2 + \|\matS_j\|_2^2 + 2\langle \matS_i, \matS_j \rangle$, the algorithm learns 
$\langle \matS_i, \matS_j \rangle$ for all $i,j$. 

Now consider the $n \times n$ matrix $\matA = \matS^T\matS$. This matrix
has rank at most $k$, since $\matS$ has rank at most $k$. By definition, $\matA_{i,j} = \langle \matS_i, \matS_j \rangle$. 
It follows that the algorithm has learned the upper $(k+1) \times (k+1)$ submatrix of $\matA$, let us call
this submatrix $\matB$. As $\matB$ has rank at most $k$, it follows there is a non-zero vector $\u \in \mathbb{R}^{k+1}$ 
in the kernel of the span of the rows of $\matB$. 

Consider the non-zero vector $\ve \in \mathbb{R}^n$ obtained by padding $\u$ with $n-(k+1)$ zero coordinates. 
We claim that $\matS \ve = 0$, which would imply that $\|\matS \ve\|_2^2$ cannot provide a relative error approximation
to $\|\ve\|_2^2$ for any $\eps > 0$. 

To prove the claim, write $\matS = [\matC, \matD]$ as a block matrix, where $\matC$ consists of the first $k+1$ columns
of $\matS$, and $\matD$ consists of the remaining $n-(k+1)$ columns of $\matS$. Then
\[ \matS^T\matS = \left ( \begin{array}{cc}
\matC^T\matC & \matC^T\matD\\
\matD^T\matC & \matD^T\matD \end{array} \right ), \]
where $\matB = \matC^T\matC$. Writing $\matC = \matU \mat\Sigma \matV^T$ in its SVD, we see that $u$ is orthogonal to the row space
of $\matV^T$, which implues $\matS \ve = 0$, as desired. 

To conclude, note that the query algorithm makes $r$ queries, the only adaptive one being $\ve$, and runs
in $O(k^3)$ time to compute the SVD of $\matB$. 
\end{proof}
While Theorem \ref{thm:jlBreak} rules out using a Johnson-Lindenstrauss transform as an $\ell_2$ sketch
which supports adaptively chosen query vectors, one could ask if a different, more carefully designed
$\ell_2$ sketch, could support adaptively chosen query vectors. Unfortunately, the answer is no, in a very
strong sense, as shown by Hardt and the author \cite{hw13}. Namely, the authors show that for {\it any} $\ell_2$ sketch, there is
an efficient query algorithm which can find a distribution on queries $\x^i$ for which with constant
probability, $f(\matS\x^i) \notin [(1-\eps)\|\x^i\|_2^2, (1+\eps)\|\x^i\|_2^2$. To avoid introducing additional
notation, we state the theorem of \cite{hw13} informally, and refer the reader to the paper for more details.  

\begin{theorem}\label{thm:adaptive}[Informal version]
There is a randomized algorithm which, given a parameter $B\ge2$ and
oracle access to an $\ell_2$ sketch that uses at most $r=n-O(\log(nB))$
rows, with high probability finds a distribution over queries on which 
the linear sketch fails to satisfy (\ref{eqn:JLproperty}) with constant probability.

The algorithm makes at most $\poly(rB)$ adaptively chosen queries to the
oracle and runs in time $\poly(rB).$ Moreover, the algorithm uses only $r$
``rounds of adaptivity'' in that the query sequence can be partitioned into at
most $r$ sequences of non-adaptive queries.
\end{theorem}
We state some of the intuition behind the proof of Theorem \ref{thm:adaptive} below. 

The problem of approximating the Euclidean norm of $\x$ is captured 
by the following game between two
players, Alice and Bob. Alice chooses an $r \times n$ sketching matrix $\matS$ from
distribution $\pi$. Bob makes a sequence of queries $\x^1, \ldots, \x^r\in\R^n$ to
Alice, who only sees $\matS\x^i$ on query $i$. This captures the fact that a sketching
algorithm only has access to $\matS\x^i$, rather than to $\x^i$ itself. The multiple
queries $\x^i$ of Bob are the vectors whose Euclidean norm one would like to approximate
using the sketching matrix $\matS$. 
Alice responds by telling Bob the
value $f(\matS\x^i)$, which is supposed to be a $(1+\eps)$-approximation to the Euclidean
norm of $\x^i$. 

Here $f$ is an arbitrary function that 
need not be efficiently computable. For simplicity of presentation, we'll just focus
on the case in which $f$ uses no randomness, though Theorem \ref{thm:adaptive}
holds also for randomized functions $f$. 
Bob will try to learn the row space $R(\matA)$ of Alice, namely the at most
$r$-dimensional subspace of $\mathbb{R}^n$ spanned by the rows of $\matA$.  If Bob
knew $R(\matA)$, he could, with probability $1/2$ query $0^n$ and with probability
$1/2$ query a vector in the kernel of $\matA$. Since Alice cannot distinguish the
two cases, and since the norm in one case is $0$ and in the other case
non-zero, she cannot provide a relative error approximation. 

Theorem \ref{thm:adaptive} provides 
an algorithm (which can be executed efficiently by Bob) that learns
$r-O(1)$ orthonormal vectors that are almost contained in $R(\matA).$ While this
does not give Bob a vector in the kernel of $\matA,$ it effectively reduces
Alice's row space to be constant dimensional thus forcing her to make a
mistake on sufficiently many queries (since the variance is large). 

\paragraph{The conditional expectation lemma.}
In order to learn $R(\matA)$, Bob's initial query is drawn from the multivariate
normal distribution $N(0, \tau \matI_{n})$, where $\tau \matI_{n}$ is the covariance
matrix, which is just a scalar $\tau$ times the identity matrix $\matI_{n}$.  This
ensures that Alice's view of Bob's query $\x$, namely, the projection $P_{\matA}\x$ of
$\x$ onto $R(\matA)$, is spherically symmetric, and so only depends on
$\|P_{\matA}\x\|_2$. Given $\|P_{\matA}\x\|_2$, Alice needs to output $0$ or $1$ depending
on what she thinks the norm of $\x$ is. Since Alice has a
proper subspace of $\mathbb{R}^n$, she will be confused into thinking $\x$ has
larger norm than it does when $\|P_{\matA}\x\|_2$ is slightly larger than its
expectation (for a given $\tau$), that is, when $\x$ has a non-trivial
correlation with $R(\matA)$. 

Formally, Theorem \ref{thm:adaptive} makes use of a conditional expectation
lemma showing that there exists a choice of $\tau$ for which 
$${\bf E}_{\x \sim N(0, \tau \matI_{r})} \left [ \|P_{\matA}\x\|_2^2 \mid f(\matA\x) = 1 \right] 
- {\bf E}_{\x \sim N(0, \tau \matI_{r})} \left[ \|P_{\matA}\x\|_2^2 \right]$$ 
is non-trivially large. This is done by showing that the sum of this
difference over all possible $\tau$ in a range $[1, B]$ is noticeably
positive. Here $B$ is the approximation factor. 
In particular, there exists a $\tau$ for which this difference is
large. To show the sum is large, for each possible condition $v =
\|P_{\matA}\x\|_2^2$, there is a probability $q(v)$ that the algorithm outputs $1$,
and as we range over all $\tau$, $q(v)$ contributes both positively and
negatively to the above difference based on $v$'s weight in the $\chi^2$-distribution 
with mean $r \cdot \tau$. The overall contribution of $v$ can be shown to be
zero. Moreover, by correctness of the sketch, $q(v)$ must typically be close to $0$ for
small values of $v,$ and typically close to $1$ for large values of $v.$ 
Therefore $q(v)$ zeros out some of the negative contributions that $v$ would otherwise
make and ensures some positive contributions in total.

\paragraph{Boosting a small correlation.}
Given the conditional expectation lemma, one then finds many independently
chosen $\x^i$ for which each $\x^i$ has a slightly increased expected projection
onto Alice's space $R(\matA)$. At this point, however, it is not clear how to
proceed unless one can aggregate these slight correlations into a single vector
which has very high correlation with $R(\matA).$ This is accomplished by arranging
all $m=\poly(n)$ positively labeled vectors $\x^i$ into an $m\times n$ matrix $\matG$ and
computing the top right singular vector~$\ve^*$ of $G.$ Note that this can be
done efficiently.  One can then show that $\|P_{\matA}\ve^*\|\ge 1-1/\poly(n).$ In
other words $\ve^*$ is almost entirely contained in $R(\matA).$ This step is crucial
as it gives a way to effectively reduce the dimension of Alice's space
by~$1$. 

\paragraph{Iterating the attack.}
After finding one vector inside Alice's space, one must iterate the argument. In fact 
Alice might initially use only a small fraction of her rows and switch to a new set
of rows after Bob learned her initial rows.
An iteration of the previously described attack is performed as follows. Bob makes
queries from a multivariate normal distribution inside of the subspace
orthogonal to the the previously found vector. In this way one effectively
reduces the dimension of Alice's space by $1$, and repeats the attack
until her space is of constant dimension, at which point a standard
non-adaptive attack is enough to break the sketch. Several complications
arise at this point. For example, each vector that we find is only approximately contained
in $R(\matA).$ One needs to rule out that this approximation error could help Alice.
This is done by adding a sufficient amount of global Gaussian noise to the query
distribution. This has the effect of making the distribution statistically
indistinguishable from a query distribution defined by vectors that are exactly
contained in Alice's space. A generalized
conditional expectation lemma is then shown for such distributions.

\section{Open Problems}
We have attempted to cover a number of examples where sketching techniques can be used to speed
up numerical linear algebra applications. We could not cover everything and have of course missed out
on some great material. We encourage the reader to look at other surveys in this area, such
as the one by Mahoney \cite{m11}, for a treatment of some of the topics that we missed. 

Here we conclude with some open problems.
\\\\
{\bf Open Question 1 (Spectral Low Rank Approximation)} We have seen in Theorem \ref{thm:lowrank} 
that it is possible to achieve a running time of 
$O(\nnz(\matA)) + n \cdot \poly(k/\eps)$ for solving the low rank approximation problem
with Frobenius norm error, namely, given an $n \times n$ matrix $\matA$, finding a (factorization of a) 
rank-$k$ matrix $\tilde{\matA}_k = \matL \matU \matR$, where $\matU$ is a $k \times k$ matrix, for which 
\begin{eqnarray*}
\FNorm{\matA-\tilde{\matA}_k} \leq (1+\eps)\FNorm{\matA-\matA_k}.
\end{eqnarray*}
On the other hand, in Theorem \ref{thm:tropp} we see that it is possible to find a projection matrix
$\matZ\matZ^T$ for which
\begin{eqnarray}\label{eqn:openSpec}
\|\matA-\matZ\matZ^T\matA\|_2 \leq (1+\eps)\|\matA-\matA_k\|_2,
\end{eqnarray}
that is the error is with respect to the spectral rather than the Frobenius norm. The latter error
measure can be much stronger than Frobenius norm error. Is it possible to achieve $O(\nnz(\matA)) + n \cdot \poly(k/\eps)$
time and obtain the guarantee in (\ref{eqn:openSpec})?
\\\\
{\bf Open Question 2. (Robust Low Rank Approximation)} 
We have seen very efficient, $O(\nnz(\matA)) + n \cdot \poly(k/\eps)$ time algorithms
for low rank approximation with Frobenius norm error, that is, for finding a factorization of a rank-$k$ matrix
$\tilde{\matA}_k = \matL \matU \matR$, where $\matU$ is a $k \times k$ matrix, for which 
\begin{eqnarray*}
\FNorm{\matA-\tilde{\matA}_k} \leq (1+\eps)\FNorm{\matA-\matA_k}.
\end{eqnarray*}
As we have seen for regression, often the $\ell_1$-norm is a more robust error measure than the $\ell_2$-norm, and
so here one could ask instead for finding a factorization of a rank-$k$ matrix $\tilde{\matA}_k$ for which
\begin{eqnarray}\label{eqn:openL1}
\|\matA-\tilde{\matA}_k\|_1 \leq (1+\eps)\|\matA-\matA_k\|_1,
\end{eqnarray}
where here for an $n \times n$ matrix $\matB$, $\|\matB\|_1 = \sum_{i, j \in [n]} |\matB_{i,j}|$ is the entry-wise $1$-norm of 
$\matB$. We are not aware of any polynomial time algorithm for this problem, nor are we aware of an NP-hardness
result. Some work in this direction is achieved by Shyamalkumar and Varadarajan \cite{SV07} 
(see also the followup papers \cite{DV07,FMSW10,FL11,VX12}) who give an algorithm which
are polynomial for fixed $k,\eps$, for the weaker error measure $\|\matB\|_V = \sum_{i = 1}^n \|\matB_{i*}\|_2,$ that is,
the $V$-norm denotes the sum of Euclidean norms of rows of $\matB$, and so is more robust than the Frobenius norm,
though not as robust as the entry-wise $1$-norm. 
\\\\
{\bf Open Question 3. (Distributed Low Rank Approximation)}
In \S\ref{sec:dislra} we looked at the arbitrary partition model. Here there
are $s$ players, 
each locally holding an $n \times d$ 
matrix $\matA^t$. Letting $\matA = \sum_{t \in [s]} \matA^t$, 
we would like for each player to compute the same rank-$k$ projection
matrix $\matW\matW^T \in \mathbb{R}^{d \times d}$, for which 
$$\FNormS{\matA-\matA\matW\matW^T} \leq (1+\eps)\FNormS{\matA-\matA_k}.$$
We presented a protocol due to Kannan, Vempala, and the author
\cite{kvw14} 
which obtained $O(sdk/\eps) + \poly(sk/\eps)$ words of communication
for solving this problem. In \cite{kvw14} a lower bound of $\Omega(sdk)$ bits
of communication is also presented. Is it possible to prove an 
$\Omega(sdk/\eps)$ communication 
lower bound, which would match the leading term of the
upper bound? Some possibly related work is an $\Omega(dk/\eps)$ 
space lower bound 
for outputting such a $\matW$ given one pass over the rows of $\matA$
presented one at a time in an arbitrary order \cite{w14}, i.e., in the
streaming model of computation. In that model, this bound is
tight up to the distinction between words and bits. 
\\\\
{\bf Open Question 4. (Sketching the Schatten-$1$ Norm)}
In Section \S\ref{sec:schatten} we looked at the Schatten norms of an
$n \times n$ matrix
$\matA$. Recall that for $p \geq 1$, 
the $p$-th Schatten norm $\|\matA\|_p$ of a rank-$\rho$ matrix $\matA$ is 
defined to be 
$$\|\matA\|_p = \left (\sum_{i=1}^{\rho} \sigma_i^p \right)^{1/p},$$
where $\sigma_1 \geq \sigma_2 \geq \cdots \geq \sigma_{\rho} > 0$ 
are the singular values of $\matA$. For
$p = \infty$, $\|\matA\|_{\infty}$ is defined to be $\sigma_1$.
Some of the
most important cases of a Schatten norm is when $p \in \{1, 2, \infty\}$, 
in which case it corresponds to the nuclear, Frobenius, and spectral norm,
respectively. For constant factor approximation, 
for $p = 2$ one can sketch $\matA$ using a constant number
of dimensions, while for $p = \infty$, we saw that $\Omega(n^2)$ dimensions 
are needed. For $p = 1$, there is a lower bound of $\Omega(n^{1/2})$ for
constant factor approximation, which can be improved to $\Omega(n^{1-\gamma})$,
for an arbitrarily small constant $\gamma > 0$, 
if the sketch is of a particular form called a ``matrix sketch'' 
\cite{lnw14}. There is no non-trivial (better than $O(n^2)$) upper bound
known. What is the optimal sketching dimension for approximating
$\|\matA\|_1$ up to a constant factor?
\\\\
{\bf Acknowledgements:} 
I would like to thank Jaroslaw Blasiok, Christos Boutsidis, 
T.S. Jayram, Jelani Nelson, Shusen Wang, Haishan Ye, Peilin Zhong,  
and the anonymous
reviewer for a careful reading and giving many helpful comments. 


\bibliographystyle{plain}
\bibliography{main}
\end{document}